\documentclass[a4paper,11pt]{article} 
\usepackage[a4paper,hmargin=2.5cm,vmargin=3cm]{geometry}
\usepackage[utf8]{inputenc} 
\usepackage[english]{babel}

\usepackage{mathtools,amsthm, amsfonts, amssymb, mathrsfs, bbm, bm}
\usepackage{authblk}
\usepackage{lineno,csquotes}
\usepackage[dvipsnames]{xcolor}
\usepackage{hyperref}
\usepackage{comment}
\usepackage{upgreek}
\usepackage{amsmath}
\usepackage{dsfont}
\usepackage{mathrsfs}
\usepackage{slashed}
\usepackage{cancel}

\definecolor{BeauBlue}{rgb}{0, 0.2, .9}
\definecolor{BeauOrange}{rgb}{.8, .1, 0}
\usepackage{hyperref}
\hypersetup{
	colorlinks = true,
	linkcolor = BeauBlue,
	urlcolor = cyan,
	citecolor = BeauOrange,
}
 \usepackage[mathcal]{euscript}

\numberwithin{equation}{section}

\theoremstyle{plain}
\newtheorem{theorem}{Theorem}[section] 
\newtheorem*{theorem*}{Theorem} 
\newtheorem{proposition}[theorem]{Proposition} 
\newtheorem{corollary}[theorem]{Corollary}
\newtheorem{lemma}[theorem]{Lemma}

\theoremstyle{definition}

\newtheorem{remark}[theorem]{Remark}

\def\a{\ensuremath\alpha}
\def\b{\ensuremath\beta}
\def\d{\ensuremath\delta}

\def\t{\ensuremath\tau}
\def\th{\ensuremath\theta}
\def\vth{\ensuremath\vartheta}

\def\s{\ensuremath\sigma}
\def\e{\ensuremath\epsilon}
\def\ve{\ensuremath\varepsilon}
\def\vphi{\ensuremath\varphi}

\def\de{\ensuremath\partial}
\def\ud{\ensuremath\underline}
\def\mc{\ensuremath\mathcal}
\def\l{\ensuremath\lambda}
\def\L{\ensuremath\Lambda}
\def\ov{\ensuremath\overline}
\def\wt{\ensuremath\widetilde}
\def\mb{\ensuremath\mathbf}

\def\dg{\ensuremath\dagger}
\def\it{\textit}
\def\mf{\ensuremath\mathfrak}
\def\wt{\ensuremath\widetilde}
\def\z{\ensuremath\zeta}
\def\mbb{\ensuremath\mathbb}

\def\eset{\ensuremath\emptyset}

\newcommand{\virg}[1]{``#1''}
\newcommand{\dist}[1]{\ensuremath |\!\!|#1|\!\!|}

\DeclareMathOperator{\val}{Val}

\usepackage{tikz}
\usetikzlibrary{shapes}
\usetikzlibrary{positioning}
\tikzset{vertex/.style={circle,fill=black,inner sep=2pt},
	ctVertex/.style={diamond,fill=black,inner sep=2pt},
    CtVertex/.style={diamond,fill=black,inner sep=3pt},
	sqVertex/.style={rectangle,fill=black,inner sep=3pt},
	bigVertex/.style={circle,fill=black,inner sep=4pt},
	E/.append style={fill=white,draw},
    F/.append style={fill=gray,draw},
	probeEP/.style={circle,fill=black,draw,inner sep=2pt,
		prefix after command= {\pgfextra{\tikzset{every pin/.style = {pin edge={decorate,decoration={snake,amplitude=2pt,segment length =4pt}}}}}}
	},
	bareProbeEP/.style={rectangle,fill=black,draw,inner sep=3pt,
		prefix after command= {\pgfextra{\tikzset{every pin/.style = {pin edge={decorate,decoration={snake,amplitude=2pt,segment length =4pt}}}}}}
	},
	nuEP/.style={circle,fill=white,draw, inner sep=2pt},
	linelabel/.style={sloped,above,very near start, inner sep=1pt,execute at begin node=$\scriptstyle,execute at end node=$},
	baseline=(current  bounding  box.center),doubled/.style={double distance= 1pt,line width=1.5pt}
}

\title{Non-perturbative renormalization for lattice massive QED$_2$: the ultraviolet problem}
\author[1]{Simone Fabbri}
\author[2]{Vieri Mastropietro}
\author[1]{Bruno Renzi}
\affil[1]{SISSA, Via Bonomea 265, 34136 Trieste, Italy}
\affil[2]{University of Rome Sapienza, P.le Aldo Moro 2, 00185 Roma, Italy}

\date{\today}

\begin{document}

\maketitle

\begin{abstract} We consider a lattice regularization, preserving Ward Identities (WI) and with a Wilson term, 
of the Massive QED$_2$, describing a fermion
with mass $m$ 
and charge $\mathsf{e}$ interacting with a vector  field
with mass $M$, in the regime $m\ll M\ll a^{-1}$ ($a$ being the lattice spacing) which is the suitable one
to mimic a realistic 4d massive gauge theory like the Electroweak sector. The presence of the lattice and of the mass $m$
breaks any solvability property. In this paper we 
prove that the effective action obtained after the integration of the ultraviolet degrees of freedom is expressed by expansions which are convergent for values of the coupling 
$|\mathsf{e}|\le \mathsf{e}_0$, with $\mathsf{e}_0$ independent on $a$ and $m$, and with cut-off-independent bare parameters. By combining this result with the analysis of the infrared part in previous papers we get a complete construction of the model
and a number of properties whose analogous are expected to hold in 4d. The analysis is done by integrating out the bosons and reducing to a fermionic theory; however, with respect to the case with
momentum regularizations (which break essential features like the
WI), the resulting effective fermionic action has not a simple
form and this requires the developments of new methods to get
the necessary bounds.
\end{abstract}

\tableofcontents

\section{Introduction and main result}

The \textit{Massive Quantum Electro-Dynamics} (QED) describes Dirac fermions interacting with a 
$U(1)$ vector boson field. Its formal Euclidean action in dimension $d\ge2$ is:
\begin{equation}
\label{xx1}
\begin{split}
\hspace{-5pt}\int_{\mbb{R}^d} d^dx \left(\frac{1}{4} \sum_{\mu,\nu=0}^{d-1} F_{\mu,\nu}F_{\mu,\nu} +   \frac{M^2}{2} \sum_{\mu=0}^{d-1} A_{\mu}^2   + \sum_{\mu=0}^{d-1}\bar\psi \gamma_{\mu} \Big(\de_{\mu} + i\mathsf{e}A_{\mu}\Big)\psi + m\bar\psi \psi \right)\hspace{-2pt},
\end{split}\end{equation} 

where: $A_\mu$ is the boson (real) field with mass $M$ and $F_{\mu,\nu}= \de_{\mu}A_{\nu}-\de_{\nu}A_{\mu}$; $\psi,\bar\psi$ is the fermionic (anti-commuting) field with mass $m$ and charge $\mathsf{e}$; $\{\gamma_{\mu}\}_{\mu=0}^{d-1}$ are a representation of the Euclidean $d-$dimensional Clifford algebra. The model is expected to be well defined
in $d=4$ only as an effective theory with a finite ultraviolet cut-off, representing the maximal energy scale at which the theory is valid; it is indeed renormalizable but not asymptotically free.
The cut-off must fulfill the following physical requirements:
\begin{itemize}
\item[\textbf{a)}] the
Ward Identities (WI) must be verified, ensuring the validity of the conservation of the current coupled to $A_{\mu}$; 
\item[\textbf{b)}] the \it{longitudinal}, non-decaying term of the boson propagator (see comments after Eq. \eqref{xx2}) must not contribute to physical observables; 
\item[\textbf{c)}]
for length scales much larger than the inverse cut-off scale the correlations  must formally reduce to the continuum ones up to sub-leading corrections; 
\item[\textbf{d)}] the cut-off has to be at least exponentially high in the inverse coupling; for small $\mathsf{e}$ this ensures that the effect of the cut-off is undetectable at the energy scale of experiments.
\end{itemize}

A natural choice for the cut-off is the lattice one, obtained 
by replacing the continuum Euclidean space-time with a discrete lattice of step $a$ and
finite length $L$, and the derivative $\partial_\mu$ appearing in 
Eq. (\ref{xx1}) with a discrete derivative\footnote{Equivalently $k_\mu$ with $\frac{1}{a}\sin (a k_\mu)$ in momentum  space.}; this however violates item  \textbf{c)}, as the lattice Dirac operator acquires spurious extra zero modes, a problem known as \it{fermion doubling}. Condition \textbf{c)} can be fulfilled adding a Wilson term to the action.  In order to fulfill condition \textbf{a)} the interaction between boson and fermion has to be defined in a non-linear way, via the exponential of $A_\mu$.
As consequence of property \textbf{a)}, property \textbf{b)} follows as well, which is crucial for ensuring the renormalizability of the theory and the validity of the property \textbf{d)}.
Note that the introduction of the Wilson term breaks the chiral symmetry  which would be present at a formal level when the electron mass $m$ is vanishing,
a fact producing the chiral anomaly.

The regularization of QED, with a Wilson term and  respecting WI is called \textit{lattice Massive QED}. Other lattice definitions could be possible, but they are expected to be equivalent as far as the above requirements are true.

Massive QED is not physically realized in nature as a theory of 
elementary particles.  However it is very much related to the Standard Model, 
the theory describing all elementary particles and forces except gravity, and in particular to the neutral Weak sector, where the force is mediated by 
an $U(1)$ massive boson, the $Z^0$ particle. Indeed both theories 
are expected to exist only as effective theories; moreover 
the
Adler-Bardeen property \cite{AB69} and the
decrease of the divergence degree, 
due to the cancellation of the contribution of the longitudinal part of the propagator, are expected to hold. 
The
lattice massive QED is therefore an ideal setting to test, in a simpler model, several 
features which are believed to be true on the basis of perturbative 
arguments, that is at the level 
of series which are not convergent.
As the ratio of the boson and the electron mass is $\sim 10^{6}$, and the energies in accelerators are of the order of $100$ times the boson mass, the range of interest of lattice massive QED is
\begin{equation}
\label{energy_regime}
m \ll M \ll a^{-1}.
\end{equation}
In addition the charge is small as its physical value is, in dimensionless units, $\mathsf{e}^2 \sim \frac{1}{137}$ (fixing the wave function renormalization to $1$).

It should be recalled that
other peculiarities of the Standard Model are instead absent in massive QED, like the fact that in the Standard Model the
interaction has a \textit{chiral} nature. This fact introduces a number of difficulties, in particular the fact that the Ward Identities are verified, even at a perturbative level, only if the chiral anomalies cancel out and if the particle charges verify certain conditions, a fact requiring the Adler-Bardeen property. 

Even if the lattice massive QED in 4d is much simpler than the Standard Model, rigorous results 
in the regime  of Eq.  (\ref{energy_regime}) are lacking, and the analysis has been limited either to the region
$m\sim 1/a$ and 
any $M\ge 0$ \cite{DH92} or with a momentum regularization with $M\sim 1/a$
and any $m\ge 0$
\cite{M23}. On the other hand 
massive QED in 2d is more tractable
and at the same time 
the analogue of a  number of features holding in 4d are true also in 2d. 

In this paper we perform the construction of 
the lattice massive QED$_2$ in the regime of Eq. (\ref{energy_regime}), and we prove a number of properties expected to be true in 4d or in the Standard Model. In particular we focus on the ultraviolet properties of the theory, obtaining results which, once combined with the infrared analysis in \cite{M22,BM,BFM07,BFM09}, provide the construction of the model and the verification at a non-perturbative level of a number of properties which are expected to hold in 4d more realistic models.

\subsection{Previous results on the Massive QED$_2$}
\label{ssect_intro1}
The Massive QED$_2$ has been extensively studied in the past, see
e.g. \cite{S} for a review with more references, and we shortly summarize some of the main results. 
\vskip.2cm
\noindent{\bf Perturbative results.}
In the continuum,  
at a classical level, the formal action of QED$_2$ is invariant under the change of variables $\psi_x\to e^{i \alpha}\psi_x$ and this, by Noether theorem, implies the conservation of the current $\sum_{\mu}\partial_\mu \jmath_{\mu}=0$, with 
$\jmath_{\mu}
=\mathsf{e}\bar\psi_x\gamma_\mu\psi_x$; in the same way if the electron mass is vanishing, $m=0$, there is invariance under $\psi_x\to e^{i \alpha\gamma_5}\psi_x$ 
which 
implies the conservation of the axial current $\sum_{\mu}\partial_\mu \jmath^5_{\mu}=0$
with $\jmath^5_{\mu,x}= \mathsf{e}\bar\psi_x\gamma_\mu\gamma_5\psi_x$. 

At a quantum level, using Feynman graphs expansion 
with dimensional regularization, the following Ward Identities (WI) are found, 
to be understood as order by order identities among formal (non-convergent) power series in the coupling:
\begin{equation}\label{xx3}
\sum_{\mu=0,1} p_\mu \langle\hat{\jmath}_{\mu,p};\hat{\jmath}_{\nu,-p}\rangle = 0, \quad\quad \sum_{\nu=0,1} p_\nu \langle\hat{\jmath}_{\mu,p};\hat{\jmath}^5_{\nu,-p}\rangle = \frac{i \mathsf{e}^2}{\pi} \sum_{\nu=0,1} \ve_{\mu,\nu} p_\nu,
\end{equation}
with $\ve_{\mu,\nu}$ being the antisymmetric tensor (namely $\ve_{0,1}=-\ve_{1,0}=1, \ve_{0,0}=\ve_{1,1}=0$). The first equation is the counterpart of the conservation of the current; the second, valid in the $m=0$ case, shows
that the axial current is not conserved at a quantum level. The coefficient $\frac{1}{\pi}$ in the 
r.h.s. is called the chiral anomaly. Such coefficient is in general a series in the coupling but the above statement says that only the lowest order contributes; this is the Adler-Bardeen's non-renormalization theorem in 2d \cite{GRcc}.
The boson propagator has the form
\begin{equation}\label{xx2}
\frac{1}{|k|^2+M^2}\left(\delta_{\mu,\nu}+ \frac{k_{\mu}k_{\nu}}{| k|^2 +M^2}\right).
\end{equation}
The term $k_{\mu}k_{\nu}/(|k|^2+M^2)$, called the \it{longitudinal} term, is not decaying for $|k|\gg M$ and the naive 
degree of divergence is $D=2-n_{\psi}/2-n^A$ ($n_\psi$ and $n_A$ the number of $\psi$ and $A$ fields), corresponding to a renormalizable theory; however the WI corresponding to the current conservation implies that the contribution from the longitudinal term is vanishing and the real degree of divergence is $D=2-n-n_{\psi}/2$, if $n$ is the perturbative order, corresponding to a super-renormalizable theory. A similar 
decrease of degree of divergence is
at the basis of the perturbative renormalizability of the 4d gauge theories;
indeed in 4d one passes from
$D=4+n- 3n_\psi/2-2 n_A$ to 
$D=4- 3n_\psi/2-n^A$, that
is there is a reduction from a non-renormalizable to a renormalizable behavior.  
\vskip.2cm
\paragraph{Exact solutions.} 
There are exact solutions for 
the massive QED$_2$ in the $m=0$ case in the continuum, which have been often used as a benchmark for perturbative analyses.
In particular a solution was found in 
\cite{Som,Bro63,Hag67}; the $M=0$ case, known as Schwinger model,  was solved in \cite{Sch62}. Solutions are based either on a combination of 
Schwinger-Dyson equations and anomalous WI, obtained by a self-consistent argument previously used in \cite{Joh61} in the Thirring model, or on bosonization. Other exact solutions were found using functional
integrals methods \cite{X1,FS76}; see also \cite{X0} for related interesting  applications.

The solutions rely on formal manipulations of diverging expressions,
and the results appear to depend on a parameter introduced by the regularization. Main properties are: 
\begin{enumerate}
\item[\textbf{1)}] a vanishing wave function renormalization  $Z=0^\eta$ with $\eta>0$, see e.g. \cite[Eq. (4.35)]{Bro63};
\item[\textbf{2)}] a value of the chiral anomaly given by 
\begin{equation}
\sum_{\nu=0,1}p_\nu \langle\hat{\jmath}_{\mu,p};\hat{\jmath}^5_{\nu,-p}\rangle = \zeta \frac{i \mathsf{e}^2}{\pi} \sum_{\nu=0,1} \ve_{\mu,\nu} p_\nu,
\end{equation}
where $\zeta$ is a parameter depending on the regularization; 
it was found
$\zeta=1/2$ in \cite{Som}, $\zeta=1$ in \cite{Bro63}
and any $\zeta$ in \cite{Hag67}. The current was found to be non-conserved.
Moreover a renormalized mass of the boson is found to be $
M^2+\zeta \frac{\mathsf{e}^2}{\pi}$.
\end{enumerate}
Note that the results found with exact solutions are somewhat different
from the ones found in perturbation theory, due to the different regularizations adopted;
in particular the conservation of the current is broken and there is no
reduction of degree of divergence, as signaled by the singular $Z$; moreover
the value of the anomaly is different.

No solution is known when $m\not=0$ or with a lattice regularization.
The issue of lattice regularization is particularly relevant as anomalies 
appear in the continuum, but in a non-perturbative construction in 4d one needs a lattice cut-off. An intense activity in understanding the role of the lattice has been devoted in recent times, mostly in the context of the Schwinger model, see e.g. 
\cite{X3,X4,K1,K2,K2a,K2b}.
\vskip.2cm
\paragraph{Renormalization Group with momentum regularization.} 
Non-perturbative Renormalization Group (RG) analysis have been performed in the continuum with momentum regularization.
In \cite{M07,BFM09} it was considered a massive QED$_2$ model in the continuum, obtained starting from the formal action in Eq. \eqref{xx1}, neglecting the longitudinal term of the propagator 
(so that the power counting is that of a super-renormalizable theory)
and with a momentum fermionic cut-off $\chi_N(k)$, given by a smooth function 
vanishing for $|k|\ge 2^N$ (the boson can be assumed with or without cut-off \cite{F10}); this model has been also called \textit{reference model}, and it was used to derive
universality relations for several systems, see e.g. 
\cite{MP,GMT}.
The analysis is done by integrating out the boson field 
and obtaining a quartic interaction in the fermions; the analysis of the ultraviolet
part is done using the non-locality of the boson propagator, see also  
\cite{L87},
while the integration of the infrared scales is done by implementing cancellations due to Ward Identities at each step of the Renormalization Group, controlling the corrections due to the presence of the cut-offs, see \cite{BM}, and proving the convergence to a line of fixed points.  
In this 
continuum massive QED$_2$ model, in which the ultraviolet limit $N\to\infty$ can be taken as well as the infinite volume and zero fermionic mass 
limit, the WI has the form:
\begin{align}
&\label{xx7} \sum_{\mu=0,1}p_\mu \langle \hat{\jmath}_{\mu,p};\hat{\jmath}_{\nu,-p}\rangle = \mathsf{e} \int \frac{d^2k}{(2\pi)^2}
\left(\slashed{k}\frac{1-\chi_N(k)}{\chi_N(k)}-(\slashed{k}+\slashed{p})\frac{1-\chi_N(k+p)}{\chi_N(k+p)} \right) \langle\hat{\bar\psi}_{k+p}\hat{\psi}_{k};\hat{\jmath}_{\nu,-p}\rangle,
\end{align}
 
with $\slashed{k}\equiv \sum_{\mu=0,1}\gamma_{\mu}k_{\mu}$. The r.h.s. of Eq. \eqref{xx7} is not small, even for large $N$, and one can prove, see 
\cite{M07,BFM09},  both the non-conservation of the current and of the axial current. The form of the limit appears to be strongly sensitive to the regularization adopted.

\subsection{The lattice massive QED$_2$}
\label{sect:model}

In this paper we study a regularization of massive QED$_2$  fulfilling all the 
physical constraints \textbf{a)}-\textbf{d)} listed few lines below Eq. \eqref{xx1}. We actually consider the lattice regularization of a generalized version of Eq. \eqref{xx1}, obtained by adding an extra term $\frac{\xi}{2} \int d^2x \big(\sum_{\mu=0,1}\de_{\mu}A_{\mu,x}\big)^2$ called \it{gauge-fixing term}, with $\xi\in[0,1]$ called \it{gauge-fixing parameter}; the original theory is recovered by setting $\xi=0$.\footnote{The addition of the gauge-fixing term does not modify the model as soon as gauge invariant observables are considered, see Lemma \ref{lemma_gauge_inv} and comments thereafter.}
The model is defined on a lattice
$\L:= a\mbb{Z}^2 \cap \big(-\frac{L}{2},\frac{L}{2}\big]^2$ with spacing $a\equiv2^{-N}\in 2^{\mbb{Z}}$. We impose periodic boundary conditions, so that $\L= (a\mbb{Z}^2)/(L\mbb{Z}^2)$ and it comes equipped with the \it{torus metric} $\dist{x-y}_{L}:= \min_{n\in\mbb{Z}^2}\sqrt{\sum_{\mu=0,1}(x_{\mu}-y_{\mu}+Ln_{\mu})^2}$. 
The generating functional for correlation functions is
\begin{equation}
\label{def_functional}
\mc{W}_{\xi}(\vphi;J;B ):= -\log \frac{\int \mc{D}\psi
\int \mbb{P}_{\xi}(\mc{D}A) e^{-I(A+J,\psi)- (\vphi,\psi)- i(B,\jmath^5(\psi))}}{\int \mc{D}\psi e^{-I(0,\psi)}}, 
\end{equation}

where:

\begin{enumerate}
\item $\psi$ is the fermion field,
identified with the Grassmann algebra generated by $\big\{\!\psi^{\pm}_{x,s}\!\big\}_{x\in\L,s\in\{1,2\}}$, where, according to the Dirac notation, letting $\psi^{\pm}_x$ be the row/column vector $(\psi^{\pm}_{x,1},\psi^{\pm}_{x,2})$, we identify $\psi^-_x\equiv \psi_x$ and $\psi^+_x\equiv \bar{\psi}_x$. The symbol $\mc{D}\psi$ denotes the Grassmann-Berezin integration over the Grassmann algebra (namely the projector onto the highest degree monomial, cf. \cite[Sect. 4.1]{GM01}).
\item $A$ is the vector-boson field, to be identified with $\big\{A_{\mu,x} \big\}_{x\in\L,\mu\in\{0,1\}}\in \mbb{R}^{\L\times\{0,1\}}$, while $\mbb{P}_{\xi}(\mc{D}A)$ is the real Gaussian integration with propagator
\begin{equation}
\hspace{-3pt} g^{A,\xi}_{\mu,\nu}(x-y)\equiv \mbb{E}_{\xi}\big(A_{\mu,x}A_{\nu,y}\big) = \hspace{-2pt} \int_{\L^*} \hspace{-5pt}dk \frac{e^{-ik\cdot(x-y)}}{|\s(k)|^2 + M^2} \left( \delta_{\mu,\nu} + (1-\xi) \frac{\s_{\mu}(k) \s_{\nu}^*(k)}{\xi|\s(k)|^2 + M^2} \right), \label{eq:boson_propag}
\end{equation}

where: $\L^*:= \frac{2\pi}{L}\mbb{Z}^2\cap\big(-\frac{\pi}{a}, \frac{\pi}{a}\big]^2$ and
\begin{equation}\int_{\L^*}dk\, \cdot\equiv \frac{1}{L^2}\sum_{k\in\L^*}\,\cdot\, ;\end{equation}

$\s_{\mu}(k)\equiv i 2^N(e^{-i2^{-N}k_{\mu}}-1)$, while $\s_{\mu}^*(k)$ denotes its complex conjugate and $|\s(k)|^2\equiv \sum_{\mu=0,1}|\s_{\mu}(k)|^2$; $M>0$ is the bare boson mass; $\xi\in[0,1]$ is the gauge-fixing parameter.

\item $J\equiv \big\{J_{\mu,x} \big\}_{x\in\L,\mu\in\{0,1\}}\in \mbb{R}^{\L\times\{0,1\}}$ is a real external field, and
\begin{equation}
\begin{split}
 &I(A,\psi) =  \frac{1}{2} 2^N \sum_{\mu=0,1}\int_{\L} dx\, \psi^+_x(\gamma_{\mu}-r)e^{i2^{-N}\mathsf{e} A_{\mu,x}} \psi^-_{x+2^{-N}\hat{e}_{\mu}}+\\
& - \frac{1}{2} 2^N \sum_{\mu=0,1}\int_{\L} dx\, \psi^+_{x+2^{-N}\hat{e}_{\mu}}(\gamma_{\mu}+r)e^{-i2^{-N}\mathsf{e} A_{\mu,x}} \psi^-_{x}
+ (m_N+ 2r2^N) \int_{\L} dx \psi^+_x\psi^-_x,
\label{eq:IA}
\end{split}\end{equation}
where:
\begin{equation}\int_{\L}dx\, \cdot\equiv a^2\sum_{x\in\L}\,\cdot\, ;\end{equation}

$\hat{e}_0= (1,0)$ and $\hat{e}_1=(0,1)$; $\mathsf{e}\in\mbb{R}$ is the electric charge; $r>0$ is the Wilson mass, to be fixed once and for all (from now on we will assume $r=1$); $m_N\in\mbb{R}$ is the bare electron mass; $\gamma_0,\gamma_1$ are a representation of the Euclidean 2D Clifford algebra, i.e. they satisfy $\{\gamma_{\mu},\gamma_{\nu}\}= 2\d_{\mu,\nu}$; we chose them as \[\gamma_0:= \left(\begin{array}{cc}
1     & 0 \\
0     & -1
\end{array}\right), \qquad \gamma_1:= \left(\begin{array}{cc}
0     & -i \\
i     & 0
\end{array}\right).\]

\item $\vphi\equiv \big\{\vphi^{\pm}_{x,s}\big\}_{x\in\L,s\in\{1,2\}}$ is an external Grassmann field and: $(\vphi,\psi)\equiv \int_{\L}dx \big(\vphi^+_x\psi^-_x+ \psi^+_x\vphi^-_x\big)$.

\item $B\equiv \big\{B_{\mu,x}\big\}_{x\in\L,\mu\in\{0,1\}}$ is an external real field coupled to the chiral current, namely $(B,\jmath^5(\psi))\equiv \int_{\L}dx B_{\mu,x} \jmath^5_{\mu,x}(\psi)$, where $\jmath^5_{\mu,x}(\psi):= \mathsf{e} Z^5_N\psi^+_x\gamma_{\mu}\gamma_5\psi^-_x$, with $\gamma_5:= i\gamma_0\gamma_1$ and $Z^5_N$ a suitable continuous function of $\mathsf{e}$, such that $Z^5_N\big|_{\mathsf{e}=0}=1$ and, say, $|Z^5_N-1|\le \frac{1}{2}$.
\end{enumerate}

\begin{remark}
The Gaussian weight of the boson measure $\mbb{P}_{\xi}$ is explicit and is given by the exponential of:
\begin{equation}
\label{eq:43}
\begin{split}
& \frac{1}{4} \sum_{\mu,\nu=0,1} \int_{\L}\!dx \Big( \nabla_{\mu}A_{\nu,x}- \nabla_{\nu}A_{\mu,x}\Big)^2 +\frac{M^2}{2}\sum_{\mu=0,1}\int_{\L}\! dx A_{\mu,x}^2 + \frac{\xi}{2}\int_{\L}\! dx \Big(\sum_{\mu=0,1} \nabla_{\mu}^TA_{\mu,x}\Big)^2
\end{split}
\end{equation}

with a minus sign, where $\nabla_{\mu}A_{\nu,x}\equiv a^{-1}\big(A_{\nu,x+a\hat{e}_{\mu}}- A_{\nu,x} \big)$ and $\nabla^T_{\mu}A_{\nu,x}\equiv a^{-1}\big(A_{\nu,x-a\hat{e}_{\mu}}- A_{\nu,x} \big)$. Eq. \eqref{eq:43} is nothing but the lattice counterpart of the free-boson action in Eq. \eqref{xx1} plus the gauge-fixing term.
\end{remark}

Physical properties can be obtained by the correlations, for instance:
\begin{align}
&\label{eq:21a}\hat{\Sigma}^{\xi}_{\mu,\nu}(p):=\int_{\L} dx \; e^{-ip\cdot x} \frac{\d^2\mc{W}_{\xi}}{\d J_{\mu,x} \d J_{\nu,0}}(0;0;0),\\
&\label{eq:21aa}\hat{\Sigma}^{\xi}_{5;\mu,\nu}(p):=\int_{\L} dx \; e^{-ip\cdot x} \frac{\d^2\mc{W}_{\xi}}{\d J_{\mu,x} \d B_{\nu,0}}(0;0;0), \\
&\label{eq:21b}\hat{\Gamma}^{\xi}_{\mu}(p;k)_{s,s'} := \int_{\L^2}dx dy\; e^{i(-p\cdot x+ k\cdot y)} \frac{\d^3\mc{W}_{\xi}}{\d J_{\mu,x} \d\vphi^-_{0,s'} \d\vphi^+_{y,s}}(0;0;0),\\
&\label{eq:21c}\hat{S}^{\xi}(k)_{s,s'} := \int_{\L}dx\; e^{ik\cdot x} \frac{\d^2 \mc{W}_{\xi}}{\d \vphi^+_{x,s}\d\vphi^-_{0,s'}}(0;0;0),
\end{align}

with $\frac{\d}{\d J_{\mu,x}}:= 2^{2N} \frac{\de}{\de J_{\mu,x}}$ and similarly for the other \it{functional derivatives}, with the understanding that derivatives w.r.t. $\vphi$ are meant as anti-commuting derivatives. 
\vskip.1cm

\begin{remark}
\label{rmk:param}
The Wilson term, that is the term proportional to $r$, is introduced in order to avoid the 
fermion doubling problem. The mass $m_N$ and $Z^5_N$ are \textit{bare parameters}
to be fixed in order to ensure the renormalization conditions. The Wilson term and the corresponding breaking of the chiral symmetry have the effect that
the massless case, $m=0$, requires a non-vanishing $m_N$, and that $Z^5_N$ is not equal to $1$.
In contrast with other related QFT in 2d, here we do not introduce in the model a bare wave function renormalization $Z_N$,
which in several cases is singular in the limit of removed cut-off, as in the Thirring model; our aim is indeed to prove that one has finite observables uniformly in $N$ by choosing $Z_N=1$.
With $L,a, M,m_N$ finite, Eq. (\ref{def_functional}) is well defined (even if the boson field is unbounded) for $\mathsf{e}$ small enough, as discussed
in the following section. 
\end{remark}

\subsection{Integration of the boson field}

Our first step towards the construction of the theory is the integration of the boson field, which yields an effective fermion interaction.
We consider the interacting part of the action,
\begin{equation}\begin{split}
V(A+J,\psi)&:= I(A+J,\psi)- I(0,\psi)\\  &=\frac{1}{2}\sum_{\mu=0,1}\int dx \psi^+_x(\gamma_{\mu}-r)\psi^-_{x+2^{-N}\hat{e}_{\mu}} G_{x,\mu,+}(A+J)+\\
&-  \frac{1}{2}\sum_{\mu=0,1}\int dx \psi^+_{x+2^{-N}\hat{e}_{\mu}}(\gamma_{\mu}+r)\psi^-_x G_{x,\mu,-}(A+J),\label{eq_V&G}
\end{split}\end{equation}
where $G_{x,\mu,\pm}(A):= 2^N\left(e^{\pm i2^{-N}\mathsf{e} A_{\mu,x}} -1 \right)$.
We use the the following conventions.
\begin{enumerate}
    \item Bosonic labels: $b\equiv (x,\mu,\e)\in \L_B\equiv \Lambda\times\{0,1\}\times\{\pm\}$. Besides, $\\ \int_{\L_B} db\equiv \sum_{\mu\in\{0,1\}}\sum_{\e\in\{\pm\}}\int_{\L} dx$. 
    \item Fermionic labels: $\eta\equiv (x,s)\in \L_F\equiv \Lambda\times\{1,2\}$. Besides, $\int_{\L_F} d\eta\equiv \sum_{s\in\{1,2\}}\int_{\L} dx$.
    \item If $\ud b\equiv(b_1,\dots, b_p)$ we let $\int_{\L_B^p}d \ud b:=\int_{\L_B} db_1\dots \int_{\L_B} db_p$, and similarly for fermionic labels. Whenever the value of $p$ and the nature of the label (bosonic or fermionic) will be clear from the context, we will simply write $\int d\ud b$ etc..
\end{enumerate}

With this understanding, we can rewrite the action as 
\begin{equation}
\begin{split}
\label{eq_Mod_action}
&V(A+J,\psi) +  (\vphi,\psi)+ i(B,\jmath^5) = \\ & \int_{\L_B} db\, O_b(\psi) G_b(A+J) + \int_{\L_F} d\eta (\vphi^+_{\eta}\psi^-_{\eta} + \psi^+_{\eta}\vphi^-_{\eta}) + \underbrace{\int_{\L_B} db B_b O_{5;b}(\psi)}_{=: (B,O_5)}, 
\end{split}
\end{equation}

where $O_b(\psi)= \int d\eta d\eta' \psi^+_{\eta} c_b(\eta,\eta') \psi^-_{\eta'}$, and, if $b=(y,\mu,\e), \eta=(x,s), \eta'=(x',s')$:

\begin{equation}
\label{def_bare_vertex}
c_b(\eta,\eta')= \left\{\begin{array}{cc}
\tfrac{1}{2}\delta(x-y)\delta(x'-y-2^{-N}\hat{e}_{\mu})(\gamma_{\mu}-r)_{s,s'}     & \e=+ \\
-\tfrac{1}{2}\delta(x'-y)\delta(x-y-2^{-N}\hat{e}_{\mu})(\gamma_{\mu}+ r)_{s,s'}     & \e=-,
\end{array}\right. 
\end{equation}

where hereafter $\d(x):= \sum_{n\in\mbb{Z}^2} 2^{2N}\d_{x_0,n_0L}\d_{x_1,n_1L}$.\footnote{For $\eta=(x,s)$ and $\eta'=(x',s')$ in $\L_F$ we also let $\d(\eta-\eta')\equiv \d_{s,s'}\d(x-x')$, while for $b=(y,\mu,\e)\in\L_B$ and $b'=(y',\mu',\e')$ in $\L_B$, $\d(b-b')\equiv\d_{\mu,\mu'}\d_{\e,\e'}\d(x-y)$.} More explicitly:
\[\left\{\begin{array}{cc}
O_{x,\mu,+}(\psi)&= \frac{1}{2}\psi^+_x(\gamma_{\mu}-r)\psi^-_{x+2^{-N}\hat{e}_{\mu}}  \\
O_{x,\mu,-}(\psi)&= -\frac{1}{2}\psi^+_{x+2^{-N}\hat{e}_{\mu}}(\gamma_{\mu}+r)\psi^-_x.
\end{array}\right. \]

Moreover, in the r.h.s. of Eq.  \eqref{eq_Mod_action}, in order for the coupling with the chiral and the vector current to possess a similar, standard form (to be also compatible with the symmetries of the model, cf. Appendix \ref{app:symmetries}), we have set $B_{x,\mu,\e}:= i\e \mathsf{e} B_{\mu,x}$ and $O_{5;x,\mu,\e}(\psi):=\int\!\! d\eta d\eta' \psi^+_{\eta} c^5_{x,\mu,\e}(\eta,\eta') \psi^-_{\eta'}$, where 
\begin{equation}
\label{def_chiral_vertex}
c^5_{y,\mu,\e}\big((x,s),(x',s')\big):= \tfrac{1}{2} \e Z^5_N \delta(x-y)\delta(x'-y) (\gamma_{\mu}\gamma_5)_{s,s'}.
\end{equation}

\medskip

The following identity holds:
\[ \int \mbb{P}_{\xi}(\mc{D}A) e^{-V(A+J,\psi)} = e^{-\mc{V}(\psi;G(J))},  \]

where $\mc{V}$ is a polynomial\footnote{Due to the finiteness of the fermionic Grassmann algebra, for any $L,N>0$ fixed the sum in the r.h.s. of Eq. \eqref{eq:2} has finitely many terms.} in $\{\psi^{\pm}_{\eta}\}_{\eta\in\L_F}$ and $\{G_b(J)\}_{b\in\L_B}$, given by
\begin{equation}
\label{eq:2}
\begin{split} 
\mc{V}(\psi;G(J))&= 2^N\left(\mu_{N,\l}-1 \right)\int db O_b(\psi)+ \mu_{N,\l} \int_{\L_B}\, db O_b(\psi) G_b(J)+ \\
&+ \sum_{n\ge 2}\sum_{0\le m\le n} 2^{N(2-n-m)} \int_{\L_B^n} d\ud b \int_{\L_B^m} d\ud b' w_{n,m}(\ud{b};\ud{b}')  O^n_{\ud b}(\psi) G^m_{\ud b'}(J)\, ,
\end{split}
\end{equation}

with $\mu_{N,\l}\equiv e^{-\frac{\l}{2}2^{-2N}g^{A,\xi}_{\mu,\mu}(0)}$, $\l\equiv\mathsf{e}^2$; given $\ud b=(b_1,\dots,b_n)$ we let $O^n_{\ud b}(\psi):=\prod_{j=1}^n O_{b_j}(\psi)$, similarly $G^m_{\ud b'}(J):=\prod_{i=1}^m G_{b'_i}(J)$. In complete analogy with \cite[Sect. 2]{M23}, as a straightforward consequence of the Battle-Brydges-Federbush formula \cite[Thm. 3.1]{Bry84}, the kernels $w_{n,m}$ can be expressed via the following formulas:
\begin{align}
&\label{eq:4a}w_{n,0}(\ud b)= \frac{1}{n!}\sum_{T\in \mc{T}(\{1,\dots,n\})} \prod_{\{i,j\}\in T} \l g^{A,\xi}_{b_i,b_j} \times \int dP_T(s) e^{-U(\ud{b};s)}\\
&\label{eq:4b}w_{n,m}(\ud b;\ud b')= \frac{1}{ m!} \sum_{\substack{i_1,\dots,i_m\in\{1,\dots,n\}\\ \text{distinct}}} \prod_{j=1}^m \d(b'_j-b_{i_j})\times w_{n,0}(\ud b), \;\; m\ge1,
\end{align}

where:
\begin{itemize}
    \item $\mc{T}(\{1,\dots,n\})$ the set of all the spanning trees on the set $\{1,\dots,n\}$ where, given $X \subset \mathbb{N}$, a \emph{spanning tree on $X$} is a set $T \subset X^2$ such that the graph  $\mc{G}=(X,T)$ is a tree (connected, acyclic). If $|X|=1$ then $\mc{T}(X)=\emptyset$ and in such case $w_{1,0}=\mu_{N,\l};$
    \item $dP_T(s)$ is a probability measure supported on $[0,1]^{\mc{P}(\{1,\dots,n\})}$, $\mc{P}(\{1,\dots,n\})$ being the set of pairs on $\{1,\dots,n\}$; 
        \item $U(\ud{b};s)\equiv \frac{1}{2}\l2^{-2N}(\sum_{j=1}^n g^{A,\xi}_{b_j,b_j} + 2\sum_{1\le i< j\le n}s_{\{i,j\}} g^{A,\xi}_{b_i,b_j})$, and this is non-negative on the support of $P_T$ for $\l\ge0$.
\end{itemize}

It is an immediate consequence of the finiteness of both $\|g^{A,\xi}_{\mu,\nu}\|_1:= \int_{\L}dx |g^{A,\xi}_{\mu,\nu}(x)|$ and $\|g^{A,\xi}_{\mu,\nu}\|_{\infty}:= \sup_{x\in\L} |g^{A,\xi}_{\mu,\nu}(x)|$ and of determinant bounds for fermionic expectations, the finiteness of the r.h.s. of Eq. 
\eqref{def_functional} for any $\xi\in[0,1]$ and $\mathsf{e}\in\mbb{R}$.

The resulting fermionic theory has an interaction involving Grassmann monomials of all possible degree, multiplied by suitable non-local kernels;
for every fixed tree $T$ in the r.h.s. of Eq. \eqref{eq:4a} and given $\ud b=(b_1,\dots,b_n)$, we can graphically represent the quantity

\[ 2^{N(2-n)} O^n_{\ud{b}}(\psi) \prod_{\{i,j\}\in T}\l g^A_{b_i,b_j}\]
 
appearing in the expression for $\mc{V}(\psi;0)$ as in Fig. \ref{fig_kernel}.(i); there  the nodes of the graph represent the points $\{1,\dots,n\}$: to each of them we assign a coordinate $b_j \in \Lambda_B$ and a 
dimensional factor $2^{N(1-d_j)}$, with $d_j$ the number of lines of the tree incident in the node. The couple of half solid lines incident in the node $j$ represents the fermionic bilinear $O_{b_j}(\psi)$; any wiggly line is a bond of the tree $T$, as in Fig. \ref{fig_tree}.(ii).

\begin{figure}[h]
    \centering
\includegraphics[scale=0.85]{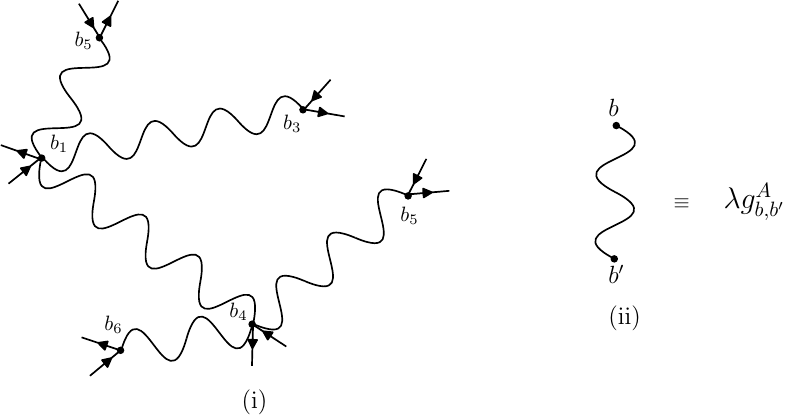}
    \caption{(i) graphical representation of the contribution to the term $\mc{O}(\psi^{12})$ of the fermionic interaction, obtained by selecting in $w_{6,0}$ the tree with bonds $\{1,2\},\{1,3\},\{1,4\},\{4,5\},\{4,6\}$; (ii) graphical representation for the boson propagator.}
    \label{fig_kernel}
\end{figure}

\subsection{Ward-Identities and $\xi$- independence}
\label{ssect_xi}

Ward identities hold, generated by the following relation (to be called \it{local phase invariance}):
\begin{equation}
\label{eq_gauge_inv}
\mc{W}_{\xi}(\vphi;J;B)= \mc{W}_{\xi}(e^{i\mathsf{e}\a}\vphi; J+ \nabla\a; B), \quad  \forall \a\in \mbb{R}^{\L},
\vspace{5pt}\end{equation}
with the understanding that $(e^{i\mathsf{e}\a}\vphi)_{x,s}^{\pm}\equiv e^{\pm i\mathsf{e}\a(x)} \vphi^{\pm}_{x,s}$,  $(A+ \nabla \a)_{\mu,x}\equiv A_{\mu,x}+ \nabla_{\mu}\a(x)$, and $\nabla_{\mu}$ the discrete-derivative operator as after Eq. \eqref{eq:43}. 
Note that $\mc{W}_{\xi}$ is well defined for finite values of the cut-off and any $\xi\in[0,1]$. Eq. \eqref{eq_gauge_inv} implies the \it{Ward identities}: 
\begin{align}
&\label{eq:WI5a}\sum_{\mu=0,1} \s_{\mu}(p) \hat{\Sigma}^{\xi}_{\mu,\nu}(p)= \sum_{\mu=0,1} \s_{\mu}(p) \hat{\Sigma}^{\xi}_{5;\mu,\nu}(p)=0, \qquad  \forall \; p\in \L^*,\\
&\label{eq:WI5b} \sum_{\mu=0,1} \s_{\mu}(p) \hat{\Gamma}^{\xi}_{\mu}(p;k)= \mathsf{e}\big(\hat{S}^{\xi}(k)- \hat{S}^{\xi}(k-p)\big), \qquad \forall p,k\in \L^*,
\end{align}

where $\hat{\Sigma}^{\xi}_{\mu,\nu},\hat{\Sigma}^{\xi}_{5;\mu,\nu},\hat{\Gamma}^{\xi}_{\mu}(p;k)$ and $\hat{S}^{\xi}(k)$ are as in Eqs. \eqref{eq:21a}-\eqref{eq:21c}. As a consequence, see App. \ref{app_gauge_inv}, one has the $\xi$-independence for the generating functional.
\begin{lemma}[$\xi-$independence]
\label{lemma_gauge_inv}
$
W_{\xi}(0;J;B)$ is constant w.r.t. $\xi\in [0,1]$.
\end{lemma}

Notice that the $\xi$-independence holds only for $\vphi=0$; more generally it holds only for correlations of \it{gauge-invariant} observables, namely quantities $F(A,\psi)$ such that $F(A+ \nabla\a, e^{i\mathsf{e}\a}\psi)= F(A,\psi)$ for every $\a\in\mbb{R}^{\L}$. Note in particular that $V(A,\psi)+ i (B,\jmath^5)$ is gauge-invariant.

Lemma \ref{lemma_gauge_inv} allows us to take $\xi=1$, so that the longitudinal term of the boson propagator is not contributing. The identity with the original model with $\xi=0$ holds only for the correlations of currents (correlations of fermions could be included if they respect the above invariance). From now on we will think of $\xi$ as fixed once and for all at the value $\xi=1$, and we will denote the boson propagator $g^{A,1}$ simply by $g^A$. The latter verifies the following bound:
\begin{equation}
\label{eq_bound_boson}
|g^A_{\mu,\nu}(x-y)|\le  \d_{\mu,\nu} \frac{K_p}{(N-h^*_M) + \big(2^{h^*_M}\dist{x-y}_L\big)^{\frac{1}{p}}} \exp\Big\{-\kappa_1 \sqrt{2^{h^*_M}\dist{x-y}_L}\, \Big\}
\end{equation}

with $h^*_M:= \lfloor \log_2 M \rfloor$, for every $1 \leq p <\infty$ and suitable constants $K_p,\kappa_1$; as a consequence  
\begin{equation}
\label{eq:bound:boson:b}
\|g^A_{\mu,\nu}\|_1=  \mc{O}(2^{-2h^*_M}), \qquad  
\|g^A_{\mu,\nu}\|_\infty=\mc{O}(N-h^*_M), 
\end{equation}

which has to be contrasted with the bounds in the $\xi=0$ case,  
$\|g^{A,0}_{\mu,\nu}\|_1= \mc{O}\big(2^{-2h^*_M}(N-h^*_M)\big)$ and
$\|g^{A,0}_{\mu,\nu}\|_\infty= \mc{O}( 2^{2N-2h^*_M})$. The boundedness of the $L^1$ norm will play a crucial role in the following analysis.
 
Note also that it is convenient to rewrite $w_{n,0}$ as a dominant term plus a remainder, the latter carrying a dimensional improvement in $N$. To do so, in Eq. \eqref{eq:4b} we rewrite $e^{-U(\ud{b};s)}= e^{-U(\ud{b};0)}+ \left(e^{-U(\ud{b};s)}-e^{-U(\ud{b};0)} \right)$, which induces the splitting 
\begin{equation}
\label{eq_16a}w_{n,0}(\ud b)= v_n(\ud b) + 2^{2h^*_M-2N}\tilde{v}_n(\ud b),  \qquad 
v_n(\ud b):= \frac{\mu_{N,\l}^n}{n!}   \sum_{T(\{1,\dots,n\})} \prod_{\{i,j\}\in T} \l g^A_{b_i,b_j},
\end{equation}

and one can easily show, using Eq. \eqref{eq_bound_boson}, that there exists a constant $C>0$ such that
\begin{equation}
\|v_n\|_1^w \le C^n \big(2^{-2h^*_M}\l\big)^{n-1}, \;\;\|\tilde{v}_n\|_1^w \le C^n \big(2^{-2h^*_M}\l\big)^n,    
\end{equation}

where, for any function $W$ (to be called \it{kernel} hereafter) over $\L_F^{2q}\times\L_B^m$ and $p>0$, we let
\begin{equation}
\label{def_Lp_norm}
\|W\|_p^{w}:= \left(\frac{1}{L^2} \int_{\L_F^{2q}} d\ud \eta\int_{\L_B^{m}} d\ud b \, \exp\Big( \tfrac{p}{2} \kappa  \sqrt{2^{h^*_M} \d^D(\ud{\eta},\ud{b})}\Big) |W(\ud{\eta};\ud{b})|^p \right)^{\frac{1}{p}} \end{equation}
with: $\d^D(\ud{\eta},\ud{b})$ the diameter distance between the $2q+m$ space coordinates of $\\ \eta_1,\dots,\eta_{2q},b_1,\dots,b_m$; $\kappa:= \min\big\{\kappa_1,\kappa_2,\kappa_3\big\}$, where $\kappa_1,\kappa_2,\kappa_3$ are suitable constants introduced in Eqs. \eqref{eq_bound_boson}, \eqref{eq_bound_propagator}, \eqref{boson:gradient} respectively.

\subsection{The ultraviolet integration}
\label{ssect_main_result}

Let
$P^{(\le N)}(\mc{D}\psi)$ be the Grassmann Gaussian integration with propagator:
\begin{equation}
\label{def_fermion_prop}
g^{(\le N)}_{s,s'}(x-y)\equiv \mc{E}_{(\le N)}\big(\psi^-_{x,s} \psi^+_{y,s'}\big)= \int_{\L^*} dk \; e^{-ik\cdot(x-y)} \Big(-i \slashed{s}(k)+ m_N + M_N(k)\Big)^{-1}_{s,s'},
\end{equation}
\begin{sloppypar} where $\slashed{s}(k):= \sum_{\mu=0,1} \gamma_{\mu} s_{\mu}(k)$, with $s_{\mu}(k):= 2^N\sin (2^{-N}k_{\mu})$,  and  $M_N(k):=  2r\sum_{\mu=0,1} 2^N \sin^2(2^{-N-1}k_{\mu})$. \end{sloppypar}

We consider $|m_N|\le M \le 2^N$ and,
recalling that $h^*_M= \lfloor\log_2M\rfloor$, we decompose the fermion propagator as $g^{(\le N)}(x-y)= g^{(\le h^*_M)}(x-y)+ g^{(h^*_M,N]}(x-y)$, with
\begin{equation} 
\label{eq:7}
\hat{g}^{(\le h^*_M)}(k):= \frac{\chi(2^{-h^*_M}|k|)}{-i\slashed{s}(k)+ m_N+ M_N(k)}, \qquad \hat{g}^{(h^*_M,N]}(k):= \frac{1-\chi(2^{-h^*_M}|k|)}{-i\slashed{s}(k)+ m_N+ M_N(k)}, 
\end{equation}

where $\chi$ is a non-negative, non-increasing Gevrey-2 function\footnote{ $f\in\mathscr{C}^{\infty}(\mbb{R})$ is said to be of Gevrey-$s$ class if $\|f^{(n)}\|_{L^{\infty}}\le C_1 C_2^n n!^s$ for every $n\ge0$ and some $C_1,C_2>0$.}, such that $\chi\equiv 1$ on $[0,\tfrac{1}{2}]$ and $\chi\equiv 0$ on $[1,\infty)$. By the addition principle for Gaussian integrations \cite[Eq. (4.21)]{GM01}, we can rewrite:
\begin{equation}\label{eq:IRinteg}
\mc{W}(\vphi,J,B)= -\log \int P^{(\le h^*_M)}(\mc{D}\psi) e^{-\mc{W}^{(u.v.)}(\psi;\vphi;J;B)}
\end{equation}
with
\begin{equation}
\label{def_functional_UV}\hspace{-5pt}\mc{W}^{(u.v.)}(\psi;\vphi;J;B):=-\log \int \hspace{-3pt}P^{(h^*_M,N]}(\mc{D}\z) \int \hspace{-3pt} \mbb{P}(\mc{D}A) e^{-V(A+J,\z + \psi) - (\vphi,\z + \psi)- i(B,\jmath^5(\z+ \psi))}.  
\end{equation}

It is an immediate consequence of the above definitions that $\mc{W}^{(u.v.)}(\psi;\vphi;J;B)$ can be written in the form:
\begin{equation}
\label{eq_main_potential}
\begin{split}
&L^2\mc{F}^{(u.v.)} + \sum_{\substack{q,q',\tilde{q},\tilde{q}',p,p'\ge0\\ q+q'+\tilde{q}+\tilde{q}'+p+p'\ge1\\ q+\tilde{q}= q'+\tilde{q}'}} \int_{\L_F^q} d\ud{\eta} \int_{\L_F^{q'}}d\ud{\eta}' \int_{\L_F^{\tilde{q}}}d\ud{\tilde{\eta}} \int_{\L_F^{\tilde{q}'}}d\ud{\tilde{\eta}}' \int_{\L_B^p} d\ud{b} \int_{\L_B^{p'}} d\ud{b}' \times\\
& \times W^{(u.v.)}_{q,q';\tilde{q},\tilde{q}';p;p'}(\ud{\eta},\ud{\eta}';\ud{\tilde{\eta}}, \ud{\tilde{\eta}}'; \ud{b};\ud{b}') \prod_{k=1}^q \psi^+_{\eta_k} \prod_{k=1}^{q'} \psi^-_{\eta'_k} \prod_{k=1}^{\tilde{q}} \vphi^+_{\tilde{\eta}_{k}} \prod_{k=1}^{\tilde{q}'}\vphi^-_{\tilde{\eta}'_{k}} \prod_{k=1}^p G_{b_k}(J) \prod_{k=1}^{p'} B_{b'_k},
\end{split}\end{equation}

for suitable complex kernels $W^{(u.v.)}_{q,q';\tilde{q},\tilde{q}';p;p'}$. Our main result is the following.

\clearpage
\begin{theorem}
\label{thm_UV}
There exist $\l_{\star},C_{\star},N_{\star}>0$ such that for any $0\le |m_N|<2^{h^*_M}$, $N\ge h^*_M+N_{\star}$, $L\ge 2^{-h^*_M}$, $\xi=1$,  and $0\le\l \le \l_{\star} 2^{2h^*_M}$, the following are true.
\begin{enumerate}
\item\label{itt:1} The kernels of $\mc{W}^{(u.v.)}$ in Eq. \eqref{eq_main_potential}  satisfy the following bounds. If $q+q'+\tilde{q}+\tilde{q}'>0$ or $p+p'>2$,
\begin{equation}
\label{eq_50}
\hspace{-11pt}\big\|W^{(u.v.)}_{q,q';\tilde{q},\tilde{q}';p;p'}\!\big\|_1^{w}\hspace{-3pt}\le 2^{h^*_M D_{sc}}\hspace{-2pt}\times \hspace{-2pt}\left\{\begin{array}{cc}
\hspace{-10pt}C_{\star}\big(2^{-2h^*_M}\l\big)^{\max\{1,\frac{2}{5}q\}}, & \hspace{-5pt}\!\!\!\!\tilde{q},\tilde{q}',p,p'=0\\
\hspace{-5pt} C_{\star}^{1+p+p'+ \tilde{q}+\tilde{q}'} \big(2^{-2h^*_M}\l\big)^{\frac{1}{8} \max\left\{q+q'+\tilde{q}+\tilde{q}'-2, 0\right\}}, & \!\!\text{otherwise},
\end{array}\right.
\end{equation}
with $D_{sc} \equiv D_{sc}(q,q';\tilde{q},\tilde{q}';p;p'):= 2- \frac{1}{2}(q+q')-\frac{3}{2}(\tilde{q}+\tilde{q}')-p-p'$. Moreover:
\begin{flalign}
&\label{eq_51a}|\mc{F}^{(u.v.)}|\le C_{\star} 2^{-2h^*_M}\l 2^{2N}; \qquad  W^{(u.v.)}_{0,0;0,0;0;1}=0; \qquad  \big\| W^{(u.v.)}_{0,0;0,0;1;0}\big\|_1^w\le C_{\star}2^N; \\
&\label{eq_51b} \big\|g^A\!*\!W^{(u.v.)}_{0,0;0,0;2;0}\big\|_1^{w}, \big\|g^A\!*\!W^{(u.v.)}_{0,0;0,0;1;1}\big\|_1^{w} \!\le C_{\star} 2^{-2h^*_M}; \quad  \big\|W^{(u.v.)}_{0,0;0,0;0;2}\big\|_1^w\! \le\! C_{\star}(\!N\!-h^*_M\!).
\end{flalign}
\item\label{itt:2} 
Letting $\hat{\Sigma}^{(u.v.)}_{\mu,\nu}$, $\hat{\Sigma}^{(u.v.)}_{5;\mu,\nu}$ and $\hat{S}^{(u.v.)}$ be defined as the r.h.s. of Eqs. \eqref{eq:21a}, \eqref{eq:21aa} and \eqref{eq:21c} respectively, with $\mc{W}_{\xi}$ in the r.h.s. replaced by $\mc{W}^{(u.v.)}$, the following relations hold:
\begin{align}
&\label{eq:23a} \big| \hat{S}^{(u.v.)}_{s,s'}(k)-  \hat{g}^{(h^*_M,N]}_{s,s'}(k) \big|\le \frac{C_{\star} 2^{-h^*_M}\l}{2^{2h^*_M}+|k|^2}, \\
&\label{eq:23b} \big| \hat{\Sigma}^{(u.v.)}_{\mu,\nu}(p)-  \l\hat{\mf{B}}^{(u.v.)}_{\mu,\nu}(p)  \big|, \, \big| \hat{\Sigma}^{(u.v.)}_{5;\mu,\nu}(p)-  \l\hat{\mf{B}}^{(u.v.)}_{5;\mu,\nu}(p)  \big|\le C_{\star} 2^{-h^*_M} \l^{\frac{3}{2}},
\end{align}
for every $p,k\in\L^*$ such that $|p|\le 2^{h^*_M}$, where $\hat{\mf{B}}^{(u.v.)}_{\mu,\nu}$ and $\hat{\mf{B}}^{(u.v.)}_{5;\mu,\nu}$ are suitable functions over $\L^*$, such that for $p \neq 0$
\begin{align}
&\lim_{M\to0^+}\lim_{m_N\to0}\lim_{L\to\infty} \hat{\mf{B}}^{(u.v.)}_{\mu,\nu}(p)= \frac{1}{\pi} \Big(\d_{\mu,\nu}- \frac{p_{\mu}p_{\nu}}{|p|^2} \Big) + R_{\mu,\nu}(p),\label{eq_24a}\\
&\lim_{M\to0^+}\lim_{m_N\to0}\lim_{L\to\infty} \hat{\mf{B}}^{(u.v.)}_{5;\mu,\nu}(p)= \frac{i}{\pi} \Big(\ve_{\nu,\mu}- \sum_{\a=0,1}\ve_{\nu,\a} \frac{p_{\mu}p_{\a}}{|p|^2} \Big) + R_{5;\mu,\nu}(p),
\label{eq_24b}
\end{align}

with $|R_{\mu,\nu}(p)|, |R_{5;\mu,\nu}(p)|\le C_{\star} (2^{-N}|p|)^{\frac{1}{2}}$.
\end{enumerate}
\end{theorem}

\vskip.1cm
\begin{remark}$\,$
\begin{itemize} 
\item[\textbf{1)}]
The finiteness of the kernels obtained after the ultraviolet integration is established with
a wave function renormalization independent of the lattice cut-off (cf. Remark \ref{rmk:param}). 
This is due to the $\xi$-independence implied by our lattice regularization, 
which is not true with other regularizations.
\item[\textbf{2)}] The result of the integration of the ultraviolet part is an 
effective
interaction  expressed by monomials which are bounded 
by the correct dimensional factor $2^{h^*_M D_{sc}}$ times constants independent on $N$, see Eq.  \eqref{eq_51a}. Our techniques can also be used for showing the existence of the limit $N\to\infty$ of the Fourier transform of the kernels  (see Appendix \ref{app:limit}).
\item[\textbf{3)}] The behavior of the two-point Schwinger function $\hat{S}^{(u.v.)}(k)$ for large momenta is the same as the non-interacting one, up to sub-leading corrections; this is a difference with the Thirring model, where the behavior in the same regime is different from the non-interacting case, modified by the presence of an anomalous critical exponent.

For small momenta instead a critical exponent for $\hat{S}^{(u.v.)}(k)$ is found, see Theorem \ref{thm:anomaly}.
\item[\textbf{4)}] The current and axial current correlation functions, $\hat\Sigma^{(u.v.)}_{\mu,\nu}$ and $\hat\Sigma^{(u.v.)}_{5;\mu,\nu}$ respectively, are close to their non-interacting counterparts up to corrections. Note that the lattice regularization ensures the same expressions
found with dimensional regularizations; with momentum regularizations
different values would be obtained, for instance the non-interacting current correlation function would be given by $\frac{1}{2\pi} \big(\d_{\mu,\nu}- 2\frac{p_{\mu}p_{\nu}}{|p|^2}\big)$,
violating the conservation of the current.
\end{itemize}
\end{remark}

\subsection{Sketch of the proof}
\label{subsec:proofsketch}

As the bosons are integrated out, the power counting is the one of a fermionic theory in 2d, with scaling exponent $D_{sc}$ as in Theorem \ref{thm_UV}.
In the resulting effective fermionic theory, Eq. \eqref{eq:2}, there are marginal terms (i.e. with $D_{sc}=0$) either quartic in $\psi$ or quadratic in $\psi$ and linear in $G(J)$, 
and relevant (i.e. with $D_{sc}>0$) if quadratic in $\psi$; all the others are irrelevant, namely with $D_{sc}<0$.

By decomposing the fermionic integration (cf.  Eq. \eqref{def_functional_UV}) in a multiscale fashion  we obtain a series expansion for the kernels  in terms of the relevant and marginal ones only. The expansion turns out to be convergent as long as the relevant and marginal kernels verify suitable bounds that are uniform w.r.t. $N$: this is the content of Proposition \ref{prop:SDB}. The main problem is to show that the relevant and marginal kernels are indeed bounded (Theorem \ref{thm:IB}). To this end the idea is to perform, at each scale, a number of manipulations within the fermionic expectations in order to \emph{extract} some fermionic or bosonic internal lines (i.e. propagators); this is done in order to gain from the non-locality of the boson propagator, $\hat{g}^A(k)\sim |k|^{-2}$, which decays faster than the fermionic one, $\hat{g}^{(\le N)}(k)\sim |k|^{-1}$, for large $k$. In doing that one has to be careful and avoid losing the good combinatorial properties of the fermionic determinants. These manipulations lead to kernels identities for the relevant and marginal ones, 
which can be seen by differentiating a suitable generalized action with respect to new auxiliary sources.

However, while in the case with momentum regularization one introduces only one source, which is coupled to the current \cite{L87,BFM09,F10}, in the present case the lattice interaction produces complicate kernels (cf. Eq. \eqref{eq:2} and Fig. \ref{fig_kernel}) so that one is forced to add source terms as in Eq. \eqref{def_potential} below, by introducing an unbounded number of auxiliary variables coupled to new, complicated fermionic terms which hardly resemble standard observables.
\begin{equation}\label{def_potential}
\hspace{-30pt}\begin{split} & \mc{ V}(\psi;G;\dot{G};\ddot{G}):= \\ & 2^N\left(\mu_{N,\l} -1 \right)\int_{\L_B} db\, O_b(\psi)+\sum_{n\ge2} 2^{N(2-n)} \int_{\L_B^n} d \ud b\, w_{n,0}(\ud b)  O^n_{\ud b}(\psi) \\&+ \sum_{n\ge1} n^3 2^{N(1-n)} \int_{\L_B^n} d\ud b\, v_n(\ud b) G^{(n)}_{b_1}  O^n_{\ud b}(\psi)+ \sum_{n\ge2} n^3 2^{N(2-n)} \int_{\L_B^n} d \ud b\, v_n(\ud b) \dot{G}^{(n)}_{b_1} O^{n-1}_{\ud b_*}(\psi)\\&+  \sum_{n\ge2} n 2^{N(1-n)} \int_{\L_B^n} d\ud b\, \tilde{v}_n(\ud b) \ddot{G}^{(n)}_{b_1} O^{n-1}_{\ud b_*}(\psi),
\end{split}\end{equation}

where $G\equiv \{G^{(n)}\}_{n\ge1}$, $\dot{G}\equiv \{\dot{G}^{(n)}\}_{n\ge2}$ and $\ddot{G}\equiv \{\ddot{G}^{(n)}\}_{n\ge2}$ ($n$ being called \it{degree} index) are the aforementioned auxiliary, complex fields\footnote{$G\equiv \{G^{(n)}\}_{n\ge1}$ should not be confused with the function $G(J)$ in Eq.  \eqref{eq:2}.}; moreover given $\ud b=(b_1,\dots,b_n)$, we let $\ud b_*:=(b_2,\dots, b_n)$, $O^{n}_{\ud b}:=\prod_{j=1}^n O_{b_j}$ and $O^{n-1}_{\ud b_*}:=\prod_{j=2}^n O_{b_j}$.
To motivate such definition (see Remark \ref{rmk_dotG} below for details), we notice that by selecting one of the fermionic lines from the kernel $v_n$ (see Fig. \ref{fig:extraction}), there are two distinct situations: when the vertex of the extracted fermionic line is a leaf of the tree or not. \begin{figure}[h]   \centering \includegraphics[width=0.92\textwidth]{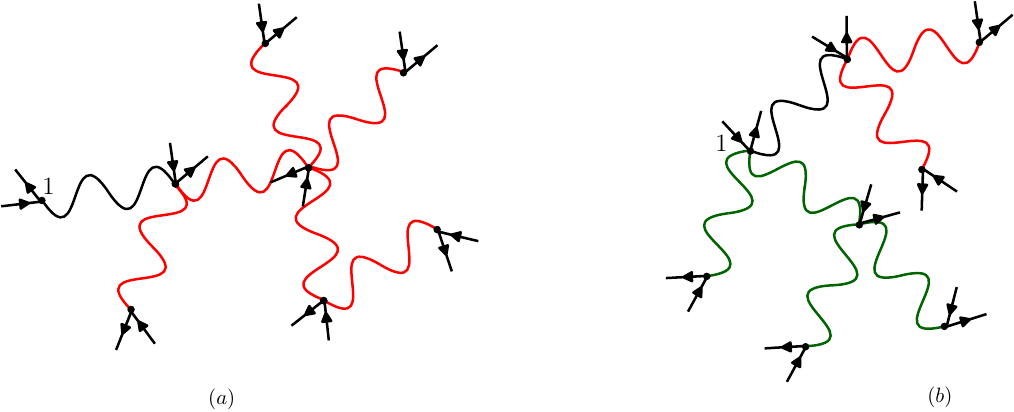} \caption{The two possibilities arising from the extraction of a fermionic line (black solid line attached to the vertex \virg{1}), depending on whether the vertex it is attached to, \virg{1}, is a leaf (a) or a not-leaf (b) of the spanning tree. The red sub-tree corresponds to a source term of type $G^{(n)}$ (with $n=7$ for (a) and $n=3$ for (b)) and the green one to a source term of type $\dot{G}^{(n)}$ (with $n=5$).} \label{fig:extraction} \end{figure} 
In the first case (Fig. \ref{fig:extraction}.(a)) the contribution from the remaining part of the kernel can be written as the expectation of the monomial in the r.h.s. of Eq. \eqref{def_potential} multiplying $G^{(n-1)}$; in the second case (Fig. \ref{fig:extraction}.(b)) it is the expectation of two monomials, the one multiplying $G^{(n_1)}$ and the one multiplying $\dot{G^{(n_2)}}$, with $n_1+n_2=n$. All such new source fields $G^{(n)}$ and $\dot{G}^{(n)}$, with $n\ge2$, have all scaling dimension  +1.\footnote{The dimension of $G^{(n)}$ and $\dot{G}^{(n)}$ is the same as in the last line of Eq. \eqref{def_potential} the number of fermionic bilinears is $n-1$.} In fact the dimensional factors $2^{N(1-n)}, 2^{N(2-n)}, 2^{N(1-n)}$ in the r.h.s of Eq. \eqref{def_potential} should be regarded as $2^{N D_{sc}}$, where the scaling dimension $D_{sc}$ of a monomial $(\psi)^{2q} (G)^{p} (\dot{G})^{\dot{p}} (\ddot{G})^{\ddot{p}}$ is defined as $ D_{sc}(2q;p;\dot{p}; \ddot{p}):= 2- q- p- \dot{p}- 2\ddot{p}$. The dimensionless factors $n^3$ and $n$ in front of the source terms are present for technical convenience, as part of the inductive structure of our proof (see Subsection \ref{ssection_improved_5}, in particular the comments after Eq. \eqref{eq_37}).  

One has therefore to add in the multiscale expansion such an unbounded number of new source terms, and prove self-consistently that \it{all} the dimensionally \it{marginal} and \it{relevant} terms, namely those with $D_{sc}=0$ and $D_{sc}>0$ respectively, remain uniformly bounded at all the steps of the multiscale integration. The main challenge is to show that this whole amount of new terms can be still controlled via the same mechanism based on the non-locality of the boson propagator. This is the main content of Theorem \ref{thm:IB}.

\subsection{The infrared problem and construction of the lattice massive QED$_2$}

Once that the ultraviolet integration is performed and 
$\mc{W}^{(u.v.)}(\psi;\vphi;J;B)$ in Eq.  
\eqref{eq_main_potential}
has been obtained, as expressed 
by Theorem \ref{thm_UV}, the full construction of the model requires the infrared integration, see Eq. \eqref{eq:IRinteg}. The  
analysis is done via a multiscale integration which is 
essentially identical, up to trivial adaptation, to the one performed (for instance) in
the case of the Thirring model, see  
\cite[Sect. 2]{BFM07}; in that case
one starts from a quartic fermionic interaction, but after the integration of the first scale, an expression similar to 
Eq. \eqref{eq_main_potential} is obtained, with kernels verifying
the bounds Eq. \eqref{eq_50}; hence the analysis in \cite[Sect. 2]{BFM07}
holds essentially identical. It should be noted that the integration of the infrared scales is completely different from the ultraviolet ones; there is no dimensional improvement due to the non-locality of the interaction, as briefly explained in the Subsection \ref{subsec:proofsketch}, hence one has to introduce a renormalized expansion in terms of a set of
running coupling constants, corresponding to the relevant and marginal terms. The running coupling constants correspond to: the quartic interaction, the wave function renormalization, the electron mass and the amplitude of the 
currents, see  \cite[Sect. 2.1]{BFM07}. The expansion in the running coupling constants is convergent if they are small enough, but one has to 
prove that the quartic running coupling constants remain small; this follows
from the asymptotic vanishing of the beta function proved in  \cite[Theorem 3.1]{BM}. While the effective quartic coupling remains close to its initial value in the multiscale analysis, the wave function renormalization increases as a power law, with a critical exponent given by a non-trivial, analytic function of the quartic coupling. 
The parameters $Z^5_N$ and $m_N$ in the r.h.s. of Eq. \eqref{def_functional} have to be fixed to ensure the natural physical requirements, see \cite{M22}.
The result can be summarized in the following theorem.

\begin{theorem}
\label{thm:anomaly}
Under the same conditions of Theorem \ref{thm_UV}, 
the limit $L\to \infty$
of the two-point function and of the current and axial current correlations exists. In the massless case, for $|k|\le M$ the two-point Schwinger function behaves like:
\begin{equation}
\label{eq:anomaly1}
\hat{S}^1_{s,s'}(k)=\mc{O}(|k|^{-1+\eta}),
\end{equation}
with $\eta=\a \lambda^2+ \mc{O}(\lambda^3)$, for
 a suitable $\a>0$, and the chiral current correlation verifies:
\begin{equation}
\label{eq:anomaly}
\sum_{\nu=0,1}p_\nu \hat{\Sigma}^{\xi}_{5;\mu,\nu}(p)= \sum_{\nu=0,1}
\tfrac{i\l}{\pi} \ve_{\nu,\mu} p_\nu+ \mc{O}\Big(\tfrac{\l|p|^{1+\th}}{M^{\th}} \Big),
\end{equation}

for some $\th\in(0,1)$.
\end{theorem}

The anomalous divergence of the  two-point Schwinger function
follows from \cite[Theorem 1]{BFM07}. The fact that, instead, no anomalous exponent appears in the ultraviolet large-momentum behavior, as stated in Eq. \eqref{eq:23a} of Theorem \ref{thm_UV},  is a crucial expected property of the massive QED$_2$ \cite{X0}.

Details for the existence of the $L\to\infty$ limit can be found e.g. in \cite[App. D]{GM2010}. Finally
the proof of Eq. \eqref{eq:anomaly} is done in
\cite[Sect. VI]{M22} (assuming the validity of Theorem \ref{thm_UV}
proved here). 
The integration of the ultraviolet degrees of freedom induces the decomposition $\hat \Sigma^\xi_{5;\mu,\nu}=\hat \Sigma^{(u.v.)}_{5;\mu,\nu}+\hat \Sigma^{(i.r.)}_{5;\mu,\nu}$, with $\hat{\Sigma}^{(u.v.)}_{5;\mu,\nu}$ as in Theorem \ref{thm_UV}, satisfying $\hat{\Sigma}^{(u.v.)}_{5;\mu,\nu}(p)- \hat{\Sigma}^{(u.v.)}_{5;\mu,\nu}(0)= \mc{O}\big(|p|/{M}\big)$. 
The analysis of $\hat{\Sigma}^{(i.r.)}_{5;\mu,\nu}$ is based on \cite[Sect. 2]{BM,BFM07} and references therein. As explained in \cite[Sect. VI]{M22}, one introduces a suitably tuned \it{reference model} (see also \cite[Sect. 1.2]{BM}) such that its infrared fixed point coincides with that of the original model. This implies that, letting $\tilde{\Sigma}_{5;\mu,\nu}$ be the correlation function of the reference model,  $\hat{\Sigma}_{5;\mu,\nu}(p)= Z^5_N \tilde{{\Sigma}}_{5;\mu,\nu}(p)+ \mc{R}_{\mu,\nu}(p)$ \cite[Eq. 44]{M22}, with $\mc{R}_{\mu,\nu}$ continuous (in particular $\mc{R}_{\mu,\nu}(p)- \mc{R}_{\mu,\nu}(0)= \mc{O}\big((|p|/{M})^{\th}\big)$) in which we can reabsorb also the contribution $\hat\Sigma^{(u.v)}_{5;\mu,\nu}$. In this way $\hat\Sigma^{\xi}_{5;\mu,\nu}(p)= Z^5_N \tilde{{\Sigma}}_{5;\mu,\nu}(p)+ \mc{R}_{\mu,\nu}(0)+ \mc{O}\big((|p|/{M})^{\th}\big) $. The contribution to the anomaly from $\tilde{{\Sigma}}_{5;\mu,\nu}$ is explicit \cite[Eq. (46)]{M22}, while the one from $\mc{R}_{\mu,\nu}(0)$ is fixed via the WI in Eq. \eqref{eq:WI5a} \cite[Eq. 50]{M22}. Combining the two, the desired result, Eq. \eqref{eq:anomaly}, follows.

\medskip 

Theorem \ref{thm:anomaly} establishes the non-perturbative validity of the Adler-Bardeen theorem in the lattice QED$_2$; note that Eq. \eqref{eq:anomaly1} ensures the validity of the anomaly non-renormalization with the coefficient in agreement with perturbation theory, Eq. 
\eqref{xx3} (and in contrast with exact solutions or momentum regularization, where a similar result is found with a different factor).
This proves that the lattice regularization of QED$_2$, fulfilling
the physical constraints \textbf{a)}-\textbf{d)} in Subsection \ref{ssect_intro1} and in the regime of Eq. \eqref{energy_regime}, verifies a number of properties at a non-perturbative level, like
the decrease of divergence degree due to the absence of contribution of longitudinal part of boson propagator, the conservation of the current and the correct value of the
anomaly (namely the Adler-Bardeen non-renormalization), whose analogous are crucial ingredients for the consistency of realistic 4d QFT models.

It is an interesting open problem if a similar result could be achieved without integrating out the bosons but with a simultaneous decomposition of bosons and fermions, which seems necessary in four dimensions. Also, it would be interesting 
to consider a 2d chiral model, where finding a regularization 
preserving WI requires the cancellation of the anomalies.

\paragraph{Organization of the paper.}

The paper is organized as follows. In Section \ref{sect_multiscale} we set up the fermionic multiscale analysis in absence of external sources $\vphi,J,B$. We rely on the \it{Gallavotti-Nicolò tree expansion}, which is reviewed in Appendix \ref{app_tree}, and we establish the well posedness of the multiscale expansion assuming the boundedness of the relevant and marginal terms. The latter property is the main focus of Section \ref{sect_improved}, and its proof relies on some exact identities for the kernels of the effective potential, which are proved in Appendix \ref{app_identities}. In Section \ref{sect_external} we reintroduce the external sources $\vphi,J,B$, thus recovering the complete generating functional, Eq. \eqref{def_functional_UV}. In Section \ref{sec:proofTHMUV} we collect all the information from the preceding sections, thus proving the claims of Theorem \ref{thm_UV}. Appendix \ref{app_gauge_inv} is devoted to the proof of two main implications of \it{local phase invariance} (cf. Eq. \eqref{eq_gauge_inv}), namely the $\xi$-independence property, Lemma \ref{lemma_gauge_inv}, and the Ward Identities, Eqs. \eqref{eq:WI5a}-\eqref{eq:WI5b}. In Appendix \ref{app_bubble} we collect some technical results about the non-interacting \it{bubble diagrams}, which are essential for the analysis in Sections \ref{sect_multiscale}-\ref{sec:proofTHMUV}. Finally, in Appendix \ref{app:symmetries} the symmetries of the lattice theory are discussed, whereas in Appendix \ref{app:limit} we sketch the proof for showing the existence of the removed-cut-off limit.

\section{The fermionic multiscale}
\label{sect_multiscale}

With no loss of generality it is sufficient to prove Theorem \ref{thm_UV} with $M\in [1,2)$, i.e. $h^*_M=0$ (cf. below Eq. \eqref{eq_bound_boson}). Indeed, as one can easily check, the rescaling
\begin{equation}
\label{eq:rescaling}
\hspace{-10pt}\left\{\begin{array}{ccc}
\hspace{-3pt}\!\!\L \ni x \mapsto y\in\L^{(R)}\!:=\! 2^{h^*_M}\L; & 
\!\! \psi_x\mapsto \psi^{(R)}_y\!:=\! 2^{-\frac{1}{2}h^*_M}\psi_{2^{-h^*_M}y}; &\!\! \vphi_{x} \mapsto \vphi^{(R)}_{y}\!:=\! 2^{-\frac{3}{2}h^*_M} \vphi_{2^{-h^*_M}y}; \\ 
\\
\hspace{-10pt} \!\! A_{\mu,x} \mapsto A^{(R)}_{\mu,y}\!:=\! A_{\mu, 2^{-h^*_M}y};&\!\!
\hspace{-8pt} J_{\mu,x} \mapsto J^{(R)}_{\mu,y}\!:=\! J_{\mu, 2^{-h^*_M}y};&\hspace{-5pt} \!\!\!
B_{\mu,x} \mapsto B^{(R)}_{\mu,y}\!:=\! 2^{-h^*_M} B_{\mu, 2^{-h^*_M}y} 
\end{array} \right. \end{equation}
maps the original theory into another one with parameters: $M\mapsto M^{(R)}:= 2^{-h^*_M}M;
m_N\mapsto m^{(R)}_N:= 2^{-h^*_M}m_N;
\mathsf{e} \mapsto \mathsf{e}^{(R)}:= 2^{-h^*_M}\mathsf{e}$. In particular $M^{(R)}\in[1,2)$. Then if Theorem \ref{thm_UV} holds for the rescaled theory, it will hold for the original one as well. 

\subsection{Multiscale decomposition}
 As a first step towards the proof of Theorem \ref{thm_UV}, we restrict to the case of zero external sources, namely to the analysis of $\mc{W}^{(u.v.)}(\psi;0;0;0)$, which can expressed as
\begin{equation}
\label{eq:5}
\mc{W}^{(u.v.)}(\psi;0;0;0)= -\log \int P^{(0,N]}(\mc{D}\z) e^{-\mc{V}(\psi+\z;0;0;0)},
\end{equation} where $\mc{V}(\psi;G;\dot{G};\ddot{G})$ is the generalized potential defined in Eq. \eqref{def_potential}.\footnote{Note that for any $N,L<\infty$, due to the finiteness of the Grassmann algebra, all the sums in the r.h.s. of Eq.  \eqref{def_potential} are actually truncated at $n_*\equiv 2 (2^NL)^2 +1$, and $\mc{V}$ depends only upon the auxiliary field variables with \it{degree} $n$ less or equal then $n_*$.} 
The external fields $\vphi,J$ and $B$ will be reintroduced in Section \ref{sect_external}. 
We start by  splitting the fermion propagator as
\[g^{(0,N]}(x-y)\equiv g^{[1, N]}(x-y)= \sum_{h=1}^N g^{(h)}(x-y),\]
where: 
\[\begin{split}
& g^{(N)}(x-y):= \int_{\L^*} dk\; e^{-ik\cdot(x-y)}\frac{1-\chi_{N-1}(|k|)}{-i\slashed{s}(k) + M_N(k) + m_N  },\\
& g^{(h)}(x-y):= \int_{\L^*} dk\; e^{-ik\cdot(x-y)}\frac{\chi_{h}(|k|)-\chi_{h-1}(|k|)}{-i\slashed{s}(k) + M_N(k) + m_N  }, \qquad   1\le h\le N-1
\end{split}\]
with $\chi_h(|k|)\equiv \chi(2^{-h}|k|)$, and $\chi$ defined after Eq. \eqref{eq:7}; recall also Eq. \eqref{def_fermion_prop}. It is possible to check that, due to the presence of the Wilson mass $M_N(k)$, 
\begin{equation}
\label{eq:bound:Wilson}
\big|\big(-i\slashed{s}(k)+ M_N(k)+ m_N\big)^{-1}_{s,s'}\big|\le \frac{\sqrt{2}\pi}{\sqrt{|k|^2 + m_N^2}}, \qquad  \forall k\in\L^*,
\end{equation}
provided that $|m_N|\le1$ and $N\ge 2\log_2\pi+3$. Moreover, from the properties of $\chi_h$, i.e. compact support and Gevrey-2 regularity, it follows that \cite[App. A]{GMR21}
\begin{equation}
\label{eq_bound_propagator}
\max_{s,s'\in\{1,2\}} |g^{(h)}_{s,s'}(x-y)|\le K 2^h e^{-\kappa_2 \sqrt{2^h|x-y|}}, \quad \forall\; 0\le h\le N,
\end{equation}

for suitable $K,\kappa_2>0$, uniformly w.r.t. $m_N\in[-1,1]$. Notice that by construction $\\\mc{W}^{(u.v.)}(\psi;0;0;0)= \mc{V}^{(0)}(\psi;0;0;0)$, where, for any $0\le h\le N-1$, 
\begin{equation}
\begin{split}
\label{eq:v_hbackward}
\mc{V}^{(h)}(\psi;G;\dot{G};\ddot{G}):=  -\log \int P^{(h,N]}(\mc{D}\z) e^{-\mc{V}(\psi+\z;G;\dot{G};\ddot{G})}
\end{split}
\end{equation}

and $P^{(h,N]}$ is the Grassmann integration with propagator $g^{(h,N]}(x-y)\equiv \sum_{k=h+1}^N g^{(k)}(x-y)$. Letting also $P^{(h)}$ be the integration with propagator $g^{(h)}(x-y)$, by the Gaussian addition principle \cite[Eq. (4.21)]{GM01}, we have that
\begin{equation}
\mc{V}^{(h)}(\psi;G;\dot{G};\ddot{G})= -\log \int P^{(h+1)}(\mc{D}\z) e^{-\mc{V}^{(h+1)}(\psi+\z;G;\dot{G};\ddot{G})}.
\end{equation}

As a standard fact \cite[Eq. (4.19)]{GM01}, the effective potential $\mc{V}^{(h)}$ can be formally computed by the formula:
\begin{equation}
\label{eq:8}
\mc{V}^{(h)}(\psi;G;\dot{G};\ddot{G})= \sum_{s\ge1} \frac{(-1)^{s+1}}{s!} \mc{E}^T_{(h+1)}\Big( \mc{V}^{(h+1)}(\psi+\cdot; G;\dot{G};\ddot{G});^s \Big),
\end{equation}

where, for Grassmann monomials $X_1(\psi,\z),\dots,X_s(\psi,\z)$, the \it{truncated expectation} $\mc{E}^T$ is recursively defined as
\begin{equation}
\label{def_truncated}
\begin{split}
\mc{E}^T_{(h+1)}\big(X_1(\psi,\cdot),\dots,&X_s(\psi,\cdot)\big):= \mc{E}_{(h+1)}\big( X_1(\psi,\cdot),\dots,X_s(\psi,\cdot)\big)+\\
&- \mathds{1}_{(s\ge2)} \sum_{\Pi\in\text{ Part}(\{1,\dots,s\})} \s_{\Pi} \prod_{j=1}^{|\Pi|} \mc{E}^T_{(h+1)}\big( X_{\pi_{j,1}}(\psi,\cdot), \dots, X_{\pi_{j,p(j)}}(\psi,\cdot) \big),
\end{split}
\end{equation}

with the following understanding. $ \mc{E}_{(h+1)}(\cdot)\equiv \int P^{(h+1)}(\mc{D}\z)\; \cdot$ is the \it{simple expectation}; $\Pi$ runs over all the partitions of $\{1,\dots,n\}$ into components $\{\pi_{j,1},\dots,\pi_{j,p(j)}\}$ indexed by $j=1,\dots,|\Pi|$, with $p(1)+\dots+p(|\Pi|)=s$; the ordering in the r.h.s. of Eq. \eqref{def_truncated} is such that $\pi_{j,1}<\pi_{j+1,1}$ and $\pi_{j,k}<\pi_{j,k+1}$ for every $j=1,\dots,|\Pi|-1$ and $k=1,\dots,p(j)-1$; $\s_{\Pi}$ is the sign needed to order the Grassmann monomials as they appear in the r.h.s of Eq. \eqref{def_truncated} starting from the ordering as in the l.h.s.. Finally, in Eq. \eqref{eq:8}, $\mc{E}^T_{(h+1)}(X(\psi,\cdot);^s):=\mc{E}^T_{(h+1)}\big(X(\psi,\cdot),\dots,X(\psi,\cdot)\big)$ when $X_1,\dots,X_s=X$. Eq. \eqref{eq:8} naturally yields an integral representation for $\mc{V}^{(h)}$:
\begin{equation}
\label{eq:9}
\begin{split}
\mc{V}^{(h)}(\psi;G;\dot{G};\ddot{G})& = \sum_{q\ge0} \sum_{\ud{n}\in\mbb{N}_{\ge1}^{\times}} \sum_{\ud{\dot{n}}\in \mbb{N}_{\ge2}^{\times}} \sum_{\ud{\ddot{n}}\in \mbb{N}_{\ge2}^{\times}} \int_{\L_F} d\ud \eta\hspace{1pt} d\ud \eta' \int_{\L_B} d\ud{b}\hspace{1pt} d\ud{\dot{b}}  \hspace{1pt} d\ud{\ddot{b}} \, \times\\
&\times W^{(h)|q,q}_{\ud{n};\ud{\dot{n}};\ud{\ddot{n}}}(\ud{\eta};\ud{\eta}';\ud{b};\ud{\dot{b}};\ud{\ddot{b}}) \prod_{k=1}^q \psi^+_{\eta_k} \prod_{k=1}^q \psi^-_{\eta'_k} \prod_{k=1}^{\dim\ud{n}} G^{(n_k)}_{b_k} \prod_{k=1}^{\dim\ud{\dot{n}}} \dot{G}^{(\dot{n})}_{\dot{b}_k} \prod_{k=1}^{\dim\ud{\ddot{n}}} \ddot{G}^{(\ddot{n})}_{\ddot{b}_k},
\end{split}
\end{equation}

where the notation is understood as follows:
\begin{itemize}
\item $\mbb{N}_{\ge p}^{\times} \equiv \emptyset \sqcup \mbb{N}_{\ge p} \sqcup (\mbb{N}_{\ge p})^2\sqcup (\mbb{N}_{\ge p})^3\sqcup \dots$ and  $\mathbb{N}_{\geq p}=\{n \in \mathbb{N}: n \geq p\};$
\item  if $\ud{n}=(n_1,\dots,n_p)$, then $|\ud{n}|\equiv n_1+\dots+n_p$ and $\dim\ud{n}\equiv p$;
\item $d\ud{b}\equiv \prod_{k=1}^{\dim\ud{n}}db_k$, $d\ud{\dot{b}}\equiv \prod_{k=1}^{\dim\ud{\dot{n}}}d\dot{b}_k$ and $d\ud{\ddot{b}}\equiv \prod_{k=1}^{\dim\ud{\ddot{n}}}d\ddot{b}_k$, and each of such variables runs in $\Lambda_B$; similarly $d \ud \eta=d\eta_1\dots d\eta_q$ and each variable runs in $ \Lambda_F$;
\item the term of the sum in the r.h.s. of Eq. \eqref{eq:9} with $q=0$ and $\ud{n}=\ud{\dot{n}}=\ud{\ddot{n}}=\eset$ must be interpreted as $W^{(h)|0,0}_{\eset;\eset;\eset}\in\mbb{C}$, and its $L^p$ norm is by definition $\big\|W^{(h)|0,0}_{\eset;\eset;\eset}\big\|_p^w:= L^{-\frac{2}{p}}\big|W^{(h)|0,0}_{\eset;\eset;\eset}\big|$ (compare with Eq. \eqref{def_Lp_norm}). 
\end{itemize}

 The kernels $W^{(h)|q,q}_{\ud{n};\ud{\dot{n}};\ud{\ddot{n}}}(\ud{\eta};\ud{\eta}';\ud{b};\ud{\dot{b}};\ud{\ddot{b}})$ are assumed to be anti-symmetric under the exchange of any couple of $\psi^+$ or $\psi^-$ labels, and symmetric under exchange of $G,\dot{G}$ and $\ddot{G}$ labels separately. 
Such symmetries can always be imposed and imply a representation for the kernels $W^{(h)}$ as functional derivatives of $\mc{V}^{(h)}$: letting $p=\dim\ud{n}, \dot{p}=\dim\ud{\dot{n}}, \ddot{p}=\dim\ud{\ddot{n}}$,  if $q+p+\dot p+\ddot p >0$, Eq. \eqref{eq:9} implies that
\begin{equation}
\label{eq_14}
\begin{split}
&W^{(h)|q,q}_{\ud{n};\ud{\dot{n}};\ud{\ddot{n}}}(\ud{\eta},\ud{\eta}';\ud{b};\ud{\dot{b}};\ud{\ddot{b}}) = \frac{(-1)^q}{q!^2p! \dot{p}! \ddot{p}!} \times \\
&\times \left(\frac{\d^{2q}}{\d\psi^+_{\eta_1}\dots\d\psi^+_{\eta_q}\d\psi^-_{\eta'_1}\dots\d\psi^-_{\eta'_q}} \frac{\d^p}{\d G^{(n_1)}_{b_1}\dots\d G^{(n_p)}_{b_p}} \frac{\d^{\dot{p}}}{\d \dot{G}^{(\dot{n}_1)}_{\dot{b}_1}\dots\d \dot{G}^{(\dot{n}_p)}_{\dot{b}_{\dot{p}}}} \frac{\d^{\ddot{p}}}{\d \ddot{G}^{(\ddot{n}_1)}_{\ddot{b}_1}\dots\d \ddot{G}^{(\ddot{n}_{\ddot{p}})}_{\ddot{b}_{\ddot{p}}}} \mc{V}^{(h)}\right)(\mb{0}),
\end{split}
\end{equation}

where $\boldsymbol{0}\equiv (0;0;0;0)$; moreover for complex fields $G^{(n)}_b=\text{Re} G^{(n)}_b+i \text{Im} G^{(n)}_b$ and Grassmann monomials the action of functional derivatives is given respectively by

\begin{equation} \hspace{-5pt}\frac{\d}{\d G^{(n)}_b}:= 2^{2N}\left(\frac{\de}{\de \text{Re} G^{(n)}_b}- i \frac{\de}{\de \text{Im} G^{(n)}_b}\right), \quad  \frac{\d}{\d\psi^{\ve}_{\eta}}\prod_{j=1}^n \psi^{\ve_j}_{\eta_j} := \sum_{k=1}^n (-1)^{k-1}\d(\eta-\eta_k) \prod_{j\ne k} \psi^{\ve_j}_{\eta_j}, \label{func_deriv}\end{equation}

the second of which extends to general Grassmann polynomials by linearity.

\subsection{Tree expansion}
\label{sect_tree_expansion}

By Eq. \eqref{eq:8}, one has that Eq. \eqref{eq:9} can be written as a recursion for the kernels $W^{(h)}$ in terms of $W^{(h+1)}$. The kernels $W^{(h+1)}$ can be either left implicit, or can be further expanded in terms of $W^{(h+2)}$. As it is customary in the renormalization group context, we expand only the irrelevant kernels, i.e. those with $D_{sc}<0$, where for a generic kernel of type $(\psi)^{2q} (G)^{p} (\dot{G})^{\dot{p}} (\ddot{G})^{\ddot{p}}$,
\begin{equation}
\label{def_Dsc}
D_{sc}(2q;p;\dot{p}; \ddot{p}):= 2- q- p- \dot{p}- 2\ddot{p}.
\end{equation}

By iterating this procedure, we find a representation for the kernels $W^{(h)}$ in terms of all the kernels at scale $N$ and the relevant and marginal kernels at scales $h'>h$.
The outcome of this expansion is conveniently represented as an infinite sum over trees for the \it{single-scale contributions} $\wt{W}^{(h)|q,q}_{\ud{n};\ud{\dot{n}};\ud{\ddot{n}}}:= W^{(h)|q,q}_{\ud{n};\ud{\dot{n}};\ud{\ddot{n}}} - W^{(h+1)|q,q}_{\ud{n};\ud{\dot{n}};\ud{\ddot{n}}}$:

\begin{equation}
\label{eq_tree2}
\wt{W}^{(h)|q,q}_{\ud{n};\ud{\dot{n}};\ud{\ddot{n}}}= \mathscr{S} \sum_{\t\in\mc{T}_{h,N}}^* \sum_{\mb{P}\in \mc{P}_{\t}}^{**} \wt{W}[\t,\mb{P}],  \end{equation}

with the following understanding.

\begin{enumerate}
\item $\mc{T}_{h,N}$ is the set \it{Gallavotti-Nicolò trees} between scale $h$ and $N$.  A tree $\t \in \mc{T}_{h,N}$ is such that the root is not a branching point and the nodes are partially ordered from the root towards the endpoints (we write $v \prec w$ if $w$ comes after $v$ on the tree). The nodes $v$ that are not the root are called \it{vertices}, the set of which we denote by 
$V(\t)$; 
each vertex $v$ comes with a \it{scale label} $h_v$, an integer between $h+1$ and $N+1$; the scale label of the root equals $h$ by definition. The $*$ over the sum over $\mc{T}_{h,N}$ is the constraint that there must always be one and one only vertex $v_0$, which is not an endpoint, following the root, with scale label $h_{v_0}= h+1$ (see Fig. \ref{fig_tree}).
\item $\mc{P}_{\t}$ is a set of decorations for the tree $\t$, each decoration $\mb{P}\equiv \{P_v,Q_v\}_{v\in V(\t)}$ being called a collection of \it{field labels}. Any $f\in\mb{P}$ comes with a sign $\ve_f$ and with a variable $\eta_f\in\L_F$, so that to any $f$ we can associate a Grassmann variable $\psi^{\ve_f}_{\eta_f}$.

At each vertex $v\in V(\t)$, $Q_v$ is the set of \it{internal fields}, i.e. those corresponding to the Grassmann variables that are integrated under the action of the truncated expectation $\mc{E}^T_{(h_v)}$; conversely $P_v$ identifies the \it{external fields} at the vertex $v$, i.e. the Grassmann variables that are not integrated under the action of $\mc{E}^T_{(h_v)}$. The field labels are subject to the following constraints.

\begin{itemize}
\item[2.1.] If $v_0$ is the vertex following the root, then $P_{v_0}$ identifies the set of fields determined by the labels $q,\ud{n},\ud{\dot{n}},\ud{\ddot{n}}$ of the kernel $W^{(h)|q,q}_{\ud{n};\ud{\dot{n}};\ud{\ddot{n}}}$ in the l.h.s. of Eq. \eqref{eq_tree2}.
\item[2.2.] If $v\in V(\t)$ is not an endpoint (we will write $v\in V_0(\t)$), then $|Q_v|\ge2$, namely at least two fields are always integrated.
\item[2.3.] If $v_1,\dots,v_s$ are the vertices following $v$ on $\t$, then $P_v \dot{\cup} Q_v= \dot{\cup}_{j=1}^s P_{v_j}$.
\item[2.4.] If $v\in V(\t)$ is an endpoint (we will write $v\in V_e(\t)$), then $Q_v=\eset$; besides, if $h_v< N+1$, $P_v$ identifies a set of fields (at most four, with at least two $\psi$ variables) corresponding to a non-irrelevant  kernel (i.e. such that $D_{sc}\equiv D_{sc}(P_v)\ge0$, see Eq. \eqref{def_Dsc}), $W^{(h_v-1)}_{P_v}$. If $h_v=N+1$ instead, $P_v$ can identify either an irrelevant or non-irrelevant kernel (i.e. $D_{sc}(P_v)<0$ or $\ge0$ respectively), $W^{(N)}_{P_v}$, related to the bare potential $\mc{V}$, Eq. \eqref{def_potential}.
\item[2.5.] If $v\in V_e(\t)$ has $P_v$ corresponding to a non-irrelevant kernel, then the vertex $v'$ which precedes $v$ has scale label $h_{v'}= h_v -1$ (this reflects the prescription above that relevant and marginal kernels are not being expanded in terms of kernels at higher scales).
\end{itemize}

The sum $\sum_{\mb{P}\in \mc{P}_{\t}}^{**}$ is simply the sum over all the choices of field labels coherent with the constraints above.
\item $\wt{W}[\t,\mb{P}]$ has the following analytical expression:
\begin{equation}
\label{eq:12}
\begin{split}
&\wt{W}[\t,\mb{P}](\ud{\eta};\ud{b};\ud{\dot{b}};\ud{\ddot{b}})=\\ & \frac{\s_{\mb{P}}}{s_{v_0}!} \int_{\L_F} d\ud{\eta}_{Q_{v_0}} \mc{E}^T_{(h_{v_0})} \big( \psi_{\tilde{P}_{w_1}},\dots,\psi_{\tilde{P}_{w_{s_{v_0}}}} \big)  \prod_{j=1}^{s_{v_0}} \wt{W}[\t_{w_j},\mb{P}_{w_j}] \big((\ud{\eta};\ud{b};\ud{\dot{b}};\ud{\ddot{b}})_{P_{w_j}}\big),
\end{split}
\end{equation}

where:
\begin{itemize}
\item[3.1.] $\s_{\mb{P}}$ is a sign;
\item[3.2.] for $v\in V_0(\t)$, $w_1,\dots,w_{s_v}$ are the vertices following $v$ on $\t$ and $\mb{P}_{w_j}\equiv \bigcup_{v\succeq w_j} P_v $;
\item[3.3.] $\tilde{P}_v:= P_v\cap Q_{v'}$, with $v'$ the vertex which precedes $v$ on $\t$, and $\psi_{\tilde{P}_v}\equiv \prod_{f\in \tilde{P}_v} \psi^{\ve_f}_{\eta_f}$, the product being performed in a suitable, prescribed order;
\item[3.4.] If $v\in V_0(\t)$, $\t_v$ is the sub-tree with $v_0=v$ and $\wt{W}[\t_v,\mb{P}_v] \big((\ud{\eta};\ud{b};\ud{\dot{b}};\ud{\ddot{b}})_{P_v}\big)$ is defined recursively via Eq. \eqref{eq:12}; if instead $v\in V_e(\t)$, $\wt{W}[\t_v,\mb{P}_v] \big((\ud{\eta};\ud{b};\ud{\dot{b}};\ud{\ddot{b}})_{P_v}\big)$ is the kernel associated with the endpoint $v$ determined by the decoration $P_v$, to be denoted by $W^{(h_v-1)}_{P_v} \big((\ud{\eta};\ud{b};\ud{\dot{b}};\ud{\ddot{b}})_{P_v}\big)$.

\end{itemize}

\item $\mathscr{S}$ is the \it{symmetrization operator}, which acts on every kernel taking the graded average over all the permutations of labels of identical Grassmann and complex fields, so that for the overall kernel Eq. \eqref{eq_14} holds.
\end{enumerate}

\begin{figure}[h]
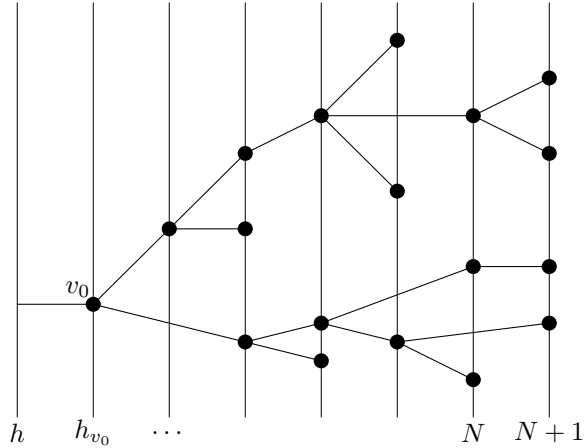

\small
    \centering
    \tikz[baseline=-2pt]{

    \foreach \x in {0,...,7}
      {
	  \draw[very thin] (\x ,-1.5) -- (\x , 4);
      }

      \draw (0,-1.7)node{$h$}; \draw (1,-1.7)node{$h_{v_0}$}; \draw (2,-1.7)node{$\dots$}; \draw (6,-1.7)node{$N$}; \draw (7,-1.7)node{$N+1$}; \draw (0.8,0.2)node{$v_0$};
      
			\draw (0,0) node {} -- (1,0) node[vertex] {} -- (2,1) node[vertex] {} -- (3,2) node[vertex] {}-- (4,2.5) node[vertex] {} -- (5,3.5) node[vertex] {} ;
   \draw (4,2.5) node {} -- (5,1.5) node[vertex]{};
   \draw (4,2.5) node {} -- (6,2.5)node[vertex]{}-- (7,3)node[vertex]{};
   \draw (6,2.5)node{}-- (7,2) node[vertex]{};
   \draw (2,1) node {} -- (3,1.0) node[vertex]{};
   \draw (1,0) node {} -- (3,-0.5) node[vertex]{}-- (4,-0.75)node[vertex]{};
   \draw (3,-0.5) node{}-- (4,-0.25) node[vertex]{}--  (6,0.5)node[vertex]{}-- (7,0.5)node[vertex]{};
   \draw (4,-0.25) node{}-- (5,-0.5)node[vertex]{} -- (6,-1)node[vertex]{};
   \draw (5,-0.5) node{} --  (7,-0.25)node[vertex]{};
   }
 \caption{\small\textit{Example of a Gallavotti-Nicolò tree $\t\in\mc{T}_{h,N}$, where the vertices $v\in V(\t)$ are outlined as black dots.}} 
 \label{fig_tree}
\end{figure}

\subsection{Bounds for the tree expansion}

Note that $\wt{W}[\t,\mb{P}]$ can be regarded as a functional in the kernels $\big\{W^{(h_v-1)}_{P_v}\big\}_{v\in V_e(\t)}$ associated with the endpoints. In fact, an iterative application of Eq. \eqref{eq:12} yields:
\begin{equation}
\label{eq:12b}
\begin{split}
&\wt{W}_{\ud{n};\ud{\dot{n}};\ud{\ddot{n}}}[\t,\mb{P}](\ud{\eta};\ud{b};\ud{\dot{b}};\ud{\ddot{b}})=\\
& \tilde{\s}_{\mb{P}} \prod_{v\in V_0(\t)} \left\{\frac{1}{s_v!} \int_{\L_F} d\ud{\eta}_{Q_v} \mc{E}^T_{(h_v)} \big( \psi_{\tilde{P}_{w_1}},\dots,\psi_{\tilde{P}_{w_{s_v}}} \big) \right\} \prod_{v\in V_e(\t)} W^{(h_v-1)}_{P_v} \big((\ud{\eta};\ud{b};\ud{\dot{b}};\ud{\ddot{b}})_{P_v}\big)
\end{split}
\end{equation}

with $\tilde{\s}_{\mb{P}}\equiv \prod_{v\in V_0(\t)} \s_{\mb{P}_v} \in \{\pm\}$. The tree expansion turns out to be a convenient tool for controlling the size of the kernels $\wt{W}^{(h)}$ as in Eq. \eqref{eq_tree2}. In fact, introducing
 \begin{align}
&\label{eq_29a}\Pi^{(h,N]}(b_1,b_2):=  \frac{1}{2} \int d\eta_1 d\eta_1' d\eta_2 d\eta_2' c_{b_1}(\eta_1,\eta_1') c_{b_2}(\eta_2,\eta_2') g^{(h,N]}_{\eta_1',\eta_2}g^{(h,N]}_{\eta_2',\eta_1},\\
&\label{eq_29b} Q^{(N)}(\eta_1,\eta_2,\eta_1',\eta_2') := -\frac{1}{4}\int db db' \l g^A_{b,b'} \Big( c_b(\eta_1,\eta_1') c_{b'}(\eta_2,\eta_2') - c_b(\eta_2,\eta_1') c_{b'}(\eta_1,\eta_2')\Big),
\end{align}
the lowest-order contributions to the kernels $W^{(h)|0,0}_{(1,1);\eset;\eset}$ and $W^{(h)|2,2}_{\eset;\eset;\eset}$ respectively, we can prove the following Proposition.
\begin{proposition}
\label{prop:SDB}
For every $\vth,\th\in(0,1)$ fixed, there exists a constant $C_0\ge1$ such that for any $R\ge1$ and $0\le h\le N$,  the following is true. Assume that  at every scale $h<k\le N$ the following bounds hold true for the non-irrelevant kernels:
\begin{align}
&\label{eq_IB_1} \big\|W^{(k)|1,1}_{\eset;\eset;\eset}\big\|_1^w\le R\l 2^{-\vth k}; \;\;\; \big\|(W^{(k)|2,2}_{\eset;\eset;\eset} - Q^{(N)})\big\|_1^w\le R\l^{\frac{5}{4}};  \\
&\label{eq_IB_2} \big\| W^{(k)|1,1}_{1;\eset;\eset} - c\big\|_1^w \le R\l; \;\;\;  \big\|W^{(k)|1,1}_{n;\eset;\eset}\big\|_1^w, \big\| W^{(k)|1,1}_{\eset;n;\eset}\big\|_1^w \le  \l^{\frac{n-1}{2}} \text{ for }  n\ge2;\\
&\label{eq_IB_5}\big\| g^A*(W^{(k)|0,0}_{(1,1);\eset;\eset}- \Pi^{(k,N]} )\big\|_1^w \le R\l;   \\
&\label{eq_IB_6}\big\|W^{(k)|0,0}_{(n_1,n_2);\eset;\eset}\big\|_1^w, \big\| W^{(k)|0,0}_{\eset;(n_1,n_2);\eset}\big\|_1^w, \big\|W^{(k)|0,0}_{n_1,n_2;\eset}\big\|_1^w \le \l^{\frac{n_1+n_2-2}{2}} \text{ for } n_1+n_2\ge3.
\end{align}

Then the following bounds hold true at scale $h$, for $\l\le (C_0R^4)^{-1}$:
\begin{align}
&\label{eq_SDB_1} \big\|W^{(h)|0;0}_{n;\eset;\emptyset}\big\|_1^w, \big\|W^{(h)|0;0}_{\emptyset;n;\emptyset}\big\|_1^w\le C_02^N \l^{\frac{n-1}{2}}; \; \big\| W^{(h)|0;0}_{\emptyset;\emptyset;n}\big\|_1^w \le C_0 \l^{\frac{5}{8}n};  \\
&\label{eq_SDB_2} \big\| g^A*W^{(h)|0;0}_{(1,1);\emptyset;\emptyset}\big\|_1^w \le C_0; \; \big\| W^{(h)|1,1}_{\emptyset;\emptyset;\emptyset}\big\|_1^w \le RC_0\l 2^h
\end{align}

and for all the other kernels:
\begin{equation}
\label{eq_SDB_3}
\big\| W^{(h)|q,q}_{\ud{n};\ud{\dot{n}};\ud{\ddot{n}}}\big\|_1^w \le C_0^{1+\dim\ud{n}+\dim\ud{\dot{n}}+\dim\ud{\ddot{n}}} \l^{d} 2^{h D_{sc}} \times \left\{\begin{array}{cc}
2^{\th(h-N)}     & \dim\ud{\ddot{n}}>0  \\
1     & \text{otherwise},
\end{array}\right.
\end{equation}

with $D_{sc}\equiv D_{sc}(2q;\dim\ud{n};\dim\ud{\dot{n}};\dim\ud{\ddot{n}})$, see Eq. \eqref{def_Dsc}, and
\begin{equation}
\label{eq:10}
\begin{split} 
&d\equiv d(2q;\ud{n};\ud{\dot{n}};\ud{\ddot{n}}):=\\
&\left\{\begin{array}{cc}
\max\{1,\frac{2}{5}q\},     &\dim\ud{n},\dim\ud{\dot{n}},\dim\ud{\ddot{n}}=0  \\
\frac{1}{4}\max\{q-1,0\}+ \frac{1}{2}(|\ud{n}|+|\ud{\dot{n}}| - \dim\ud{n}-\dim\ud{\dot{n}})+ \frac{5}{8}|\ud{\ddot{n}}|, & \text{otherwise.}
\end{array}\right.\end{split}\end{equation}
\end{proposition}

Observe that Eqs. \eqref{eq_IB_1}-\eqref{eq_IB_6} are stronger than Eqs. \eqref{eq_SDB_1}-\eqref{eq_SDB_2} with $h=k$, for the non-irrelevant kernels. For this reason,
Eqs. \eqref{eq_SDB_1}-\eqref{eq_SDB_2} will be called \it{standard dimensional bounds}, while we will refer to Eqs. \eqref{eq_IB_1}-\eqref{eq_IB_6} as \it{improved bounds}. Proposition \ref{prop:SDB} establishes the validity of the \it{standard dimensional bounds} taking as input the \it{improved bounds} at higher scales. The proof of the latter is the content of Section \ref{sect_improved}. 

\begin{proof}

Proposition \ref{prop:SDB} is actually an immediate consequence of the following bounds for the kernels $\wt{W}^{(h)}$ (cf. Eq. \eqref {eq_tree2}):
\begin{equation}
\label{eq_tree4}
\hspace{-6pt} \big\|\wt{W}^{(h)|q,q}_{\ud{n};\ud{\dot{n}}; \ud{\ddot{n}}} \big\|_1^w\le \big(1-2^{\th-1}\big) C_0^{1+\dim\ud{n}+\dim\ud{\dot{n}}+\dim\ud{\ddot{n}}} 2^{h D_{sc}} \l^{d} \times \left\{\begin{array}{cc} \hspace{-5pt}R & |\ud{n}|,|\ud{\dot{n}}|,|\ud{\ddot{n}}|=0,\, q=1
\\ \hspace{-5pt} 2^{\th(h-N)}     & |\ud{\ddot{n}}|\ge1 \\
1     & \text{otherwise},
\end{array}\right. 
\end{equation}

under the same hypotheses of Proposition \ref{prop:SDB}. Eq. \eqref{eq_tree4} is proved in Appendix \ref{app_tree}. 

For any irrelevant kernel with $|\ud{\ddot{n}}|\ge1$, by construction we have:
\[W^{(h)|q,q}_{\ud{n};\ud{\dot{n}};\ud{\ddot{n}}}=
\sum_{k=h}^{N-1} \wt{W}^{(k)|q,q}_{\ud{n};\ud{\dot{n}};\ud{\ddot{n}}} + W^{(N)|q,q}_{\ud{n};\ud{\dot{n}};\ud{\ddot{n}}},\] 
so that, combining Eq.  \eqref{eq_tree4} with the bounds for the kernels at scale $N$, we find:
\begin{equation}
\label{eq:11}
\begin{split} \|W^{(h)|q,q}_{\ud{n};\ud{\dot{n}};\ud{\ddot{n}}}\|_1^w \le &\sum_{k=h}^{N-1} 2^{k D_{sc}} \|\wt{W}^{(k)|q,q}_{\ud{n};\ud{\dot{n}};\ud{\ddot{n}}}\|_{(k)} + \|W^{(N)|q,q}_{\ud{n};\ud{\dot{n}};\ud{\ddot{n}}}\|_1^w  \\ \le &\sum_{k=h}^{N} 2^{k D_{sc}} \big(1-2^{\th-1}\big) C_0^{1+\dim\ud{n}+ \dim\ud{\dot{n}}+ \dim\ud{\ddot{n}}}  \l^d 2^{\th(k-N)}\\ \le & C_0^{1+\dim\ud{n}+ \dim\ud{\dot{n}}+ \dim\ud{\ddot{n}}} \l^d 2^{h D_{sc}} 2^{\th(h-N)},
\end{split}\end{equation}

where we used that $\|W^{(N)|q,q}_{\ud{n};\ud{\dot{n}};\ud{\ddot{n}}}\|_1^w\le 2^{N D_{sc}} \big(1-2^{\th-1}\big) C_0^{1+\dim\ud{n}+ \dim\ud{\dot{n}}+ \dim\ud{\ddot{n}}} \l^d$ as well, and that for irrelevant kernels $D_{sc}\le-1$. For $W^{(h)|0,0}_{\eset;\eset;n}$ and for the irrelevant kernels such that $|\ud{\ddot{n}}|=0$, the estimate follows the same lines as in Eq. \eqref{eq:11}, except for the absence of the \emph{short memory factor} $2^{\th(h-N)}$. Following the same steps as Eq. \eqref{eq:11}, we also find the estimates in Eq. \eqref{eq_SDB_1} for the relevant kernels $W^{(h)|0,0}_{n;\eset;\eset}$ and $W^{(h)|0,0}_{\eset;n;\eset}$, which are however unbounded in $N$, since $D_{sc}=+1$. In fact for the other non-irrelevant kernels we cannot proceed as in Eq. \eqref{eq:11}. Let us e.g. look at the kernel $\psi^2$, which is relevant. Writing $W^{(h)|1,1}_{\eset;\eset;\eset}= \wt{W}^{(h)|1,1}_{\eset;\eset;\eset}+ W^{(h+1)|1,1}_{\eset;\eset;\eset}$, we use Eq. \eqref{eq_tree4} for estimating the first term, while for the second one we exploit the first bound in Eq. \eqref{eq_IB_1}:
\[\begin{split}
&\|W^{(h)|1,1}_{\eset;\eset;\eset}\|_1^w\le \|\wt{W}^{(h)|1,1}_{\eset;\eset;\eset}\|_1^w+ \|W^{(h+1)|1,1}_{\eset;\eset;\eset}\|_1^w\le \big(1-2^{\th-1}\big) C_0 R\l 2^h + R\l 2^{-\vth(h+1)} \le C_0 R\l 2^h,
\end{split}\]
if $C_0\geq 2^{-2\theta+1}$.
The estimates for the other non-irrelevant kernels qualitatively follow the very same lines, so they will be omitted for sake of brevity. All in all one recovers the bounds in Eqs. \eqref{eq_SDB_1}- \eqref{eq_SDB_2}, with $C_0$ (the same in Proposition \ref{prop:SDB} and Eq. \eqref{eq_tree4}) independent of $R$, and $\l$ smaller than $(C_0R^4)^{-1}$, thus proving the claim of Proposition \ref{prop:SDB}.
\end{proof}

Proposition \ref{prop:SDB} shows that the multiscale expansion is well posed and well behaved provided that the relevant and marginal kernels (i.e. those with $D_{sc}\le-1$ according to Eq. \eqref{eq_SDB_3}), satisfy suitable bounds at all the higher scales. 
This property, whose proof is actually a crucial point of the present work and will be discussed in the next section, is in general a crucial point in every renormalization group analysis.

\section{Bounds for the relevant and marginal terms}
\label{sect_improved}

In the present work, in the same way as \cite{BFM09,M07} and similarly to \cite{L87}, the key underlying idea for controlling the non-irrelevant terms is to exploit the non-locality of the boson propagator (see Eq. \eqref{eq_bound_boson} and lines thereafter).

\begin{theorem}
\label{thm:IB}
For every $\vth,\th\in(0,1)$ fixed, there exist $N_0\ge1$, $R\ge1$ and $\l_0>0$ such that, for any $N\ge N_0$ and $\l\le \l_0$, the bounds in Eqs. \eqref{eq_IB_1}-\eqref{eq_IB_6} and Eqs.  \eqref{eq_SDB_1}-\eqref{eq_SDB_3} for $k=h$ hold together at every scale $0\le h\le N$.
\end{theorem}

For the implications of Theorem \ref{thm:IB} on the kernels bounds appearing in the main Theorem \ref{thm_UV}, see Subsection \ref{subsec:boundkernels}.

\medskip 

\emph{Proof of Theorem \ref{thm:IB}.} The proof is by induction over the scale index $h$ of the kernel:  we will think of $R$ entering the bounds in the inductive hypothesis as a free parameter and the idea is to show that the inductive step can be performed if $\l \leq \l_0$ with $\l_0$ small enough with respect to $R$, and $R$ large but fixed. It turns out that a working choice for $R$ and $\l_0$ is \begin{equation}R\geq \mf{C}:= \max_{j=0}^4 \mf{c}_j C_0^8, \qquad  \l_0= \min\Big\{(C_0{ R}^4)^{-1},\min_{j=0}^7 (\mf{c}'_j C_0^{12}{R})^{-1}\Big\},
\label{eq:parameterchoice}\end{equation}
where $C_0\geq 1$ is the same as in Proposition \ref{prop:SDB} while $\mf{c}_j$ and $\mf{c}'_j$ are constants independent of $C_0,R$ that will appear along the proof in the present section.

To see the origin of Eq. \eqref{eq:parameterchoice}, we start by observing that for the base of the induction, i.e. at scale $h=N$, since the potential has the explicit expression in Eq.  \eqref{def_potential}, it is possible to find constants $\mf{c}_0,\mf{c}'_0\ge1$ such that bounds in Eqs. \eqref{eq_SDB_1}-\eqref{eq_SDB_3}, \eqref{eq_IB_1}-\eqref{eq_IB_6},   hold for $\l\le (\mf{c}'_0)^{-1}$ and  $R\geq \mf{c}_0$.

 For the inductive step, we assume the claim true up to scale $h+1$, and we prove its validity at scale $h$. As a first step, we can apply Proposition \ref{prop:SDB}, which implies the validity of bounds in Eqs. \eqref{eq_SDB_1}-\eqref{eq_SDB_3} at scale $h$,  for $\l\le \l_0$ with $R,\l_0$ as above (note that by construction $\l_0\le (C_0R^4)^{-1}$). The nontrivial point is now to show the validity of bounds in Eqs. \eqref{eq_IB_1}-\eqref{eq_IB_6} at scale $h$, taking for granted Eqs. \eqref{eq_SDB_1}-\eqref{eq_SDB_3} at the same scale. The remaining part of this section is dedicated to the proof of this fact. For sake of clarity, in order to avoid confusion with constants, throughout this section we will regard $R$ and $C_0$ as generic parameters greater than 1 appearing in the bounds in Eqs. \eqref{eq_SDB_1}-\eqref{eq_SDB_3}, so we will track explicitly their dependence in the upcoming estimates. Conversely, we will use the symbols $\mf{c},\mf{c}',\mf{c}'',\dots$ to denote quantities that are independent on $R$ and $C_0$.

\subsection{The electron self-energy}
\label{ssection_improved_1}

We start by proving the bound regarding the kernel $W^{(h)|1,1}_{\eset;\eset;\eset}$ associated to the Grassmann monomial $\psi^2$, also called \it{electron self-energy}. The goal of this section is to show:

\begin{equation}
\big\|W^{(h)|1,1}_{\eset;\eset;\eset}\big\|_1^w\le \mf{c}_1 C_0^5 \l 2^{-\vth h}, \label{eq:goal_bound_psiquadro}\end{equation} 

with $C_0$ as in Proposition \ref{prop:SDB}, $\l\le (\mf{c}'_1C_0^2R)^{-1}$, and suitable constants $\mf{c}_1,\mf{c}'_1\ge1$ independent on $C_0,R$. The starting point consists of deriving some convenient representation for the kernel.

\begin{lemma}
\label{lemma_id_self_energy}
The following identity holds: 
\begin{equation}
\label{eq_15}
\begin{split}
&W^{(h)|1,1}_{\eset;\eset;\eset}(\eta_1,\eta_2)= \sum_{j=1}^4 \mc{A}_j(\eta_1,\eta_2)
\end{split}\end{equation}

where, recalling  $\mu_{N,\l}=e^{-\frac{1}{2}\l 2^{-2N}g^A_{\mu,\mu}(0)},$ $\mc{A}_j(\eta_1,\eta_2)$ are given by
\begin{equation}
\hspace{-5pt} \mc{A}_1(\eta_1,\eta_2)= 2^N\left(\mu_{N,\l} -1 \right) \int db \left( c_{b} (\eta_1,\eta_2)- \int d\eta d\eta' c_{b}(\eta_1,\eta) g^{(h,N]}_{\eta,\eta'} W^{(h)|1,1}_{\eset;\eset;\eset}(\eta',\eta_2)\right),
\end{equation}
\begin{equation}
\begin{split}
\hspace{-5pt}\mc{A}_2(\eta_1,\eta_2)&=\sum_{n\ge1} \frac{\mu_{N,\l}}{n^2}  \int db_1 db \, \l g^A_{b_1,b} \left\{ c_{b_1}(\eta_1,\eta_2) W^{(h)|0,0}_{n;\eset;\eset}(b) + \int d\eta d\eta' c_{b_1}(\eta_1,\eta) g^{(h,N]}_{\eta,\eta'} \times \right.\\
&\times\left. \left[W^{(h)|1,1}_{n;\eset;\eset}(\eta',\eta_2;b) - W^{(h)|1,1}_{\eset;\eset;\eset}(\eta',\eta_2) W^{(h)|0,0}_{n;\eset;\eset}(b) \right] \right\},
\end{split}
\end{equation}
\begin{equation}
\begin{split}
&\mc{A}_3(\eta_1,\eta_2)=\\
&=2^{-N} \sum_{n_1\ge1,n_2\ge2} \frac{\mc{N}_{n_1,n_2}}{(n_1+n_2)^2} \int db_1 db db' \l g^A_{b_1,b} \d(b_1-b')\times \\
&\times \left\{ c_{b_1}(\eta_1,\eta_2) \left[-W^{(h)|0,0}_{n_1;n_2;\eset}(b;b') + W^{(h)|0,0}_{n_1;\eset;\eset}(b) W^{(h)|0,0}_{\eset;n_2;\eset}(b') \right]+  \int d\eta d\eta' c_{b_1}(\eta_1,\eta) g^{(h,N]}_{\eta,\eta'}\times\right. \\
&\times \left.\left[-W^{(h)|1,1}_{n_1;n_2;\eset}(\eta',\eta_2;b;b') + W^{(h)|1,1}_{n_1;\eset;\eset}(\eta',\eta_2;b) W^{(h)|0,0}_{\eset;n_2;\eset}(b')+ W^{(h)|0,0}_{n_1;\eset;\eset}(b) W^{(h)|1,1}_{\eset;n_2;\eset}(\eta',\eta_2;b') \right] \right.\\
&\left. -\int d\eta d\eta' c_{b_1}(\eta_1,\eta) g^{(h,N]}_{\eta,\eta'} W^{(h)|1,1}_{\eset;\eset;\eset}(\eta',\eta_2) \left[  -W^{(h)|0,0}_{n_1;n_2;\eset}(b;b') +  W^{(h)|0,0}_{n_1;\eset;\eset}(b) W^{(h)|0,0}_{\eset;n_2;\eset}(b') \right]\right\},
\end{split}\label{eq_A3}
\end{equation}
\begin{equation}
\begin{split}
\label{eq_A4}
\mc{A}_4(\eta_1,\eta_2)&= \sum_{n\ge2} 2^{-N}\int db_1 db \delta(b_1-b)\left\{ c_{b_1}(\eta_1,\eta_2) W^{(h)|0,0}_{\eset;\eset;n}(b)+ \int d\eta d\eta' c_{b_1}(\eta_1,\eta) \times \right.\\
&\left.  \times g^{(h,N]}_{\eta,\eta'}\left[ W^{(h)|1,1}_{\eset;\eset;n}(\eta',\eta_2;b) - W^{(h)|1,1}_{\eset;\eset;\eset}(\eta',\eta_2) W^{(h)|0,0}_{\eset;\eset;n}(b) \right]\right\},
\end{split}
\end{equation}

where $\mc{N}_{n_1,n_2}$ are suitable positive coefficients, such that $\mc{N}_{n_1,n_2}\le \frac{(n_1+n_2)^2}{n_1^2 n_2^2}$.

\end{lemma}

The proof of Lemma \ref{lemma_id_self_energy} given in Appendix \ref{app_identities}. 

\paragraph{Notation.} 
In order to show Eq. \eqref{eq:goal_bound_psiquadro} we find convenient to exploit a graphical representation associated to the decomposition in Eq. \eqref{eq_15}, as shown in Fig. \ref{fig_self_energy_1}.
 For $\sharp=a,b,\dots, o$, we will denote by $\val(\sharp)(\eta,\eta')$, the analytical value of the diagram $(\sharp)$ in Fig. \ref{fig_self_energy_1}, i.e. the corresponding term in the decomposition of Eq. \eqref{eq_15} and by $\| \val (\sharp) \|_1^w$ its weighted norm, cf.   Eq. 
 \eqref{def_Lp_norm}.   Even though there are 15 diagrams to discuss, we will see that only a fraction of them is independent, in a suitable sense (cf. Remark \ref{rmk_2}).

\begin{figure}[ht]
\centering \hspace{-20pt}
\includegraphics[scale=0.8]{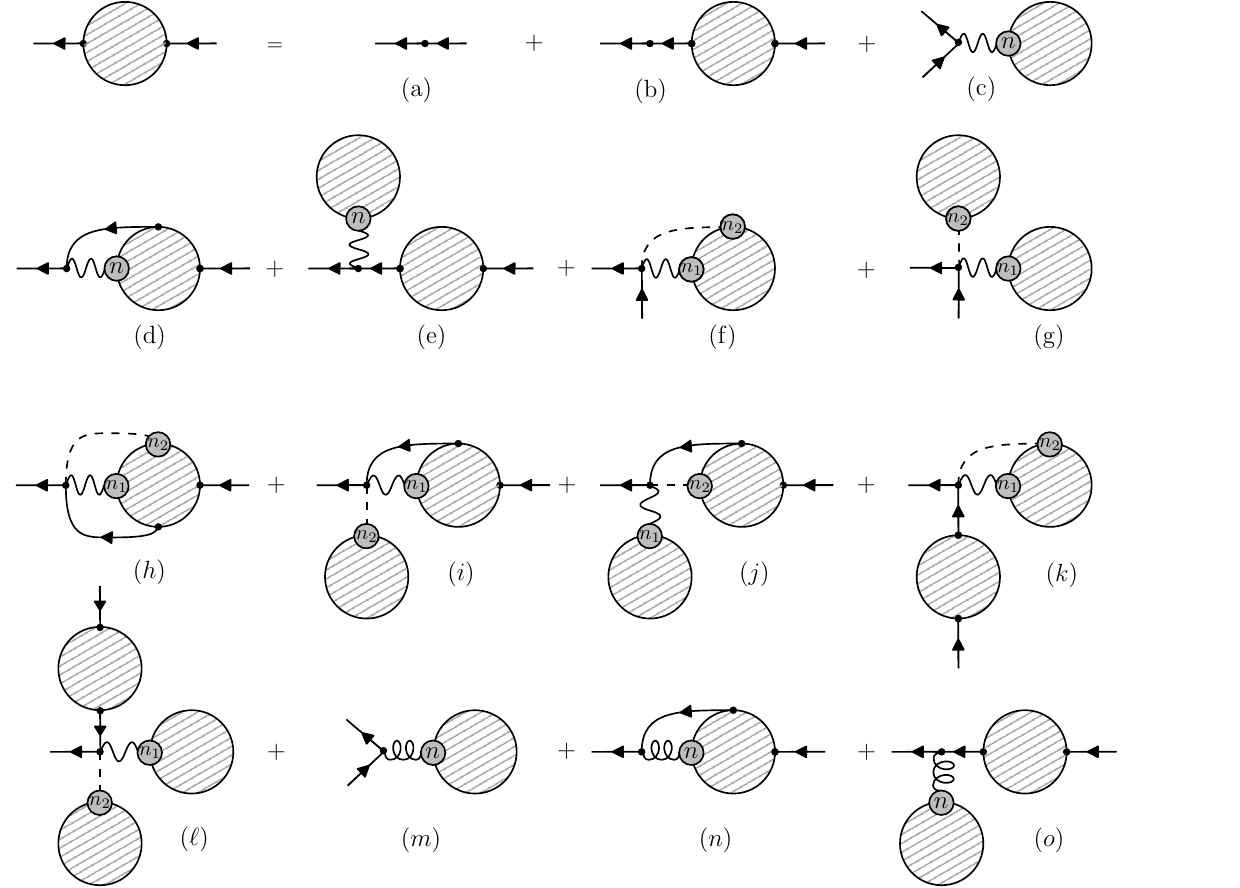}
\caption{Graphical representation of Eq. \eqref{eq_15}. The graphs are numbered as in the order they appear in the expansion of Eq.  \eqref{eq_15}: graphs $(a),(b)$ are relative to $\mc{A}_1$,  $(c)-(e)$ to $\mc{A}_2$, $(f)-(l)$ to $\mc{A}_3$ and  $(m)-(o)$  to $\mc{A}_4$. Solid and wiggly lines, with dotted endpoints, represent fermion and boson propagators respectively, while dashed and curly lines, with black dotted endpoints, represent delta functions. Blobs (big North East line patterned circles) represent kernels, associated to external fields distinguished by the type of attached lines to them: each solid, wiggly, dashed and curly line attached to a blob identifies a $\psi,G,\dot{G}$ and $\ddot{G}$ field respectively.    }
\label{fig_self_energy_1}
\end{figure}

\begin{remark}
\label{rmk_dotG}
The decomposition in Fig. \ref{fig_self_energy_1} is paradigmatic for understanding the combinatorial difficulty emerging from the lattice regularization of the theory. In the continuum version of this model \cite{F10}, the only term in the r.h.s. would be diagram $(d)$ with $n=1$ (see \cite[Fig.2]{F10}); similarly in \cite{M22}, where only the terms $\psi^2,\psi^4$ and $\psi^2G(J)$ are retained in the fermionic interaction (cf. Eq. \eqref{eq:2}), one would get only diagrams $(a),(b),(d)$ and $(e)$ with $n=1$. In the present work, the multitude of terms in the expansion is due to the form of the \virg{bare kernels} $v_n$, Eq. \eqref{eq_16a}, which involve sums over trees of arbitrary order. If we adopt the representation in Fig. \ref{fig_diag_2}.(a) for the bare kernels $v_n$, we see that terms like the one in Fig. \ref{fig_diag_2}.(b) (corresponding to diagram $(h)$ of Fig. \ref{fig_self_energy_1}) can emerge as contributions to the kernel $W^{(h)|1,1}_{\eset;\eset;\eset}$.

\begin{figure}[h]
    \centering
\includegraphics[width=0.9\textwidth]{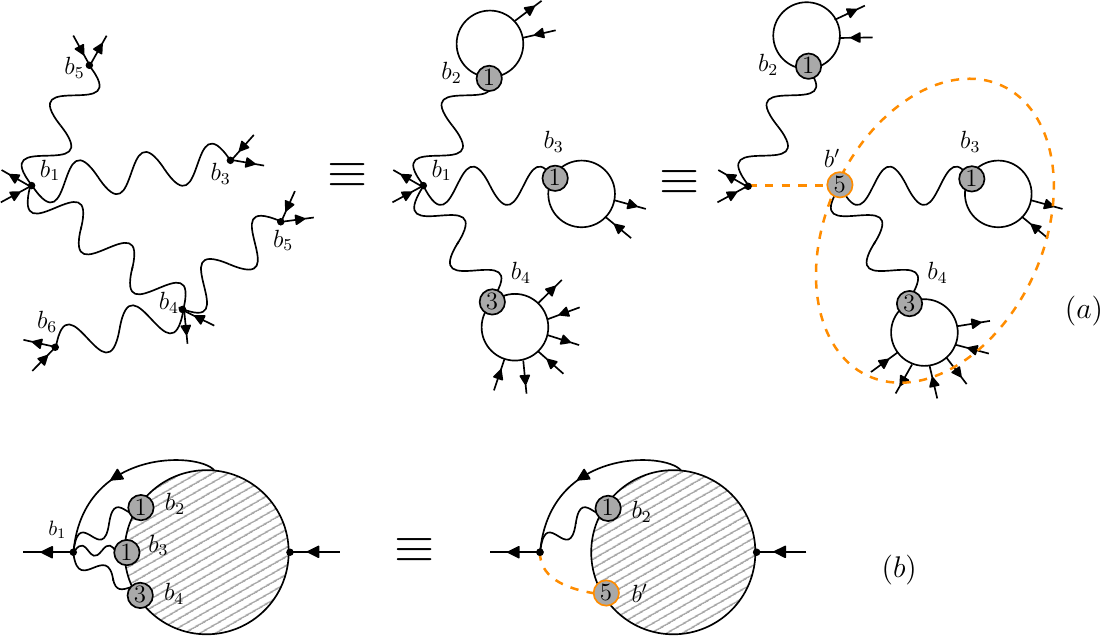}
    \caption{(a): synthetic representation for $v_n(\ud{b})O^n_{\ud{b}}$, obtained from Fig. \ref{fig_kernel} by first collapsing each sub-tree rooted in $b_2,b_3$ and $b_4$ into one dark blob (the label $(n)$ stands for the order of the sub-tree), and then by thinking of the sub-trees rooted at $b_3$ and $b_4$ as one only tree rooted at $b'$, with $b'=b$ (the dashed line stands for a Dirac delta).
    (b): contribution to the kernel $W^{(h)|1,1}_{\eset;\eset;\eset}$ arising from the bare kernel in (a).}
    \label{fig_diag_2}
\end{figure}

The blob in the l.h.s. of Fig \ref{fig_diag_2}.(b) should be regarded as

\[ \frac{\d^3 \mc{V}^{(h)}}{\d G^{(1)}_{b_2}\d G^{(1)}_{b_3} \d G^{(3)}_{b_4}}(\mb{0}) \equiv 3! W^{(h)|0,0}_{(1,1,3);\eset;\eset}(b_2,b_3,b_4). \]

The factor 3 in the factorial corresponds to the number of $G$ fields, which equals the number of wiggly lines attached to the blob in the diagram. Now, since in general we must expect an arbitrary number (say $s$) of wiggly lines in the l.h.s. of Fig. \ref{fig_diag_2}.(b), following this expansion we would end up with unbounded factorials (say $s!$), which would induce a combinatorial divergence of the expansion. The way we overcome this problem is by rewriting the diagram in the l.h.s. of Fig.  \ref{fig_diag_2}.(b) as in the r.h.s., namely by \virg{hiding} the bunch of wiggly lines at the end of the dashed line (which stands for a Dirac delta): in this way no combinatorial problem arises, being the number of points of the blob bounded. The role of the auxiliary source $\dot{G}$ is exactly to perform analytically what here is explained here by a graphical intuition.
\end{remark}

\paragraph{Diagram $(a)$.}
This is simply the bare quartic kernel, whose weighted norm can be bounded by (cf. Eq. \eqref{eq_bound_boson}) \[ \|\val(a)\|_1^w\leq 2^N\|c\|_1^w|\mu_{N,\l}-1|\leq 12 K' \l N  2^{-N}\leq \mf{c}_{1,1} \l 2^{-\vth N},\] 
where we used definitions in Eqs. \eqref{def_bare_vertex}, \eqref{def_Lp_norm} (recall $r=1$, cf. below Eq. \eqref{eq:IA}) to get $\|c\|_1^w\leq 6 e^{\hat \kappa 2^{-\frac N2-1}}\leq 12$ for $N$ large enough, the explicit expression of $\mu_{N,\l}$ (below Eq. \eqref{eq:15}) and the bound on $|g^A_{\mu,\mu}(0)|\leq \|g^A\|_\infty \leq K' N$ (cf. Eq. \eqref{eq:bound:boson:b}).

\paragraph{Diagram $(b)$.} The bound in this case is \[\|\val(b)\|_1^w\leq 2^N\|c\|_1^w|\mu_{N,\l}-1|\|g^{(h,N]}\|_1^w \|W^{(h)|1,1}_{\eset;\eset;\eset}\|_1^w\leq \tilde K C_0 R \l 2^N\|c\|_1^w|\mu_{N,\l}-1|,\] where we used that for some\footnote{One can choose, e.g., $\tilde K=320 K \hat \kappa^{-2}$,  cf. Eq. \eqref{eq_bound_propagator}.} $\tilde K>0$,   $\|g^{(h,N]}\|_1^w\le \tilde K 2^{-h}$  and that, by assumption, if $\l \leq (C_0R^4)^{-1}$ we have $\|W^{(h)|1,1}_{\eset;\eset;\eset}\|_1^w\le C_0 R\l 2^h$. Taking $\l \leq \min\{(C_0R^4)^{-1}, (RC_0 \tilde K)^{-1}\}$ we see that $\val(b)$ can be bounded as $\val(a)$, i.e.
\begin{equation}
    \|\val(b)\|_1^w \leq \mf{c}_{1,1} \l 2^{-\vartheta N}.
\end{equation}
The constraint on $\l$ implies that  $\mf{c}'_1 \geq  \tilde K$.

\begin{remark}
\label{rmk_2}
Diagram (b) is the first example that not all the graphs of Fig. \ref{fig_self_energy_1} have independent bounds: each time that in a diagram we have a ``dressed fermionic line'' (first diagram in the first line of Fig. \ref{fig_self_energy_1}), in order to obtain  an upper bound, we can graphically delete the dressed fermionic line  
at the price of shrinking $\l$, which, in the end, reflects in changing the values of the  constant $ \mf{c}'_1$ in Eq. \eqref{eq:goal_bound_psiquadro}.  In this way, we can say that graphs $(b)$, $(k)$ and $(o)$ are apriori bounded by $(a)$, $(f)$ and $(m)$ respectively.
\end{remark}

\paragraph{Diagram $(c)$.} In this case $\val(c)(\eta_1,\eta_2)=0$. This is actually a consequence of the \it{charge conjugation} symmetry (cf. Lemma \ref{lemma:symm}), which establishes the following cancellation:

\begin{equation}
\label{eq:Furry}
\int_{\L_B^p} db'_1\dots db'_p  W^{(h)|0,0}_{(n_1,\dots,n_p);\eset;\eset}(b'_1,\dots,b'_p)\prod_{j=1}^p g^A_{b_j,b'_j}=0, 
\end{equation}

for every $p\ge1$ odd, $n_1,\dots,n_p\ge1$ and $b_1,\dots,b_p\in\L_B$. In particular for $p=1$ we find $\val(c)(\eta_1,\eta_2)=0$. 
The proof of Eq. \eqref{eq:Furry} follows immediately by noting that, in force of the \it{charge conjugation} symmetry (cf. Lemma \ref{lemma:symm}), 

\[W^{(h)|0;0}_{(n_1,\dots,n_p);\eset;\eset}(b'_1,\dots,b'_p) = W^{(h)|0;0}_{(n_1,\dots,n_p);\eset;\eset}(\ov{b}'_1,\dots,\ov{b}'_p),\]

where recall that if $b=(x,\mu,\e)$, $\ov{b}\equiv (x,\mu,-\e)$. Using that $g^A_{b,b'}=-g^A_{b,\ov{b}'}$ and the fact that $p$ is odd, we find:
\[\begin{split}
& \int d\ud{b}' \prod_{j=1}^p g^A_{b_j,b'_j} \times W^{(h)|0;0}_{\ud{n};\eset;\eset}(\ud{b}')= \int d\ud{b}' \prod_{j=1}^p g^A_{b_j,b'_j} \times W^{(h)|0;0}_{\ud{n};\eset;\eset}(\ov{\ud{b}}')=\\
&=\int d\ud{b}' \prod_{j=1}^p g^A_{b_j,\ov{b'_j}} \times W^{(h)|0;0}_{\ud{n};\eset;\eset}(\ud{b'})= - \int d\ud{b}' \prod_{j=1}^p g^A_{b_j,b'_j} \times W^{(h)|0;0}_{\ud{n};\eset;\eset}(\ud{b}'),
\end{split} \]
as desired.

\paragraph{Diagram $(d)$.} This is the first case where we must use the non-locality of $g^A$.  
The analytical value of the diagram is
\[ \text{Val}(d)(\eta_1,\eta_2)= \sum_{n\ge1} \frac{\mu_{N,\l}}{n^2}  \int dbdb' \int d\eta d\eta' c_b(\eta_1,\eta) \l g^A_{b,b'} g^{(h,N]}_{\eta,\eta'} W^{(h)|1,1}_{n;\eset;\eset}(\eta',\eta_2;b'), \]

which can be bounded as
\[\begin{split} 
&\|\val(d)\|_1^w \le \\
&\le\sum_{n\ge1} \frac{\l}{n^2}\int \frac{db db'}{L^2} \int d\eta d\eta' d\eta_1 d\eta_2 e^{\frac{\kappa}{2}\sqrt{\delta^D(\eta,\eta',\eta_1,\eta_2,b,b')}}|c_b(\eta_1,\eta)| |g^A_{b,b'}| |g^{(h,N]}_{\eta,\eta'}| |W^{(h)|1,1}_{n;\eset;\eset}(\eta',\eta_2;b')| \\ 
&\le \sum_{n\ge1} \frac{\l}{n^2}\int \frac{db db'}{L^2} \int d\eta d\eta' d\eta_1 d\eta_2 \left(e^{\frac{\kappa}{2}\sqrt{\delta^D(b,\eta,\eta_1)}}|c_b(\eta_1,\eta)|\right) \left(e^{\frac{\kappa}{2}\sqrt{\delta^D(b-b')}}|g^A_{b,b'}|\right) |g^{(h,N]}_{\eta,\eta'}| \times\\
&\times \left(e^{\frac{\kappa}{2}\sqrt{\delta^D(\eta',\eta_2,b')}} |W^{(h)|1,1}_{n;\eset;\eset}(\eta',\eta_2;b')|\right)=\\
&= \sum_{n\ge1} \frac{\l}{n^2}\int \frac{db db'}{L^2} \int  d\eta' d\eta_2 \tilde{g}^{(h,N]}_{b,\eta'} \left(e^{\frac{\kappa}{2}\sqrt{\d^D(b,b')}}|g^A_{b,b'}|\right)   \left(e^{\frac{\kappa}{2}\sqrt{\delta^D(\eta',\eta_2,b')}} |W^{(h)|1,1}_{n;\eset;\eset}(\eta',\eta_2;b')|\right),
\end{split}\]

with $\tilde{g}^{(h,N]}_{b,\eta'}:= \int d\eta d\eta_1 e^{\frac{\kappa}{2}\sqrt{\delta^D(b,\eta,\eta_1)}}|c_b(\eta_1,\eta)| |g^{(h,N]}_{\eta,\eta'}|$.
We proceed by applying the H{\"o}lder's inequality with weights $p=\frac{2}{1-\vth}$ and $p'=\frac{2}{1+\vth}$:
\[\begin{split}
&\|\text{Val}(d)\|_1^w\le  \sum_{n\ge1} \frac{\l}{n^2} \|g^A\|_{p}^w \|\tilde{g}^{(h,N]}\|_{p'} \|W^{(h)|1,1}_{n;\eset;\eset}\|_1^w\le \sum_{n\ge1} \frac{\l}{n^2} \tilde K_{\vartheta}C_0^2 2^{-\vth h} \l^{\frac{n-1}{2}} \le \l \mf{c}_{1,2} C_0^2 2^{-\vth h},
\end{split}\]
where we used Eq.  \eqref{eq_SDB_3} with $\l\le (C_0R^4)^{-1}$, denoted with 
$\tilde K_{\vth}$  the constant coming from $\|\cdot\|_{p},\|\cdot\|_{p'}$ (by using  Eq. \eqref{eq_bound_boson}) and used the fact that $\tilde{g}^{(h,N]}$ has the same dimensional bounds as $g^{(h,N]}$, being $c_b(\eta_1,\eta)$ local; finally $  \mf{c}_{1,2}=\tilde K_\theta  \pi^2/6 $ is independent of $C_0,R$.

\paragraph{Diagram $(f)$.} This is the first diagram with a dashed line,  corresponding to a delta function (cf. the first term in $\mc{A}_3$, Eq. \eqref{eq_A3}). In bounding such term, this forces us to take the $L^{\infty}$ norm of the boson propagator associated with the wiggly line:  
\[\begin{split}
& \|\val(f)\|_1^w \leq \\ & \sum_{\substack{n_1\ge1 \\ n_2\ge2}} \frac{\mc{N}_{n_1,n_2}}{(n_1+n_2)^2} 2^{-N} \l\|g^A\|_{\infty} \|W^{(h)|0,0}_{n_1;n_2;0}\|_1^w\le \sum_{\substack{n_1\ge1 \\ n_2\ge2}} \frac{\l}{n_1^2 n_2^2} 2^{-N}K'N C_0^3\l^{\frac{n_1+n_2-2}{2}} \le \mf{c}_{1,3}C_0^3\l^{3/2} 2^{-\vth N},
\end{split}\]

where we used Eqs. \eqref{eq_bound_boson}, \eqref{eq_SDB_3}, that $2^{-N}N\leq 2^{-\vartheta N}$ for any $\vartheta\in(0,1)$ if $N$ is large enough, and the fact that $\mc{N}_{n_1,n_2}\le \frac{(n_1+n_2)^2}{n_1^2 n_2^2}$; $\mf{c}_{1,3}=K' (\pi^2/6)^2$ does not depend on $C_0,R$.
\begin{remark}
Diagram $(f)$ is the first nontrivial case (after diagram $(a)$) of a class of graphs which not only have the desired bound at any fixed scale $h$, but they are actually vanishing as $N\to \infty.$ In view of proving Theorem \ref{thm:IB} it is sufficient to replace the factor $2^{-\vth N}$ above with $2^{-\vth h}$; however, it is worth mentioning that tracking small terms in $N$ would be crucial for the analysis of the $N \to \infty$ limit of the theory (see Appendix \ref{app:limit} for an overview).

\label{rmk:o(N)diagrams}
\end{remark}
\paragraph{Diagram $(h)$.}  
Again, due to the delta function associated with the dashed line, we are forced to take the $L^{\infty}$ norm of both $g^A$ and $g^{(h,N]}$. We find:
\begin{equation}
\label{eq_17}
\|\val(h)\|_1^w\le \sum_{\substack{n_1\ge 1 \\  n_2\ge2}} 2^{-N} \frac{\mc{N}_{n_1,n_2}}{(n_1+n_2)^2} \l \|g^A\|_{\infty} \|g^{(h,N]}\|_{\infty} \|W^{(h)|1,1}_{n_1;n_2;\eset}\|_1^w\le C_0^3 K N \mf{c}_{1,3} \l 2^{-h}, 
\end{equation}
with $\mf{c}_{1,3}$ as in the previous item, and we used that $\|g^{(h,N]}\|_\infty \leq 2K 2^N$ (cf. Eq. \eqref{eq_bound_propagator}) and that $\|W^{(h)|1,1}_{n_1;n_2;\eset}\|_1^w\leq C_0^3 \l^{\frac{n_1+n_2-2}{2}}2^{-h}$, cf. Eq. \eqref{eq_SDB_3}.  This apparent  logarithmic divergence is related to the definition of the external source $\dot G$, which imposes to take the $L^{\infty}$ norm of both $g^A$ and $g^{(h,N]}$. We will show at the end of this paragraph (see  Eq. \eqref{eq_18} and below) how to solve this divergence; for the moment let us continue with the other graphs of Fig. \ref{fig_self_energy_1}.

\paragraph{Diagram $(i)$.} The analytical value in this case is:
\begin{equation}
\label{eq_39}
\begin{split}
&\val (i)(\eta_1,\eta_2) = \\
&\sum_{\substack{n_1\ge1\\ n_2\ge2}} \frac{\mc{N}_{n_1,n_2}}{(n_1+n_2)^2} 2^{-N} \int d\eta d\eta' \int db db' c_b(\eta_1,\eta) \l g^A_{b,b'} g^{(h,N]}_{\eta,\eta'} W^{(h)|1,1}_{n_1;\eset;\eset}(\eta',\eta_2;b') W^{(h)|0,0}_{\eset;n_2;\eset}(b)
\end{split}
\end{equation}
and by noting that from translation invariance and Items \ref{item:symm:2}, \ref{item:symm:3} of  Lemma \ref{lemma:symm}, one has $W_{\emptyset,n_2,\emptyset}(b)=W_{\emptyset,n_2,\emptyset}((0,0,+))$, the bound for this graph can be deduced from the one for graph $(d)$:
\[\begin{split}
\|\text{Val}(i)\|_1^w & \le \sum_{\substack{n_1\ge1\\ n_2\ge2}} \frac{2^{-N}\mc{N}_{n_1,n_2}} {(n_1+n_2)^2}  \|W^{(h)|0,0}_{\eset;n_2;\eset}\|\int \frac{d\eta_1 d\eta_2}{L^2} e^{\frac{\kappa}{2}\sqrt{\d^D(\eta_1,\eta_2)}} \times\\
&\times \left| \int d\eta d\eta' \int db db' c_b(\eta_1,\eta) \l g^A_{b,b'} g^{(h,N]}_{\eta,\eta'} W^{(h)|1,1}_{n_1;\eset;\eset}(\eta',\eta_2;b') \right| \le\\
&\le \sum_{n_2\ge2} \frac{1}{n_2^2} 2^{-N} \|W^{(h)|0,0}_{\eset;n_2;\eset}\|_1^w |\mu_{N,\l}^{-1}|\|\text{Val}(d)\|_1^w\leq C_0^3 (\pi^2/3) \mf{c}_{1,2} \l^{3/2} 2^{-\vth h} 
\end{split}\]

where we used the bound of diagram $(d)$, again that $\mc{N}_{n_1,n_2}\le (n_1+n_2)^2 (n_1n_2)^{-2}$, that $|\mu_{N,\l}^{-1}|\leq 2$ for $N$ large, and that $\|W^{(h)|0,0}_{\eset;n_2;\eset}\|_1^w\le C_0 2^N\l^{\frac{n_2-1}{2}}$ by Eq. \eqref{eq_SDB_1}. Note that imposing also $\l \leq (C_0^2\pi^4/9)^{-1}$ one has $\|\val(i)\|_1^w\leq C_0^2\mf{c}_{1,2}\l 2^{-\vth h}$ which is the same as case $(d).$ 

\begin{remark}
\label{rmk_1} We observe a second reduction of the number of diagrams to bound related to the presence of a \emph{dashed tadpole graph}, i.e. a blob with only one dashed external line, as in $(g),(i),(\ell).$ By possibly shrinking $\l$, i.e. changing the value of $\mf{c}'_1$ in Eq. \eqref{eq:goal_bound_psiquadro}, one has that $\val(g),\val(i),\val(\ell)$ can be bounded exactly as $\val(e),\val(d),\val(c)$ respectively (in particular $\val(g)=0$).
\end{remark}

\paragraph{Diagram $(m)$.} We have to bound the first of Eq. \eqref{eq_A4}  which immediately gives 
\[\|\val(m)\|_1^w \leq \sum_{n\ge2}2^{-N} \|c\|_1^w \|W^{(h)|0,0}_{\eset;\eset;n}\|_1^w\le 12 \sum_{n\ge2} 2^{-N} C_0\l^{\frac{5}{8}n} \le 12 C_0 \l^{5/4} 2^{-\vth h},\]
where we used $\|c\|_1^w\leq 12 $ as in case $(a)$, Eq. \eqref{eq_SDB_1} and $\l \leq 1/2^{8/5}$ so that $(1-\l^{5/8})^{-1}\leq 2$: this implies that $\mf{c}'_1 \geq 2^{8/5}.$

\paragraph{Diagram $(n)$.} The curly line is again associated to a delta function, which forces us to take the $L^{\infty}$ norm of the fermion propagator. The bound is:
\[\|\val(n)\|_1^w\leq \sum_{n\ge2}2^{-N}\|g^{(h,N]}\|_{\infty} \|W^{(h)|1,1}_{\eset;\eset;n}\|_1^w\le \sum_{n\ge2}2K C_0^2\l^{\frac{5}{8}n} 2^{\th(h-N)}2^{-h}\le 4K  C_0^2\l^{5/4}2^{-\vth h},\]
where we used Eq.  \eqref{eq_SDB_3} and the same hypothesis on $\l$ of the previous item.

\paragraph{Improvement for $(h)$.}
Finally, we return to diagram $(h)$.
The way to overcome the apparent divergence is to extract a second line, i.e. to perform an expansion for the kernel $W^{(h)|1,1}_{n_1;n_2;\eset}$: this allows to solve the apparent logarithmic divergence with an expansion in terms of a finite number of extra diagrams, each having good dimensional estimates. In analogy with Eq. \eqref{eq_15}, it is possible to obtain the following identity, proved in Appendix \ref{app_identities}:
\begin{equation}
\label{eq_18}
\begin{split}
&W^{(h)|1,1}_{n_1;n_2;\eset}(\eta,\eta';b_1;b_2)=\\
&= \frac{n_2^2}{(n_2-1)^2}\mu_{N,\l} \int db \l g^A_{b_2,b}  2 W^{(h)|1,1}_{(n_1,n_2-1);\eset;\eset}(\eta,\eta';b_1,b)+\\
&-  \sum_{\substack{n'_1\ge1, n'_2\ge2\\ n'_1+n'_2=n_2}} 2^{-N} \mc{N}_{n'_1,n'_2}  \int db'_1 db'_2 \l g^A_{b_2,b'_1} \d(b_2-b'_2) \left[2 W^{(h)|1,1}_{(n_1,n_1');n_2';\eset}(\eta,\eta';b_1,b'_1;b'_2)+ \right.\\
&\left.- 2 W^{(h)|1,1}_{(n_1,n_1');\eset;\eset}(\eta,\eta';b_1,b'_1) W^{(h)|0,0}_{\eset;n_2';\eset}(b'_2) - 2 W^{(h)|0,0}_{(n_1,n_1');\eset;\eset}(b_1,b'_1) W^{(h)|1,1}_{\eset;n_2';\eset}(\eta,\eta';b'_2)+ \right.\\
&\left.- W^{(h)|1,1}_{n'_1;\eset;\eset}(\eta,\eta';b'_1) W^{(h)|0,0}_{n_1;n_2';\eset}(b_1;b'_2)- W^{(h)|0,0}_{n'_1;\eset;\eset}(b'_1) W^{(h)|1,1}_{n_1;n_2';\eset}(\eta,\eta';b_1;b'_2)\right],
\end{split}\end{equation}

which admits the graphical representation in Fig. \ref{fig_self_energy_2}. 

\begin{figure}[h]
\centering
\includegraphics[scale=0.8]{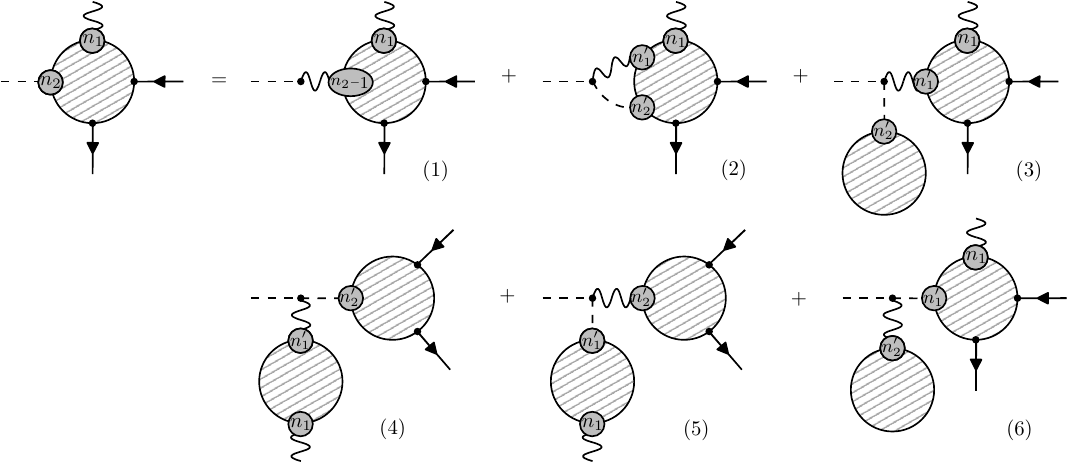}
\caption{Graphical representation of Eq. \eqref{eq_18}.}
\label{fig_self_energy_2}
\end{figure}

Let us denote by $(h1),(h2),\dots, (h6)$  the graphs obtained from $(h)$ (Fig. \ref{fig_self_energy_1}) by replacing $W^{(h)|1,1}_{n_1;n_2;\eset}$ with each  of the terms in its expansion $(1),\dots,(6)$ (as in Fig. \ref{fig_self_energy_2}). Similarly we denote by $\val(h1),\dots,\val(h6)$ their analytical value, as in the notation below Eq. \eqref{eq_A4}. 
\begin{figure}
    \centering
{\includegraphics[scale=0.85]{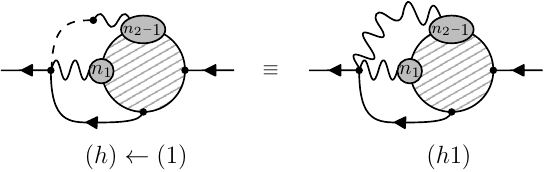}}
\caption{Diagram $(h1)$ obtained by plugging $(1)$ in $(h)$ of Fig. \ref{fig_self_energy_1}.}    \label{fig:placeholder}
\end{figure}
\paragraph{Diagram $(h1)$.}
Explicitly,
\begin{equation}
\begin{split}&
-\val(h1)=\\ & \sum_{\substack{n_1\geq 1\\ n_2\geq 2}}  \frac{2^{-N+1}\mu_{N,\l} n_2^2\, \mathcal{N}_{n_1,n_2}}{(n_1+n_2)^2(n_2-1)^2} \int db_1 db d \tilde b d \eta d\eta' \l^2 g^A_{b_1,b} g^A_{b_1,\tilde b} c_{b_1}(\eta_1,\eta)  g^{(h,N]}_{\eta_1,\eta'} W^{(h)|(1,1)}_{(n_1,n_2-1);\emptyset;\emptyset} (\eta',\eta_2;b;b'), \nonumber 
\end{split}
\end{equation}
which can be bounded by taking the $L^1$ norm of the fermionic propagator and take the $L^{\infty}$ norm of the two boson propagators:
\[\begin{split}
\|\val(h1)\|_1^w & \leq \sum_{\substack{n_1\ge1\\n_2\ge2}} \frac{2^{-N+1}|\mu_{N,\l}| n_2^2\, \mc{N}_{n_1,n_2} }{(n_1+n_2)^2 (n_2-1)^2}  \|c\|_1^w \l^2 \|g^A\|_{\infty}^2\|g^{(h,N]}\|_1^w \|W^{(h)|1,1}_{(n_1,n_2-1);\eset;\eset}\|_1^w\le\\
&\leq 24 K'^2 \tilde K C_0^3\l^2 2^{-2h}N^22^{-N} \sum_{n_1,n_2\ge1} \frac{1}{n_1^2n_2^2}\l^{\frac{n_1+n_2-2}{2}}  \le C_0^3\mf{c}_{1,4}\l^2 2^{-\theta N} 2^{-2h},
\end{split}\]
where $\mf{c}_{1,4}=(4/3)\pi^4 K'^2 \tilde K $ and the constants are the same appearing in the previous items.

\paragraph{Diagrams $(h2)-(h6)$.}
Observe that each of the remaining terms in the r.h.s. of Eq. \eqref{eq_18} have a dimensional factor $2^{-N}$, which graphically corresponds to the leftmost node with two dashed lines and a wiggly line in diagrams $(2)-(6)$  in Fig. \ref{fig_self_energy_2}.
In this cases it is enough to note that each of such diagrams $(2)-(6)$ have good bounds for their analytical values:
\begin{equation}
\label{eq_25}
\begin{split}
& \big\|\val(2)_{n_1,n_2}\big\|_1^w\le \sum_{\substack{n_1'\ge1, n_2'\ge2\\ n_1'+n_2'=n_2}} 2^{-N}\mc{N}_{n_1',n_2'} \l \|g^A\|_{\infty} \big\|W^{(h)|1,1}_{(n_1,n_1');n_2';\eset}\big\|_1^w \le \\ & \sum_{\substack{n_1'\ge1, n_2'\ge2\\ n_1'+n_2'=n_2}} K'N2^{-N}\l \frac{n_2^2}{n_1'^2n_2'^2}  \l^{\frac{n_1+n_1'+n_2'-3}{2}} \big\|W^{(h)|1,1}_{(n_1,n_1');n_2';\eset}\big\|_1^w \le C_0^4(4\pi^2/3) K'2^{-N-2h}N \l^{\frac{n_1+n_2-1}{2}},
\end{split}
\end{equation}

where we  used Eq. \eqref{eq_SDB_3} and the fact that $\sum_{\ell=1}^{n-2} \frac{n^2}{\ell^2(n-\ell)^2}\leq 4\pi^2/3$. It is easy to check that also the other graphs  $(3)-(6)$, in the same way, are bounded as \begin{equation}
    \sup_{j=3}^6\|\val(j)_{n_1,n_2}\|_1^w \leq C_0^5 (4\pi^2/3) K' N2^{-N} \l^{\frac{n_1+n_2-1}{2}}.  
\end{equation} Hence when inserted in graph $(h)$ of Fig. \ref{fig_self_energy_1}, they produce a bound:
\begin{equation}
    \begin{split}
\sup_{j=2}^6\|\val(hj)\|_1^w & \leq \sum_{\substack{n_1\geq 1\\ n_2 \geq 2}} \frac{2^{-N} \mc{N}_{n_1,n_2}}{(n_1+n_2)^2} \l \|g^A\|_\infty \|g^{(h,N]}\|_\infty \sup_{j=2}^6 \|\val(j)_{n_1,n_2}\|_1^w  \\ & \leq  C_0^5 \mf{c}_{1,5}  \l^{3/2} N2^{-N} 
    \end{split}
\end{equation}
with $\mf{c}_{1,5}=K'^2K (\pi^2/3)^3$. \medskip Collecting the bounds for diagrams $(a)$ to $(o)$ in Fig. \ref{fig_self_energy_1} one realizes that all the graphs except $(d)$ admit a bound with at a factor $\l^{1+\eta}$ with $\eta \geq 1/4$. The extra small factor $\l^{1/4}$ 
can be used to reduce, in each diagram but $(d)$, the overall constant to $\mf{c}_{1,2}/14$ at the price of choosing $\mf{c}_1'$ depending on $\mf{c}_{1,1},\mf{c}_{1,3},\mf{c}_{1,4}$ and such that $\mf{c}'_1\geq \max\{K',2^{8/5}\}$ (see diagrams $(b)$ and $(m)$) . Summing all the contributions we obtain that if $\l\le (\mf{c}'_1 C_0^2 R)^{-1}$   
\begin{equation}
\label{eq_19}
\|W^{(h)|1,1}_{\eset;\eset;\eset}\|_1^w\le 2\mf{c}_{1,2} C_0^5 \l 2^{-\vth h},
\end{equation}
thus proving Eq. \eqref{eq:goal_bound_psiquadro} with $\mf{c}_1 \geq 2\mf{c}_{1,2}$, with $\mf{c}_{1,2}$ appearing in diagram $(d).$ 

\subsection{The lowest-degree vertex function}
\label{ssection_improved_2}

Now we focus on the kernel $G^{(1)}\psi^2$. Recalling the definition of the bare vertex kernel $c$,  in Eq. \eqref{def_bare_vertex},  the goal is to prove that  
\begin{equation}
\label{eq:18}
\big\|W^{(h)|1,1}_{1;\eset;\eset}-c\big\|_1^w\le C_0^{6} \mf{c}_2\l,
\end{equation}

with $C_0$ as in Proposition \ref{prop:SDB}, $\l\le (\mf{c}'_2C_0^2R)^{-1}$, and suitable constants $\mf{c}_2,\mf{c}'_2\ge1$ independent on $C_0,R$.

Again, the strategy is to derive a convenient identity for $W^{(h)|1,1}_{1;\eset;\eset}$. In analogy with Lemma \ref{lemma_id_self_energy}, by applying $\d/\d\psi^-_{\eta_2}$ to Eq. \eqref{eq_id_4},   one can prove that 
\begin{equation} 
\label{eq:14}
W^{(h)|1,1}_{1;\eset;\eset}(\eta_1,\eta_2;b)= \mu_{N,\l} c_b(\eta_1,\eta_2)+ \mf{D} W^{(h)|1,1}_{1;\eset;\eset}(\eta_1,\eta_2;b), 
\end{equation}

with $\mu_{N,\l}$ as below Eq. \eqref{eq_15} and where $\mf{D}W^{(h)|1,1}_{1;\eset;\eset}$ is a sum of terms, each interpretable as  the analytical value $\val(\sharp)(\eta_1,\eta_2;b)$ of a graphical diagram $(\sharp)$ with the same understanding of Fig. \ref{fig_self_energy_1} and the notation below it. Also in this case, using Remarks \ref{rmk_2} and \ref{rmk_1}, one can obtain a bound for  $\mf{D}W^{(h)|1,1}_{1;\eset;\eset}$ by studying a \emph{subclass} of terms contributing to it: by possibly shrinking $\l,$ i.e. gauging the value of $\mf{c}'_2$ above, it is sufficient to bound the diagrams labeled by $(\sharp)=(a),\dots,(k)$ appearing in Fig. \ref{fig_kernel_vertex}, which, except $(b)$\footnote{Diagram $(b)$  survives this procedure because, in order to recover the bound in Eq. \eqref{eq:18}, one does not have any extra power of $\l$ to apply Remark \ref{rmk_2} as it follows  from the bound of $\|W^{(h)|1,1}_{\eset;\eset;\eset}\|_1^w$ being of order $\l$. }, are obtained after the removal of a dashed tadpole or a dressed fermionic insertion from the general graph in $\mf{D}W^{(h)|1,1}_{1;\eset;\eset}$.

\begin{figure}[ht]
    \centering
\includegraphics[width=0.99\textwidth]{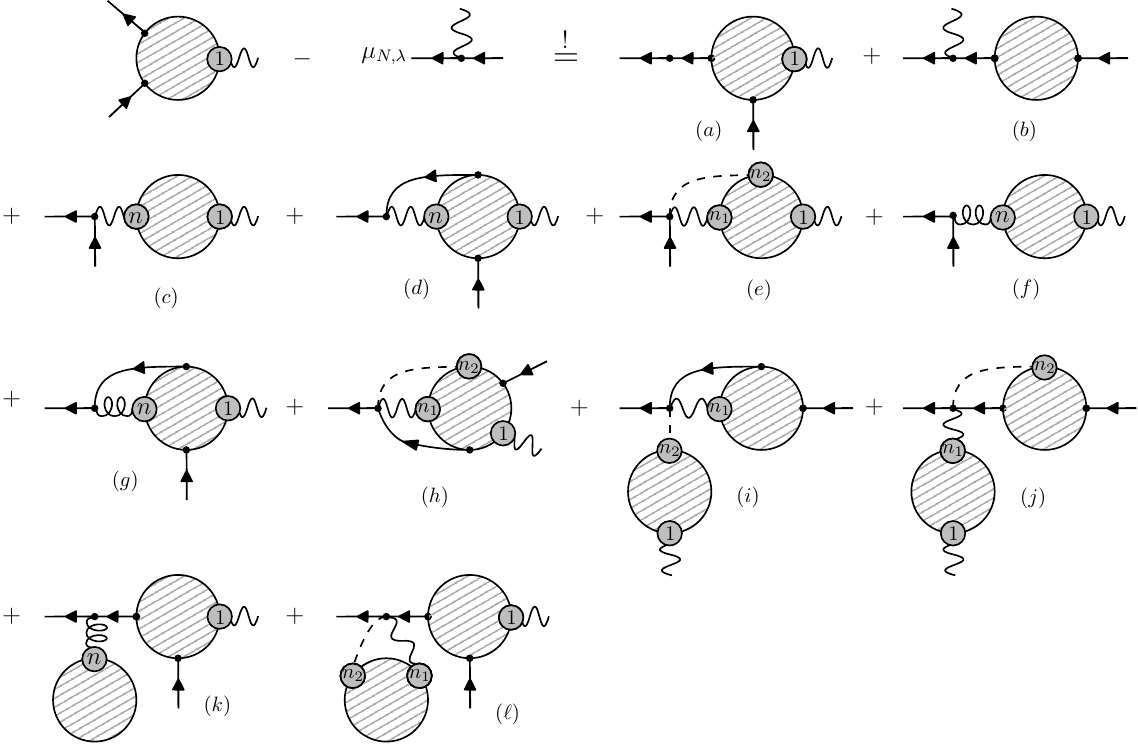}
    \caption{The diagram expansion for $\mf{D} W^{(h)|1,1}_{1;\eset;\eset}(\eta_1,\eta_2;b)$. The symbol $\stackrel{!}{=}$ reminds that the identity holds up to diagrams that can be bounded in terms of the ones listed, i.e.,  only if one considers all the diagrams before the reduction procedure summarized in Remarks \ref{rmk_2} and \ref{rmk_1} by deleting dashed tadpoles and replacing dressed fermionic lines with bare ones.}
    \label{fig_kernel_vertex}
\end{figure}

The estimates follow the same lines as Subsection \ref{ssection_improved_1} and will not be repeated full detail: we only discuss the main novelties.  In order to deduce Eq. \eqref{eq:18} from the diagrams in Fig. \ref{fig_kernel_vertex}, we start by observing the following.

\begin{remark}[Bubble insertions]
Differently from the analysis of the previous section, now some of the graphs in Fig. \ref{fig_kernel_vertex} involve also the \virg{bubble kernel} $W^{(h)|0,0}_{(1,1);\eset;\eset}$ which, by direct inspection, appears always convoluted with a bosonic propagator, i.e. in the form $g^A*W^{(h)|0,0}_{(1,1);\eset;\eset}$. This implies that a bound on $\|g^A*W^{(h)|0,0}_{(1,1);\eset;\eset}\|_1^w$, as provided in Eq.  \eqref{eq_SDB_3}, is sufficient\footnote{By bounding  $g^A*W^{(h)|0,0}_{(1,1);\eset;\eset}$ as $\|g^A\|_1^w\|W^{(h)|0,0}_{(1,1);\eset;\eset}\|_1^w$ one obtains, apriori, a logarithmic divergent factor $|h-N|$, as can be checked at first order in $\l$ for the norm of the bubble diagram.} for a bound of the whole graph: for concreteness, the simplest example is graph $(c)$ which reads
\[\begin{split}
& \text{Val}(c)(\eta_1,\eta_2;b)= \mu_{N,\l}\sum_{n\ge1} \frac{\l}{n^2} \int_{\L_B^2}db'db''  c_{b'}(\eta_1,\eta_2) g^A_{b',b''} W^{(h)|0,0}_{(n,1);\eset;\eset}(b'',b),
\end{split}\]

so that, using the bounds for the kernels $W^{(h)|0,0}_{(n,n');\eset;\eset}$ (Eqs. \eqref{eq_SDB_2}-\eqref{eq_SDB_3}) and the fact that according to Eq. \eqref{eq_bound_boson}, $\|g^A\|_1^w\le \frac{32\pi}{\kappa^2}$, we find: 
\[\begin{split}
\big\| \val (c)\big\|_1^w &\le \|c\|_1^w\l \Big( \big\|g^A*W^{(h)|0,0}_{(1,1);\eset;\eset}\big\|_1^w + \sum_{n\ge2} \frac{1}{n^2} \|g^A\|_1^w \big\| W^{(h)|0,0}_{(n,1);\eset;\eset}\big\|_1^w  \Big)\\
&\le 12\l \Big( C_0+ \sum_{n\ge2}\tfrac{1}{n^2} \tfrac{32\pi}{\kappa^2} C_0^3 \l^{\frac{n-1}{2}} \Big)\le \mf{c}_{2,1}C_0^3 \l,
\end{split}\]

for some $\mf{c}_{2,1}$ independent of $C_0,R$.
\end{remark}
Let us stress that, in order to close the inductive procedure, in bounding the analytical value of diagram $(b)$ one has to really use the just obtained bound Eq. \eqref{eq_19} for the kernel $W^{(h)|1,1}_{\eset;\eset;\eset}$ instead of using the inductive hypothesis in Eq. \eqref{eq_SDB_1}.

Also, a direct bound on graph $(h)$ shows a logarithmic divergence, in complete analogy with graph $(h)$ of the previous subsection (see Fig. \ref{fig_self_energy_1}). To solve such divergence, one has to extract another wiggly line, by expanding the kernel $W^{(h)|1,1}_{(1,n_1);n_2;\eset}$ that appears in graph $(h)$, following the same steps as in Eq. \eqref{eq_18} and below. 

Finally, in analogy with Remark \ref{rmk:o(N)diagrams},
a careful analysis shows that the graphs of type $\sharp\in\{a,e,f,g,h,i,k,\ell\}$ actually admit a bound which is exponentially suppressed in $N$, namely $\|\text{Val}(\sharp)\|_1^w\le \mf{c}_{2,1}C_0^6\l N^22^{-\th N}$, for $\l\le (\mf{c}'_{2,1}C_0^2R)^{-1}$, with some $\mf{c}_{2,2},\mf{c}'_{2,1}\ge1$.

Collecting all the bounds one finds that there exists $\mf{c}_{2,3},\mf{c}'_{2}\geq1$ independent of $C_0,R$ such that for $\l \leq (\mf{c}'_{2}C_0^2 R)^{-1}$,
\[\big\| \mf{D} W^{(h)|1,1}_{1;\eset;\eset} \big\|_1^w\le \mf{c}_{2,3}C_0^6\l, \]

so that from Eq.  \eqref{eq:14}, 
\begin{equation}
\label{eq_28}
\begin{split}
\big\|W^{(h)|1,1}_{1;\eset;\eset}-  c\big\|_1^w&\le  \big\|\mf{D} W^{(h)|1,1}_{1;\eset;\eset} \big\|_1^w  + \left(1- \mu_{N,\l}\right)\|c\|_1^w\\
 &\le \big\|\mf{D}W^{(h)|1,1}_{1;\eset;\eset} \big\|_1^w + \l2^{-2N}\|g^A\|_{\infty}\|c\|_1^w \le \mf{c}_2\l C_0^6,
\end{split}
\end{equation}
for some $\mf{c}_2 \geq \mf{c}_{2,3}.$

\subsection{The lowest-degree polarization bubble}
\label{ssection_improved_3}
We now analyze the kernel $W^{(h)|0,0}_{(1,1);\eset;\eset}$ associated with the monomial $(G^{(1)})^2$, also called the \it{polarization bubble}. We will show that
\begin{equation}
\label{eq:13}
\big\|g^A*(W^{(h)|0,0}_{(1,1);\eset;\eset} - \Pi^{(h,N]})\big\|_1^w\le \mf{c}_3 C_0^8\l,
\end{equation}

with $\Pi^{(h,N]}$ the \it{non-interacting bubble} defined in Eq. \eqref{eq_29a}, $C_0$ as in Proposition \ref{prop:SDB}, $\l\le (\mf{c}'_3C_0^2 R)^{-1}$, and suitable constants $\mf{c}_3,\mf{c}'_3\ge1$ independent on $C_0,R$. In complete analogy with Lemma \ref{lemma_id_self_energy}, by applying $\d/\d G^{(1)}_{b'}$ to Eq.  \eqref{eq_id_2}, for $n=1$, one finds the following identity:
\begin{equation}
\label{eq_30}
\hspace{-5pt} W^{(h)|0,0}_{(1,1);\eset;\eset}(b,b') = \frac{1}{2} \int_{\L_F^2} d\eta_1 d\eta_2 \mu_{N,\l} c_b(\eta_1,\eta_2) \int_{\L_F^2} d\eta_1' d\eta_2' g^{(h,N]}_{\eta_2',\eta_1} g^{(h,N]}_{\eta_2,\eta_1'} W^{(h)|1,1}_{1;\eset;\eset}(\eta_1',\eta_2';b'),
\end{equation}

which can be graphically represented as in Fig. \ref{fig:bolla_1}.
\begin{figure}[ht]
    \centering
    \includegraphics[width=0.42\linewidth]{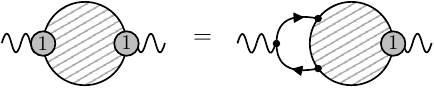}
    \caption{Expansion of the polarization bubble diagram.}
    \label{fig:bolla_1}
\end{figure}

\medskip

Combining Eqs.  \eqref{eq:14},\eqref{eq_30} we find that
\begin{equation}
\begin{split}
&W^{(h)|0,0}_{(1,1);\eset;\eset}(b,b')=\mu_{N,\l} \Pi^{(h,N]}(b,b')+  \mf{D}W^{(h)|0,0}_{(1,1);\eset;\eset}(b,b')
\end{split}
\end{equation}

where
\begin{equation}
\label{eq_31}
\begin{split}
&\mf{D}W^{(h)|0,0}_{(1,1);\eset;\eset}(b,b')=\\
&\frac{1}{2} \int_{\L_F^2} d\eta_1 d\eta_2 \mu_{N,\l} c_b(\eta_1,\eta_2) \int_{\L_F^2} d\eta_1' d\eta_2' g^{(h,N]}_{\eta_2',\eta_1} g^{(h,N]}_{\eta_2,\eta_1'} \; \mf{D}W^{(h)|1,1}_{1;\eset;\eset}(\eta_1',\eta_2';b').
\end{split}\end{equation}

In this way we have:
\begin{equation}
\begin{split}
&\big\|g^A*(W^{(h)|0,0}_{(1,1);\eset;\eset} - \Pi^{(h,N]})\big\|_1^w \le |\mu_{N,\l}-1|\big\|g^A*\Pi^{(h,N]}\big\|_1^w + \big\|g^A* \mf{D}W^{(h)|0,0}_{(1,1);\eset;\eset}\big\|_1^w.
\end{split}
\end{equation}

Using the finiteness of the non-interacting bubble, see Corollary \ref{cor_bubble} in Appendix \ref{app_bubble}, namely $\|g^A*\Pi^{(h,N]}\|_1^w\le C'_{\Pi}$ for some constant $C'_{\Pi}\ge1$, one finds that
\begin{equation}
\label{eq_32}
\begin{split}
\big\|g^A*(W^{(h)|0,0}_{(1,1);\eset;\eset} - \Pi^{(h,N]})\big\|_1^w &\le C'_{\Pi}K' \l N2^{-2N}+ \big\|g^A* \mf{D}W^{(h)|0,0}_{(1,1);\eset;\eset}\big\|_1^w,
\end{split}
\end{equation}

where we also used that $|\mu_{N,\l}-1|\le K'N\l 2^{-2N-1}$ (see the analysis of diagram (a) in Subsection \ref{ssection_improved_1}). 

Now the idea is to control the quantity $\big\|g^A* \mf{D}W^{(h)|0,0}_{(1,1);\eset;\eset}\big\|_1^w$ by expanding the kernel $\mf{D}W^{(h)|1,1}_{1;\eset;\eset}$ in the r.h.s. of Eq.  \eqref{eq_31} as we did in Subsection \ref{ssection_improved_2}. The outcome can be expressed as a sum over graphs, obtained by replacing the blob in the r.h.s. of Fig. \ref{fig:bolla_1},
by the sum over the diagrams in Fig. \ref{fig_kernel_vertex}. For instance, recalling that the graphs of type $\sharp\in\{a,e,f,g,h,i,k,\ell\}$ are bounded as $\|\text{Val}(\sharp)\|_1^w\le \mf{c}_{2,1}C_0^6\l N^22^{-\th N}$ for $\l\le (\mf{c}'_{2,1}C_0^2R)^{-1}$, their contribution to $\big\|\mf{D}W^{(h)|0,0}_{(1,1);\eset;\eset}\big\|_1^w$ can be easily bounded as follows:
\begin{equation}
\label{eq:15}
\begin{split}
& \frac{1}{L^2} \int_{\L_B} dbdb' e^{\frac{\kappa}{2}\sqrt{\d^D(b,b')}} \int_{\L_F^2} d\eta_1 d\eta_2 d\eta_1' d\eta_2'  \mu_{N,\l} |c_b(\eta_1,\eta_2)|\times\\ 
&\times \sum_{k,j=h+1}^N \big|g^{(k)}_{\eta_2',\eta_1}\big| \big|g^{(j)}_{\eta_2,\eta_1'}\big| \big|\text{Val}(\sharp)(\eta_1',\eta_2';b')\big|\le \\
& \sum_{h<k\le j\le N} \|c\|_1^w  \big\|g^{(k)}\big\|_{\infty} \big\|g^{(j)}\big\|_1^w \big\|\text{Val}(\sharp)\big\|_1^w + \sum_{h<j\le k\le N} \|c\|_1^w  \big\|g^{(j)}\big\|_{\infty} \big\|g^{(k)}\big\|_1^w \big\|\text{Val}(\sharp)\big\|_1^w\ \\
&\le 2 \sum_{h<k\le j\le N} \big\|g^{(k)}\big\|_{\infty} \big\|g^{(j)}\big\|_1^w \mf{c}_{2,1} C_0^6 \l N^22^{-\th N} \le  \sum_{h<k\le j\le N} \mf{c}_{3,1}C_0^6 2^{k-j} \l N^2 2^{-\th N}\le\\
&\leq 2\mf{c}_{3,1}C_0^6\l N^3 2^{-\th N},
\end{split}
\end{equation}

for a suitable $\mf{c}_{3,1}\ge1$. The trick of expanding the product $g^{(h,N]} g^{(h,N]}$ in terms of single-scale contributions, namely $\sum_{h<k,j\le N} g^{(k)} g^{(j)}$, and then to take the $L^{\infty}$ (resp. $L^1$) norm of the propagator at the lowest (resp. highest) scale will be recurrently used throughout this paragraph.

\medskip

We are left with bounding the contributions to $\big\|g^A*\mf{D}W^{(h)|0,0}_{(1,1);\eset;\eset}\big\|_1^w$ in the r.h.s. of Eq.  \eqref{eq_32} which derive from diagrams $(b),(c),(d),(j)$ of Fig. \ref{fig_kernel_vertex}.  The \virg{composite} diagrams, obtained by combining $(b),(c),(d),(j)$ of Fig. \ref{fig_kernel_vertex} with the r.h.s of Fig. \ref{fig:bolla_1}, are reported in Fig. \ref{fig_kernel_0200b}. Note that differently from the previous section, where in Fig. \ref{fig_kernel_vertex} we could omit (according to Remark \ref{rmk_2}) the diagram $(\tilde c)$, obtained by attaching a dressed fermionic line to the bare vertex\footnote{Exactly as $(b)$ is obtained by $(a)$ in Fig. \ref{fig_self_energy_1}.} in diagram $(c)$, now one has to keep track of such diagram in the bound for $\big\|g^A* \mf{D}W^{(h)|0,0}_{(1,1);\eset;\eset}\big\|_1^w$. The reason is that a bound for $\|\val (\tilde c)\|_1^w$ does not imply a finite bound on $\|\val(\tilde c')\|_1^w$ of Fig. \ref{fig_kernel_0200b} but one instead has to exploit the topological structure of the resulting composite diagram, as explained shortly.

\begin{figure}[h]
    \centering
\includegraphics[width=0.83\textwidth]{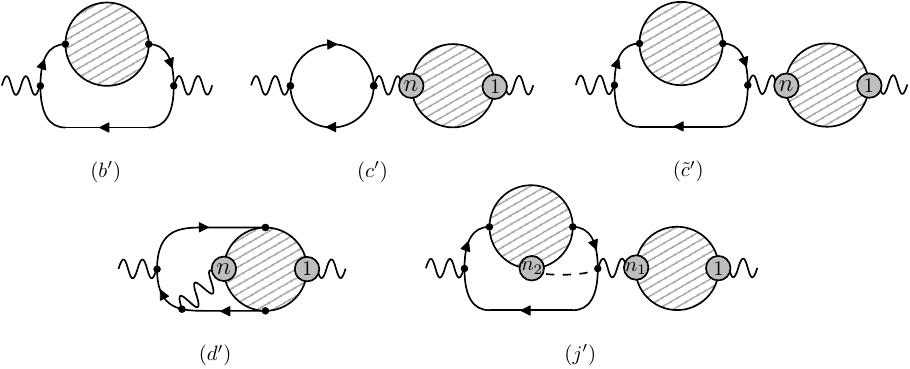}
    \caption{Graphs with non-trivial bounds (i.e. after the reduction of Remarks \ref{rmk_2}, \ref{rmk_1}) obtained by combining $(b),(c),(\tilde c),(d),(j)$ in Fig. \ref{fig_kernel_vertex} with Fig. \ref{fig_kernel_0200b} and contributing to $\mf{D}W^{(h)|0,0}_{(1,1);\eset;\eset}(b,b')$.
    }
    \label{fig_kernel_0200b}
\end{figure}

\paragraph{Diagram $(b')$.} Similarly to Eq.  \eqref{eq:15}, we write each of the three propagators $g^{(h,N]}$ as $\sum_{h'=h+1}^N g^{(h')}$, and for any triple $g^{(h_1)}, g^{(h_2)}, g^{(h_3)}$ we take the $L^{\infty}$ norm of the propagator at lowest scale, and $L^1$ norm for the other two. In this way the bound for diagram $(b')$ becomes:
\[\begin{split}
&\big\|g^A*\text{Val}(b')\big\|_1^w\le\\
&\sum_{h< h_1\le h_2, h_3\le N} \big(\|c\|_1^w\big)^2 \|g^A\|_1^w \big\|g^{(h_1)}\big\|_{\infty} \big\|g^{(h_2)}\big\|_1^w \big\|g^{(h_3)}\big\|_1^w \big\|W^{(h)|1,1}_{\eset;\eset;\eset}\big\|_1^w\le\\
& \sum_{h< h_1\le h_2, h_3\le N} \mf{c}_{3,2} 2^{h_1-h_2-h_3} \big\|W^{(h)|1,1}_{\eset;\eset;\eset}\big\|_1^w \le 4\mf{c}_{3,2} 2^{-h} \|W^{(h)|1,1}_{\eset;\eset;\eset}\|_1^w\le 4\mf{c}_{3,2}\mf{c}_1 C_0^5 \l 2^{-(1+\vth)h},
\end{split}\]

with a suitable $\mf{c}_{3,2}\ge1$ and for $\l\le (\mf{c}'_1C_0^2R)^{-1}$, having used the improved bound for $W^{(h)|1,1}_{\eset;\eset;\eset}$, Eq. \eqref{eq:goal_bound_psiquadro}.

\paragraph{Diagram $(c')$.} Explicitly we have that
\[\begin{split}
\val (c')(b_1,b_2)=&\sum_{n\ge1}   \frac{\l}{n^2} \mu_{N,\l}^2 \int_{\L_B} db db'\, \Pi^{(h,N]}(b_1,b)g^A_{b,b'}W^{(h)|0,0}_{(n,1);\eset;\eset}(b',b_2),
\end{split}\]

and exploiting the finiteness of the non-interacting bubble, namely $\|g^A*\Pi^{(h,N]}\|_1^w\le C'_{\Pi}$ (cf. Corollary \ref{cor_bubble}), we find that
\[\begin{split}
\big\|g^A*\val (c')\big\|_1^w&\le \l  \big\| g^A*\Pi^{(h,N]}\big\|_1^w \Big( \big\|g^A* W^{(h)|0,0}_{(1,1);\eset;\eset}\big\|_1^w + \sum_{n\ge2}   \tfrac{1}{n^2}  \|g^A\|_1^w \big\| W^{(h)|0,0}_{(n,1);\eset;\eset}\big\|_1^w \Big)\\
& \le \l C'_{\Pi} \Big(C_0+ \sum_{n\ge1}\tfrac{1}{n^2} \tfrac{32\pi}{\kappa^2} C_0^3 \l^{\frac{n-1}{2}}  \Big)\le \l \mf{c}_{3,3} C_0^3,
\end{split}\]

with some constant $\mf{c}_{3,3}\ge1$.

\paragraph{Diagram $(\tilde{c}')$.} Analogously to diagram $(b')$, we use the trick of splitting the propagator $g^{(h,N]}$ into its single-scale counterparts: 
\begin{equation}
\label{eq_33}
\begin{split}
\text{Val}(\tilde{c}')(b_1,b_2)=
&\sum_{n\ge1}\sum_{h<h_1,h_2,h_3\le N} \mu_{N,\l} \int_{\L_F^6} d\eta_1d\eta_2 d\eta_1' d\eta_2' d\eta_1''d\eta_2'' \int_{\L_B^2} db_3 db_4  \times \\
&\times c_{b_1}(\eta_1,\eta_2) c_{b_3}(\eta_1'',\eta_2'') g^{(h_1)}_{\eta_2',\eta_1} g^{(h_2)}_{\eta_2,\eta_1''} g^{(h_3)}_{\eta_2'',\eta_1'} g^A_{b_3,b_4} W^{(h)|1,1}_{(1,n);\eset;\eset}(b_2,b_4).
\end{split}
\end{equation}

In the r.h.s. of Eq. \eqref{eq_33} we must distinguish the three possible cases: $h_1\le h_2,h_3$ or $h_2\le h_1,h_3$ or $h_3\le h_1,h_2$. For instance, if $h_1\le h_2,h_3$, we can take the $L^{\infty}$ norm of $g^{(h_1)}$, the $L^1$ norm of $g^{(h_2)}$ and, similarly to diagram $(d)$ of Subsection \ref{ssection_improved_1}, use the H{\"o}lder's inequality for the product $g^A\tilde{g}^{(h_3)}$ with weights $p=\frac{2}{1-\vth}$ and $p'=\frac{2}{1+\vth}$:
\[\begin{split}
&\big\|g^A*\text{Val}(\tilde{c}'|_{h_1\le h_2,h_3})\big\|_1^w\le\\
& \sum_{h<h_1\le h_2,h_3\le N} \sum_{n\ge1}\frac{\l}{n^2} \|g^A\|_1^w \|g^{(h_1)}\|_{\infty} \|g^{(h_2)}\|_{1}^w \|g^A\|_p^w \|\tilde{g}^{(h_3)}\|_{p'} \big\|W^{(h)|1,1}_{(1,n);\eset;\eset}\big\|_1^w\le\\
& \sum_{h<h_1\le h_2,h_3\le N}\sum_{n\ge 1}\frac{1}{n^2} \mf{c}_{3,4} C_0^3 \l \l^{\frac{n-1}{2}} 2^{h_1-h_2-\vth h_3} 2^{-h} \le \frac{\pi^2\mf{c}_{3,4}}{3(1-2^{-\vth})^2} C_0^3 \l 2^{-(1+\vth)h},
\end{split}\]

where $\tilde{g}^{(h_3)}_{b,\eta'}\equiv \int d\eta d\eta_1 e^{\frac{\kappa}{2}\sqrt{\delta^D(b,\eta,\eta_1)}}|c_b(\eta_1,\eta)| |g^{(h_3)}_{\eta,\eta'}|$ and $\mf{c}_{3,4}$ is a suitable constant greater than 1. The regimes $h_2\le h_1,h_3$ and $h_3\le h_1,h_2$ can be worked out in the same way, taking the $L^{\infty}$ norm of the propagator at the lowest scale, and they yield the same bound. 

\paragraph{Diagram $(d')$.} Here again we exploit the finiteness of the non-interacting bubble, namely $\|g^A*\Pi^{(h,N]}\|_1^w\le C'_{\Pi}$ (cf. Corollary \ref{cor_bubble}), so that 
\[\begin{split}
&\big\|g^A* \text{Val}(d')\big\|_1^w\le\\
&\sum_{n\ge1} \frac{\l}{n^2}  \int_{\L_B^3} \frac{db_1db_2db_3}{L^2} e^{\frac{\kappa}{2}\sqrt{\d^D(b_1,b_2)}+ \frac{\kappa}{2}\sqrt{\d^D(b_2,b_3)}} \Big|\big(g^A*\Pi^{(h,N]}\big)(b_1,b_2) \big(g^A*W^{(h)|0,0}_{n;\eset;\eset}\big)(b_2,b_3)\Big|\le\\
&\sum_{n\ge1} \frac{\l}{n^2} \big\|g^A*\Pi^{(h,N]}\big\|_1^w \big\|g^A*W^{(h)|0,0}_{(n,1);\eset;\eset}\big\|_1^w\le C'_{\Pi}\l \|g^A* W^{(h)|0,0}_{(n,1);\eset;\eset}\|_1^w + \\
&+C'_{\Pi}\l\sum_{n\ge2}\frac{1}{n^2} \|g^A\|_1^w \|W^{(h)|0,0}_{(n,1);\eset;\eset}\|_1^w \le \l C'_{\Pi} C_0 + \sum_{n\ge2}\frac{1}{n^2} C'_{\Pi} \|g^A\|_1^w C_0^3\l^{\frac{n+1}{2}}\le \mf{c}_{3,5}C_0^3\l,
\end{split}\]

for a suitable $\mf{c}_{3,5}\ge1$.

\paragraph{Diagram $(d'')$.} The bound for this diagram can be deduced from the bound for diagram $(b')$. Indeed:
\[\begin{split}
&\text{Val}(d'')(b_1,b_2)= \sum_{n\ge1} \frac{\l}{n^2} \int_{\L_B^2} db_3db_4 \big(\text{Val}(b')\big)(b_1,b_3) g^A_{b_3,b_4} W^{(h)|0,0}_{(n,1);\eset;\eset}(b_4,b_2)
\end{split}\]

from which: \[
\begin{split}
\big\|g^A*\text{Val}(d'')\big\|_1^w &\le\sum_{n\ge1} \frac{\l}{n^2}  \big\|g^A*\text{Val}(b')\big\|_1^w \big\| g^A* W^{(h)|0,0}_{(n,1);\eset;\eset}\big\|_1^w\le\\
&\sum_{n\ge1} \frac{\l}{n^2} 4\mf{c}_{3,2} \mf{c}_1 C_0^5\l 2^{-(1+\vth)h} C_0^3 \l^{\frac{n-1}{2}}\le \l \mf{c}_{3,6}C_0^8  2^{-(1+\vth)h},  
\end{split}
\]

with $\mf{c}_{3,6}=\frac{2\pi^2}{3} \mf{c}_{3,2}\mf{c}_1$.

\paragraph{Diagram $(j')$.} Similarly to diagram $(h)$ of Fig. \ref{fig_self_energy_1},  the present graph shows a superficial logarithmic divergence, which can be solved by a second extraction of line. Expanding the kernel $W^{(h)|1,1}_{\eset;n_2;\eset}$ which appears in the diagram, with a formula similar to Eq. \eqref{eq_18}, we find that diagram $(j')$ can be graphically represented as a sum of two graphs (see Fig. below)

\begin{center}
\raisebox{-0.3\height}{\includegraphics[width=0.7\textwidth]{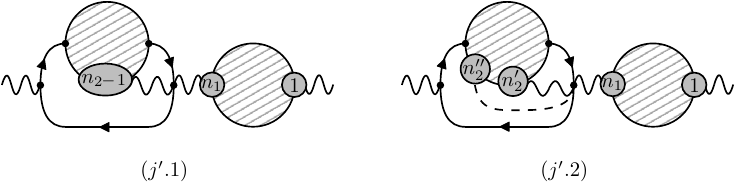}}
\end{center}

having analytical value, respectively:
\begin{equation}
\label{eq_34a}
\begin{split}
&\text{Val}(j'.1)(b,b')=\\
&\sum_{n_1\ge1,n_2\ge2} \frac{\mc{N}_{n_1,n_2}}{(n_1+n_2)^2} \mu_{N,\l}^2 2^{-N}\l^2\int_{\L_F^6} d\eta_1 d\eta_2 d\eta_1' d\eta_2'd\eta_3 d\eta_4\int_{\L_B^3} db_1 db_2 db_3 c_b(\eta_1,\eta_2)  \times\\
&\times c_{b_1}(\eta_1',\eta_2') g^A_{b_1,b_2} g^A_{b_1,b_3}  g^{(h,N]}_{\eta_2,\eta'_1} g^{(h,N]}_{\eta_4,\eta_1}   g^{(h,N]}_{\eta_2',\eta_3} W^{(h)|1,1}_{n_2-1;\eset;\eset}(\eta_3,\eta_4;b_2)   W^{(h)|0,0}_{(n_1,1);\eset;\eset}(b_3,b');
\end{split}
\end{equation}
\begin{equation}
\label{eq_34b}
\begin{split}
&\text{Val}(j'.2)(b,b')=\\
& \sum_{\substack{n_1,n_2'\ge1, n_2''\ge2}}  \frac{ \mc{N}_{n_1,n_2'+n_2''} \mc{N}_{n_2',n_2''}}{(n_1+n_2'+n_2'')^2}  \mu_{N,\l} 2^{-2N}\l^2 \int_{\L_F^6} d\eta_1 d\eta_2 d\eta_1' d\eta_2'd\eta_3 d\eta_4\times \\
&\times \int_{\L_B^4} db_1 db_2 db_3 db_4 c_b(\eta_1,\eta_2) c_{b_1}(\eta_1',\eta_2') g^A_{b_1,b_2} g^A_{b_1,b_3} g^{(h,N]}_{\eta_2,\eta'_1} g^{(h,N]}_{\eta_4,\eta_1}     g^{(h,N]}_{\eta_2',\eta_3} \d(b_1-b_4) \times\\
&\times   W^{(h)|1,1}_{n_2';n_2'';\eset}(\eta_3,\eta_4;b_2;b_4) W^{(h)|0,0}_{(n_1,1);\eset;\eset}(b_3,b').
\end{split}
\end{equation}

Both these graphs can be bounded by writing $g^{(h,N]}_{\eta_1',\eta_2}$ and $g^{(h,N]}_{\eta_1,\eta'_2}$ as the sum of their single-scale contributions, and taking the $L^{\infty}$ norm of the propagator at the lowest scale. In diagram $(j'.1)$ one exploits the extraction of the second line to take the $L^{1}$ norm of $g^{(h,N]}_{\eta_2',\eta_3}$, which gives $\mc{O}(2^{-h})$ and the $L^{\infty}$ norm of $g^A_{b_1,b_3}$, which gives $ \mc{O}(N)$, so that the finiteness of the bound follows, due to the dimensional factor $2^{-N}$ in the r.h.s. of Eq. \eqref{eq_34a}.  In diagram $(j'.2)$ the dashed line, associated with a Dirac delta, forces one to bound in $L^{\infty}$ norm both $g^{(h,N]}_{\eta'_2,\eta_3}$ and $g^A_{b_1,b_3}$, so obtaining a factor $\mc{O}(N2^N)$, which is however compensated by the dimensional factor $\mc{O}(2^{-2N})$ in the r.h.s. of Eq. \eqref{eq_34b}. Summarizing,
\[\big\|g^A*\text{Val}(g') \big\|_1^w\le \mf{c}_{3,7}  C_0^7 \l^2 2^{-N}N^2 2^{-h}, \]

with a suitable $\mf{c}_{3,7}\ge1$. Recalling Eq. \eqref{eq_32}, and collecting the bounds obtained for the various contributions to $\mf{D} W^{(h)|0,0}_{(1,1);\eset;\eset}$, the desired bound for the kernel $W^{(h)|0,0}_{(1,1);\eset;\eset}$ follows:
\begin{equation}
\label{eq:16}
\begin{split}
&\big\|g^A*\big(W^{(h)|0,0}_{(1,1);\eset;\eset} - \Pi^{(h,N]}\big)\big\|_1^w\le\\
&C'_{\Pi}K' \l N2^{-2N}+ 18 \|g^A\|_1^w \mf{c}_{3,1}C_0^6\l N^3 2^{-\th N} + 2\sum_{\sharp\in\{b',c',d',d'',g'\}} \big\|g^A*\text{Val}(\sharp)\big\|_1^w\le \mf{c}_3 C_0^8 \l, 
\end{split}
\end{equation}

for a suitable $\mf{c}_3$ bigger than all the constants $\mf{c}_{3,1},\mf{c}_{3,2},\dots$ appeared throughout the paragraph, and for $\l\le \Big(\max\big\{\mf{c}'_{1,2}, \mf{c}'_1 \big\}C_0^2R\Big)^{-1}$. The factor 2 multiplying the sum over the graphs in the r.h.s. of Eq. \eqref{eq:16} takes into account the fact that all the diagrams containing a dashed tadpole have not been considered in Fig. \ref{fig_kernel_0200b}, coherently with the discussion in Remark \ref{rmk_1} of Subsection \ref{ssection_improved_1}. 

\subsection{The effective quartic interaction}
\label{ssection_improved_4}

We would like to prove the bound: 
\begin{equation}
\label{eq:17}
\big\|W^{(h)|2,2}_{\eset;\eset;\eset} - \mc{Q}^{(N)}\big\|_1^w\le \mf{c}_4 C_0^8\l^{\frac{5}{4}},
\end{equation}

with $\mc{Q}^{(N)}$ defined in Eq.  \eqref{eq_29b}, $C_0$ as in Proposition \ref{prop:SDB}, $\l\le (\mf{c}'_4C_0^2 R)^{-1}$, and suitable constants $\mf{c}_4,\mf{c}'_4\ge1$ independent on $C_0,R$. In analogy with Lemma \ref{lemma_id_self_energy}, one can prove that

\[ W^{(h)|2,2}_{\eset;\eset;\eset}(\eta_1,\eta_2,\eta'_1,\eta'_2)= \mu_{N,\l} \mc{Q}^{(N)}(\eta_1,\eta_2,\eta'_1,\eta'_2) + \mf{D}W^{(h)|2,2}_{\eset;\eset;\eset}(\eta_1,\eta_2,\eta'_1,\eta'_2), \]

where $\mf{D}W^{(h)|2,2}_{\eset;\eset;\eset}$ can be graphically represented as a sum of graphs, whose analytical value is deduced in the same way as in Subsection \ref{ssection_improved_1}. Such graphs are collected in Fig. \ref{fig_quartic_1} where, as before (see Remarks \ref{rmk_2} and \ref{rmk_1}), we have omitted all those
that after the elimination of a \virg{dashed tadpole} or the replacement of a \virg{dressed line} by a \virg{bare line}, reduce to one of the graphs already present.

\begin{figure}[h]
    \centering
\includegraphics[width=0.99\textwidth]{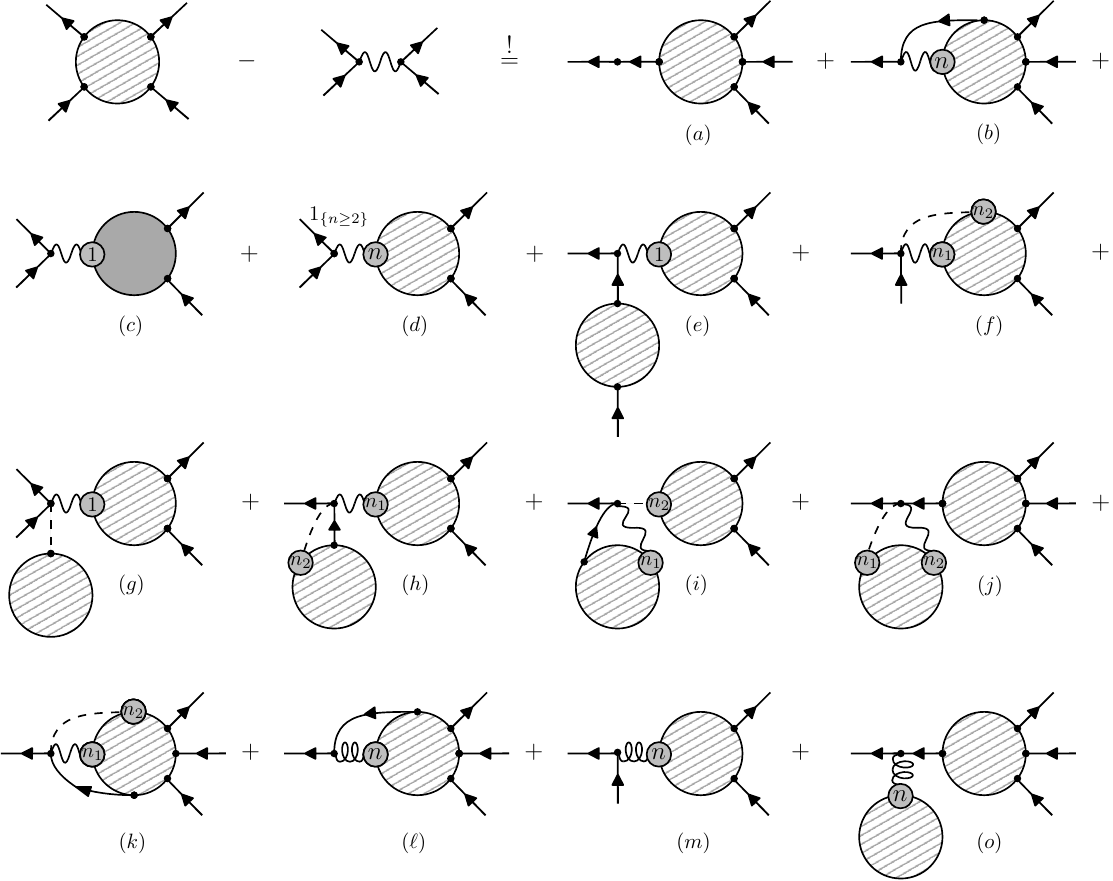}
    \caption{Graphs contributing to $\mf{D}W^{(h)|2,2}_{\eset;\eset;\eset}(\eta_1,\eta_2,\eta'_1,\eta'_2)$. The darker blob in diagram (c) is the difference kernel $\mf{D}W^{(h)|1,1}_{1;\eset;\eset}$. The exiting half-line attached to the bare vertex is always associated with the coordinate label $\eta_1$, while the labels $\eta_1',\eta_2'$ can be associated to the entering lines in all the two possible ways.}
    \label{fig_quartic_1}
\end{figure}

The estimates follow the same techniques used so far in this section, with no conceptual complication, and thus they will not be repeated in detail. Let us stress that in graphs $(c)$ and $(e)$ one has to use the improved bounds for $\|W^{(h)|1,1}_{1;\eset;\eset}- c\|_1^w$ and $\|W^{(h)|1,1}_{\eset;\eset;\eset}\|_1^w$ respectively, namely Eqs. \eqref{eq:18} and \eqref{eq_19}.  Also, in analogy with diagrams $(h)$ of Fig. \ref{fig_self_energy_1} and Fig. \ref{fig_kernel_vertex}, here we have that diagram $(k)$ of Fig. \ref{fig_quartic_1} displays a superficial log-divergence, which is readily solved by a second extraction of line, i.e. by the expansion of the kernel $W^{(h)|2,2}_{n_1;n_2;\eset}$ (the analogous procedure was explained with all the details in Eq. \eqref{eq_18} and below). Collecting the bounds from all the diagrams of Fig. \ref{fig_quartic_1}, one finds:
\begin{equation}
\label{eq_35}
\big\|\mf{D}W^{(h)|2,2}_{\eset;\eset;\eset}\big\|_1^w\le \mf{c}_{4,1}C_0^8\l^{\frac{5}{4}}\end{equation}

for $\l\le (\mf{c}'_4C_0^2R)^{-1}$, with suitable constants $\mf{c}_{4,1}, \mf{c}'_4\ge1$.
The exponent $5/4$ in $\l$ takes into account the lowest power of $\l$ among all the contributions to the l.h.s. of Eq. \eqref{eq_35}, namely diagram $(b)$ with $n=1$ and diagram $(m)$ with $n=2$. From Eq. \eqref{eq_35} we get the desired bound:
\[\begin{split} &\big\|W^{(h)|2,2}_{\eset;\eset;\eset}-  \mc{Q}^{(N)}\big\|_1^w \le\\
&|\mu_{N,\l}-1| \|\mc{Q}^{(N)}\|_1^w+ \big\|\mf{D}W^{(h)|2,2}_{\eset;\eset;\eset}\big\|_1^w \le  \mf{c}_{4,2}\l^2N2^{-2N}+ \mf{c}_{4,1}C_0^8 \l^{\frac{5}{4}}\le (\mf{c}_{4,1}+\mf{c}_{4,2})C_0^8 \l^{\frac{5}{4}}
\end{split}\]

for $\l\le (\mf{c}'_4C_0^2R)^{-1}$, where we also used the fact that $\|\mc{Q}^{(N)}\|_1^w \le \frac{\l}{2}(\|c\|_1^w)^2\|g^A\|_1^w$ is bounded by constant times $\l$.

\subsection{The higher-degree marginal kernels}
\label{ssection_improved_5}

So far we have been discussing the kernels which were already present in the quartic theory \cite{BFM09,F10,M22}. We are left with analyzing the marginal kernels due to the source fields $\{G^{(n)},\dot{G}^{(n)}\}_{n\ge2}$, namely $W^{(h)|1,1}_{n;\eset;\eset}, W^{(h)|1,1}_{\eset;n;\eset}$ with $n\ge2$ and $W^{(h)|0,0}_{(n,n');\eset;\eset}, W^{(h)|1,1}_{\eset;(n,n');\eset}, W^{(h)|1,1}_{n;n';\eset}$ with $n+n' \geq 3$. Even though the aforementioned fields enter in the initial potential, Eq. \eqref{def_potential}, via irrelevant monomials, they all have scaling dimension 1, so they produce non-irrelevant terms at scales $h<N$. It turns out that the bounds for these extra marginal terms are somehow easier than those discussed so far in this section, the reason being that the fields $G^{(n)},\dot{G}^{(n)}$ are coupled to kernels $\{v_n\}_{n\ge2}$ (cf. Eq. \eqref{eq_16a}) which, differently from $v_1$, include \it{at least} one wiggly line. As it already emerged throughout this section, the presence of wiggly lines generally yields dimensional gains, due to the non-locality of the boson propagator.

\medskip

Since the analysis follows the same lines of the previous subsections, we will discuss with some more detail only the kernel $W^{(h)|1,1}_{n;\eset;\eset}$ and only give ideas on how to proceed for the other cases.

\paragraph{The higher-degree vertex.}

We consider the kernel $W^{(h)|1,1}_{n;\eset;\eset}$, with $n\ge2$ and we aim to show that
\begin{equation}
\label{eq:19}
\big\|W^{(h)|1,1}_{n;\eset;\eset}\big\|_1^w\le \mf{c}_5C_0^5\l^{\frac{1}{2}}\l^{\frac{n-1}{2}}, \qquad n\ge2,
\end{equation}

for $\l\le (\mf{c}''_5C_0 R)^{-1}$, and suitable constants $\mf{c}_5,\mf{c}''_5\ge1$ independent on $C_0,R$. Notice that Eq. \eqref{eq:19} directly implies the desired bound in Eq. \eqref{eq_IB_2} by requiring $\l\le (\mf{c}'_5C_0^{10}R)^{-1}$, with $\mf{c}'_5=\max\{\mf{c}_5^2,\mf{c}''_5\}$. Let us recall that the choice of losing a fractional power of $\l$,  so to gain in the r.h.s. of Eq. \eqref{eq_36} a pre-factor which is exactly 1, is only technical and made so that the inductive structure of our proof results more tractable.

Again, the starting point is to derive an expansion for the kernel obtained from Lemma \ref{lemma_id_B} by applying $\d^2/\d\psi^-_{\eta_2}\d\psi^+_{\eta_1}$ to Eq.  \eqref{eq_id_2}.  The output of the expansion is represented in Fig. \ref{fig_kernel_vertex_2} where, as before (cf. Remarks \ref{rmk_2} and \ref{rmk_1}), we have omitted all those diagrams
that after the elimination of a \virg{dashed tadpole} or the replacement of a \virg{dressed line} by a \virg{bare line}, reduce to one of the graphs already present.

In order to read the value of a diagram,  differently from the previous cases referred to Figures  \ref{fig_self_energy_1},\ref{fig_kernel_vertex},\ref{fig_kernel_0200b} and \ref{fig_quartic_1}, now the nodes of the graphs in Fig. \ref{fig_kernel_vertex_2} are interpreted as follows. If the incoming and outgoing solid lines at the node have fermionic labels $\eta_2,\eta_1$ respectively, then: 
\begin{itemize}
\item a node with two wiggly lines, with labels $n$ and $n-1$, is associated to a factor $\\2^{-N}\frac{n^2}{(n-1)^2}\mu_{N,\l} c_b(\eta_1,\eta_2)$ ;
\item a node with  two wiggly lines and a dashed line,  having labels $n,n_1,n_2$, respectively, and satisfying $n_1+n_2=n$, is associated to a factor $2^{-2N} \mc{N}_{n_1,n_2}c_b(\eta_1,\eta_2)$ where we recall that $\mc{N}_{n_1,n_2}= \frac{(n_1+n_2)^2}{(n_1+n_2-1)n_1n_2^2}\le \frac{(n_1+n_2)^2}{n_1^2n_2^2}$.
\end{itemize}

\begin{figure}[h]
\centering
\includegraphics[width=0.92\textwidth]{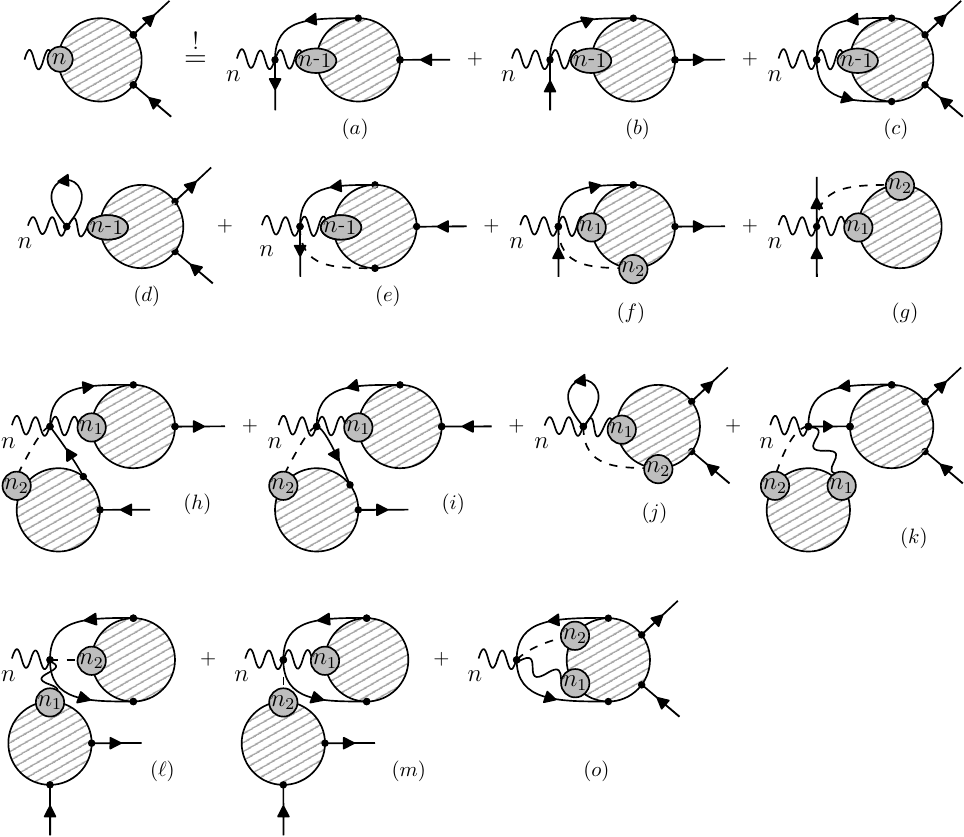}
\caption{Graphs contributing to the kernel $W^{(h)|1,1}_{n;\eset;\eset}$, with $n\ge2$. The $n$ label of the leftmost wiggly line in each diagram reminds that the external field is of type $G^{(n)}.$ The symbol $\stackrel{!}{=}$ reminds that the identity is up to diagrams which are not drawn, according to the reduction procedure explained below Eq. \eqref{eq:19}.}
\label{fig_kernel_vertex_2}
\end{figure}

\paragraph{Diagram $(a)$.}
In this case we get 
\begin{equation}
\label{eq_36}
\begin{split} 
\big\|\text{Val}(a)\big\|_1^w& \le 2^{-N}n^2(n-1)^{-2} \|c\|_1^w \l\|g^A\|_\infty \|g^{(h,N]}\|_1^w \big\|W^{(h)|1,1}_{n-1;\eset;\eset}\big\|_1^w \\
&\le  24 \lambda K' \tilde K N 2^{-N} 2^{-h} C_0^2\l^{\frac{n-2}{2}}\leq  \mf{c}_{5,1}C_0^2 \l^{\frac{1}{2}} \l^{\frac{n-1}{2}}, \end{split}\end{equation}
where $\mf{c}_{5,1}\ge1$ is a suitable constant depending on $\tilde K,K'$, already appearing along Section \ref{ssection_improved_1}. Note that we neglected the vanishing factor $ 2^{-\theta N-h}$ as $N\to \infty$ (see Remark \ref{rmk:o(N)diagrams}) Graphs $(b)$ and $(d)$ follow the same type of estimates as $(a)$.

\paragraph{Diagram $(c)$.} Here we use the same idea of Eq. \eqref{eq:15} and lines below by splitting in scales each $g^{(h,N]}$ and taking the $L^\infty$ norm of $g^A$, 
\[\begin{split}
\big\|\text{Val}(c)\big\|_1^w& \le \lambda 2^{-N+1} n^2(n-1)^{-2}  \sum_{h \leq i \leq j \leq N} \|g^A\|_{\infty}\|g^{(i)}\|_\infty  \|\tilde{g}^{(j)}\|_1^w \big\|W^{(h)|2,2}_{n-1;\eset;\eset}\big\|_1^w \\
&\le \mf{c}_{5,2}C_0^2 2^{-N} N^2 \l^{\frac{n-2}{2}}\le \mf{c}_{5,2}C_0^2 \l^{\frac{1}{2}} \l^{\frac{n-1}{2}}, 
\end{split}\]

for a suitable $\mf{c}_{5,2}\ge1$.

\medskip

\paragraph{Diagram $(e)$.} The presence of the dashed line forces us to take the $L^{\infty}$ norm of both $g^{(h,N]}$ and $g^A$:
\begin{equation}
\label{eq_37}
\begin{split}
\big\|\text{Val}(e)\big\|_1^w& \le \sum_{\substack{n_1\ge1,n_2\ge2\\ n_1+n_2=n}} 2^{-2N} \mc{N}_{n_1,n_2} \|g^{(h,N]}\|_{\infty}\l\|g^A\|_{\infty} \big\|W^{(h)|1,1}_{n_1;n_2;\eset}\big\|_1^w \\
&\le \sum_{\substack{n_1\ge1,n_2\ge2\\ n_1+n_2=n}} \tfrac{n^2}{n_1^2 n_2^2} \mf{c}_{5,3}\l 2^{-N}N  C_0^3 \l^{\frac{n_1+n_2-2}{2}}2^{-h}\le \tfrac{4}{3}\pi^2\mf{c}_{5,3}C_0^3 \l^{\frac{1}{2}} \l^{\frac{n-1}{2}} 2^{-h},
\end{split}
\end{equation}

for a suitable $\mf{c}_{5,3}\ge1$, having also used the fact that $\sum_{n_1+n_2=n}\frac{n^2}{n_1^2 n_2^2}\le \frac{4}{3}\pi^2$. 

\begin{remark}\label{rmk_n^3}
The origin of the quantity $\mc{N}_{n_1,n_2}\leq (n_1+n_2)^2n_1^{-2}n_2^{-2}$ is in the numerical factors $n^3$ in the definition of the source terms with $G^{(n)}$ and $\dot{G}^{(n)}$ in Eq. \eqref{def_potential}, as it can be explicitly derived from Lemmas \ref{lemma_id_A} and \ref{lemma_id_B}. If in place of $n^3$ we had put $n^k$, in the r.h.s. of Eq.  \eqref{eq_37} we would have found: 
\[\sum_{n_1+n_2=n}\left(\frac{n}{n_1n_2}\right)^{k-1},\]

 which is \it{bounded} in $n$ for every $k\ge3$, but for $k=0,1,2$ it is actually \it{unbounded}. For instance, with the naive choice $k=0$, the estimate for graph $(d)$ would have been:

\[\big\|\text{Val}(d)\big\|_1^w\le  \mc{O}(n) C_0^3 N 2^{-N-h} \l^{\frac{1}{2}} \l^{\frac{n-1}{2}}, \]

which has no hope to be consistent with the desired bound in Eq. \eqref{eq_IB_2}. With our choice $k=3$ instead, we manage to get $\mc{O}(1)$ in place of $\mc{O}(n)$.
\end{remark}

\medskip

The estimates for diagrams $(f)$-$(o)$ can be performed similarly and we will not belabor the details. We mention that a direct bound on diagram $(o)$ shows superficial divergences, which can be solved by expanding the kernel $W^{(h)|2,2}_{\eset;\eset;\eset}$ and exploiting the presence of an extra boson propagator (the analogous procedure for diagram $(h)$ of Fig. \ref{fig_self_energy_1} was explained in Eq. \eqref{eq_18} and below).

Collecting all the contributions in Fig. \ref{fig_kernel_vertex_2} one finds the bound in Eq. \eqref{eq:19}, from which the desired bound in Eq.  \eqref{eq_IB_2} follows.

\paragraph{Remaining higher-degree kernels.}
The analysis of the \emph{higher order bubbles}, i.e. the kernels $W^{(h)|0,0}_{(n,n');\eset;\eset}$ with $n+n'\ge3$, is almost identical to that for the higher-degree vertices and one can find that
\begin{equation}
\big\|W^{(h)|0,0}_{(n,n');\eset;\eset}\big\|_1^w\le \mf{c}_6C_0^6\l^{\frac{n+n'-2}{2}}, \qquad  n+n'\ge3,
\end{equation}

for $\l\le (\mf{c}''_6C_0R)^{-1}$ and with a suitable $\mf{c}''_6\ge1$. As before, by requiring $\l\le (\mf{c}'_6C_0^{12}R)^{-1}$, with $\mf{c}'_6=\max\{\mf{c}''_6,\mf{c}_6^2\}$, the desired bound in Eq. \eqref{eq_IB_6} follows.

\medskip

Concerning instead the kernels $W^{(h)|1,1}_{\eset;n;\eset}$, $W^{(h)|0,0}_{\eset;(n,n');\eset}$ and $W^{(h)|0,0}_{n';n;\eset}$, their expansion is similar to  Eq. \eqref{eq_18} (depicted in  Fig. \ref{fig_self_energy_2}) and is shown in Fig. \ref{fig_dot_kernels}.

\begin{figure}[h]
\centering
\includegraphics[width=0.98\textwidth]{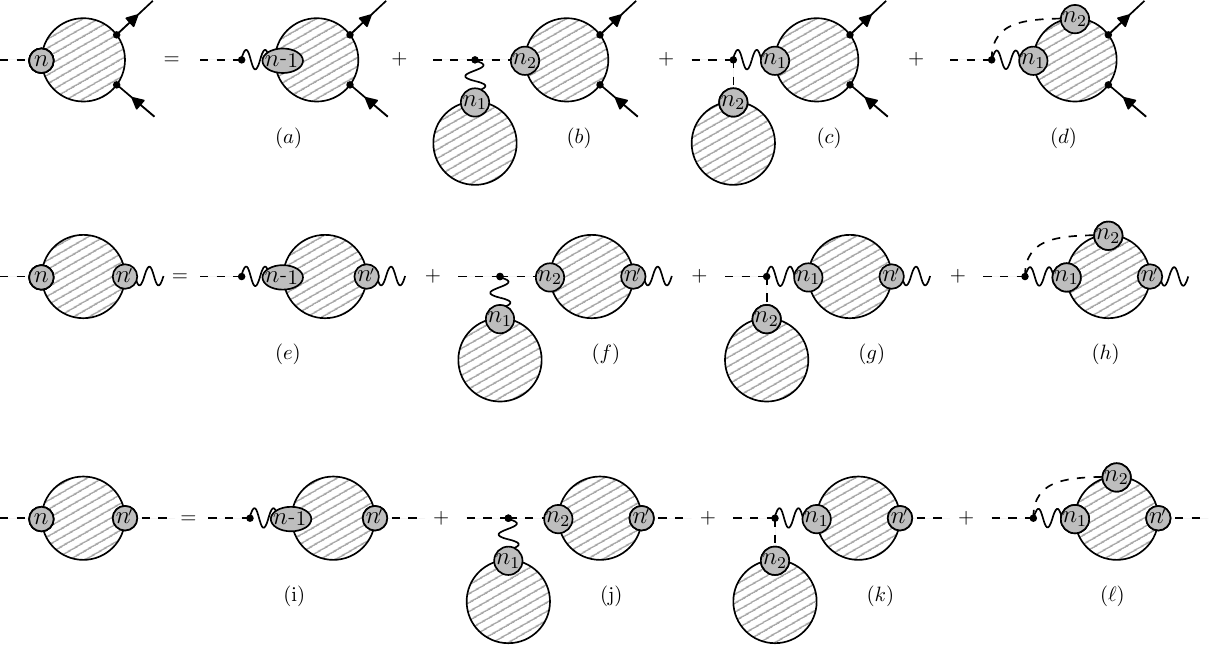}
\caption{Graphical expansions for the kernels $W^{(h)|0,0}_{\eset;n;\eset}$, $W^{(h)|1,1}_{\eset;n;\eset}$, $W^{(h)|0,0}_{n';n;\eset}$ and $W^{(h)|0,0}_{\eset;(n',n);\eset}$ respectively.}
\label{fig_dot_kernels}
\end{figure}

Using the same tools as in the previous cases, one can check that each diagram in the r.h.s. of each line of Fig. \ref{fig_dot_kernels} is bounded by $\mc{O}\big(C_0^4\l^{\frac{1}{2}} \l^{\frac{n-1}{2}}\big)$, without the need to extract any further boson line, so that:
\begin{equation}
\label{eq:20}
\big\|W^{(h)|1,1}_{\eset;n;\eset} \big\|_1^w \le \l^{\frac{n-1}{2}}, \;\; \big\|W^{(h)|1,1}_{\eset;(n,n');\eset} \big\|_1^w\le  \l^{\frac{n+n'-2}{2}},\;\; \big\|W^{(h)|1,1}_{n';n;\eset} \big\|_1^w \le \l^{\frac{n+n'-2}{2}},
\end{equation}

for any $n\ge2$ and $\l\le (\mf{c}'_7C_0^8)^{-1}$, with a suitable $\mf{c}'_7\ge1$.

\bigskip

Summarizing, with Eqs. \eqref{eq:goal_bound_psiquadro}, \eqref{eq:18}, \eqref{eq:13}, \eqref{eq:17}, \eqref{eq:19} and \eqref{eq:20}, we have shown the validity of the bounds in Eqs. \eqref{eq_IB_1}-\eqref{eq_IB_6} with $R= \max_{j=1}^4 \mf{c}_j C_0^8$ and for  $\lambda$ such that $\l\le \min\big\{(C_0{\hat R}^4)^{-1},\min_{j=1}^7 (\mf{c}'_j C_0^{12}{\hat R})^{-1}\big\}$. This establishes the validity of the inductive step introduced at the beginning of Section \ref{sect_improved}, thus concluding the proof of Theorem \ref{thm:IB}.

\section{The complete generating functional}
\label{sect_external}

In this section we reintroduce the external sources of the theory, i.e. the background vector field $J$, the chiral vector $B$ and the Grassmann field $\vphi$ with the aim of  completing the proof of Theorem \ref{thm_UV} for the full generating functional: \begin{equation}\mc{W}^{(u.v.)}(\psi;\varphi;J;B)= -\log \int P^{(0,N]}(\mc{D}\z) e^{-\mc{V}(\psi+\z;G(J))+(\vphi,\z + \psi)- (B,O_5(\z+ \psi))}\label{eq:44}\end{equation} (cf. Eqs.  \eqref{def_functional_UV},  \eqref{eq:2}). 
We will follow two intermediate steps:  we will first reintroduce the fields $\vphi$ and $B$ by keeping $J=0$, and in a second moment we will discuss the case $J\neq 0$. The reason for this subdivision is that the multiscale analysis in presence of $\vphi$ and $B$, at $J=0$, requires a procedure which is very close to the one already explained in Sections \ref{sect_multiscale} and \ref{sect_improved} for the purely fermionic interaction.  In particular,  we will use a combination of new \it{standard dimensional bounds} and \it{improved bounds}, analogous to Eqs. \eqref{eq_SDB_1}-\eqref{eq_SDB_3} and Eqs. \eqref{eq_IB_1}-\eqref{eq_IB_6} which provide in the end the bounds for the new $\varphi,B$ kernels appearing in Theorem \ref{thm_UV}. 
On the other hand the reintroduction of the field $J$ and the bounds for the related kernels, appearing in Theorem \ref{thm_UV}, will require a simpler strategy based only on standard estimates with no \emph{improved bounds} needed.

\subsection{The chiral and fermion sources}
\label{ssection_external_1}

At this stage we add to the fermionic potential $\mc{V}(\psi;0)$, Eq. \eqref{eq:2}, the two source terms:

\[ \int_{\L_F}d\eta (\vphi^+_{\eta}\psi^-_{\eta}+ \psi^+_{\eta}\vphi^-_{\eta}) + \int_{\L_B} db B_b O_{5;b}(\psi). \]

The scaling dimension of these external fields is defined so that the monomials $\vphi\psi$ and $B\psi^2$ are both marginal. In this way, the scaling dimension of a monomial $(\psi)^{q} (\vphi)^{\tilde{q}} (B)^p$ is $2-q-\frac{3}{2}\tilde{q}- p$. In analogy with Section \ref{sect_multiscale}, it is convenient to add to the potential some auxiliary source terms to help controlling the relevant and marginal terms. It turns out that it is enough to introduce exactly the same source fields $G,\dot{G},\ddot{G}$ of Section \ref{sect_multiscale}, with potential
\begin{equation}
\label{def_potential_2}
\begin{split}
\mc{V}(\psi;G;\dot{G};\ddot{G};\vphi;B)&:= \mc{V}(\psi;G;\dot{G};\ddot{G})+ \int_{\L_F}d\eta (\vphi^+_{\eta}\psi^-_{\eta}+ \psi^+_{\eta}\vphi^-_{\eta}) + \int_{\L_B} db B_b O_{5;b}(\psi),
\end{split}
\end{equation}

where $\mc{V}(\psi;G;\dot{G};\ddot{G})$ is the same as Eq. \eqref{def_potential}. In analogy with Section \ref{sect_multiscale}, we construct the effective potential at scale $h$:
\begin{equation}
\mc{V}^{(h)}(\psi;G;\dot{G};\ddot{G};\vphi;B)= -\log \int P^{(h,N]}(\mc{D}\z) e^{-\mc{V}(\psi+\z;G;\dot{G};\ddot{G};\vphi;B)}, \label{eq:VhconvphiBmaJnullo}\end{equation}

and we denote by $W^{(h)|q,q';\tilde{q},\tilde{q}'}_{\ud{n};\ud{\dot{n}};\ud{\ddot{n}};p}(\ud{\eta},\ud{\eta}';\ud{\tilde{\eta}},\ud{\tilde{\eta}}';\ud{b};\ud{\dot{b}};\ud{\ddot{b}}; \ud{b}')$ the generic kernel of $\mc{V}^{(h)}$ associated with the monomial
\begin{equation}
\begin{split}
\prod_{k=1}^q \psi^+_{\eta_k}  \prod_{k=1}^{\tilde{q}} \vphi^+_{\tilde{\eta}_k} 
\prod_{k=1}^{q'} \psi^-_{\eta'_k}
\prod_{k=1}^{\tilde{q}'} \vphi^-_{\tilde{\eta}'_k} \prod_{k=1}^{\dim\ud{n}} G^{(n_k)}_{b_k} \prod_{k=1}^{\dim\ud{\dot{n}}} \dot{G}^{(\dot{n}_k)}_{\dot{b}_k} \prod_{k=1}^{\dim\ud{\ddot{n}}} \ddot{G}^{(\ddot{n}_k)}_{\ddot{b}_k}
\prod_{k=1}^p B_{b'_k},
\end{split}
\end{equation}

where we are keeping the same notation of Eq. \eqref{eq:9} and lines thereafter. The scaling dimension associated with the kernel $W^{(h)|q,q';\tilde{q},\tilde{q}'}_{\ud{n};\ud{\dot{n}};\ud{\ddot{n}};p}$ is
\begin{equation}
\label{def_Dsc_2}
D_{sc}:= 2- \tfrac{1}{2}(q+q')- \tfrac{3}{2}(\tilde{q}+\tilde{q}') - \dim\ud{n}- \dim\ud{\dot{n}}- 2\dim\ud{\ddot{n}}- p.
\end{equation}

The new \emph{relevant} and \emph{marginal} (i.e. non-irrelevant) kernels that emerge are those satisfying $D_{sc}\ge0$, which are graphically represented in Fig. \ref{fig:newrelevantmarginal}. 
\begin{figure}[h]
    \centering
    \includegraphics[width=0.77\linewidth]{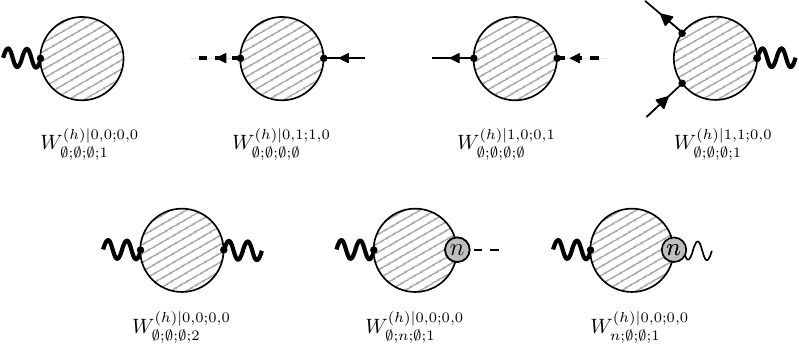}
    \caption{The graphical representation of the marginal and relevant kernels associated to the fields $\vphi,B$, which are respectively represented as thick dashed and wiggly external lines. The thick dashed line is distinguishable from the dashed line associated to $\dot G^{(n)}$ thanks to the circled $n$ label attached to the latter.}
    \label{fig:newrelevantmarginal}
\end{figure}

Now we state the analogue of Theorem \ref{thm:IB}, providing bounds for all the kernels $W^{(h)|q,q';\tilde{q},\tilde{q}'}_{\ud{n};\ud{\dot{n}};\ud{\ddot{n}};p}$. Such bounds involve either the whole kernels, or their deviation from the value of their lowest-order counterparts. Specifically,
\begin{equation}
\Pi_5^{(h,N]}(b_1,b_2):=  \frac{1}{2} \int d\eta_1 d\eta_1' d\eta_2 d\eta_2' c_{b_1}(\eta_1,\eta_1') c^5_{b_2}(\eta_2,\eta_2') g^{(h,N]}_{\eta_1',\eta_2}g^{(h,N]}_{\eta_2',\eta_1}\label{eq_chiralbubble}
\end{equation}

is the lowest-order counterpart of the kernel $W^{(h)|0,0;0,0}_{1;\eset;\eset;1}$, with $c^5$ the bare chiral vertex\footnote{Recall that by assumption $|Z^5_N|\le\frac{3}{2}$, cf. Section \ref{sect:model}.} given by Eq. \eqref{def_chiral_vertex}, which is also the lowest-order contribution to the kernel $W^{(h)|1,1;0,0}_{\eset;\eset;\eset;1}$.
\begin{proposition}
\label{prop:IB2}
For every $\vth,\th\in(0,1)$ fixed, there exist $N'_0\ge1$, $R',C_0'\ge1$ and $\l_0'>0$ such that, for any $N\ge N'_0$, $\l\le \l_0'$ and $0\le h\le N$, the following is true.

\begin{enumerate}
\item  The non-irrelevant kernels represented in Fig. \ref{fig:newrelevantmarginal} admit the following bounds:
\begin{align}
&\label{eq_IB_2:1}\big\|W^{(h)|0,1;1,0}_{\eset;\eset;\eset;0}- \delta\big\|_1^w \le R'\l 2^{-(1+\vth)h}; \;\; \big\|W^{(h)|1,0;0,1}_{\eset;\eset;\eset;0}- \delta\big\|_1^w \le R'\l 2^{-(1+\vth)h};\\
&\label{eq_IB_2:2} \big\|W^{(h)|1,1;0,0}_{\eset;\eset;\eset;1} - c^5\big\|_1^w \le R'\l; \;\; \big\|g^A*(W^{(h)|0,0;0,0}_{1;\eset;\eset;1}- \Pi_5^{(h,N]} )\big\|_1^w \le R'\l; \\
&\label{eq_IB_2:3}\big\|W^{(h)|0,0;0,0}_{n;\eset;\eset;1}\big\|_1^w\le \l^{\frac{n-1}{2}} \; \forall n\ge2; \;\; \big\|W^{(h)|0,0;0,0}_{\eset;n;\eset;1}\big\|_1^w\le \l^{\frac{n-1}{2}} \; \forall n\ge2;
\end{align}

with $\d$ the lattice Dirac delta $\d(\eta-\eta')$ over $\L_F$.

\item  All the irrelevant and non-irrelevant kernels not included in Eqs. \eqref{eq_SDB_1}-\eqref{eq_SDB_3}, admit the following bounds: 
\begin{equation}
\label{eq_SDB_2:1}
\big\|g^A*W^{(h)|0,0;0,0}_{1;\eset;\eset;1}\big\|_1^w\le C_0', \;\; \big\|W^{(h)|0,0;0,0}_{\eset;\eset;\eset;2} \big\|_1^w\le C_0'N, \;\;W^{(h)|0,0;0,0}_{\eset;\eset;\eset;1}=0\end{equation} 

and in the remaining cases
\begin{equation}
\label{eq_SDB_2:2}
\hspace{-5pt}\big\|W^{(h)|q,q';\tilde{q},\tilde{q}'}_{\ud{n};\ud{\dot{n}};\ud{\ddot{n}};p}\big\|_1^w \le C_0'^{1+\dim\ud{n}+\dim\ud{\dot{n}}+\dim\ud{\ddot{n}}+ \tilde{q}+ \tilde{q}'+p} \l^{\tilde d} 2^{h D_{sc}} \times \left\{\begin{array}{cc}
2^{\th(h-N)}     & \dim\ud{\ddot{n}}>0  \\
1     & \text{otherwise},
\end{array}\right.
\end{equation}

with $D_{sc}$ as in Eq. \eqref{def_Dsc_2} and  $\tilde d := \tfrac{1}{8}\max\{0,q+q'+\tilde{q}+\tilde{q}'-2\}+ \tfrac{1}{2}(|\ud{n}|+|\ud{\dot{n}}| - \dim\ud{n}-\dim\ud{\dot{n}})+ \tfrac{3}{4}|\ud{\ddot{n}}|$.

\end{enumerate}
\end{proposition}

As a direct consequence one can derive bounds for the kernels of the original generating functional $\mc{W}^{(u.v.)}(\psi;\vphi;0;B)$ since $\mc{W}^{(u.v.)}(\psi;\vphi;0;B)=\mc{V}^{(0)}(\psi;0;0;0;\vphi;B)$, see Subsection \ref{subsec:boundkernels}.

Proposition \ref{prop:IB2} is the analogue of Theorem \ref{thm:IB} for the kernels involving also the external fields $\vphi$ and $B$, at $J=0$, and its proof follows the same scheme as in Subsection \ref{sect_multiscale} and Section \ref{sect_improved}, thus we only provide a sketch of it. 

\begin{proof}[Sketch of the proof]
The strategy is made of two steps.

\begin{enumerate}
\item First  one derives the analogue of Proposition \ref{prop:SDB}, namely that there exists $C_0'\ge1$ (the same $C_0'$ as in Proposition \ref{prop:IB2}) such that, assuming the bounds in Eqs. \eqref{eq_IB_2:1}-\eqref{eq_IB_2:3} at scales higher than $h$, then at scale $h$ the bounds in Eqs.  \eqref{eq_SDB_2:1}-\eqref{eq_SDB_2:2} hold true, for $\l$ small enough, say $\l\le (C_0'R')^4$. The arguments used to prove this fact are essentially the same as for Proposition \ref{prop:SDB}, which are based on the expansion in Gallavotti-Nicolò trees (see Appendix \ref{app_tree}) and will not be rediscussed.

\item 
As a second step, one has to show the validity of the bounds in Eqs. \eqref{eq_IB_2:1}-\eqref{eq_IB_2:3} for some $R'$ large enough and $\l$ small enough, namely the analogue of Proposition \ref{prop:IB2}. This is achieved by induction over the scale $h$, as explained at the beginning of Section \ref{sect_improved}. Again, the proof of this fact is ultimately based on the non-locality of the boson propagator, and can be exploiting crucial identities for the non-irrelevant kernels in the l.h.s. of Eqs. \eqref{eq_IB_2:1}-\eqref{eq_IB_2:3}. Such identities are written as expansions in terms of all the kernels of $\mc{V}^{(h)}$, for which we shall use: 

\begin{itemize}
\item the bounds in Eqs. \eqref{eq_SDB_2:1}-\eqref{eq_SDB_2:2}, for the kernels involving at least a variable $\vphi^{\pm}$ or $B$;
\item the bounds in Eqs. \eqref{eq_SDB_1}-\eqref{eq_SDB_2}, with $R$ defined at the beginning of Section \ref{sect_improved}, for the kernels not involving any $\vphi^{\pm}$ or $B$ variables.
\end{itemize}
\end{enumerate}

Let us begin by discussing the bound for the kernels $\vphi^{\pm}\psi^{\mp}$, Eq. \eqref{eq_IB_2:1}. In analogy with Lemma \ref{lemma_id_self_energy}, one can obtain the following representation for these kernels by differentiating  Eq. \eqref{eq:VhconvphiBmaJnullo} with respect to $\psi^\pm, \vphi^{\mp}$:
\begin{align}
&\label{eq:24a}W^{(h)|0,1;1,0}_{\eset;\eset;\eset;0}(\eta';\tilde{\eta}) = \delta(\eta'-\tilde{\eta}) - \int d\eta'' g^{(h,N]}_{\tilde{\eta},\eta''} W^{(h)|1,1;0,0}_{\eset;\eset;\eset;0}(\eta'',\eta'),\\
&\label{eq:24b}W^{(h)|1,0;0,1}_{\eset;\eset;\eset;0}(\eta;\tilde{\eta}') = \delta(\eta-\tilde{\eta}') - \int d\eta'' g^{(h,N]}_{\eta'',\tilde{\eta}'} W^{(h)|1,1;0,0}_{\eset;\eset;\eset;0}(\eta,\eta''),
\end{align}
which are graphically represented as in the figure below.

\begin{center}
\raisebox{-0.3\height}{\includegraphics[width=0.57\textwidth]{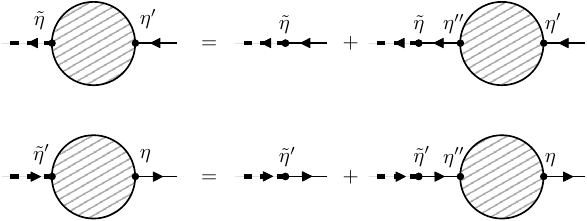}}
\end{center}

We therefore get the following estimate in $L^1$ norm:
\begin{equation}
\label{eq_40}
\big\|W^{(h)|0,1;1,0}_{\eset;\eset;\eset;0} - \d\big\|_1^w\le \|g^{(h,N]}\|_1^w \big\|W^{(h)|1,1;0,0}_{\eset;\eset;\eset;0}\big\|_1^w. 
\end{equation}

Recall that the kernel $W^{(h)|1,1;0,0}_{\eset;\eset;\eset;0}$ coincides with $W^{(h)|1,1}_{\eset;\eset;\eset}$ of the effective potential Eq.  \eqref{eq:9}, which, in force of Theorem \ref{thm:IB}, admits the bound $\big\|W^{(h)|1,1}_{\eset;\eset;\eset}\big\|\le R\l 2^{-\vth h}$ for $\l\le \l_0$, where the quantities $R$ and $\l_0$ are fixed by the analysis of Section \ref{sect_improved}. Recalling that $\|g^{(h,N]}\|_1^w\le \tilde{K}2^{-h}$ for a suitable constant $\tilde{K}\ge1$, from Eq. \eqref{eq_40} we get
\[
\big\|W^{(h)|0,1;1,0}_{\eset;\eset;\eset;0} - \d\big\|_1^w\le \tilde{K}R\l 2^{-(1+\vth)h}. 
\]

The same argument applies to the kernel $\psi^+\vphi^-$, leading to the same bound. 

\medskip

The marginal kernels with one field $B$ can be treated almost identically to the kernels where $B$ is replaced by $G^{(1)}$, which have been discussed in Section \ref{sect_improved}. For instance, about the kernel $B\psi^2$, in complete analogy with the kernel $G^{(1)}\psi^2$ discussed in Subsection \ref{ssection_improved_2}, one finds that the difference
\[ W^{(h)|1,1;0,0}_{\eset;\eset;\eset;1}(\eta_1,\eta_2;b)- c^5_b(\eta_1,\eta_2)
\]

admits an expansion which is graphically reported in Fig. \ref{fig_chiral_vertex}, where the graphical rules are analogous to those of Section \ref{sect_improved}  (see also the notation of Fig. \ref{fig:newrelevantmarginal} for the new kernels), and the understanding that the node with a thick wiggly line and two solid lines is associated to the bare chiral vertex $c^5_b(\eta_1,\eta_2)$.

\begin{figure}[h]
    \centering
\includegraphics[width=0.99\textwidth]{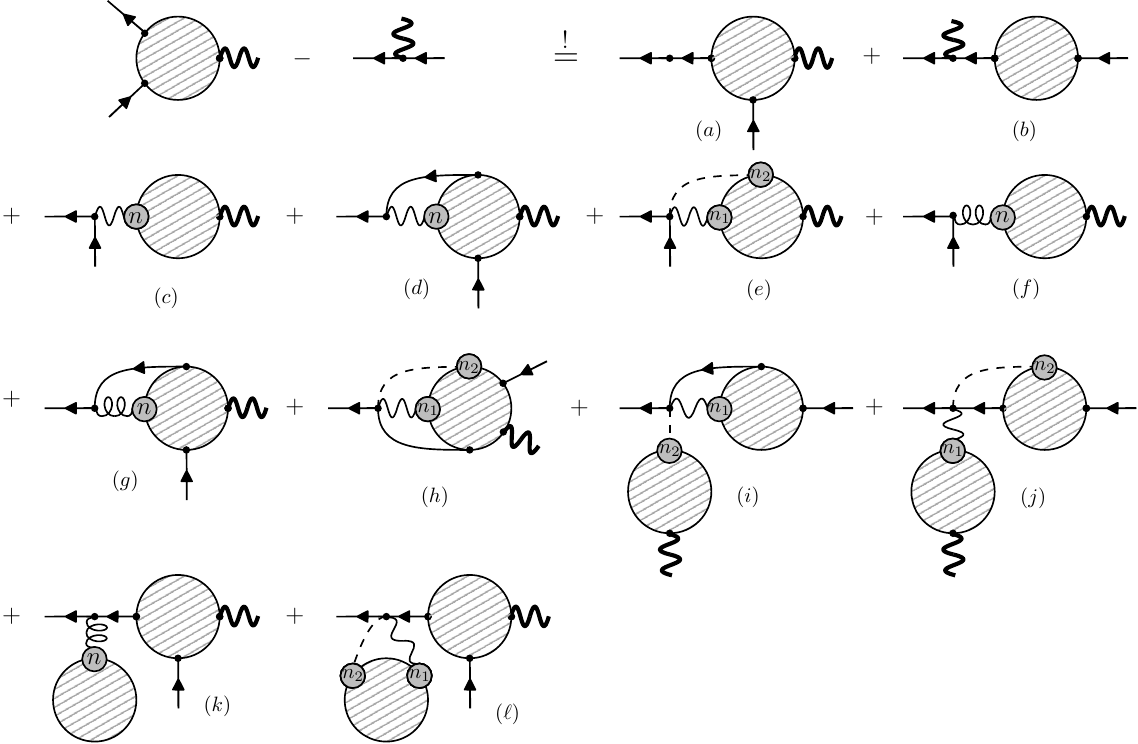}
    \caption{graphs contributing to the difference $W^{(h)|1,1;0,0}_{\eset;\eset;\eset;1}(\eta_1,\eta_2;b)- c^5_b(\eta_1,\eta_2)$.}
    \label{fig_chiral_vertex}
\end{figure}

The graphs of Fig. \ref{fig_chiral_vertex} are in one to one correspondence with those in Fig. \ref{fig_kernel_vertex}, an can be in fact bounded in the very same way. Each of them is indeed checked to be bounded by $\max\big\{C_0,C'_0\big\}^6 \mf{c}_8\l$ for $\l\le \Big(\mf{c}'_8 \max\big\{C_0,C'_0 \big\}^2 \max\big\{R',R\big\} \Big)^{-1}$, for suitable constants $\mf{c}_8,\mf{c}'_8\ge1$, independent of $R'$ and $C'_0$. Such a bound leads to the same qualitative estimate as Eq. \eqref{eq:18}, with $W^{(h)|1,1}_{1;\eset;\eset}$ and $c$ in the l.h.s. replaced by $W^{(h)|1,1;0,0}_{\eset;\eset;\eset;1}$ and $c^5$ respectively, and $C_0$ in the r.h.s. replaced by $\max\big\{C_0,C'_0\big\}$.

\medskip

Similarly, one can study the kernels $G^{(n)}B$. For instance, in the case $n=1$, following the same steps as in Subsection \ref{ssection_improved_3}, one can expand the difference $W^{(h)|0,0;0,0}_{1;\eset;\eset;1}(b;b')- \mu_{N,\l} \Pi_5^{(h,N]}(b,b')$ ending up with the same graphs as Fig. \ref{fig_kernel_0200b}, where now the external wiggly line attached to one of the blobs is thicker, corresponding to the external field $B$, in place of $G^{(1)}$. Again, the estimates for the graphs are formally identical to those in Subsection \ref{ssection_improved_3}, so they will not be repeated here. One crucial fact is that, as well as the \it{vector bubble} diagram $\Pi^{(h,N]}$, also its \it{chiral} counterpart $\Pi_5^{(h,N]}$ is uniformly bounded, in the sense that $\big\|g^A*\Pi_5^{(h,N]}\big\|_1^w\le C'_{\Pi}$ (see Corollary \ref{cor_bubble}).
\end{proof}

\begin{remark}
The property $W^{(h)|0,0;0,0}_{\eset;\eset;\eset;1}=0$, namely the third of Eq. \eqref{eq_SDB_2:1}, is a simple consequence of the \it{charge conjugation} symmetry (cf. Eq. \ref{symm:cc}). Under this transformation, $O_{5;x,\mu,\e}(\psi) \mapsto O_{5;x,\mu,\e}(\mc{Q}\psi)= O_{5;x,\mu,-\e}(\psi)$. This immediately implies that $W^{(h)|0,0;0,0}_{\eset;\eset;\eset;1}(x,\mu,\e)= W^{(h)|0,0;0,0}_{\eset;\eset;\eset;1}(x,\mu,-\e)$. On the other hand, since $O_{5;x,\mu,-\e}(\psi)= -O_{5;x,\mu,\e}(\psi)$, we also have that $W^{(h)|0,0;0,0}_{\eset;\eset;\eset;1}(x,\mu,\e)= -W^{(h)|0,0;0,0}_{\eset;\eset;\eset;1}(x,\mu,-\e)$, which implies $W^{(h)|0,0;0,0}_{\eset;\eset;\eset;1}(x,\mu,\e)= 0$.
\end{remark}

\subsection{The vector source}
\label{ssection_external_2}

According to Eq. \eqref{eq:2}, in order to restore the original dependence upon $J$ of the generating functional, we must consider as a starting potential the r.h.s. of Eq.  \eqref{eq:2}:
\begin{equation}
\label{eq_41}
\mc{V}(\psi;G(J))= \mc{V}(\psi;0)+ \sum_{n\ge 1}\sum_{1\le p\le n} 2^{N(2-n-p)} \int_{\L_B^n} d\ud{b} \int_{\L_B^p} d\ud{b}' w_{n,p}(\ud{b};\ud{b}') O^n_{\ud{b}}(\psi) G^p_{\ud{b}'}(J),
\end{equation}

with the kernels $w_{n,p}$ given by Eqs. \eqref{eq:4a},  \eqref{eq:4b}. The presence of infinitely many monomials in Eq. \eqref{eq_41} is again a consequence of \it{local phase invariance}, Eq. \eqref{eq_gauge_inv}. For dealing with such a combinatorial complication we adopt the following strategy. 

\begin{enumerate}
\item\label{it:source:1} All the terms in Eq. \eqref{eq_41} that are linear in $G(J)$ will be recast in a form which is compatible with the r.h.s. of Eq. \eqref{def_potential} after the replacement of $G^{(n)}$ by $G(J)$. In this way, for all the non-irrelevant terms produced by the multiscale integration involving only the variables $G^{(n)}$ and $\psi$, we will rely on the bounds provided by Theorem \ref{thm:IB}.

\item\label{it:source:2} All the monomials in Eq. \eqref{eq_41} involving at least two variables $G(J)$ will be checked not to produce any relevant or marginal term at lower scales (with a single exception which is however easily manageable), hence we will be able to control them by a standard multiscale analysis based on the tree-expansion techniques presented in Appendix \ref{app_tree}.
\end{enumerate}

Practically, in the r.h.s. of Eq. \eqref{eq_41} we isolate the terms that are linear in $G(J)$ and by splitting $w_{n,0}(\ud{b})= v_n(\ud{b})+ 2^{-2N}\tilde{v}_n(\ud{b})$ (cf. Eq.   \eqref{eq_16a}), we rewrite them as
\[\begin{split}
\sum_{n\ge 2} n2^{N(1-n)} \int d\ud{b} w_{n,0}(\ud{b}) G_{b_1} O_{\ud{b}}^n(\psi) \equiv \sum_{n\ge 2} n2^{N(1-n)} \int d\ud{b} \big( v_n(\ud{b})+ 2^{-2N}\tilde{v}_n(\ud{b}) \big) G_{b_1} O_{\ud{b}}^n(\psi),
\end{split}\]
so that we can rewrite Eq. \eqref{eq_41} as
\begin{equation}
\label{eq_42}
\begin{split} 
\mc{V}(\psi;G(J)) &= \mc{V}(\psi;0)+\\
&+\sum_{n\ge 1} n2^{N(1-n)} \int_{\L_B^n} d\ud{b}  v_n(\ud{b}) G_{b_1}(J) O_{\ud{b}}^n(\psi)+ \\ &+\sum_{n\ge 1} n2^{-N(1+n)} \int_{\L_B^n} d\ud{b}  \tilde{v}_n(\ud{b}) G_{b_1}(J) O_{\ud{b}}^n(\psi)+\\
& +\sum_{n\ge 2}\sum_{2\le p\le n} 2^{N(2-n-p)} \int_{\L_B^n} d\ud{b} \int_{\L_B^p} d\ud{b}' w_{n,p}(\ud{b};\ud{b}') O_{\ud{b}}^n(\psi) G_{\ud{b}'}^p(J).
\end{split}
\end{equation}

Now the idea is to treat differently the source field $G(J)$ in the three lines in the last r.h.s. of Eq. \eqref{eq_42}. To this purpose, we introduce three families of auxiliary fields:
\[ H\equiv\{H^{(n)}_b\}_{n\ge1, b\in \L_B}, \;  \dot{H}\equiv \{\dot{H}^{(n)}_b\}_{n\ge2, b\in\L_B}, \; \ddot{H}\equiv \{\ddot{H}_b\}_{b\in \L_B}, \]

and the generalized potential 
\begin{equation}
\label{def_potential_3}
\begin{split}
\mc{U}(\psi;H;\dot{H};\ddot{H};\vphi;B)&:= \mc{V}(\psi;0)+  \int_{\L_F} d\eta\; (\vphi^+_{\eta}\psi^-_{\eta} + \psi^+_{\eta}\vphi^-_{\eta})+ \int_{\L_B} db\; B_b O_{5;b}(\psi)+\\
&+\sum_{n\ge1}n 2^{N(1-n)} \int_{\L_B^n} d\ud{b}\; v_n(\ud{b}) H^{(n)}_{b_1} O_{\ud{b}}^n(\psi)+\\ &+  \sum_{n\ge2} n 2^{-(n+1)N} \int_{\L_B^n} d\ud{b}\; \tilde{v}_n(\ud{b}) \dot{H}^{(n)}_{b_1} O_{\ud{b}}^n(\psi)+\\
& +\sum_{2\le p\le n} 2^{N(2-n-p)} \int_{\L_B^n} d\ud{b}\; \int_{\L_B^p} d\ud{b}' w_{n,p}(\ud{b};\ud{b}') O_{\ud{b}}^n(\psi) \ddot{H}_{\ud{b}'}^p.
\end{split}
\end{equation}

Note that by construction $\mc{V}\big(\psi;G(J);\vphi;B\big)= \mc{U}\big(\psi;  G(J); G(J); G(J); \vphi;B\big)$\footnote{ $H=G(J)$ means that $H^{(n)}_b$ is replaced by $G_b(J)$ for every $n\geq 1$, and similarly for $\dot H,\ddot H$.}. We then consider, for every $0\le h\le N$,
\begin{equation} \mc{U}^{(h)}(\psi;H;\dot{H};\ddot{H};\vphi;B):= -\log\int P^{(h,N]}(\mc{D}\z) e^{-\mc{U}(\psi+\z;H;\dot{H};\ddot{H};\vphi;B)} \label{eq:45} \end{equation}

so that, in particular, the original generating functional $\mc{W}^{(u.v.)}(\psi;J;\vphi;B)$ can be recovered via the identification \begin{equation}\mc{W}^{(u.v.)}(\psi;J;\vphi;B)=
    \mc{U}^{(0)}\big(\psi;G(J);G(J);G(J);\vphi;B\big).\label{eq:46}
    \end{equation} As usual $\mc{U}^{(h)}$ can be expressed in an integral form analogous to Eq. \eqref{eq:9}, where now we use the symbol $U^{(h)|q,q';\tilde{q},\tilde{q}'}_{\ud{n};\ud{\dot{n}};\ddot{p};p}(\ud{\eta},\ud{\eta}';\ud{\tilde{\eta}},\ud{\tilde{\eta}}';\ud{b};\ud{\dot{b}};\ud{\ddot{b}}; \ud{b}')$ to denote the generic kernel of $\mc{U}^{(h)}$ associated with the monomial
\begin{equation}
\begin{split}
\prod_{k=1}^q \psi^+_{\eta_k} \prod_{k=1}^{\tilde{q}} \vphi^+_{\tilde{\eta}_k}  \prod_{k=1}^{q'} \psi^-_{\eta'_k} \prod_{k=1}^{\tilde{q}'} \vphi^-_{\tilde{\eta}'_k} \prod_{k=1}^{\dim\ud{n}} H^{(n_k)}_{b_k} \prod_{k=1}^{\dim\ud{\dot{n}}} \dot{H}^{(\dot{n}_k)}_{\dot{b}_k} \prod_{k=1}^{\ddot{p}} \ddot{H}_{\ddot{b}_k}
\prod_{k=1}^{p} B_{b'_k}.
\end{split}
\end{equation}

The scaling dimension associated with the kernels $U^{(h)|q,q';\tilde{q},\tilde{q}'}_{\ud{n};\ud{\dot{n}};\ddot{p};p}$ is defined, as usual, so to fit the dimensional factors in the r.h.s. of Eq. \eqref{def_potential_3}:
\begin{equation}
\label{def_Dsc_3}
D_{sc}\equiv D_{sc}:= 2- \tfrac{1}{2}(q+q')- \tfrac{3}{2}(\tilde{q}+\tilde{q}') - \dim\ud{n}- 3\dim\ud{\dot{n}}- \ddot{p}- p.
\end{equation}

\begin{remark}
\label{rmk_4}
As anticipated in Item \ref{it:source:1} below Eq. \eqref{eq_41}, the different labeling of the field $G(J)$ in terms of the auxiliary variables $H,\dot{H},\ddot{H}$, has the purpose of providing an identification between $H^{(n)}$ and the auxiliary field $G^{(n)}$ introduced in Eq. \eqref{def_potential_2}, namely
\[\mc{U}\big(\psi; \{n^2 G^{(n)}\}_{n\ge1};0;0;\vphi;B\big)= \mc{V}\big(\psi;\{G^{(n)}\}_{n\ge1};0;0;\vphi;B\big).\]

As a consequence, the kernels of $\mc{U}^{(h)}$ are related to those of $\mc{V}^{(h)}$ (cf. Eq. \eqref {eq:VhconvphiBmaJnullo} and below) via the relation:
\begin{equation}
U^{(h)|q,q';\tilde{q},\tilde{q}'}_{\ud{n};\eset;0;p}= \prod_{j=1}^s n_j^{-2} \times W^{(h)|q,q';\tilde{q},\tilde{q}'}_{\ud{n};\eset;\eset;p}, \qquad \ud{n}\equiv(n_1,\dots,n_s).
\end{equation}

The kernels $W^{(h)|q,q';\tilde{q},\tilde{q}'}_{\ud{n};\eset;\eset;0}$ have already been studied in the previous analysis. In particular for those with $\tilde{q}+\tilde{q}'=0$, we have the identification $W^{(h)|q,q;0,0}_{\ud{n};\eset;\eset;0}\equiv W^{(h)|q,q}_{\ud{n};\eset;\eset}$ and Theorem \ref{thm:IB} applies; on the other hand, for the kernels with $\tilde{q}+\tilde{q}'\ge1$ we can use Proposition \ref{prop:IB2}.
\end{remark}

We now state the result concerning the kernels not covered by the analysis carried out so far, namely those associated with at least one variable $\dot{H}$ or $\ddot{H}$.

\begin{proposition}
\label{prop_multiscale_ext}
For every $\vth,\th\in(0,1)$ fixed, there exist $C_0''\ge1$ and $\l_0''>0$ such that for $\l\le \l_0''$, the kernels $U^{(h)}$ admit the following bounds. $\big\|U^{(h)|0,0;0,0}_{\eset;\eset;2;0}\big\|_1^w\le \frac{1-2^{\th-1}}{1-2^{-\th}} C_0''^3 \l^{\frac{1}{2}}$, while for any other kernel with $|\ud{\dot{n}}|+\ddot{p}\ge1$,
\begin{equation}
\label{eq_45}
\big\|U^{(h)|q,q';\tilde{q},\tilde{q}'}_{\ud{n};\ud{\dot{n}};\ddot{p};p}\big\|_1^w\le C_0''^{1+\dim\ud{n}+\dim\ud{\dot{n}}+\ddot{p}+p} \l^{d} 2^{h D_{sc}} 2^{\th(h-N)}
\end{equation}

with $D_{sc}$ as in Eq. \eqref{def_Dsc_3} and 
$d:=
\tfrac{1}{2}(|\ud{n}|-\dim\ud{n})+ \tfrac{5}{8}|\ud{\dot{n}}|+ \tfrac{1}{4}\ddot{p} + \tfrac{1}{8}\max\left\{0, q+q'+\tilde{q}+\tilde{q}'-2\right\}$.
\end{proposition}
In virtue of Eq. \eqref{eq:46}, Proposition \ref{prop_multiscale_ext} directly implies the bounds of the kernels appearing in Theorem \ref{thm_UV}, see Subsection \ref{subsec:boundkernels} for more details.

\begin{proof}[Proof of Proposition \ref{prop_multiscale_ext}] The proof goes along the same lines as that of Proposition \ref{prop:SDB}. It is convenient to introduce the single-scale contributions to the kernels, $\wt{U}^{(h)|q,q';\tilde{q},\tilde{q}'}_{\ud{n};\ud{\dot{n}};\ddot{p};p}:= U^{(h)|q,q';\tilde{q},\tilde{q}'}_{\ud{n};\ud{\dot{n}};\ddot{p};p}- U^{(h+1)|q,q';\tilde{q},\tilde{q}'}_{\ud{n};\ud{\dot{n}};\ddot{p};p}$, for $0\le h\le N-1$. For these kernels, by using the tools of the tree expansion (see Appendix \ref{app_tree}), one gets the following bounds:
\begin{equation}
\label{eq:tree:H}
\big\|\wt{U}^{(h)|q,q';\tilde{q},\tilde{q}'}_{\ud{n};\ud{\dot{n}};\ddot{p};p; \ud{\ddot{n}}} \big\|_1^w\le \big(1-2^{\th-1}\big) C_0''^{1+\dim\ud{n}+\dim\ud{\dot{n}}+ \ddot{p}+ p} 2^{h D_{sc}} \l^{d} 2^{\th(h-N)},
\end{equation}

for $\l\le \l_0''$. The proof of Eq. \eqref{eq:tree:H} is very similar to that of Eq. \eqref{eq_tree4}, and thus will be omitted. We just mention, as a main difference from Eq. \eqref{eq_tree4}, that in the r.h.s. of Eq. \eqref{eq:tree:H} the \virg{short memory factor} $2^{\th(h-N)}$ is always present. This is is trace of the fact that the monomials involving $\dot{H}$ and $\ddot{H}$, in the r.h.s. of Eq. \eqref{def_potential_3}, are all irrelevant and do not produce, under the multiscale integration, any \virg{running coupling function}, namely a relevant or marginal term involving $\psi$.

Let us consider a kernel $U^{(h)}$ which is irrelevant, namely has $D_{sc}\le-1$. Using Eq. \eqref{eq:tree:H}, we find:
\[\begin{split} \big\|U^{(h)|q,q';\tilde{q},\tilde{q}'}_{\ud{n};\ud{\dot{n}};\ddot{p};p}\big\|_1^w &\le\sum_{h'=h}^{N-1} \big\|\wt{U}^{(h')|q,q';\tilde{q},\tilde{q}'}_{\ud{n};\ud{\dot{n}};\ddot{p};p}\big\|_1^w + \big\|U^{(N)|q,q';\tilde{q},\tilde{q}'}_{\ud{n};\ud{\dot{n}};\ddot{p};p}\big\|_1^w \\
&\le \sum_{h'=h}^N \big(1- 2^{\th-1}\big) C_0''^{1+\dim\ud{n}+\dim\ud{\dot{n}}+\ddot{p}+p} \l^{d} 2^{h' D_{sc}} 2^{\th(h'-N)}\\
&\le \big(1- 2^{\th-1}\big) C_0''^{1+\dim\ud{n}+\dim\ud{\dot{n}}+\ddot{p}+p} \l^{d} 2^{h D_{sc}} 2^{\th(h-N)} \sum_{h'=h}^N 2^{(h'-h)(\th+D_{sc})}\\
&\le C_0''^{1+\dim\ud{n}+\dim\ud{\dot{n}}+\ddot{p}+p} \l^{d} 2^{h D_{sc}} 2^{\th(h-N)},
\end{split}\]

where we used that $\|U^{(N)|q,q';\tilde{q},\tilde{q}'}_{\ud{n};\ud{\dot{n}};\ddot{p};p}\|_1^w$ admits the same bound as the r.h.s. of Eq. \eqref{eq_45} with $h=N$. Now we must analyze the non-irrelevant terms, namely those with $D_{sc}\ge0$. Observe that since (by construction) there are no terms with a single $\ddot{H}$ field in the r.h.s. of Eq. \eqref{def_potential_3}, nor they can be generated by the multiscale integration, the only non-irrelevant kernel with $|\ud{\dot{n}}|+\ddot{p}\ge1$ is the one with $|\ud{\dot{n}}|=0,\ddot{p}=2$. Even though this kernel is marginal, the \virg{short memory factor} $2^{\th(h-N)}$ in the r.h.s. of Eq. \eqref{eq:tree:H} is sufficient to provide a finite bound:
\[\begin{split}
& \big\|U^{(h)|0,0;0,0}_{\eset;\eset;2;0}\big\|_1^w\le \sum_{h'=h}^{N-1}\big\|\wt{U}^{(h')|0,0;0,0}_{\eset;\eset;2;0}\big\|_1^w\le \sum_{h'=h}^{N-1} \big(1- 2^{\th-1}\big) C_0''^{3} \l^{\frac{1}{2}} 2^{\th(h'-N)}\le \frac{1-2^{\th-1}}{1-2^{-\th}} C_0''^3\l^{\frac{1}{2}}.
\end{split}\]
\end{proof}

\section{Proof of Theorem \ref{thm_UV}.}
\label{sec:proofTHMUV}

In this section we combine the partial results obtained along Sections \ref{sect_multiscale}-\ref{sect_external} in order to prove the statements of Theorem \ref{thm_UV}. For the sake of clarity we stress the dependence on $h^*_M$ of the generating functional in Eq.  \eqref{def_functional_UV}, by writing it as $\mc{W}^{(u.v.|h^*_M)}$, and similarly for its kernels appearing in Theorem \ref{thm_UV}. Recall that by Eq. \eqref{eq:rescaling} the kernels of $\mc{W}^{(u.v.|h^*_M)}$ and $\mc{W}^{(u.v.|0   )}$ are related by a simple rescaling.

\subsection{Bounds for the kernels}
\label{subsec:boundkernels}
In order to prove Eqs. \eqref{eq_50}-\eqref{eq_51b} we follow the structure of the paper by collecting the desired bounds first for $\varphi,J,B=0$ (i.e. $\tilde q,\tilde q',p,p'=0$, cf. Sections  \ref{sect_tree_expansion}-\ref{sect_improved}) then for $J=0$ (i.e. $p=0$, cf. Section \ref{ssection_external_1}) and finally in the  most general case  (i.e. $q, q',\tilde q,\tilde q',p, p' \neq 0$, cf. Section \ref{ssection_external_2}).
\paragraph{Case $\varphi,J,B=0$.} If also $\psi=0$, by definition (cf. Eqs.  \eqref{eq_main_potential}, \eqref{eq:9} for $h=0$ and Eq. \eqref{eq_tree2}): 
\begin{equation}
\mc{F}^{(u.v|0)}=L^{-2}W^{(0)|0,0}_{\eset;\eset;\eset}=\sum_{h=0}^{N-1} L^{-2} \widetilde W^{(0)|0,0}_{\eset;\eset;\eset},
\end{equation}

which, after Eq. \eqref{eq_tree4}, implies that $|\mc{F}^{(u.v.|0)}|\leq \frac{1}{3}(1-2^{-\theta})C_0 \lambda 2^{2N}$. Rescaling with Eq. \eqref{eq:rescaling}, $\mc{F}^{(u.v.|h^*_M)}=2^{2h^*_M}\mc{F}^{(u.v.|0)}$ and $\lambda \mapsto \lambda 2^{-2h^*_M}$, $N \mapsto N-h^*_M$ so that $|\mc{F}^{(u.v.|h^*_M)}|\leq C_\star (\l2^{-2h^*_M})2^{2N}$, for $C_\star \geq \tfrac{1}{3}(1-2^{-\theta})C_0$, i.e. the first of Eq. \eqref{eq_51a}.

If $\psi \neq 0$, combining Eqs. \eqref{eq:5}, \eqref{eq:v_hbackward} and Eq. \eqref{eq:9}, and recalling the rescaling in Eq. \eqref{eq:rescaling}, the desired bounds in this case follow from Theorem \ref{thm:IB} (applied with $h=0$): we thus find the constraint that $N_\star \geq N_0$, $\l_\star\leq \l_0$ and $C_\star \geq C_0 \max\big\{R,\frac{1}{3}(1+2^{-\theta})\big\}.$

\paragraph{Case $J=0$, $(\vphi,B)\ne0$.}
\medskip
In this case, since $\mc{W}^{(u.v.|0)}(\psi;\vphi;0;B)= \mc{V}^{(0)}(\psi;0;0;0;\vphi;B)$ (cf. Eqs.  \eqref{eq:44}, \eqref{eq:VhconvphiBmaJnullo}), Proposition \ref{prop:IB2} readily implies the claims of Theorem \ref{thm_UV} concerning the kernels with zero $J$ variables, namely $W^{(u.v.|0)}_{q,q';\tilde{q},\tilde{q}';0;p}$, with
the constraint $N_{\star}\ge N'_0, \l_{\star}\le\l_0'$ and $C_{\star}\ge C_0'$. For $W^{(u.v.|h^*_M)}_{q,q';\tilde{q},\tilde{q}';0;p}$ one simply uses the rescaling in Eq. \eqref{eq:rescaling}.

\paragraph{Case $J\ne0$.}
In this case the kernels  $U^{(0)|q,q';\tilde{q},\tilde{q}'}_{\ud{n};\ud{\dot{n}};\ddot{p};p'}$ of $\mc{U}^{(0)}$ (cf. Eq. \eqref{eq:45})  collect all the information about $W^{(u.v.|0)}_{q,q';\tilde{q},\tilde{q};p;p'}$,
since $
\mc{W}^{(u.v.|0)}(\psi;J;\vphi;B)= \mc{U}^{(0)}\big(\psi;G(J);G(J);G(J);\vphi;B\big)$ (compare Eqs. \eqref{eq:44}, \eqref{eq:45} and below). We discuss first the strategy in two cases of importance.

\medskip 

\emph{The tadpole $G(J)$.} 
We have that
\begin{equation}
\label{eq:33}
\begin{split}
W^{(u.v.|0)}_{0,0;0,0;1;0}(b)&= \sum_{n\ge1}U^{(0)|0,0;0,0}_{n;\eset;\eset}(b)+ \sum_{\dot{n}\ge2} U^{(0)|0,0;0,0}_{\eset;\dot{n};\eset}(b)\\
&= \sum_{n\ge1}\frac{1}{n^2} W^{(0)|0,0}_{n;\eset;\eset}(b)+ \sum_{\dot{n}\ge2} U^{(0)|0,0;0,0}_{\eset;\dot{n};\eset}(b).
\end{split}
\end{equation}
Note that the above kernels are actually constant w.r.t. $b\in\L_B$ due to the \it{charge conjugation} and the \it{axes flip} symmetries (cf. Lemma \ref{lemma:symm}) and the translation invariance of the theory. Using the bounds for the kernels $W^{(0)|0,0}_{n;\eset;\eset}$ from Theorem \ref{thm:IB}, and the bounds for $U^{(0)|0,0;0,0}_{\eset;\dot{n};0;0}$ from Proposition \ref{prop_multiscale_ext}, we find:
\[\begin{split}
\big\|W^{(u.v.|0)}_{0,0;0,0;1;0}\big\|_1^w&\le \sum_{n\ge1} \tfrac{1}{n^2} \big\|W^{(0)|0,0}_{n;\eset;\eset}\big\|_1^w + \sum_{\dot{n}\ge2} \big\|U^{(0)|0,0;0,0}_{\eset;\dot{n};0;0}\big\|_1^w \\
&\le \sum_{n\ge1} \tfrac{1}{n^2} C_0 2^N  \l^{\frac{n-1}{2}}+ \sum_{\dot{n}\ge2} C_0''^2\l^{\frac{5}{8}\dot{n}} 2^{-\th N}\le 2^N\Big(\tfrac{\pi^2}{6}C_0^2+ \tfrac{\sqrt{2}}{\sqrt{2}-1} C_0''^2\Big),
\end{split}\]
for $\l\le \min\{\l_0,\l_0'',\tfrac{1}{2}\}$; by rescaling $\|W^{(u.v.|h^*_M)}_{0,0;0,0;1;0}\|^w_1=2^{h^*_M}\|W^{(u.v.|0)}_{0,0;0,0;1;0}\|^w_1$ one obtains the third of Eq. \eqref{eq_51a}.

\emph{The polarization bubble $G(J)G(J)$.} We have that
\begin{equation}\label{eq:30}\begin{split}W^{(u.v.|0)}_{0,0;0,0;2;0}(b,b')& =W^{(0)|0,0}_{(1,1);\eset;\eset}(b,b')+\sum_{\substack{n_1,n_2\ge1\\ n_1+n_2\ge3}} \tfrac{1}{n_1^2 n_2^2} W^{(0)|0,0}_{(n_1,n_2);\eset;\eset}(b,b')  \\&+ U^{(0)|0,0;0,0}_{\eset;\eset;2;0}(b,b')+\hspace{-5pt} \sum_{n\ge1,\dot{n}\ge2} U^{(0)|0,0;0,0}_{n;\dot{n};0;0}(b,b'),\end{split}\end{equation}which, again in force of Theorem \ref{thm:IB} and Proposition \ref{prop_multiscale_ext}, implies for $\l\le \min\{\l_0,\l_0'',\tfrac{1}{2}\}$:
\begin{equation}
\label{eq:47}
\begin{split}&
\big\|g^A*W^{(u.v.|0)}_{0,0;0,0;2;0}\big\|_1^w\le\\&C_0^3+ \sum_{\substack{n_1,n_2\ge1\\ n_1+n_2\ge3}} \frac{\l^{\frac{n_1+n_2-2}{2}}}{n_1^2 n_2^2} + C_0''^3\l^{\frac{1}{2}} \sum_{n\ge1,\dot{n}\ge2} C_0''^3 \l^{\frac{1}{2}(n-1)+ \frac{5}{8}\dot{n}}\le C_0^3+ C_0''^3+ \tfrac{\sqrt{2}}{\sqrt{2}-1}(1+ C_0''^3). \end{split}\end{equation}

By the rescaling, $g^A_{\mu,\nu}(x-y)\mapsto g^A_{\mu,\nu}(2^{h^*_M}(x-y))$ and $W^{(u.v.|0)}_{0,0;0,0;2;0}\mapsto W^{(u.v.|h^*_M)}_{0,0;0,0;2;0}$, with

\[W^{(u.v.|h^*_M)}_{0,0;0,0;2;0}\big((x,\mu,\e),(y,\nu,\e')\big)= 2^{2h^*_M}W^{(u.v.|0)}_{0,0;0,0;2;0}\big((2^{h^*_M}x,\mu,\e),(2^{h^*_M}y,\nu,\e')\big),\]

so that $\big\|g^A*W^{(u.v.|h^*_M)}_{0,0;0,0;2;0}\big\|_1^w$ is bounded by the r.h.s. of Eq. \eqref{eq:47} times an extra $2^{-2h^*_M}$, proving the first bound in Eq. \eqref{eq_51b} for $C_{\star}$ greater than the r.h.s. of Eq. \eqref{eq:47}.

The desired bound for the kernel $G(J)B$, i.e. the second of Eq. \eqref{eq_51b}, follows analogously.

\medskip

\emph{Remaining kernels.}  Similar (though more cumbersome) steps yield finite bounds for all the remaining kernels of arbitrary order in the fields $G(J)$. As above one has to express the kernels of $\mc{W}^{(u.v.|0)}$ in terms of those of $\mc{U}^{(0)}$, via the relation Eq. \eqref{eq:46}, and then exploit Theorem \ref{thm:IB} and Propositions \ref{prop:IB2}, \ref{prop_multiscale_ext} for bounding the various terms that appear. All in all one finds the validity of the whole collection of bounds in Item \ref{itt:1} of Theorem \ref{thm_UV}, for $N_\star$ large enough and suitable constants $C_{\star}$ large enough and $\l_{\star}$ small enough.

\subsection{Correlation functions}
Here we discuss the Item \ref{itt:2} of Theorem \ref{thm_UV}, concerning the two-point function $\hat{S}^{(u.v.)}$, the current-current function $\hat{\Sigma}^{(u.v.)}_{\mu,\nu}$ and its chiral counterpart. We keep using the same conventions introduced at the beginning of Section \ref{sec:proofTHMUV}.

\paragraph{The two-point Schwinger function.} By construction we have that \sloppy $\hat{S}^{(u.v.|0)}_{s,s'}(k)= \int_{\L} dx\, e^{ik\cdot x} W^{(u.v.|0)}_{0,0;1,1;0;0}\big((x,s),(0,s')\big)$. In analogy with Eqs. \eqref{eq:24a}, \eqref{eq:24b}, it is straightforward to check that the kernel $W^{(u.v.|0)}_{0,0;1,1;0;0}$ admits the expansion:
\begin{equation}
\begin{split}
&W^{(u.v.|0)}_{0,0;1,1;0;0}\big((x,s),(0,s')\big)=\\
&g^{(0,N]}_{s,s'}(x) - \sum_{s_1,s_2\in\{1,2\}}\int_{\L} dydz\, g^{(0,N]}_{s,s_1}(x-y) W^{(u.v.|0)}_{1,1;0,0;0;0}\big((y,s_1), (z,s_2)\big) g^{(0,N]}_{s_2,s'}(z).
\end{split}
\end{equation}

It follows that
\[\begin{split}
&\hat{S}^{(u.v.|0)}_{s,s'}(k)= \hat{g}^{(0,N]}_{s,s'}(k) - \sum_{s_1,s_2\in\{1,2\}} \hat{g}^{(0,N]}_{s,s_1}(k) \hat{g}^{(0,N]}_{s_2,s'}(k) \int_{\L}dy\, e^{ik\cdot y} W^{(u.v.|0)}_{1,1;0,0;0;0}\big((y,s_1), (0,s_2)\big).
\end{split}\]

Note that since $1-\chi\equiv0$ on $[0,\frac{1}{2}]$, the bound in Eq. \eqref{eq:bound:Wilson} implies:
\[\begin{split}
|\hat{g}^{(0,N]}_{s,s'}(k)|&= \big(1-\chi(|k|)\big) \big|\big(-i\slashed{s}(k)+ m_N+ M_N(k) \big)^{-1}_{s,s'}\big|\le \sqrt{2}\pi \frac{1-\chi(|k|)}{\sqrt{|k|^2 + m_N^2}}\le \frac{4\sqrt{2}\pi}{|k|+1}.
\end{split}\]

Using also the bound for the kernel $W^{(u.v.|0)}_{1,1;0,0;0;0}\equiv W^{(0)|1,1}_{\eset;\eset;\eset}$ from Theorem \ref{thm:IB}, we find:
\[\begin{split}
&\big|\hat{S}^{(u.v.|0)}_{s,s'}(k)- \hat{g}^{(0,N]}_{s,s'}(k)\big| \le \sum_{s_1,s_2\in\{1,2\}} \frac{32\pi^2}{1+|k|^2} \big\|W^{(u.v.|0)}_{1,1;0,0;0;0}\big\|_1^w\le \frac{32\pi^2R\l}{1+|k|^2},
\end{split}\]

which, after the rescaling in Eq. \eqref{eq:rescaling}, yields the desired bound in Eq. \eqref{eq:23a} for $\hat{S}^{(u.v.|h^*_M)}_{s,s'}(k)$, if $C_{\star}\ge 32\pi^2R$.

\paragraph{The non-interacting current functions.} For the sake of concreteness we discuss only the \it{vector current function}, namely we prove the first bound of Eq.  \eqref{eq:23b} and Eq.  \eqref{eq_24a} since the second of Eq. \eqref{eq:23b} and Eq. \eqref{eq_24b}, concerning the \it{chiral current function}, can be proved with the same strategy. 

\medskip \emph{Proof of Eq. \eqref{eq:23b}.} 
From the definition $\hat{\Sigma}^{(u.v.|0)}_{\mu,\nu}(p):= \int_{\L}dx e^{-ip\cdot x} \frac{\d^2\mc{W}^{(u.v.|0)}}{\d J_{\mu,x}\d J_{\nu,0}}(0)$, and the structure of $\mc{W}^{(u.v.|0)}$ (cf. Eq. \eqref{eq_main_potential}), it is straightforward to check that
\begin{equation}
\label{eq:32}
\hspace{-7pt}\hat{\Sigma}^{(u.v.|0)}_{\mu,\nu}(p)= -2\l \Big( 2^{-N}\d_{\mu,\nu} W^{(u.v.|0)}_{0,0;0,0;1;0} + \hspace{-3pt} \sum_{\e,\e'=\pm} \hspace{-3pt}\e\e' \hspace{-4pt}\int_{\L} \hspace{-3pt} dx\, e^{-ip\cdot x} W^{(u.v.|0)}_{0,0;0,0;2;0}\big((x,\mu,\e),(0,\nu,\e') \big) \Big), \end{equation}

where recall that the kernel $W^{(u.v.|0)}_{0,0;0,0;1;0}$ is actually a constant (see Eq. \eqref{eq:33} and below). Now we explicitly characterize the dominant term in $\lambda$ of the r.h.s. of Eq. \eqref{eq:32}. 
About the first term in the r.h.s., we start from the decomposition in Eq. \eqref{eq:33} for the kernel $W^{(u.v.|0)}_{0,0;0,0;1;0}$ and we apply Lemma \ref{lemma_id_B} to further expand the kernel $W^{(0)|0,0}_{1;\eset;\eset}$:
\begin{equation}
\label{eq:34}
\hspace{-4pt} W^{(0)|0,0}_{1;\eset;\eset}(b)= \mu_{N,\l} 2^N \mc{T}^{(0,N]} +  \mu_{N,\l}\int d\eta_1 d\eta_2 c_b(\eta_1,\eta_2)\int d\eta d\eta' g^{(h,N]}_{\eta,\eta_1} g^{(h,N]}_{\eta_2,\eta'} W^{(0)|1,1}_{\eset;\eset;\eset}(\eta',\eta),\end{equation}

where 
\begin{equation}
\label{def:tadpole}
\mc{T}^{(h,N]}:= 2^{-N}\mc{E}^T_{(h,N]}(O_b)\equiv - 2^{-N}\int_{\L_F^2} d\eta_1 d\eta_2 c_b(\eta_1,\eta_2) g^{(h,N]}_{\eta_2,\eta_1}.
\end{equation}

Plugging the decomposition in Eq. \eqref{eq:34} into Eq. \eqref{eq:33}, we find
\begin{equation}
\label{eq:35}
\begin{split}
& \big|2^{-N}\big(W^{(u.v.|0)}_{0,0;0,0;1;0}- \mc{T}^{(0,N]}\big)\big|\le\\
&(1- \mu_{N,\l})2^{-N}\|c\|_1^w \|g^{(0,N]}\|_{\infty}+ 2^{-N}\|c\|_1^w \|g^{(0,N]}\|_1^w \|g^{(0,N]}\|_{\infty}\big\|W^{(0)|1,1}_{\eset;\eset;\eset}\big\|_1^w+\\
& \sum_{n\ge2} \tfrac{1}{n^2} \big\|W^{(0)|0,0}_{n;\eset;\eset}\big\|_1^w + \sum_{\dot{n}\ge2} \big\|U^{(0)|0,0;0,0}_{\eset;\dot{n};0;0}\big\|_1^w\le C_1\l^{\frac{1}{2}},
\end{split}
\end{equation}

for a suitable constant $C_1\ge1$, where in the last step we used Theorem \ref{thm:IB} for the bounds of $W^{(0)|1,1}_{\eset;\eset;\eset}$ and $W^{(0)|0,0}_{n;\eset;\eset}$, and Proposition \ref{prop_multiscale_ext} for the bound of $U^{(0)|0,0;0,0}_{\eset;\dot{n};0;0}$. 

About the second summand in the r.h.s. of Eq. \eqref{eq:32}, by Eq. \eqref{eq:30} we have that
\[\begin{split}
&\big\|g^A*\big(W^{(u.v.|0)}_{0,0;0,0;2;0}- \Pi^{(0,N]} \big)\|_1^w\le\big\|g^A*\big( W^{(0)|0,0}_{(1,1);\eset;\eset} - \Pi^{(0,N]} \big)\big\|_1^w+ \\
&+\|g^A\|_1^w \left( \sum_{\substack{n_1,n_2\ge1\\ n_1+n_2\ge3}} \tfrac{1}{n_1^2 n_2^2}  \big\|W^{(0)|0,0}_{(n_1,n_2);\eset;\eset}\big\|_1^w + \big\|U^{(0)|0,0;0,0}_{\eset;\eset;2;0}\big\|_1^w +\hspace{-5pt} \sum_{n\ge1,\dot{n}\ge2} \big\|U^{(0)|0,0;0,0}_{n;\dot{n};0;0}\big\|_1^w\right)\le C_2\l^{\frac{1}{2}},
\end{split}\]

for a suitable constant $C_2\ge1$. In momentum space this bound reads:
\begin{equation}
\label{eq:31}
\begin{split}
& \Big|\sum_{\e=\pm} \e \int_{\L}dx\, e^{-ip\cdot x} W^{(u.v.|0)}_{0,0;0,0;2;0}\big((x,\mu,\e),(0,\nu,\e') \big) - \e' \hat{\Pi}^{(0,N]}_{\mu,\nu}(p) \Big|\le \big(|p|^2+4)C_2 \l^{\frac{1}{2}},
\end{split}
\end{equation}

where we used the fact that $|\s(p)|^2+M^2\le |p|^2 +4$ (recall that we are assuming $h^*_M=0$) and we have introduced:
\begin{equation}
\label{def:bubble:Fourier}
\begin{split}
\hat{\Pi}^{(h,N]}_{\mu,\nu}(p)&:= \frac{1}{2}\sum_{\e,\e'=\pm} \e\e' \int_{\L} \frac{dx dy}{L^2} e^{-ip\cdot(x-y)} \Pi^{(h,N]}\big((x,\mu,\e),(y,\nu,\e')\big)\\
&\equiv \sum_{\e=\pm} \e\e' \int_{\L} \frac{dx dy}{L^2} e^{-ip\cdot(x-y)} \Pi^{(h,N]}\big((x,\mu,\e),(y,\nu,\e')\big),
\end{split}
\end{equation}

with $\Pi^{(h,N]}$ defined in Eq. \eqref{eq_29a}, and the last equality holds in force of the \it{charge conjugation} symmetry in Eq. \eqref{symm:cc}. By simply reverting the rescaling in Eq. \eqref{eq:rescaling}, Eqs.  \eqref{eq:34}, \eqref{eq:31} become:
\begin{align}
&\big|2^{-N}\big(W^{(u.v.|h^*_M)}_{0,0;0,0;1;0}- \mc{T}^{(h^*_M,N]}\big)\big|\le C_1(2^{-2h^*_M}\l)^{\frac{1}{2}},\\
& \Big|\sum_{\e=\pm} \e \int_{\L}dx\, e^{-ip\cdot x} W^{(u.v.|h^*_M)}_{0,0;0,0;2;0}\big((x,\mu,\e),(0,\nu,\e') \big) - \e' \hat{\Pi}^{(h^*_M,N]}_{\mu,\nu}(p) \Big|\le 5C_2 (2^{-2h^*_M}\l)^{\frac{1}{2}},
\end{align}

for every $p\in\L^*$ such that $|p|\le 2^{h^*_M}$. Plugging these bounds back into Eq. \eqref{eq:32}, and letting: 
\begin{equation}
\label{eq:41}
\hat{\mf{B}}^{(u.v.|h^*_M)}_{\mu,\nu}(p,M,m_N,L):= -2\Big(\d_{\mu,\nu} \mc{T}^{(h^*_M,N]}(m_N,L)+ 2 \hat{\Pi}^{(h^*_M,N]}_{\mu,\nu}(p,m_N,L) \Big),
\end{equation}

we find: 
\begin{equation}
\label{eq:42}
\begin{split}
&\Big|\hat{\Sigma}^{(u.v.|h^*_M)}_{\mu,\nu}(p) - \l\hat{\mf{B}}^{(u.v.|h^*_M)}_{\mu,\nu}(p,M,m_N,L)\Big|\le 2\l 2^{-N}\d_{\mu,\nu} \Big| W^{(u.v.|h^*_M)}_{0,0;0,0;1;0}- \mc{T}^{(h^*_M,N]}\Big| + \\
&+2\l \Big|\sum_{\e,\e'=\pm} \e\e' \int_{\L}dx\, e^{-ip\cdot x} W^{(u.v.|h^*_M)}_{0,0;0,0;2;0}\big((x,\mu,\e),(0,\nu,\e') \big) - 2\hat{\Pi}^{(h^*_M,N]}_{\mu,\nu}(p) \Big|\\ &\le 2(C_1+5C_2)2^{-h^*_M}\l^{\frac{3}{2}},
\end{split}
\end{equation}

thus proving the first of Eq. \eqref{eq:23b}, for $C_{\star}\ge 2(C_1+5C_2)$.

\medskip

For the incoming analysis it is also useful to introduce $\hat{\mf{B}}_{\mu,\nu}(p,M,m_N,L):=\lim_{\l \to 0^+} \frac{1}{\lambda} \hat{\Sigma}_{\mu,\nu}(p)$, which, in force of Eq. \eqref{eq:WI5a},  satisfies the following Ward Identity:
\begin{equation}
\label{eq:WI:bubble}
\sum_{\mu=0,1} \s_{\mu}(p) \hat{\mf{B}}_{\mu,\nu}(p,M,m_N,L)=0, \; \forall p\in \tfrac{2\pi}{L}\mbb{Z}^2 \cap (-\tfrac{\pi}{a},\tfrac{\pi}{a}].
\end{equation}

It is straightforward to check that
\begin{equation}
\hat{\mf{B}}_{\mu,\nu}(p,M,m_N,L)= -2\Big(\d_{\mu,\nu} \mc{T}^{(\le N)}(m_N,L)+ 2 \hat{\Pi}^{(\le N)}_{\mu,\nu}(p,m_N,L) \Big),
\end{equation}

where $\mc{T}^{(\le N)}$ and $\hat{\Pi}^{(\le N)}_{\mu,\nu}$ are as in the r.h.s. of Eqs. \eqref{def:tadpole} and \eqref{def:bubble:Fourier} respectively, with the fermionic propagator $g^{(h,N]}$ in their definition replaced by $g^{(\le N)}$. In order to proceed we make use of the following proposition for the characterization of $\hat{\mf{B}}_{\mu,\nu}$ and $\hat{\mf{B}}^{(u.v.)}_{\mu,\nu}$.

\begin{proposition}
\label{prop_limitebollaGG}
For every $M,m_N,L$ as in Theorem \ref{thm_UV}, with $m_N\ne0$, the functions $\L^*\ni p\mapsto \hat{\Pi}^{(h^*_M,N]}(p,m_N,L)$ and $\L^*\ni p\mapsto \hat{\Pi}^{(h^*_M,N]}(p,m_N,L)$ admit a continuous extension over $(-\frac{\pi}{a},\frac{\pi}{a}\big ]^2$, which, with some abuse of notation, we will denote by the same symbol; their $L\to\infty$ limit  exists and is reached uniformly w.r.t. $p\in(-\frac{\pi}{a},\frac{\pi}{a}]^2$. Besides, for $p\ne0$,
\begin{equation}
\begin{split}
    & \lim_{M \to 0^+}\lim_{m_N\to 0}\lim_{L \to \infty} \hat \Pi^{(h_M^*,N]}_{\mu,\nu}(p,m_N,L)=\lim_{m_N\to 0}\lim_{L \to \infty} \hat \Pi^{(\leq N)}_{\mu,\nu}(p,m_N,L)=:\hat \Pi_{\mu,\nu}(p), \\ & \lim_{M \to 0^+}\lim_{m_N\to 0}\lim_{L \to \infty} \mc{T}^{(h_M^*,N]}(m_N,L)=\lim_{m_N\to 0}\lim_{L \to \infty} \mc{T}^{(\leq N)}(m_N,L)=: \mc{T}, 
    \end{split}\label{eq:bubbletadpolelimits}
\end{equation}
where 
\begin{equation}
    \hat \Pi_{\mu,\nu}(p)=\frac{1}{8\pi |p|^2}\left ( 2p_\mu p_\nu -|p|^2 \d_{\mu,\nu} \right )+\tilde R_{\mu,\nu}(ap),\label{eq:bubblecomput}
\end{equation}
for a suitable function $\tilde{R}_{\mu,\nu}$ on $(-\pi,\pi]^2$, such that $|\tilde R_{\mu,\nu}(q)-\tilde R_{\mu,\nu}(0)|\leq C|q|^{\frac{1}{2}}$, for some constant $C>0$ independent of $N$.
\end{proposition}

The proof is postponed to Appendix \ref{subapp:proofprop_limitebollaGG}. Let us see how Proposition \ref{prop_limitebollaGG} implies Eq. \eqref{eq_24a} (the derivation of Eq. \eqref{eq_24b} can be done analogously). First note that after taking the limit $L\to\infty$ at both sides of Eq. \eqref{eq:WI:bubble}, in force of the uniform convergence established by Proposition \ref{prop_limitebollaGG}, the equality above extends to every $p\in (-\frac{\pi}{a},\frac{\pi}{a}]^2$. Taking also the limit $m_N\to0$ and then $M\to 0^+$, we obtain:
\[
\sum_{\mu=0,1} \tfrac{i}{a}\big(e^{-ia p_{\mu}}-1\big) \left( \frac{1}{4\pi |p|^2}\left ( 2p_\mu p_\nu -|p|^2 \d_{\mu,\nu} \right )+ 2\tilde R_{\mu,\nu}(ap) +\mc{T} \d_{\mu,\nu} \right)=0,
\]

for every $p \in(-\tfrac{\pi}{a},\tfrac{\pi}{a}]^2\setminus\{0\}$. Using the continuity in 0 of $\tilde{R}_{\mu,\nu}$, we find that
\[
\sum_{\mu=0,1} p_{\mu} \left( \frac{1}{4\pi |p|^2}\left ( 2p_\mu p_\nu -|p|^2 \d_{\mu,\nu} \right )+ 2\tilde R_{\mu,\nu}(0) +\mc{T} \d_{\mu,\nu} \right)+ \mc{O}\big(a^{\frac{1}{2}}|p|^{\frac{3}{2}} \big)=0,
\]

and for arbitrariness of $p$, we must have that $2\tilde{R}_{\mu,\nu}(0)+ \mc{T}\d_{\mu,\nu} = -\frac{1}{4\pi}\d_{\mu,\nu}$. Hence, going back to $\hat{\mf{B}}^{(u.v.|h^*_M)}_{\mu,\nu}$, we find that
\[
\begin{split}
\lim_{M \to 0^+}\lim_{m_N\to 0}\lim_{L \to \infty}\hat{\mf{B}}^{(u.v.|h^*_M)}_{\mu,\nu}(p,M,m_N,L)&= \!-2\! \left(\frac{1}{4\pi |p|^2}\left ( 2p_\mu p_\nu -|p|^2 \d_{\mu,\nu} \right )+ 2\tilde R_{\mu,\nu}(ap) +\mc{T} \d_{\mu,\nu}\right)\\
&= \frac{1}{\pi}\left( \d_{\mu,\nu}- \frac{p_{\mu} p_{\nu}}{|p|^2} \right) - 4\Big(\tilde{R}_{\mu,\nu}(ap)- \tilde{R}_{\mu,\nu}(0) \Big),
\end{split}\]

for every $p\ne0$, which proves Eq. \eqref{eq_24a} with $R_{\mu,\nu}(p)= - 4\Big(\tilde{R}_{\mu,\nu}(ap)- \tilde{R}_{\mu,\nu}(0) \Big)$.

\appendix

\section{Consequences of the local phase invariance}
\label{app_gauge_inv} 

In this section we discuss two main consequences of the local phase invariance of the theory as in Eq. \eqref{eq_gauge_inv}, namely the $\xi-$independence (Lemma \ref{lemma_gauge_inv}) and the Ward identities (Eqs. \eqref{eq:WI5a} and \eqref{eq:WI5b}). 

\begin{proof}[Proof of Lemma \ref{lemma_gauge_inv}.]

Letting
\begin{equation} 
\label{eq_46}
Z_{\xi}(\vphi;J;B):= e^{-\mc{W}_{\xi}(\vphi;J;B)} \equiv \frac{\int \mbb{P}_{\xi}(\mc{D}A) \int \mc{D}\psi \; e^{-I(A+J,\psi)-(\vphi,\psi)-i(B,\jmath^5(\psi))} }{\int \mc{D}\psi \; e^{-I(0,\psi)}},
\end{equation}

we are going to show that $Z_{\xi}(0;J;B)$ is constant w.r.t. $\xi\in [0,1]$. Observe that $Z_{\xi}(\vphi;J;B)$, as a function of $\xi$, is continuous in $[0,1]$ and continuously differentiable in $(0,1)$, with 
\begin{equation}
\label{eq_48}
\de_{\xi}Z_{\xi}(\vphi;J;B)= \frac{1}{2} \sum_{\mu,\nu\in\{0,1\}} \int_{\L^2}dxdy \; \de_{\xi} g^{A,\xi}_{\nu,\mu}(y-x) \frac{\d^2Z_{\xi}}{ \d J_{\mu,x} \d J_{\nu,y}}(\vphi;J;B),
\end{equation}

as a standard formula for derivatives of Gaussian measures, where $\frac{\d}{\d J_{\mu,x}}\equiv 2^{2N}\frac{\de}{\de J_{\mu,x}}$. Hence it suffices to show that $\de_{\xi}Z_{\xi}(0;J;B)=0$ for every $\xi\in(0,1)$. Eq. \eqref{eq_48} can be rewritten in Fourier space:
\begin{equation}
\label{eq_49}
\de_{\xi}Z_{\xi}(\vphi;J;B)= \frac{1}{2}\sum_{\mu,\nu} \int_{\L^*} dk\; \de_{\xi}\left(\frac{1-\xi}{\xi|\s(k)|^2+M^2} \right) \frac{\s_{\nu}(k)\s_{\mu}^*(k)}{|\s(k)|^2+M^2} \hat{K}^{\xi}_{\mu,\nu}(k,-k)_{\vphi;J;B},
\end{equation}

where $\hat{K}^{\xi}_{\mu,\nu}(k,p)_{\vphi;J;B}:= \int_{\L^2} dxdy e^{ik\cdot x +ip\cdot y} \frac{\d^2 Z_{\xi}}{\d J_{\mu,x} \d J_{\nu,y}}(\vphi;J;B)$ and $\s_{\mu}(k)$ defined after Eq. \eqref{eq:boson_propag}. Now one has the following equation:
\begin{equation}
\label{eq:WI3}
\sum_{\mu=0,1}\s_{\mu}^*(k) \hat{K}^{\xi}_{\mu,\nu}(k,p)_{0;J;B}=0, \qquad\forall k,p \in \L^*,
\end{equation}

which readily implies the vanishing of the r.h.s. of Eq. \eqref{eq_49} at $\vphi=0$, which in turn implies the independence of $Z_{\xi}(0;J;B)$ upon $\xi\in[0,1]$. In order to prove Eq. \eqref{eq:WI3}, we must exploit the local phase invariance of the theory, namely
\begin{equation}
\label{eq:WI1}
Z_{\xi}(e^{i\mathsf{e}\a}\vphi;J+ \nabla\a;B)= Z_{\xi}(\vphi;J;B), \qquad \forall \a\in\mbb{R}^{\L}
\end{equation}

with the understanding that $(e^{i\mathsf{e}\a}\vphi)_{x,s}^{\pm}\equiv e^{\pm i\mathsf{e}\a(x)} \vphi^{\pm}_{x,s}$,  $(A+ \nabla \a)_{\mu,x}\equiv A_{\mu,x}+ \nabla_{\mu}\a(x)$, and where $\nabla_{\mu}\a(x)\equiv a^{-1}\big(\a(x+a\hat{e}_{\mu})- \a(x) \big)$. By differentiating both sided of Eq. \eqref{eq:WI1} w.r.t. $\a(x)$ and setting $\a=0$, we find:
\begin{equation}
\label{eq:WI2}
\begin{split}
&a^{-1}\sum_{\mu=0,1} \Bigg( \frac{\d Z_{\xi}}{\d J_{\mu,x-a\hat{e}_{\mu}}}(\vphi;J;B)- \frac{\d Z_{\xi}}{\d J_{\mu,x}}(\vphi;J;B) \Bigg)  =\\
& -i\mathsf{e} \sum_{s=1,2} \Bigg(\vphi^+_{x,s}\frac{\d Z_{\xi}}{\d \vphi_{x,s}^+}(\vphi;J;B)- \vphi^-_{x,s}\frac{\d Z_{\xi}}{\d \vphi_{x,s}^-}(\vphi;J;B) \Bigg).
\end{split}\end{equation}

If we further differentiate both sides of Eq. \eqref{eq:WI2} w.r.t. $J_{\nu,y}$ and we set $\vphi=0$, we get:

\begin{equation}
\begin{split}
a^{-1}\sum_{\mu=0,1} \Bigg( \frac{\d^2 Z_{\xi}}{\d J_{\mu,x-a\hat{e}_{\mu}} \d J_{\nu,y}}(0;J;B)- \frac{\d^2 Z_{\xi}}{\d J_{\mu,x} \d J_{\nu,y}}(0;J;B) \Bigg)  =0.
\end{split}\end{equation}

By taking the Fourier transform, we finally recover Eq. \eqref{eq:WI3}:
\begin{equation*}
\begin{split}
&0= a^{-1}\sum_{\mu=0,1} \int_{\L^2} dx dy\; e^{i(k\cdot x+ p\cdot{y})} \Bigg( \frac{\d^2 Z_{\xi}}{\d J_{\mu,x-a\hat{e}_{\mu}} \d J_{\nu,y}}(0;J;B)- \frac{\d^2 Z_{\xi}}{\d J_{\mu,x} \d J_{\nu,y}}(0;J;B) \Bigg)  =\\
& \sum_{\mu=0,1}\frac{e^{ia k_{\mu}}-1}{a} \int_{\L^2} dx dy\; e^{i(k\cdot x+ p\cdot{y})} \frac{\d^2 Z_{\xi}}{\d J_{\mu,x} \d J_{\nu,y}}(0;J;B)= i \sum_{\mu=0,1} \s_{\mu}^*(k) \hat{K}^{\xi}_{\mu,\nu}(k,p)_{0;J;B}.
\end{split}\end{equation*}
\end{proof}

We now turn our focus to the generating functional, defined formally as $\mc{W}_{\xi}(\vphi;J;B)= -\log Z_{\xi}(\vphi;J;B)$. We are going to assume that $\mc{W}_{\xi}$ is a well defined object, analytic w.r.t. $J$ and $B$ close enough to zero: note that this is apriori true only for $|\mathsf{e}|$ \it{very} small depending on the cut-offs of the model, namely $N,L,m_N$; however, one of the nontrivial outcomes of our analysis is the well posedness of $\mc{W}_1$ for $|\mathsf{e}|\ll M$, uniformly in $N,L,m_N$. 

\begin{proof}[Proof of the Ward Identities: Eqs. \eqref{eq:WI5a},\eqref{eq:WI5b}.]

After taking the logarithm at both sides of Eq. \eqref{eq:WI1}, it follows that $\mc{W}_{\xi}(e^{i\mathsf{e}\a}\vphi;J+ \nabla\a;B)= \mc{W}_{\xi}(\vphi;J;B)\; \forall \a\in\mbb{R}^{\L}$. This implies the analogue of Eq. \eqref{eq:WI2} with $Z_{\xi}$ replaced by $\mc{W}_{\xi}$, which, after differentiation of both sides w.r.t. $J_{\nu,0}$, reads:
\begin{equation}
\label{eq:WI4}
\begin{split}
&a^{-1}\sum_{\mu=0,1} \Bigg(\frac{\d}{\d J_{\mu,x- a\hat{e}_{\mu}}} - \frac{\d}{\d J_{\mu,x}} \Bigg)\Bigg( \frac{\d \mc{W}_{\xi}(\vphi;J;B)}{\d J_{\nu,0}} \Bigg)  =\\
& -i\mathsf{e} \sum_{s=1,2} \Bigg(\vphi^+_{x,s}\frac{\d }{\d \vphi_{x,s}^+}- \vphi^-_{x,s}\frac{\d }{\d \vphi_{x,s}^-} \Bigg) \Bigg( \frac{\d\mc{W}_{\xi}(\vphi;J;B)}{\d J_{\nu,0}}\Bigg).
\end{split}
\end{equation}

As a first step, we set $\vphi,J,B$ to zero in both sides of Eq. \eqref{eq:WI4}; then if we also take the Fourier transform, we get:
\[\begin{split}
& a^{-1}\sum_{\mu=0,1} \int_{\L} dx\; e^{-ip\cdot x}  \Bigg( \frac{\d^2\mc{W}_{\xi}}{\d J_{\mu,x-a\hat{e}_{\mu}}\d J_{\nu,0}}(0;0;0)- \frac{\d^2\mc{W}_{\xi}}{\d J_{\mu,x}\d J_{\nu,0}}(0;0;0) \Bigg)=0
\end{split}\]

which readily implies the first of  Eq.  \eqref{eq:WI5a}, namely $\sum_{\mu=0,1} \s_{\mu}(p) \hat{\Sigma}^{\xi}_{\mu,\nu}(p)=0$ for every $p\in\L^*.$

\medskip
The second of Eq.  \eqref{eq:WI5a}, concerning $\hat{\Sigma}^{\xi}_{5;\mu,\nu}$, follows in full analogy. Let us now consider the analogue of Eq. \eqref{eq:WI4} with the derivative w.r.t. $J_{\nu,0}$ replaced by the second derivative w.r.t. $\vphi^-_{0,s'}$ and $\vphi^+_{y,s''}$:  
\[
\begin{split}
&a^{-1}\sum_{\mu=0,1} \Bigg(\frac{\d^3\mc{W}_{\xi}}{\d J_{\mu,x- a\hat{e}_{\mu}}\d\vphi^-_{0,s'}\d\vphi^+_{y,s''}}(0;0;0) - \frac{\d^3\mc{W}_{\xi}}{\d J_{\mu,x}\d\vphi^-_{0,s'}\d\vphi^+_{y,s''}}(0;0;0) \Bigg)  =\\
& i\mathsf{e} \sum_{s=1,2} \Bigg(\d(y-x)\d_{s,s''} \frac{\d^2 \mc{W}_{\xi} }{\d \vphi_{x,s}^+ \d\vphi^-_{0,s'}}(0;0;0) - \d(x)\d_{s,s'}\frac{\d^2 \mc{W}_{\xi} }{\d\vphi^+_{y,s''} \d \vphi_{x,s}^-}(0;0;0) \Bigg).
\end{split}
\]

Then, taking the Fourier transform,
\[
\begin{split}
&\sum_{\mu=0,1} \frac{e^{-iap_{\mu}}-1}{a} \int_{\L^2} dx dy\; e^{i(-p\cdot x+ k\cdot y)} \frac{\d^3\mc{W}_{\xi}}{\d J_{\mu,x}\d\vphi^-_{0,s'}\d\vphi^+_{y,s''}}(0;0;0)  =\\
& i\mathsf{e} \Bigg( \int_{\L} dx\; e^{i (-p+k)\cdot x} \frac{\d^2 \mc{W}_{\xi} }{\d \vphi_{x,s''}^+ \d\vphi^-_{0,s'}}(0;0;0) - \int_{\L} dy\; e^{ik\cdot y} \frac{\d^2 \mc{W}_{\xi} }{\d\vphi^+_{y,s''} \d \vphi_{0,s'}^-}(0;0;0)\Bigg), 
\end{split}
\]

which is exactly Eq. \eqref{eq:WI5b}, namely
\[\sum_{\mu=0,1} \s_{\mu}(p) \hat{\Gamma}^{\xi}_{\mu}(p;k)= -\mathsf{e}\big(\hat{S}^{\xi}(k-p)- \hat{S}^{\xi}(k)\big), \qquad \forall p,k\in \L^*.\]
\end{proof}

\section{Identities for the effective potential}\label{app_identities}

Here we give the proof of the identities among the kernels of the effective potential, which have been used in Section \ref{sect_improved}. Since, as long as $N,L$ are finite, 
$\mc{V}^{(h)}$ depends only upon field sources with degree $n$ less or equal than $n_*\equiv 2(2^NL)^2+1$, throughout the section,  functional derivatives and kernels related to $G^{(n)},\dot{G}^{(n)},\ddot{G}^{(n)}$ will be understood as zero for $n>n_*$. For completeness, recall that 
\[\mc{V}^{(h)}(\psi;G;\dot{G};\ddot{G})= -\log \mathcal{E} \Big( e^{-\mc{V}(\psi+\, \cdot\, ;G;\dot{G};\ddot{G})} \Big), \] 
where in this section $\mc{E}(F)\equiv \mc{E}_{(h,N]}(F):=\int P^{(h,N]}(\mc{D} \z) F(\z)$ for any Grassmann polynomial $F(\z),$ while  $\mc{V}$ is as in  Eq. \eqref{def_potential}.

The incoming analysis is based on the following two fundamental identities.

\begin{enumerate}

\item \it{Grassmann integration by parts.} Given any Grassmann 
    polynomial $F(\z)$,
\begin{align}
&\label{int_part_1} \int P^{(h,N]}(\mc{D}\z) \z^-_{\eta} F(\z)= \int_{\L_F} d\eta' g^{(h,N]}_{\eta,\eta'} \int P^{(h,N]}(\mc{D}\z) \frac{\d F}{\d \z^+_{\eta'}}(\z)\\
&\label{int_part_2} \int P^{(h,N]}(\mc{D}\z) \z^+_{\eta} F(\z)= -\int_{\L_F} d\eta' g^{(h,N]}_{\eta',\eta} \int P^{(h,N]}(\mc{D}\z) \frac{\d F}{\d \z^-_{\eta'}}(\z).
\end{align}

\item \it{Tree property.} Let $I:=\{1,\dots,n\},I':=I\setminus\{1\}$, $g^A_T:=\prod_{(i,j) \in T} \l g^A_{b_i,b_j}$ and recall that $\mc{T}(I)$ denotes set of all spanning trees on $I$ (cf. below Eq.  \eqref{eq:4b}). If $n=1$ by definition $\mc{T}(I)=\emptyset$, otherwise:
\begin{equation}
\begin{split}
& \sum_{T\in\mc{T}(I)} g_T^A= \\ &=\sum_{k \in I'} g^A_{b_1,b_k}\sum_{T\in\mc{T}(I')} g^A_T\; +\mathds{1}_{(n \geq 3)}\sum_{\substack{X_1\dot{\cup} X_2=I' \\ X_1 \ni 2}}
\sum_{k \in X_1} g^A_{b_1,b_k} \sum_{T_1\in\mc{T}(X_1)} g^A_{T_1}\sum_{T_2\in\mc{T}(X_2\cup\{1\})}g^A_{T_2}.  
\label{tree_property}\end{split}\end{equation}
\end{enumerate}
While the \it{Grassmann integration by parts} is a standard fact (see e.g. \cite[Eq. (A4)]{F10}), the \it{tree property} is not completely straightforward, and since it is a crucial building block for the incoming analysis of this section, we are going to present its proof. 

\begin{proof}[Proof of the tree property]

There is a one-to-one correspondence between $\mc{T}(I)$ and
\[ \mc{C}(I):= \Big\{(X_1,T_1,T_2,k): \; X_1\subseteq I', 2\in X_1, T_1\in\mc{T}(X_1), T_2\in\mc{T}(I'\setminus X_1),k\in X_1 \Big\}.\]

Indeed, we can define a map $\mc{T}(I)\mapsto \mc{C}(I)$ obtained by setting, for any $T\in\mc{T}(I)$, 

\begin{itemize}
    \item $X_1\subseteq I'$ is the maximal connected set such that $2\in X_1$.
    \item $T_1$ is the restriction of $T$ to $X_1$, namely $T_1=\big\{(i,j)\in T: \; i,j\in X_1\big\}$.
    \item $k\in X_1$ is the unique vertex such that $(1,k)\in T$.
    \item Letting $X_2:= I'\setminus X_1$, $T_2$ is the restriction of $T$ to $X_2\cup\{1\}$.
\end{itemize}

See Fig. \ref{fig:spanningtree} for a graphical representation. On the other hand, we can consider the map $\mc{C}(I)\mapsto \mc{T}(I)$ defined as $T:= T_1\dot{\cup} T_2\dot{\cup} (1,k)$. It is then straightforward to check that the maps $\mc{T}(I) \mapsto \mc{C}(I)$ and $\mc{C}(I) \mapsto \mc{T}(I)$ are one the inverse of the other.

\begin{figure}[h]
    \centering
\includegraphics[width=0.8\textwidth]{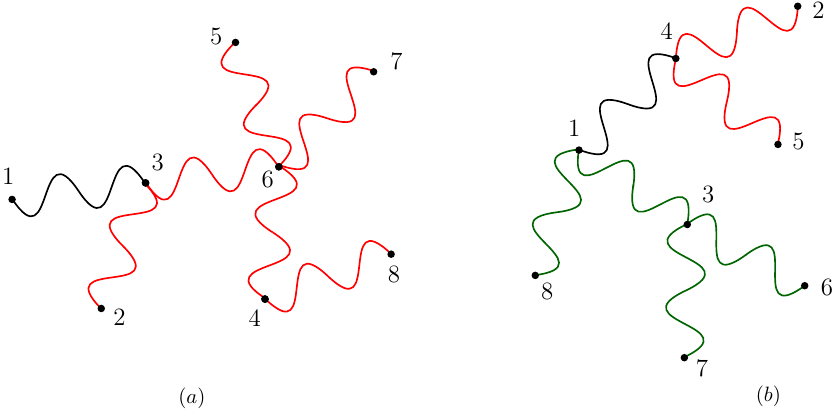}
    \caption{Graphical representation of the correspondence $\mc{T}(I)\leftrightarrow \mc{C}(I)$. In the example (a), the vertex $1$ is a leaf of $T$, $X_1=I',X_2=\eset$, $T_1=\{(2,3),(3,6),(6,5),(6,7),(6,4),(4,8) \}$ and $k=3$. In the example (b), $1$ is not a leaf, $X_1=\{2,4,5\}$, $X_2=\{3,6,7,8\}$, $T_1=\{(4,2),(4,5)\}$, $T_2=\{(1,3),(3,6),(3,7),(1,8)\}$ and $k=4$.}
    \label{fig:spanningtree}
\end{figure}

The bijection between $\mc{C}(I)$ and $\mc{T}(I)$ allows us to write:

\begin{equation}
\label{eq:tree1}
\sum_{T\in\mc{T}(I)} g^A_T= \sum_{(X_1,T_1,T_2,k) \in \mc{C}(I)} g^A_{T_{(X_1,T_1,T_2,k)}},
\end{equation}

where $T_{(X_1,T_1,T_2,k)}$ is the tree in $\mc{T}(I)$ uniquely identified by $(X_1,T_1,T_2,k)\in \mc{C}(I)$. Now it is convenient to distinguish whether $X_2\equiv I'\setminus X_1$ is empty or not (i.e. the vertex $1$ is a leaf or not, see Fig. \ref{fig:spanningtree}). In the first case we simply have: $g^A_{T(X_1,T_1,T_2,k)}= g^A_{(1,k)} g^A_{T_1}$, while, if $X_2\ne \eset$,  $g^A_{T(X_1,T_1,T_2,k)}= g^A_{(1,k)} g^A_{T_1} g^A_{T_2}$. Therefore Eq. \eqref{eq:tree1} becomes:

\begin{equation}
\label{eq:tree2}
\sum_{T\in\mc{T}(I)} g^A_T= \sum_{\substack{(X_1,T_1,T_2,k) \in \mc{C}(I)\\ X_1=I'}} g^A_{(1,k)} g^A_{T_1}+ \sum_{\substack{(X_1,T_1,T_2,k) \in \mc{C}(I)\\ X_1\subset I'}} g^A_{(1,k)} g^A_{T_1} g^A_{T_2}.
\end{equation}

Noting that the second sum in the r.h.s. of Eq. \eqref{eq:tree2} is nonempty only if $n=3$, the desired \it{tree property}, Eq. \eqref{tree_property}, follows.
\end{proof}

\subsection{Preliminary identities}

We begin with two intermediate lemmas. The first one, Lemma \ref{lemma_id_A}, characterizes the derivatives of $\mc{V}^{(h)}$ w.r.t. $\psi$ and $G^{(1)}$. The next one, Lemma \ref{lemma_id_B}, characterizes the derivatives of $\mc{V}^{(h)}$ w.r.t. $G^{(n)}$ and $\dot{G}^{(n)}$, with $n\ge2$.  We will use the notation $\mc{V}^{(h)}(\psi;\mb{0}):=\mc{V}^{(h)}(\psi;G;\dot G; \ddot G)\rvert_{G,\dot G, \ddot G\equiv 0}$ or $\mc{V}^{(h)}(\psi;G^{(1)};\mb{0})$ when $G^{(n)},\dot G^{(n)},\ddot G^{(n)} \equiv 0$ for $n \geq 2,$ etc.

\begin{lemma}
\label{lemma_id_A}
The following identities hold for the effective potential $\mc{V}^{(h)}$.

\begin{equation}
\label{eq_id_1}
\begin{split}
&\frac{\d\mc{V}^{(h)}}{\d\psi^+_{\eta_1}}(\psi;\mb{0}) =
e^{\mc{V}^{(h)}(\psi;\mb{0})} \int db_1 d\eta c_{b_1}(\eta_1,\eta) \left( \psi^-_{\eta} + \int d\eta' g^{(h,N]}_{\eta,\eta'} \frac{\d}{\d\psi^+_{\eta'}} \right)\times \\ & \times \left\{ 2^N\left(\mu_{N,\l} -1 \right)-\sum_{n\ge1} \frac{ \mu_{N,\l}}{n^2} \int db \l g^A_{b_1,b} \frac{\d}{\d G^{(n)}_{b}} +\right.\\
&+ \sum_{n_1 \geq 1, n_2 \geq 2}\left.  2^{-N}\frac{\mc{N}_{n_1,n_2}}{(n_1+n_2)^2} \int db \l g^A_{b_1,b}  \frac{\d^2}{\d G^{(n_1)}_{b}\d \dot{G}^{(n_2)}_{b_1}} -\left. \sum_{n\ge2} 2^{-N} \frac{\d}{\d \ddot{G}^{(n)}_{b_1}} \right\} e^{-\mc{V}^{(h)}}\right|_{(\psi;\mb{0})},
\end{split}
\end{equation}

where  $\mc{N}_{n_1,n_2}=\frac{(n_1+n_2)^2}{(n_1+n_2-1)n_1 n_2^2}\le \frac{(n_1+n_2)^2}{n_1^2n_2^2}$. Moreover:

\begin{equation}
\label{eq_id_4}
\begin{split}
&\frac{\d^2\mc{V}^{(h)}}{\d G^{(1)}_b \d\psi^+_{\eta_1}} (\psi;\mb{0})=\\
&\frac{\d}{\d G^{(1)}_b}\left\{ e^{\mc{V}^{(h)}} \int db_1 d\eta c_{b_1}(\eta_1,\eta) \left( \psi^-_{ \eta} + \int d\eta' g^{(h,N]}_{\eta,\eta'} \frac{\d}{\d\psi^+_{\eta'}} \right)\times \right.\\
&\times \left[ 2^N\left(\mu_{N,\l}  -1 \right) -\sum_{n\ge1} \frac{\mu_{N,\l}}{n^2}  \int db \l g^A_{b_1,b} \frac{\d}{\d G^{(n)}_{b}} +\right.\\
&+ \left.\left.\left.  \sum_{n_1 \geq 1, n_2 \geq 2}  2^{-N}\frac{\mc{N}_{n_1,n_2}}{(n_1+n_2)^2} \int db \l g^A_{b_1,b}  \frac{\d^2}{\d G^{(n_1)}_{b}\d \dot{G}^{(n_2)}_{b_1}} -  \sum_{n\ge2} 2^{-N} \frac{\d}{\d \ddot{G}^{(n)}_{b_1}}\right]\right\} e^{-\mc{V}^{(h)}}\right|_{(\psi;\mb{0})}+\\
&+ \mu_{N,\l}\int d\eta c_b(\eta_1,\eta) \left(\psi^-_{\eta}- \int d\eta' g^{(h,N]}_{\eta,\eta'} \frac{\d\mc{V}^{(h)}}{\d\psi^+_{\eta'}}(\psi;\mb{0}) \right). 
\end{split}
\end{equation}
\end{lemma}

\begin{proof}
We start from the identity:

\[ \frac{\d \mc{V}^{(h)}}{\d\psi^+_{\eta_1}}(\psi;G;\dot{G};\ddot{G})= e^{\mc{V}^{(h)}(\psi;G;\dot{G};\ddot{G})} \mc{E} \Big( \frac{\d \mc{V}}{\d\psi^+_{\eta_1}}(\psi+\, \cdot \, ;G;\dot{G};\ddot{G}) e^{-\mc{V}(\psi+\, \cdot \, ;G;\dot{G};\ddot{G})}\Big) \]   

and we set to zero all the fields except $\psi$ and $G^{(1)}$. From Eq. \eqref{def_potential} we find:
\begin{equation}
\label{eq_id_5}
\begin{split}
& \frac{\d\mc{V}^{(h)}}{\d\psi^+_{\eta_1}}(\psi;G^{(1)};\mb{0})= \\
& e^{\mc{V}^{(h)}} \Bigg\{ \Bigg[2^N\left(\mu_{N,\l} -1 \right)\int db_1  \mc{E}\Big( \frac{\d O_{b_1}}{\d\psi^+_{\eta_1}}(\psi+\cdot) e^{-\mc{V}(\psi+\, \cdot\, ;G^{(1)};\mb{0})} \Big)+\\
&\sum_{n\ge2} n2^{N(2-n)} \int d\ud b v_n(\ud b)  \mc{E} \Big(\frac{\d O_{b_1}}{\d\psi^+_{\eta_1}}(\psi+\cdot) O_{\ud b'}(\psi+\cdot) e^{-\mc{V}(\psi+\, \cdot \, ; G^{(1)};\mb{0})}\Big)+\\
& \sum_{n\ge2} n2^{-nN} \int d\ud b\tilde{v}_n(\ud b) \mc{E} \Big( \frac{\d O_{b_1}}{\d\psi^+_{\eta_1}}(\psi+\cdot) O_{\ud b'}(\psi+\cdot) e^{-\mc{V}(\psi+\, \cdot \, ; G^{(1)};\mb{0})}\Big) \Bigg] \Bigg\}+\\
& e^{\mc{V}^{(h)}(\psi; G^{(1)};\mb{0})} \int db\; G^{(1)}_b v_1(b) \mc{E}\Big( \frac{\d O_b}{\d\psi^+_{\eta_1}}(\psi+\cdot) e^{-\mc{V}(\psi+\, \cdot\, ; G^{(1)};\mb{0})}\Big),
\end{split} 
\end{equation}

where we have used the decomposition $w_{n,0}= v_n+ 2^{-2N}\tilde{v}_n$ according to Eq. \eqref{eq_16a}. Noting that $\frac{\d O_{b}}{\d\psi^+_{\eta_1}}(\psi+\z)= \int d\eta c_{b}(\eta_1,\eta) (\psi^-_{\eta}+ \z^-_{\eta})$ and using Eq. \eqref{int_part_1}, we find that the first line in the r.h.s. of Eq. \eqref{eq_id_5} becomes

\[ 2^N\left(\mu_{N,\l} -1 \right)\int db \int d\eta c_b(\eta_1,\eta) \left(\psi^-_{\eta}+ \int d\eta' g^{(h,N]}_{\eta,\eta'} \frac{\d}{\d\psi^+_{\eta'}}\right) e^{-\mc{V}^{(h)}(\psi; G^{(1)};\mb{0})}. \]

Let us proceed with the second line in the r.h.s. of Eq. \eqref{eq_id_5}. Recalling that $v_n(\ud b)=\sum_{T\in\mc{T}(I)} \frac{\mu_{N,\l}^n}{n!} g_T^A$ (cf. Eq. \eqref{eq_16a}) and using Eq. \eqref{tree_property} we find that
\begin{equation}
\hspace{-20pt}\begin{split}
     v_n(\ud b)&=\mu_{N,\l} \mathds{1}_{(n=1)}+ \mathds{1}_{(n\geq 2)}\sum_{k \in I} \frac{\l \mu_{N,\l}}{n} g^A_{b_1,b_k} v_{n-1}(\ud b')+\\ &+\mathds{1}_{(n \geq 3)}\sum_{\substack{X_1 \dot{\cup} X_2=I' \\ 2 \in X_1}} \binom{n}{|X_1|}^{-1} \sum_{k \in X_1}\l g^A_{b_1,b_k} v_{|X_1|}(\ud b_{X_1})v_{|X_2|+1}(\ud b_{X_2 \cup \{1\}}),
\end{split}\label{v_n_rewriting}
\end{equation}

where $\ud b_{X_1}=\{b_i: i \in X_1\}$ and recall that $v_n$ is symmetric under permutations of its variables. Thus we can rewrite the second line in the r.h.s. of Eq. \eqref{eq_id_5} as
\begin{equation}
\label{eq_id_13}
\begin{split}
&    \sum_{n\ge2} \int d \ud b  e^{\mc{V}^{(h)}(\psi; G^{(1)};\mb{0})}\mc{E}\Big( \frac{\d O_{b_1}}{\d\psi^+_{\eta_1}}(\psi+\cdot)   \left\{ \mu_{N,\l}(n-1) \l g^A_{b_1,b_2} \left(2^{N(2-n)} v_{n-1}(\ud b')  O_{\ud b'}(\psi+\cdot)  \right) \right.\\
&+ 2^{-N}n \mathds{1}_{(n \geq 3)}  \sum_{\substack{X_1{\dot{\cup}}X_2= I'\\ X_1\ni 2}}\binom{n}{|X_1|}^{-1} |X_1| \l g^A_{b_1,b_2} 2^{N(1-|X_1|)} v_{|X_1|}(\ud{b}_{X_1}) O_{\ud b_{X_1}}(\psi+\cdot) \times \\
&\times\left. 2^{N(1-|X_2|)} v_{|X_2|+1} (\ud{b}_{X_2\cup\{1\}}) O_{\ud b_{X_2}}(\psi+\cdot) \right\}e^{-\mc{V}(\psi+\, \cdot\,; G^{(1)}; \mb{0} )} \Big),
\end{split}
\end{equation}
where we used the symmetry of $v_n$ and $O_{\ud b_{X_1}}:=\prod_{i \in X_1} O_{b_i}$ etc.. Eq. \eqref{eq_id_13} can be further rewritten by identifying the terms in parentheses as derivatives of $\mc{V}$ w.r.t. $G$ and $\dot{G}$ (cf. Eq. \eqref{def_potential}) as
\[\begin{split}
&\int db_1 db\; e^{\mc{V}^{(h)}(\psi; G^{(1)};\mb{0})}\mc{E}\Big(  \frac{\d O_{b_1}}{\d\psi^+_{\eta_1}}(\psi+\cdot)\Big[ -\sum_{n\ge2} \frac{\mu_{N,\l}}{(n-1)^2}  \l g^A_{b_1,b} \frac{\d}{\d G^{(n-1)}_{b}}+\\
& +\sum_{n\ge3} 2^{-N} \sum_{\substack{n_1 \geq 1, n_2 \geq 2\\ n_1+n_2=n}} \frac{\mc{N}_{n_1,n_2}}{(n_1+n_2)^2} \l g^A_{b_1,b} \frac{\d}{\d G^{(n_1)}_{b}}  \frac{\d}{\d\dot{G}^{(n_2)}_{b_1}} \Big] e^{-\mc{V}(\psi+\, \cdot\, ; G^{(1)};\mb{0})} \Big),
\end{split}\]

where
\begin{equation}
\label{eq_id_10}
\begin{split}
\mc{N}_{n_1,n_2}= \frac{(n_1+n_2)^3 n_1}{\binom{n_1+n_2}{n_1} n_1^3n_2^3} \sum_{\substack{X_1\dot{\cup}X_2=I'\\ X_1\ni 2 \\ |X_1|=n_1}}   1 =\frac{\binom{n_1+n_2-2}{n_1-1}(n_1+n_2)^3}{\binom{n_1+n_2}{n_1}n_1^2 n_2^3}=\frac{(n_1+n_2)^2}{(n_1+n_2-1)n_1n_2^2}.
\end{split}
\end{equation}

Expanding again $\frac{\d O_{b}}{\d\psi^+_{\eta_1}}(\psi+\z)= \int d\eta c_{b}(\eta_1,\eta) (\psi^-_{\eta}+ \z^-_{\eta})$ and using Eq. \eqref{int_part_1}, we find that the second line of Eq. \eqref{eq_id_5} can be written as
\[\begin{split}
&e^{\mc{V}^{(h)}(\psi; G^{(1)};\mb{0})} \int db_1 db d\eta \; c_{b_1}(\eta_1,\eta)  \left(\psi^-_{\eta}+ \int d\eta' g^{(h,N]}_{\eta,\eta'} \frac{\d}{\d\psi^+_{\eta'}} \right) \times\\
&\times\left( - \sum_{n\ge2} \frac{\mu_{N,\l} \l g^A_{b_1,b}}{(n-1)^2}    \frac{\d}{\d G^{(n-1)}_{b}} +\sum_{n_1 \geq 1, n_2 \geq 2} 2^{-N} \frac{\mc{N}_{n_1,n_2}}{(n_1+n_2)^2}  \l g^A_{b_1,b}  \frac{\d^2}{\d G^{(n_1)}_{b} \d \dot{G}^{(n_2)}_{b_1}}  \right )e^{-\mc{V}^{(h)}(\psi; G^{(1)};\mb{0})}.
\end{split}\]

Let us now discuss the third line in the r.h.s. of Eq. \eqref{eq_id_5}, which can be rewritten as
\[\begin{split}
& \sum_{n\ge2} 2^{-N} \int d\ud b  e^{\mc{V}^{(h)}(\psi; G^{(1)};\mb{0})} \mc{E} \Big( \int d\eta c_{b_1}(\eta_1,\eta) \Big(\psi^-_{\eta}+ \int d\eta' g^{(h,N]}_{\eta,\eta'} \frac{\d}{\d\psi^+_{\eta'}}\Big)  \times \\
&\times \left(n 2^{N(1-n)} \tilde{v}_n(\ud b)O_{\ud b'}(\psi+\cdot)\right) e^{-\mc{V}(\psi+\, \cdot\, ; G^{(1)};\mb{0})} \Big)=\\
& -\sum_{n\ge2} 2^{-N} \int db_1  d\eta c_{b_1}(\eta_1,\eta) e^{\mc{V}^{(h)}} \Big(\psi^-_{\eta}+ \int d\eta' g^{(h,N]}_{\eta,\eta'} \frac{\d}{\d\psi^+_{\eta'}}\Big) \frac{\d}{\d \ddot{G}_{b_1}^{(n)}} e^{-\mc{V}^{(h)}(\psi; G^{(1)};\mb{0})}.
\end{split}\]

Finally, the fourth line in the r.h.s. of Eq. \eqref{eq_id_5} can be easily handled using again Eq. \eqref{int_part_1} and recalling that $v_1(b)=\mu_{N,\l}$: 
\[\begin{split}
&e^{\mc{V}^{(h)}(\psi; G^{(1)};\mb{0})} \mu_{N,\l} \int db G^{(1)}_b   \mc{E} \Big( \frac{\d O_b}{\d\psi^+_{\eta_1}}(\psi+\cdot) e^{-\mc{V}(\psi+\cdot; G^{(1)}; \ud  0)}\Big) =\\
&e^{\mc{V}^{(h)}(\psi; G^{(1)};\mb{0})} \mu_{N,\l} \int db d\eta G^{(1)}_b c_b(\eta_1,\eta) \Big(\psi^-_{\eta}+ \int d\eta' g^{(h,N]}_{\eta,\eta'} \frac{\d}{\d\psi^+_{\eta'}} \Big) e^{-\mc{V}^{(h)}(\psi ; G^{(1)};\mb{0})}.
\end{split}\]

All in all we have rewritten Eq. \eqref{eq_id_5} as
\[\begin{split}
&\frac{\d\mc{V}^{(h)}}{\d\psi^+_{\eta_1}}(\psi; G^{(1)};\mb{0})=\\
&e^{\mc{V}^{(h)}(\psi; G^{(1)};\mb{0})} \int db_1 d\eta c_{b_1}(\eta_1,\eta) \Big( \psi^-_{\eta} + \int d\eta' g^{(h,N]}_{\eta,\eta'} \frac{\d}{\d\psi^+_{\eta'}} \Big)\Bigg \{ 2^N\left(\mu_{N,\l} -1 \right)+ \\ 
& -\sum_{n\ge1} \frac{\mu_{N,\l}}{n^2}  \int db \l g^A_{b_1,b} \frac{\d}{\d G^{(n)}_{b}}  +  2^{-N} \sum_{\substack{n_1\ge1 \\ n_2\ge2}} \frac{\mc{N}_{n_1,n_2}}{(n_1+n_2)^2} \int db \l g^A_{b_1,b} \frac{\d^2}{\d G^{(n_1)}_{b}\d \dot{G}^{(n_2)}_{b_1}}+\\
&- 2^{-N}\sum_{n\ge2} \frac{\d}{\d \ddot{G}^{(n)}_{b_1}} \Bigg \} e^{-\mc{V}^{(h)}(\psi;G^{(1)};\mb{0})}+ \\ 
&+ e^{\mc{V}^{(h)}(\psi;G^{(1)};\mb{0})} \mu_{N,\l} \int db d\eta G^{(1)}_b c_b(\eta_1,\eta) \left (\psi^-_{\eta}+ \int d\eta' g^{(h,N]}_{\eta,\eta'} \frac{\d}{\d\psi^+_{\eta'}} \right ) e^{-\mc{V}^{(h)}(\psi;G^{(1)};\mb{0})}.
\end{split}\]

 By setting also $G^{(1)}=0$ we obtain Eq. \eqref{eq_id_1}, while,  deriving both sides w.r.t. $G^{(1)}$ and then setting $G^{(1)}=0$, we find Eq. \eqref{eq_id_4}.
\end{proof}

\begin{lemma}
\label{lemma_id_B}
The following identities hold for the effective potential $\mc{V}^{(h)}$.
\begin{equation}
\label{eq_id_2}
\begin{split}
&\frac{\d \mc{V}^{(h)}}{\d G^{(n)}_{b_1}}=\\
&e^{\mc{V}^{(h)}} \int d\eta_1 d\eta_2 c_{b_1}(\eta_1,\eta_2)  \left(\psi^+_{\eta_1}\psi^-_{\eta_2} + \psi^-_{\eta_2}\int d\eta g^{(h,N]}_{\eta,\eta_1} \frac{\d}{\d\psi^-_{\eta}}+ \psi^+_{\eta_1}\int d\eta' g^{(h,N]}_{\eta_2,\eta'} \frac{\d}{\d\psi^+_{\eta'}}+\right.\\
&\left.- g^{(h,N]}_{\eta_2,\eta_1} + \int d\eta d\eta' g^{(h,N]}_{\eta,\eta_1} g^{(h,N]}_{\eta_2,\eta'} \frac{\d^2}{\d\psi^+_{\eta'}\d\psi^-_{\eta}}\right)\times \\
&\times \left\{ \mathds{1}_{(n=1)}\mu_{N,\l} - \mathds{1}_{(n\ge2)} \frac{n^2}{(n-1)^2} 2^{-N} \mu_{N,\l} \int db \l g^A_{b_1,b}  \frac{\d}{\d G^{(n-1)}_{b}}   \right.+\\
&+ \left.  \mathds{1}_{(n\ge3)} 2^{-2N} \sum_{\substack{n_1\ge1,n_2\ge2\\ n_1+n_2=n}} \mc{N}_{n_1,n_2}  \int db \l g^A_{b_1,b} \frac{\d^2}{\d G^{(n_1)}_{b}\d \dot{G}^{(n_2)}_{b_1}}  \right\}e^{-\mc{V}^{(h)}},  
\end{split}\end{equation}

where $\mc{N}_{n_1,n_2}$ are as in Lemma \ref{lemma_id_A}. Moreover:
\begin{equation}
\label{eq_id_3}
\begin{split}
&\frac{\d \mc{V}^{(h)}}{\d \dot{G}^{(n)}_{b_1}}= \\ &     e^{\mc{V}^{(h)}} \hspace{-5pt}\int\hspace{-3pt} db  \Bigg\{-  \frac{ \mu_{N,\l} n^2}{(n-1)^2}   \l g^A_{b_1,b}  \frac{\d}{\d G^{(n-1)}_{b}}   +  2^{-N}\mathds{1}_{(n\ge3)} \sum_{\substack{n_1\ge1 \\  n_2\ge2\\ n_1+n_2=n}} \mc{N}_{n_1,n_2}\l g^A_{b_1,b} \frac{\d^2}{\d G^{(n_1)}_{b}\d \dot{G}^{(n_2)}_{b_1}}  \Bigg\}e^{-\mc{V}^{(h)}}.  
\end{split}\end{equation}
\end{lemma}

\begin{proof}
The proof follows the same ideas used for Lemma \ref{lemma_id_A} and will not be belabored in detail. The starting point is the identity:
\[ \frac{\d \mc{V}^{(h)}}{\d G^{(n)}_{b_1}}(\psi;G;\dot{G};\ddot{G})= e^{\mc{V}^{(h)}(\psi;G;\dot{G};\ddot{G})} \mc{E} \Big(\frac{\d \mc{V}}{\d G^{(n)}_{b_1}}(\psi+\cdot\, ;G;\dot{G};\ddot{G}) e^{-\mc{V}(\psi+\cdot \, ;G;\dot{G};\ddot{G})} \Big) \]  

for proving Eq. \eqref{eq_id_2}, and the same one with the derivative w.r.t. $\dot{G}^{(n)}$ instead of $G^{(n)}$ for proving Eq. \eqref{eq_id_3}. Note that:
\[\begin{split}
&\frac{\d\mc{V}}{\d G^{(n)}_{b_1}}(\psi+\z;G;\dot{G};\ddot{G}) = n^3 2^{N(1-n)} \int d \ud b' v_n(\ud b) O_{\ud b}(\psi +\z).
\end{split}\]

Using Eq. \eqref{v_n_rewriting} for rewriting $v_n$ and following the same steps as for the proof of Lemma \ref{lemma_id_A}, we find:
\begin{equation}
\label{eq_id_11}
\begin{split}  &\frac{\d \mc{V}^{(h)}}{\d G^{(n)}_{b_1}}(\psi;G;\dot{G};\ddot{G}) 
=e^{\mc{V}^{(h)}}  \Bigg\{ \mathds{1}_{(n=1)}\mu_{N,\l} -  2^{-N} \mathds{1}_{(n\ge2)} \frac{\mu_{N,\l} \, n^2}{(n-1)^2} \int db  \l g^A_{b_1,b}  \frac{\d}{\d G^{(n-1)}_{b}}   +\\
&+ 2^{-2N} \mathds{1}_{(n\ge3)} \int db  \sum_{\substack{n_1\ge1, n_2\ge2\\ n_1+n_2=n}} \mc{N}_{n_1,n_2} \l g^A_{b_1,b}  \frac{\d^2}{\d G^{(n_1)}_{b} \d\dot{G}^{(n_2)}_{b_1}} \Bigg\} \mc{E} \Big( O_{b_1}(\psi+\cdot \, ) e^{-\mc{V}(\psi+\, \cdot \, ;G;\dot{G};\ddot{G})} \Big),   
\end{split}
\end{equation}

and the same expression holds for $\frac{\d\mc{V}^{(h)}}{\d\dot{G}^{(n)}_{b_1}}$ without the factor $O_{b_1}(\psi+\z)$ within the Grassmann integral, yielding immediately Eq. \eqref{eq_id_3}. Letting $\mc{V}=\mc{V}(\psi+\z;G;\dot{G};\ddot{G})$, we can rewrite:
\[\begin{split}
&\mc{E} \Big( O_{b_1}(\psi+\cdot) e^{-\mc{V}}\Big)  = \int d\eta_1 d\eta_2 c_{b_1}(\eta_1,\eta_2) \mc{E} \Big( (\psi^+_{\eta_1}+(\cdot)^+_{\eta_1})(\psi^-_{\eta_2}+ (\cdot)^-_{\eta_2})  e^{-\mc{V}}\Big) \\
&= \int d\eta_1 d\eta_2 c_{b_1}(\eta_1,\eta_2) \mc{E} \Big( \Big[-(\psi^-_{\eta_2}+ (\cdot)^-_{\eta_2}) \Big(\psi^+_{\eta_1}-\int d\eta g^{(h,N]}_{\eta,\eta_1}\frac{\d}{\d(\cdot)^-_{\eta}}\Big) - g^{(h,N]}_{\eta_2,\eta_1}\Big]  e^{-\mc{V}}\Big)=\\
&= \int d\eta_1 d\eta_2 c_{b_1}(\eta_1,\eta_2) \left(\psi^+_{\eta_1}\psi^-_{\eta_2} + \psi^-_{\eta_2} \int d\eta g^{(h,N]}_{\eta,\eta_1}\frac{\d}{\d\psi^-_{\eta}} + \psi^+_{\eta_1}\int d\eta' g^{(h,N]}_{\eta_2,\eta'}\frac{\d}{\d\psi^+_{\eta'}} + \right.\\
&\left. + \int d\eta d\eta' g^{(h,N]}_{\eta,\eta_1} g^{(h,N]}_{\eta_2,\eta'}\frac{\d^2}{\d\psi^+_{\eta'}\d\psi^-_{\eta}} - g^{(h,N]}_{\eta_2,\eta_1} \right) e^{-\mc{V}^{(h)}},
\end{split}\]

which, once plugged in the r.h.s. of Eq. \eqref{eq_id_11}, yields Eq. \eqref{eq_id_2}.
\end{proof}

\subsection{Identities for the kernels}

Lemma \ref{lemma_id_A} and Lemma \ref{lemma_id_B} can be used to infer the identities among the kernels of $\mc{V}^{(h)}$, by means of Eq. \eqref{eq_14}.

\begin{proof}[Proof of Lemma \ref{lemma_id_self_energy}]

This is actually a corollary of Eq. \eqref{eq_id_1}:
\begin{equation}
\label{eq_id_12}
\begin{split}
\frac{\d\mc{V}^{(h)}}{\d\psi^+_{\eta_1}}(\psi;\mb{0})&= \int db_1 \int d\eta c_{b_1}(\eta_1,\eta) \left( \psi^-_{\eta}- \int d\eta' g^{(h,N]}_{\eta,\eta'} \frac{\d\mc{V}^{(h)}}{\d\psi^+_{\eta'}} + \int d\eta' g^{(h,N]}_{\eta,\eta'} \frac{\d}{\d\psi^+_{\eta'}} \right)\times\\
&\times \left\{ 2^N\left(\mu_{N,\l} -1 \right) +\sum_{n\ge1} \frac{\mu_{N,\l}}{n^2}  \int db \l g^A_{b_1,b} \frac{\d\mc{V}^{(h)}}{\d G^{(n)}_{b}} +\right.\\
&-  2^{-N} \sum_{n_1\ge1,n_2\ge2} \frac{\mc{N}_{n_1,n_2}}{(n_1+n_2)^2} \int db \l g^A_{b_1,b}    \left(\frac{\d^2\mc{V}^{(h)}}{\d G^{(n_1)}_{b}\d \dot{G}^{(n_2)}_{b_1}} + \frac{\d\mc{V}^{(h)}}{\d G^{(n_1)}_{b}} \frac{\d\mc{V}^{(h)}}{\d \dot{G}^{(n_2)}_{b_1}} \right) +\\
&\left.+ \sum_{n\ge2} 2^{-N} \frac{\d\mc{V}^{(h)}}{\d \ddot{G}^{(n)}_{b_1}} \right\}(\psi;\mb{0}). 
\end{split}
\end{equation}

Eq. \eqref{eq_15} follows by deriving both sides of Eq. \eqref{eq_id_12} w.r.t. $\psi^-_{\eta_2}$ and then setting $\psi=0$. $\mc{A}_1(\eta_1,\eta_2)$ comes from the first term in curly bracket in the r.h.s. of Eq. \eqref{eq_id_12}:
\[\begin{split}
\mc{A}_1(\eta_1,\eta_2)&=2^N\left(\mu_{N,\l} -1 \right) \int db_1 \int d\eta c_{b_1}(\eta_1,\eta)  \frac{\d}{\d\psi^-_{\eta_2}}\left[ \psi^-_{\eta} - \int d\eta' g^{(h,N]}_{\eta,\eta'} \frac{\d\mc{V}^{(h)}}{\d\psi^+_{\eta'}} \right](\mb{0})\\
&= 2^N\left(\mu_{N,\l} -1 \right) \int db_1 \int d\eta c_{b_1}(\eta_1,\eta)  \left[ \d(\eta-\eta_2) - \int d\eta' g^{(h,N]}_{\eta,\eta'} W^{(h)|1,1}_{\eset;\eset;\eset}(\eta',\eta_2) \right].
\end{split}\]

$\mc{A}_2(\eta_1,\eta_2)$ comes from the second term in curly bracket in the r.h.s. of Eq. \eqref{eq_id_12}:
\[\begin{split}
\mc{A}_2(\eta_1,& \eta_2)= \sum_{n\ge1} \frac{\mu_{N,\l}}{n^2} \int db db_1\int d\eta  \l g^A_{b_1,b}  c_{b_1}(\eta_1,\eta)\times\\
&\times \left[\frac{\d}{\d\psi^-_{\eta_2}} \left( \psi^-_{\eta}- \int d\eta' g^{(h,N]}_{\eta,\eta'} \frac{\d\mc{V}^{(h)}}{\d\psi^+_{\eta'}} + \int d\eta' g^{(h,N]}_{\eta,\eta'} \frac{\d}{\d\psi^+_{\eta'}} \right)  \frac{\d\mc{V}^{(h)}}{\d G^{(n)}_{b}}\right](\mb{0})=\\
& =\sum_{n\ge1} \frac{\mu_{N,\l}}{n^2}  \int db db_1\int d\eta  \l g^A_{b_1,b}  c_{b_1}(\eta_1,\eta) \left[  \d(\eta-\eta_2)\frac{\d\mc{V}^{(h)}}{\d G^{(n)}_{b}}(\mb{0})+\right.\\
&\left.- \int d\eta' g^{(h,N]}_{\eta,\eta'} \frac{\d^2\mc{V}^{(h)}}{\d\psi^-_{\eta_2}\d\psi^+_{\eta'}}(\mb{0}) \frac{\d\mc{V}^{(h)}}{\d G^{(n)}_{b}}(\mb{0}) + \int d\eta' g^{(h,N]}_{\eta,\eta'} \frac{\d^3\mc{V}^{(h)}}{\d\psi^-_{\eta_2}\d\psi^+_{\eta'}\d G^{(n)}_b}(\mb{0})   \right].
\end{split}\]

Similarly, $\mc{A}_3(\eta_1,\eta_2)$ and $\mc{A}_4(\eta_1,\eta_2)$ come from the third and fourth terms respectively in the curly bracket in the r.h.s. of Eq. \eqref{eq_id_12}.
\end{proof}

\begin{proof}[Proof of Eq. \eqref{eq_18}]
This is a corollary of Eq. \eqref{eq_id_3}:
\[
\begin{split}
\frac{\d \mc{V}^{(h)}}{\d \dot{G}^{(n_2)}_{b_2}}&=        \frac{n_2^2}{(n_2-1)^2}  \mu_{N,\l} \int db \l g^A_{b_2,b}  \frac{\d\mc{V}^{(h)}}{\d G^{(n_2-1)}_{b}}+\\
&-  2^{-N}  \sum_{\substack{n'_1\ge1, n'_2\ge2\\ n'_1+n'_2=n_2}} \mc{N}_{n'_1,n'_2}\int db_1' db'_2\l g^A_{b_2,b_1'} \d(b_2-b'_2) \left( \frac{\d^2\mc{V}^{(h)}}{\d G^{(n'_1)}_{b_1'}\d \dot{G}^{(n'_2)}_{b'_2}}- \frac{\d\mc{V}^{(h)}}{\d G^{(n'_1)}_{b'_1}} \frac{\d\mc{V}^{(h)}}{\d \dot{G}^{(n'_2)}_{b'_2}}  \right).  
\end{split}\]

Deriving both sides w.r.t. $\psi^+_{\eta}$, $\psi^-_{\eta'}$ and $G^{(n_1)}_{b_1}$, and setting the fields to zero, we find:
\[
\begin{split}
&\frac{\d^4\mc{V}^{(h)}}{\d\psi^-_{\eta'}\d\psi^+_{\eta}\d G^{(n_1)}_{b_1}\d\dot{G}^{(n_2)}_{b_2}}(\mb{0})=\\
&\frac{n_2^2}{(n_2-1)^2}  \mu_{N,\l} \int db \l g^A_{b_2,b}  \frac{\d^4\mc{V}^{(h)}}{\d\psi^-_{\eta'}\d\psi^+_{\eta}\d G^{(n_1)}_{b_1}\d G^{(n_2-1)}_{b}}(\mb{0}) +\\
&-  \sum_{\substack{n'_1\ge1, n'_2\ge2\\ n'_1+n'_2=n_2}} 2^{-N} \mc{N}_{n'_1,n'_2}  \int db'_1 db'_2 \l g^A_{b_2,b'_1} \d(b_2-b'_2) \left[\frac{\d^5\mc{V}^{(h)}}{\d\psi^-_{\eta'}\d\psi^+_{\eta}\d G^{(n_1)}_{b_1} \d G^{(n'_1)}_{b'_1} \d \dot{G}^{(n'_2)}_{b'_2}}(\mb{0}) + \right.\\
&\left.- \frac{\d^4\mc{V}^{(h)}}{\d\psi^-_{\eta'}\d\psi^+_{\eta}\d G^{(n_1)}_{b_1} \d G^{(n'_1)}_{b'_1}}(\mb{0}) \frac{\d\mc{V}^{(h)}}{\d \dot{G}^{(n'_2)}_{b'_2}}(\mb{0}) - \frac{\d^2\mc{V}^{(h)}}{\d G^{(n_1)}_{b_1} \d G^{(n'_1)}_{b'_1}}(\mb{0}) \frac{\d^3\mc{V}^{(h)}}{\d\psi^-_{\eta'}\d\psi^+_{\eta}\d \dot{G}^{(n'_2)}_{b'_2}}(\mb{0}) +\right.\\
&\left.- \frac{\d^3\mc{V}^{(h)}}{\d\psi^-_{\eta'}\d\psi^+_{\eta} \d G^{(n'_1)}_{b'_1}}(\mb{0}) \frac{\d^2\mc{V}^{(h)}}{G^{(n_1)}_{b_1} \d \dot{G}^{(n'_2)}_{b'_2}}(\mb{0}) - \frac{\d\mc{V}^{(h)}}{\d G^{(n'_1)}_{b'_1}}(\mb{0}) \frac{\d^4\mc{V}^{(h)}}{\d\psi^-_{\eta'}\d\psi^+_{\eta}\d G^{(n_1)}_{b_1}  \d \dot{G}^{(n'_2)}_{b'_2}}(\mb{0}) \right],
\end{split}
\]

which gives exactly Eq. \eqref{eq_18}, in force of. Eq. \eqref{eq_14}.
\end{proof}

\section{The non-interacting bubbles}
\label{app_bubble}

This appendix is dedicated to the characterization of the non-interacting bubbles $\Pi^{(h,N]}$ and $\Pi_5^{(h,N]}$:
\[\begin{split}
&\Pi^{(h,N]}(b_1,b_2)=-\frac{1}{2}\mc{E}^T_{(h,N]}(O_{b_1},O_{b_2}) \equiv  \frac{1}{2} \int d\eta_1 d\eta_1' d\eta_2 d\eta_2' c_{b_1}(\eta_1,\eta_1') c_{b_2}(\eta_2,\eta_2') g^{(h,N]}_{\eta_1',\eta_2}g^{(h,N]}_{\eta_2',\eta_1}\\
&\Pi^{(h,N]}_5(b_1,b_2)=-\frac{1}{2}\mc{E}^T_{(h,N]}(O_{b_1},O_{5;b_2}) \equiv  \frac{1}{2} \int d\eta_1 d\eta_1' d\eta_2 d\eta_2' c_{b_1}(\eta_1,\eta_1') c^5_{b_2}(\eta_2,\eta_2') g^{(h,N]}_{\eta_1',\eta_2}g^{(h,N]}_{\eta_2',\eta_1},
\end{split}\]

(cf. Eq. \eqref{eq_29a}) with $c$ and $c^5$ defined in Eq. \eqref{def_bare_vertex} and Eq. \eqref{def_bare_vertex} respectively. 

\subsection{The finiteness of the non-interacting bubbles}

In the first part of this appendix we establish the validity of the following result.

\begin{lemma}
\label{lemma_bubble} 
There exists $C_{\Pi}\ge1$ such that, for every $0\le h\le N$, letting

\begin{equation}
\label{eq_bubble6}
\hat{\Pi}^{(h,N]}_{\sharp; \mu,\nu}(p):= \frac{1}{2}\sum_{\e,\e'=\pm} \e\e' \int_{\L} \frac{dx dy}{L^2} e^{-ip\cdot(x-y)} \Pi^{(h,N]}_{\sharp}\big((x,\mu,\e),(y,\nu,\e')\big),
\end{equation}

(cf. Eq. \eqref{def:bubble:Fourier}) with $\sharp$ either blank or 5, the following items are true.

\begin{enumerate}
\item\label{it:bubble:1} $\hat{\Pi}^{(h,N]}_{\mu,\nu}(0)= \d_{\mu,\nu}\hat{\Pi}^{(h,N]}_{0,0}(0)$; $\hat{\Pi}^{(h,N]}_{5;\mu,\nu}(0)= \ve_{\mu,\nu}\hat{\Pi}^{(h,N]}_{5;0,1}(0)$.

\item\label{it:bubble:2} $|\hat{\Pi}^{(h,N]}_{\sharp;\mu,\nu}(0)|\le C_{\Pi}$.

\item\label{it:bubble:3} For every $v:\L_B\mapsto\mbb{C}$ such that $v\big((x,\mu,\e),(y,\nu,\e')\big)= \e\e' \tilde{v}_{\mu,\nu}(x-y)$ for some $\tilde{v}:\L\times\{0,1\} \mapsto \mbb{C}$, we have that
\[\big\|v*\Pi_{\sharp}^{(h,N]}- v \hat{\Pi}_{\sharp}^{(h,N]}(0)\big\|_1^w\le C_{\Pi} 2^{-h} \|\nabla v\|_1^w,\] 

where $\nabla_{\a} v\big((x,\mu,\e),(y,\nu,\e')\big):= 2^{N}\e\e' \big( \tilde{v}_{\mu,\nu}(x+ 2^{-N}\hat{e}_{\a})- \tilde{v}_{\mu,\nu}(x) \big)$ and $\|\nabla v\|_1^w :=  \sum_{\a=0,1} \|\nabla_{\a}v\|_1^w$.
\end{enumerate}

\end{lemma}

As a straightforward corollary of Lemma \ref{lemma_bubble} we have, as a crucial property of the theory, the finiteness of the quantity $\big\|g^A*\Pi_{\sharp}^{(h,N]}\big\|_1^w$, which is a main building block for the estimates in Sects. \ref{sect_improved} and \ref{sect_external}.

\begin{corollary}
\label{cor_bubble}
There exists a constant $C'_{\Pi}\ge1$ such that for every $0\le h\le N$, one has: $\|g^A*\Pi_{\sharp}^{(h,N]}\|_1^w \le C'_{\Pi} $.
\end{corollary}

\begin{proof}
From Lemma \ref{lemma_bubble} we find:
\[\begin{split} \big\|g^A*\Pi^{(h,N]}_{\sharp}\big\|_1^w &\le \big\| g^A \hat{\Pi}^{(h,N]}_{\sharp}(0) \big\|_1^w + \big\|g^A*\Pi^{(h,N]}_{\sharp}- g^A\hat{\Pi}^{(h,N]}_{\sharp}(0) \big\|_1^w  \\
&\le  \Big(2C_{\Pi} \|g^A\|_1^w +C_{\Pi}2^{-h} \|\nabla g^A\|_1^w\Big).
\end{split}\]

We already know that $\|g^A\|_1^w$ is bounded by a constant (cf. Eq. \eqref{eq_bound_boson} and comments thereafter); about $\nabla g^A$, one can easily show that

\begin{equation}
\label{boson:gradient}
|\nabla_{\a}g^A_{\mu,\nu}(x-y)|\le \frac{K_{\nabla} \d_{\mu,\nu}}{2^{-N}+ \dist{x-y}_L} e^{-\kappa_3 \sqrt{\dist{x-y}_L}}
\end{equation}

for suitable $K_{\nabla},\kappa_3>0$. Eq. \eqref{boson:gradient} readily implies $\|\nabla g^A\|_1^w\le K_{\nabla}'$, for a suitable $K_{\nabla}'>0$, thus the claim of Corollary \ref{cor_bubble}.
\end{proof}

\begin{proof}[Proof of Lemma \ref{lemma_bubble}]$\,$

\emph{Item \ref{it:bubble:1}.} This is a consequence of two symmetries of the theory, namely \it{axes flip} and \it{parity}, defined in Eqs. \eqref{symm:rot}, \eqref{sym:par} respectively. First, under the \it{axes flip} transformation, one has: $O_{x,\mu,\e}(\psi) \mapsto O_{x,\mu,\e}(\mc{F}\psi)= O_{Tx,1-\mu,\e}(\psi)$, where, if $b=(x,\mu,\e)$, $Tb\equiv (Tx,1-\mu,\e)$. Moreover $O_{5;x,\mu,\e}(\psi) \mapsto O_{5;x,\mu,\e}(\mc{F}\psi)= -O_{5;Tx,1-\mu,\e}(\psi)$. As a consequence:
\[\begin{split}
\hat{\Pi}^{(h,N]}_{\mu,\nu}(0)=&
-\frac{1}{4}\sum_{\e,\e'=\pm} \e\e' \int \frac{dx dy}{L^2} \mc{E}^T_{(h,N]}(O_{x,\mu,\e},O_{y,\nu,\e'})=\\
&-\frac{1}{4}\sum_{\e,\e'=\pm} \e\e' \int \frac{dx dy}{L^2} \mc{E}^T_{(h,N]}(O_{Tx,1-\mu,\e},O_{Ty,1-\nu,\e'})= \hat{\Pi}^{(h,N]}_{1-\mu,1-\nu}(0).
\end{split}\]

Similarly one finds that $\hat{\Pi}^{(h,N]}_{5;\mu,\nu}(0)=- \hat{\Pi}^{(h,N]}_{5;1-\mu,1-\nu}(0)$. On the other hand, under the \it{parity} transformation one has $O_b(\psi)\mapsto O_b(\mc{P}\psi)= O_{Pb}(\psi)$, where if $b=(x,\mu,\e)$, then $Pb\equiv(Px-\delta_{\mu,1}2^{-N}\hat{e}_1, \mu, (-1)^{\mu}\e)$. Moreover $O_{5;x,\mu,\e}(\psi) \mapsto O_{5;x,\mu,\e}(\mc{P}\psi)= (-1)^{\mu-1} O_{5;Px,\mu,\e}(\psi)$. It follows that
\[\begin{split}
\hat{\Pi}^{(h,N]}_{0,1}(0)=&
-\frac{1}{4}\sum_{\e,\e'=\pm} \e\e' \int \frac{dx dy}{L^2} \mc{E}^T_{(h,N]}(O_{x,0,\e},O_{y,1,\e'})=\\
&-\frac{1}{4}\sum_{\e,\e'=\pm} \e\e' \int \frac{dx dy}{L^2} \mc{E}^T_{(h,N]}(O_{Px,0,\e},O_{Py- 2^{-N}\hat{e}_1,1,-\e'})= \\
&\frac{1}{4}\sum_{\e,\e'=\pm} \e\e' \int \frac{dx dy}{L^2} \mc{E}^T_{(h,N]}(O_{Px,0,\e},O_{Py- 2^{-N}\hat{e}_1,1,\e'})= -\hat{\Pi}^{(h,N]}_{0,1}(0).
\end{split}\]

Similarly, $\hat{\Pi}^{(h,N]}_{5;0,0}(0)= - \hat{\Pi}^{(h,N]}_{5;0,0}(0)$. All in all: $\hat{\Pi}^{(h,N]}_{\mu,\nu}(0)= \d_{\mu,\nu}\hat{\Pi}^{(h,N]}_{0,0}(0)$ and $\hat{\Pi}^{(h,N]}_{5;\mu,\nu}(0)= \ve_{\mu,\nu}\hat{\Pi}^{(h,N]}_{5;0,1}(0)$.

\medskip

\emph{Item \ref{it:bubble:2}.} First of all note that by \it{charge conjugation} symmetry, Eq. \eqref{sym:par},\footnote{Note that under the \it{charge conjugation} transformation, $O_{5;x,\mu,\e}(\psi) \mapsto O_{5;x,\mu,\e}(\mc{Q}\psi)= O_{5;x,\mu,-\e}(\psi)$.} we have (compare also with Eq. \eqref{def:bubble:Fourier}):
\[\hat{\Pi}^{(h,N]}_{\sharp; \mu,\nu}(p)= \sum_{\e=\pm} \e \int_{\L} \frac{dx dy}{L^2} e^{-ip\cdot(x-y)} \Pi^{(h,N]}_{\sharp}((x,\mu,\e),(y,\nu,+)).\]

Let us begin with the \it{vector bubble}, $\Pi^{(h,N]}$. Recalling the explicit expression for $c_b(\eta_1,\eta_2)$, Eq. \eqref{def_bare_vertex}, we have that
\begin{equation}
\label{eq_bubble2}
\begin{split}
& \hat{\Pi}^{(h,N]}\big((\mu,+),(\nu,+)\big)(0)\equiv L^{-2} \int_{\L} dx dy\; \Pi^{(h,N]}((x,\mu,+),(y,\nu,+))=\\
& \frac{1}{8L^2} \int dzdy\, \mbox{Tr}\left\{(\gamma_{\mu}-r) g^{(h,N]}(x+2^{-N}\hat{e}_{\mu}-y) (\gamma_{\nu}-r) g^{(h,N]}(y+ 2^{-N} \hat{e}_{\nu} -x) \right\}=\\
& \frac{1}{8} \int dz\, \mbox{Tr}\left\{(\gamma_{\mu}-r) g^{(h,N]}(z)(\gamma_{\nu}-r) g^{(h,N]}(-z) \right\}+\\
&+ \frac{1}{8}\sum_{j,j'=h+1}^N \int dz\,  \mbox{Tr}\left\{(\gamma_{\mu}-r) g^{(j)}(z)(\gamma_{\nu}-r) \left[ g^{(j')}(-z)- g^{(j')}(-z +2^{-N}(\hat{e}_{\mu}+\hat{e}_{\nu}))\right] \right\}. 
\end{split}
\end{equation}

Note that the last line at the r.h.s. of Eq. \eqref{eq_bubble2} can be bounded by
\[ \frac{1}{2} 2^{-N}\sum_{h+1\le j\le j'\le N} \|g^{(j)}\|_{\infty} \sum_{\mu=0,1}\|\nabla_{\mu}g^{(j')}\|_1 \le  \sum_{h+1\le j\le j'\le N} \mc{O}(2^{-N+j}) =\mc{O}(1),  \]

where we used that $\|\nabla_{\mu}g^{(j')}\|_1 =\mc{O}(1)$, with $\mc{O}(1)$ denoting quantities which are bounded by universal constants. Therefore:
\begin{equation}
\label{eq_bubble3a}
\hat{\Pi}^{(h,N]}\big((\mu,+),(\nu,+)\big)(0)= 
\frac{1}{8} \int dz' \mbox{Tr}\left\{(\gamma_{\mu}-r) g^{(h,N]}(z')(\gamma_{\nu}-r) g^{(h,N]}(-z') \right\} + \mc{O}(1).
\end{equation}

Similarly, one finds:
\begin{equation}
\label{eq_bubble3b}
\hat{\Pi}^{(h,N]}\big((\mu,-),(\nu,+)\big)(0)= -\frac{1}{8}
 \int dz' \mbox{Tr}\left\{(\gamma_{\mu}+r) g^{(h,N]}(z')(\gamma_{\nu}-r) g^{(h,N]}(-z') \right\} + \mc{O}(1).
\end{equation}

Collecting Eqs. \eqref{eq_bubble3a} and \eqref{eq_bubble3b}, we find:
\begin{equation}
\label{eq_bubble4}
\begin{split}
&\sum_{\e=\pm} \e \hat{\Pi}^{(h,N]}\big((\mu,\e),(\nu,+)\big)(0) =\\
& \frac{1}{8} \sum_{\e=\pm} \int dz' \text{Tr}\left\{(\gamma_{\mu}-\e r) g^{(h,N]}(z')(\gamma_{\nu}- r) g^{(h,N]}(-z') \right\}+ \mc{O}(1)=\\
& \frac{1}{4} \int_{\L^*} dk \frac{\chi^2_{(h,N]}(k)}{\big(|s(k)|^2+ (m_N+M_N(k))^2\big)^2}\times\\
&\times \mbox{Tr}\Big\{ \gamma_{\mu} \left(i\slashed{s}(k) + m_N+M_N(k)\right) (\gamma_{\nu}- r) \left(i\slashed{s}(k) + m_N+M_N(k)\right) \Big\} + \mc{O}(1).
\end{split}
\end{equation}

The trace in the r.h.s. of Eq. \eqref{eq_bubble4} can be computed using that $\text{Tr}\{\gamma_{\mu}\gamma_{\nu}\}= 2\d_{\mu,\nu}$, $\text{Tr}\{\gamma_{\mu}\}= \text{Tr}\{\gamma_{\mu}\gamma_{\nu} \gamma_{\a}\}= 0$ and $\text{Tr}\{\gamma_{\mu}\gamma_{\a}\gamma_{\nu}\gamma_{\beta}\}= 2( \d_{\mu,\a}\d_{\nu,\b}- \d_{\mu,\nu} \d_{\a,\b} + \d_{\mu,\b}\d_{\nu,\a})$. We find that
\[\begin{split}
& \sum_{\e=\pm} \e \hat{\Pi}^{(h,N]}\big((\mu,\e),(\nu,+) \big)(0)= \\
&=-\frac{1}{2} \int_{\L^*} dk \frac{\chi^2_{(h,N]}(k)}{\big(|s(k)|^2+ (m_N+M_N(k))^2\big)^2} \Big( 2s_{\mu}(k)s_{\nu}(k) - \delta_{\nu,\mu} |s(k)|^2  \Big) +\\
&-ir \int_{\L^*} dk \; \chi^2_{(h,N]}(k)\frac{(m_N+M_N(k))s_{\mu}(k)}{\big(|s(k)|^2+ (m_N+M_N(k))^2\big)^2}+\\
&+\frac{1}{2} \delta_{\mu,\nu} \int_{\L^*} dk\; \chi^2_{(h,N]}(k)\frac{(m_N+M_N(k))^2}{\big(|s(k)|^2+ (m_N+M_N(k))^2\big)^2} + \mc{O}(1).
\end{split}\]

Now, the first two terms in the r.h.s. above are checked to be zero after \it{axes flip}: $(k_0,k_1)\mapsto (k_1,k_0)$ and \it{axes inversion}: $(k_0,k_1)\mapsto-(k_0,k_1)$.\footnote{Note that the vanishing of the first term would not hold in dimension $d\ne2$.} The third term is non-zero, but still finite:
\[\begin{split}
& \int_{\L^*} \!\! dk\; \chi^2_{(h,N]}(k)\frac{\big(m_N+M_N(k)\big)^2}{\big(|s(k)|^2+ (m_N+M_N(k))^2\big)^2} \le  4\pi^2 \int_{\L^*} \!\! dk\, \chi^2_{(0,N]}(k) \frac{m_N^2+ r^2 2^{-2N}|k|^4}{\big(|k|^2+ m_N^2\big)^2}= \mc{O}(1),
\end{split}\]
where we used Eq. \eqref{eq:bound:Wilson}; it follows that $\hat{\Pi}^{(h,N]}_{\mu,\nu}(0)=\mc{O}(1)$. 

Now let us discuss the \it{chiral bubble} $\Pi^{(h,N]}_5$. In analogy with the \it{vector bubble}, we have:
\[\begin{split}
& \hat{\Pi}^{(h,N]}_5 \big((\mu,\e),(\nu,+)\big)(0)= \frac{Z^5_N}{8} \int_{\L} dz\, \mbox{Tr}\Big\{(\e\gamma_{\mu}-r) g^{(h,N]}(z) \gamma_{\nu}\gamma_5 g^{(h,N]}(-z) \Big\} + \mc{O}(1).
\end{split}\]

Note that $\gamma_{\nu}\gamma_5= i\sum_{\s=0,1} \ve_{\nu,\s} \gamma_{\s}$, therefore:
\[\begin{split}
&\sum_{\e=\pm} \e \hat{\Pi}^{(h,N]}_5\big((\mu,\e),(\nu,+)\big)(0)= \frac{iZ^5_N}{4} \sum_{\s=0,1}\ve_{\nu,\s} \int_{\L^*} dk \frac{\chi^2_{(h,N]}(k)}{\big(|s(k)|^2 + (m_N+M_N(k))^2 \big)^2}\times\\
&\times \text{Tr}\Big\{\gamma_{\mu}\big(i\slashed{s}(k)+ m_N+ M_N(k)\big) \gamma_{\s} \big(i\slashed{s}(k)+ m_N+ M_N(k)\big) \Big\} +\mc{O}(1)
\end{split}\]

and the computation reduces exactly to the one for the simple bubble, so that:
\[\hat{\Pi}^{(h,N]}_{5;\mu,\nu}(0)\equiv \sum_{\e=\pm} \e \hat{\Pi}^{(h,N]}_5\big((\mu,\e),(\nu,+)\big)(0)= \mc{O}(1).\]

\medskip

\emph{Item \ref{it:bubble:3}.} Letting $P_{\sharp}^{(h)}(b,b'):= \big(v*\Pi_{\sharp}^{(h,N]}\big)(b,b') - \big(v\hat{\Pi}_{\sharp}^{(h,N]}(0)\big)(b,b')$, we have that
\begin{equation}
\label{eq_A8}
\begin{split}
& P^{(h)}_{\sharp}\big((x,\mu,\e),(y,\nu,\e')\big)=\\
&\sum_{\a\in\{0,1\}} \sum_{\e''\in\{\pm\}} \e\e'' \int_{\L} dz \Big(v_{\mu,\a}(x-z) - v_{\mu,\a}(x-y) \Big) \Pi^{(h,N]}_{\sharp}\big((z,\a,\e''),(y,\nu,\e')\big).
\end{split}
\end{equation}

We rewrite: $v_{\mu,\a}(x-z) - v_{\mu,\a}(x-y)= \int_{\mc{C}_{y\rightarrow z}} d\ell \cdot \nabla v_{\mu,\a}(x-\ell)$, with the following understanding.

\begin{itemize}
\item $\mc{C}_{y\rightarrow z}$ is a path on the discrete torus $\L$ connecting $y$ to $z$, with minimal length. More precisely, $\mc{C}_{y\rightarrow z}$ is a collection of ordered couples: 
\[\big( (x_1,x_2), (x_2,x_3),\dots,(x_{n-1},x_n), (x'_1,x'_2), (x'_2,x'_3),\dots,(x'_{m-1},x'_m)\big)\]

with the constraint that $x_1=y, x'_m=z$; $x_{i+1}-x_i= 2^{-N}\s_0 \hat{e}_0$ and $x'_{i+1}-x'_i= 2^{-N}\s_1\hat{e}_1$, for some $0\le m,n\le \frac{2^N L}{2}+1$ and $\s_0,\s_1\in\{\pm\}$. 
\item If $\mc{C}_{y\rightarrow z}=\{(x_1,x_2), (x_2,x_3),\dots,(x_n,x_{n+1})\}$, then
\[ \int_{\mc{C}_{y\rightarrow z}} d\ell \cdot \nabla v_{\mu,\a}(x-\ell)\equiv \sum_{i=1}^n 2^{-N}\left( \tfrac{v_{\mu,\a}(x-x_{i+1})- v_{\mu,\a}(x-x_i)}{2^{-N}}\right), \]

where note that the difference $\left( \frac{v_{\mu,\a}(x-x_{i+1})- v_{\mu,\a}(x- x_i)}{2^{-N}}\right)$ equals, up to a sign, the discrete gradient $\nabla_{\nu} v_{\mu,\a}(x-u)\equiv \left( \frac{v_{\mu,\a}(x-u+2^{-N}\hat{e}_{\nu})- v_{\mu,\a}(x-u)}{2^{-N}}\right)$, with either $x_{i+1}+2^{-N}\hat{e}_{\nu}=x_i=u$ or $x_i+2^{-N}\hat{e}_{\nu}=x_{i+1}=u$.
\end{itemize}

Using this interpolated expansion for the difference of $v$'s, we can estimate:
\[\begin{split}
\big\|P^{(h)}_{\sharp}\big\|_1^w &= \frac{1}{L^2} \sum_{\mu,\nu\in\{0,1\}} \sum_{\e,\e'\in\{\pm\}} \int_{\L^2} dxdy\, e^{\frac{\kappa}{2}\sqrt{\dist{x-y}_L}} \left| P^{(h)}_{\sharp}\big((x,\mu,\e),(y,\nu,\e')\big) \right|\le \\
&\le \frac{1}{L^2} \sum_{\mu,\nu,\a\in\{0,1\}} \sum_{\e,\e',\e''\in\{\pm\}} \int_{\L^3} dxdydz\, e^{\frac{\kappa}{2}\sqrt{\dist{x-y}_L}} \Big|\int_{\mc{C}_{y\rightarrow z}} d\ell\cdot \nabla v_{\mu,\a}(x-\ell)\Big| \times\\
&\times  \left|\Pi^{(h,N]}_{\sharp}\big((z,\a,\e''),(y,\nu,\e')\big)\right|
\end{split}\]

Note that for any $u\in\mc{C}_{y\rightarrow z}$, the quantity $\sqrt{\dist{x-y}_L}$ is smaller than $\sqrt{\dist{x-u}_L}+\sqrt{\dist{z-y}_L}$. Therefore, using that the length of the path $\mc{C}_{y\rightarrow z}$ is smaller than $\sqrt{2}\dist{z-y}_L$, we find that
\begin{equation}
\label{eq_bubble5}
\begin{split}
\|P^{(h)}_{\sharp}\|_1^w  &\le C \|\nabla v\|_1^w \sum_{\nu,\a\in\{0,1\}} \sum_{\e',\e''\in\{\pm\}} \int_{\L^2} \frac{dydz}{L^2}   \left|\Pi^{(h,N]}_{\sharp}\big((z,\a,\e''),(y,\nu,\e')\big)\right| e^{\frac{\kappa}{2}\sqrt{\dist{y-z}_L}} \dist{y-z}_L\\
&\equiv C \|\nabla v\|_1^w  \int_{\L_B^2} \frac{dbdb'}{L^2}   \left|\Pi^{(h,N]}_{\sharp}(b,b')\right| e^{\frac{\kappa}{2}\sqrt{\d^D(b,b')}} \d^D(b,b').
\end{split}\end{equation}

for a suitable constant $C>0$. Expanding the expression of $\Pi_{\sharp}^{(h,N]}$ and decomposing in scale each propagator in the same fashion as Eq. \eqref{eq:15} and lines below, 
we can estimate the r.h.s. of Eq. \eqref{eq_bubble5} by
\[\begin{split}
& C'\|\nabla v\|_1^w \sum_{h+1\le j'\le j\le N} 2^{-j} \|g^{(j')}\|_{\infty} \big\|g^{(j)}(\,\cdot\,) e^{\frac{\kappa}{2}\sqrt{2^j \dist{\,\cdot\,}_L}} \big\|_1 \le \frac{C_{\Pi}}{4}\|\nabla v\|_1^w \sum_{h+1\le j'\le j\le N}  2^{j'-2j},
\end{split}\]

for some $C_{\Pi}>0$, where we used the fact that, according to Eq. \eqref{eq_bound_propagator},

\[ \big\|g^{(j)}(\,\cdot\,) e^{\frac{\kappa}{2}\sqrt{2^j \dist{\,\cdot\,}_L}} \big\|_1 \equiv \sum_{s,s'\in\{1,2\}} \int_{\L}dx |g^{(j)}_{s,s'}(x)| e^{\frac{\kappa}{2}\sqrt{2^j\dist{x}_L}} \le \tilde{K} 2^{-j},\]

for a suitable constant $\tilde{K}>0$. Hence, as desired:

\[\|P^{(h)}_{\sharp}\|_1^w \le  \frac{C_{\Pi}}{4} \|\nabla v\|_1^w \sum_{h+1\le j'\le j\le N}   2^{j'-2j} \le C_{\Pi} 2^{-h}\|\nabla v\|_1^w.\]
\end{proof}

\subsection{Proof of Proposition \ref{prop_limitebollaGG}}
\label{subapp:proofprop_limitebollaGG}
For definiteness we discuss only the claims concerning $\hat \Pi^{(h^*_M,N]}_{\mu,\nu}(p,M,m_N,L)$ and $\hat \Pi^{(\leq N)}_{\mu,\nu}(p,m_N,L)$; the discussion for $\mc{T}^{(h^*_M,N]}$ and $\mc{T}^{(\le N)}$ follows similarly.
\medskip 
\par \emph{Limits: proof of Eq. \eqref{eq:bubbletadpolelimits}.}
By exploiting the \it{charge conjugation} symmetry, Eq. \eqref{symm:cc}, combining the definition in Eq. \eqref{eq_bubble6} with the notations in Eq.  \eqref{def_bare_vertex}, one explicitly obtains that, for every $p\in \frac{2\pi}{L}\mbb{Z}^2 \cap (-\frac{\pi}{a},\frac{\pi}{a}]^2$,
\begin{equation}  
\begin{split}
& \hat \Pi^{(\leq N)}_{\mu,\nu}(p,m_N,L)= \\ &\frac{1}{8} \int_{\Lambda^*} dk\, \sum_{\epsilon=\pm} e^{ia (\epsilon k_\mu-k_\nu-p_\nu+\frac{1-\epsilon}{2} p_\mu)} \text{Tr}\left \{(\gamma_\mu-\epsilon r)\hat g^{(\le N)}(k) (\gamma_\nu-r) \hat g^{(\le N)}(k+p)\right \}, \label{eq:C_1}
\end{split}
\end{equation}

 and $\Pi^{(h^*_M, N]}_{\mu,\nu}(p,M,m_N,L)$ having the same expression with an extra factor $(1-\chi_{h^*_M}(|k|)(1-\chi_{h^*_M}(|k+p|))$ under integral sign. Notice that the r.h.s. of Eq. \eqref{eq:C_1} makes sense for every $p\in(-\frac{\pi}{a},\frac{\pi}{a}]^2$, thus defining an extension of $\hat{\Pi}^{(\le N)}_{\mu,\nu}$ which is continuous w.r.t. $p$, as long as $m_N\ne0$. By rescaling $k \to ak=q \in a\L^*\equiv \frac{2\pi a}{L}\mbb{Z}^2\cap (-\pi,\pi]^2$, one finds that
\begin{equation} 
\label{eq:C_1b}
\begin{split}
& \hat \Pi^{(\leq N)}_{\mu,\nu}(p,m_N,L)= \\ &\frac{1}{8} \int_{a\Lambda^*} dq\, \sum_{\epsilon=\pm} e^{i (\epsilon q_\mu-q_\nu-ap_\nu+\frac{1-\epsilon}{2} ap_\mu)} \text{Tr}\left \{(\gamma_\mu-\epsilon r) \tilde{g}(q) (\gamma_\nu-r) \tilde{g}(k+ap)\right \},
\end{split}
\end{equation} 

\begin{sloppypar}
    where $\tilde{g}(q)= 2^N\hat{g}^{(\le N)}(2^Nq)\equiv\big(-i \sum_{\mu=0,1}\gamma_{\mu}\sin (q_{\mu})+ am_N+ 2r\sum_{\mu=0,1}\sin^2(q_{\mu}/2)\big)^{-1}$. 
\end{sloppypar}  The limit $\lim_{L\to\infty}\hat{\Pi}^{(\le N)}_{\mu,\nu}(p,m_N,L) =: \hat{\Pi}^{(\le N)}_{\mu,\nu}(p,m_N)$ can be characterized via the Poisson summation formula \cite[App. D]{GM2010}, which establishes the convergence to the same expression as the r.h.s. of Eq. \eqref{eq:C_1b} where the Riemann sum over $a\L^*$ is replaced by the integral over $(-\frac{\pi}{a},\frac{\pi}{a}]^2$. Besides the convergence is uniform w.r.t. $p\in(-\frac{\pi}{a},\frac{\pi}{a}]^2$ for any $a>0$ and $m_N\ne0$ fixed. Similar considerations hold for $\hat\Pi^{(h^*_M, N]}_{\mu,\nu}(p,M,m_N):=\lim_{L \to \infty}\hat\Pi^{(h^*_M, N]}_{\mu,\nu}(p,M,m_N,L)$.  The existence of $\hat\Pi^{(h^*_M, N]}_{\mu,\nu}(p,M):=\lim_{m_N \to 0}\hat\Pi^{(h^*_M, N]}_{\mu,\nu}(p,M,m_N)$ instead trivially follows from the fact that $M >0$ acts as an infrared cut-off. Finally, for $p \neq 0$, the Dominated Convergence Theorem\footnote{It sufficient to observe that if $p \neq 0$ then $\text{Tr}\big\{(\gamma_\mu-\epsilon r)\tilde g(q) (\gamma_\nu-r) \tilde g(q+ap)\big\}$ has two simple integrable singularities at $q=0,q=-ap$.} implies that $\lim_{M \to 0^+} \hat\Pi^{(h^*_M, N]}_{\mu,\nu}(p,M)$ and $\lim_{m_N \to 0} \hat\Pi^{(\leq N)}_{\mu,\nu}(p,m_N)$ exist and are both given by
 \begin{equation*}
 \hat \Pi_{\mu,\nu}(p):= \frac{1}{8} \int_{(-\pi,\pi]^2} \frac{d^2q}{(2\pi)^2}\, \sum_{\epsilon=\pm}e^{i(\epsilon q_\mu-q_\nu-ap_\nu+\frac{1-\epsilon}{2}ap_\mu)} \text{Tr}\left \{(\gamma_\mu-\epsilon r)\tilde g(q)(\gamma_\nu-r)\tilde g(q+ap)\right \}
 \end{equation*}

 with the understanding that $\tilde{g}(q)$ and $\tilde{g}(q+ap)$ are both computed at $m_N=0$.
\medskip
 \par \emph{Computation of $\hat \Pi_{\mu,\nu}(p)$: proof of Eq.  \eqref{eq:bubblecomput}.}
 We start by rewriting $\hat \Pi_{\mu,\nu}(p)=\hat \Pi^{\text{(sing)}}_{\mu,\nu}(p)+\hat \Pi^{\text{(reg)}}_{\mu,\nu}(p)$, where 
 \begin{equation}
 \begin{split}
      \hat \Pi^{\text{(sing)}}_{\mu,\nu}(p):=& \frac{1}{8} \int_{(-\pi,\pi]^2} \frac{d^2q}{(2\pi)^2}\, \sum_{\epsilon=\pm} \text{Tr} \left \{ (\gamma_\mu-\epsilon r) \frac{\chi(|q|)}{-i \slashed{q}}(\gamma_\nu-r) \frac{\chi(|q+ap|)}{-i (\slashed{q}+a\slashed{p})}\right \}= \\  -& \sum_{\alpha,\beta=0,1}\int_{(-\pi,\pi]^2} \frac{d^2q}{(2\pi)^2}\frac{\chi(|q|)\chi(|q+ap|)q_\a(q_\beta+ap_\beta)}{4 |q|^2 |q+p|^2} \text{Tr}\left \{\gamma_\mu \gamma_\alpha\gamma_\nu \gamma_\beta \right \}, \, \,  \forall p\ne 0,\label{eq:C_12} 
     \end{split}
 \end{equation}
with $\slashed{q}\equiv\sum_{\mu=0,1} \gamma_\mu q_\mu$ and we used that $\text{Tr}\{\gamma_\mu \gamma_\nu \gamma_\a\}=0.$ Besides it is possible to show that

 \begin{equation}
 \label{eq:48}
 \begin{split}
\hat{\Pi}^{\text{(reg)}}_{\mu,\nu}(p)&:=\frac{1}{8}\sum_{\e=\pm} \int_{(-\pi,\pi]^2} \frac{d^2q}{(2\pi)^2} \left(- \text{Tr} \Big\{ (\gamma_\mu-\epsilon r) \tfrac{\chi(|q|)}{-i \slashed{q}}(\gamma_\nu-r) \tfrac{\chi(|q+ap|)}{-i (\slashed{q}+a\slashed{p})}\Big\} \right.\\
&\left. + \, \,  e^{i(\epsilon q_\mu-q_\nu-ap_\nu+\frac{1-\epsilon}{2}ap_\mu)} \text{Tr}\Big\{(\gamma_\mu-\epsilon r)\tilde g(q)(\gamma_\nu-r)\tilde g(q+ap) \Big\} \right)
\end{split}\end{equation}

is a well defined function of $p\in(-\frac{\pi}{a},\frac{\pi}{a}]^2$, included the case $p=0$, where the integrand is checked to be singular only at $q=0$ with an integrable singularity, see Eqs. \eqref{eq:49}-\eqref{eq:37} and comments nearby. To conclude it is enough to show the two following properties: 
\begin{align}
&\label{eq:36a} \hat \Pi_{\mu,\nu}^{\text{(sing)}}(p)=\frac{1}{8\pi|p|^2}(2p_\mu p_\nu -|p|^2\d_{\mu,\nu})+\tilde{R}'_{\mu,\nu}(ap), \qquad  p\ne0, \\ 
&\label{eq:36b} \big|\hat \Pi_{\mu,\nu}^{\text{(reg)}}(p)-\hat \Pi_{\mu,\nu}^{\text{(reg)}}(0)\big|\leq C (a|p|)^{1/2}, 
 \end{align} 
 
 for some positive constant $C$ and a function $\tilde{R}'_{\mu,\nu}$ such that $|\tilde{R}'_{\mu,\nu}(q)|\leq C|q|^{1/2}$. Notice that this concludes the proof of Eq.  \eqref{eq:bubblecomput} with $\tilde R_{\mu,\nu}=\tilde R'_{\mu,\nu}+\hat \Pi^{(\text{reg})}_{\mu,\nu}$, with the understanding that $\tilde R'_{\mu,\nu}(0)\equiv 0$.
 \medskip
 \par \emph{Proof of Eq. \eqref{eq:36a}.} Given $\omega=\pm$, we introduce $\gamma_\omega:=\frac{1}{2}(-i \gamma_0-\omega \gamma_1)$ so that $\gamma_\mu=\sum_{\omega=\pm} \Omega_{\mu,\omega} \gamma_\omega$, with $\Omega_{\mu, \omega}=i\d_{\mu,0}-\omega\d_{\mu,1}$. Using that $\text{Tr}\{\gamma_{\mu}\gamma_\alpha\gamma_\nu\gamma_\beta\}=2(\d_{\mu,\alpha}+\delta_{\mu,\beta}\delta_{\nu, \alpha}-\delta_{\mu,\nu}\delta_{\alpha, \beta})$, one can check that $\text{Tr}\{\gamma_\omega \gamma_\alpha \gamma_{\omega'} \gamma_\beta\}=\delta_{\omega,\omega'}(i\omega\d_{\alpha,0}+\delta_{\alpha,1})(i\omega\delta_{\beta, 0}+\delta_{\beta, 1}),$ which means that we can rewrite Eq. \eqref{eq:C_12} as
 \begin{equation}
 \hat \Pi_{\mu,\nu}^{\text{(sing)}}(p)=\sum_{\omega}\frac{\Omega_{\mu\omega}\Omega_{\nu\omega}}{4}\int_{\mathbb{R}^2} \frac{d^2q}{(2\pi)^2} \frac{\chi(|q|)\chi(|q+ap|)}{D_\omega(q)D_\omega(q+ap)}, \qquad D_\omega(q):=-iq_0+\omega q_1. \label{eq:C_13}
 \end{equation}
 Now it is known \cite[App. B.1]{BFM07} that the integral in Eq. \eqref{eq:C_13} is given by  $\frac{1}{4\pi} \frac{D_{-\omega}(p)}{D_\omega(p)}+\tilde{R}'_\omega(ap)$, with $|\tilde{R}'_\omega(q)|\leq C|q|^{1/2},$ so that letting $4\tilde R'_{\mu,\nu}(q):=\sum_{\omega} \Omega_{\mu, \omega}\Omega_{\nu,\omega}\tilde{R}'_\omega(q)$,
 \begin{equation}
 \begin{split} \hat \Pi_{\mu,\nu}^{\text{(sing)}}(p)-\tilde R'_{\mu,\nu}(ap)&=\frac1{16\pi|p|^2}\sum_{\omega=\pm} \Omega_{\mu,\omega}\Omega_{\nu,\omega}(p_1^2-p_0^2-2i\omega p_0p_1)\\ &=\frac{1}{8\pi |p|^2}\Big((p_1^2-p_0^2)(-\d_{\mu0}\d_{\nu0}+\d_{\mu1}\d_{\nu1})-2p_0 p_1(\d_{\mu0}\d_{\nu1}+\d_{\mu1}\d_{\nu0}) \Big) \\ &= \frac{1}{8\pi |p|^2}(2p_\mu p_\nu-\d_{\mu,\nu}|p|^2). \nonumber
 \end{split}
 \end{equation}

\par \emph{Proof of Eq. \eqref{eq:36b}.} For $p\in(-\frac{\pi}{a},\frac{\pi}{a}]^2$ and $|p|\ge \frac{1}{2}a^{-1}$, by simple dimensional arguments one can easily check that $\hat{\Pi}_{\mu,\nu}(p)$ and $\hat{\Pi}^{\text{(sing)}}_{\mu,\nu}(p)$ are both bounded by constant, so that Eq. \eqref{eq:36b} trivially holds for some constant $C>0$. Hence we must analyze the case $|p|\le \frac{1}{2} a^{-1}$. Explicitly, we have that the expression in the r.h.s. can be further expanded as $\hat{\Pi}^{\text{(reg)}}_{\mu,\nu}(p)= \hat{\Pi}^{\text{(reg,1)}}_{\mu,\nu}(p)+ \hat{\Pi}^{\text{(reg,2)}}_{\mu,\nu}(p)$, where:
\begin{equation}
\label{eq:49}
\begin{split}
\hat{\Pi}^{\text{(reg,1)}}_{\mu,\nu}(p)&= \frac{1}{8}\sum_{\e=\pm} \int_{(-\pi,\pi]^2} \frac{d^2q}{(2\pi)^2} \Big[ \big(1-\chi(|q+ap|)\big) + \chi(|q+ap|)\big(1-\chi(|q|)\big) \Big] \times\\
&\times e^{i(\epsilon q_\mu-q_\nu-ap_\nu+\frac{1-\epsilon}{2}ap_\mu)} \text{Tr}\Big\{(\gamma_\mu-\epsilon r)\tilde g(q)(\gamma_\nu-r)\tilde g(q+ap) \Big\},
\end{split}
\end{equation}
\begin{equation}
\label{eq:37}
\begin{split}
\hat{\Pi}^{\text{(reg,2)}}_{\mu,\nu}(p)&= \frac{1}{8}\sum_{\e=\pm} \int_{(-\pi,\pi]^2} \frac{d^2q}{(2\pi)^2}  \Bigg( e^{i(\epsilon q_\mu-q_\nu-ap_\nu+\frac{1-\epsilon}{2}ap_\mu)} \text{Tr}\Big\{(\gamma_\mu-\epsilon r)\tilde g(q)(\gamma_\nu-r)\tilde g(q+ap) \Big\} \\
& - \text{Tr} \Big\{ (\gamma_\mu-\epsilon r) \tfrac{1}{-i \slashed{q}}(\gamma_\nu-r) \tfrac{1}{-i (\slashed{q}+a\slashed{p})}\Big\}  \Bigg) \chi(|q|)\chi(|q+ap|).
\end{split}
\end{equation}

Now, since for $|p|\le \frac{1}{2}a^{-1}$ we have that $\tilde{g}(\cdot+ap)$ is never singular on the support of $1-\chi(|\cdot+ ap|)$, it follows that $\hat{\Pi}^{\text{(reg,1)}}_{\mu,\nu}$ is actually $\mathscr{C}^{\infty}$ over the ball of radius $\frac{1}{2}a^{-1}$;\footnote{The smoothness is obvious for the contribution to the r.h.s. of Eq. \eqref{eq:49} from the term $1-\chi(|q+ap|)$ in square brackets. For the contribution from $\chi(|q+ap|)\big(1-\chi(|q|)\big)$, one can first extend the integral over $\mbb{R}^2$ and then perform the change of variable $q\mapsto q-ap$, from which the smoothness of the integral follows.} $\big|\hat{\Pi}^{\text{(reg,1)}}_{\mu,\nu}(p)- \hat{\Pi}^{\text{(reg,1)}}_{\mu,\nu}(0)\big|= \mc{O}(|ap|)$. In order to evaluate the continuity at zero of $\hat{\Pi}^{\text{(reg,2)}}_{\mu,\nu}$ we further decompose, for $p'\in\{0,p\}$, $\hat{\Pi}^{\text{(reg,2)}}_{\mu,\nu}(p')=\hat{\Pi}^{\text{(reg,2,1)}}_{\mu,\nu}(p';p)+ \hat{\Pi}^{\text{(reg,2,2)}}_{\mu,\nu}(p';p)$, where the former (resp. the latter) is defined as the r.h.s. of Eq. \eqref{eq:37} with the integral restricted to the ball (resp. the complement of the ball) of radius $\sqrt{a|p|}$. Since the integrand in the r.h.s. of Eq. \eqref{eq:37} can be bounded, up to a constant, by $(|q||q+ap|)^{-1}$, we find that
\[\begin{split}
&\big|\hat{\Pi}^{\text{(reg,2,1)}}_{\mu,\nu}(p;p)- \hat{\Pi}^{\text{(reg,2,1)}}_{\mu,\nu}(0;p) \big|\le \\
& \big|\hat{\Pi}^{\text{(reg,2,1)}}_{\mu,\nu}(p;p)\big|+ \big|\hat{\Pi}^{\text{(reg,2,1)}}_{\mu,\nu}(0;p) \big|\le \int_{|q|\le \sqrt{a|p|}} \frac{d^2q}{(2\pi)^2} \left[\mc{O}\Big(\tfrac{1}{|q+ap|}\Big)+ \mc{O}\Big(\tfrac{1}{|q|}\Big) \right]= \mc{O}\big(\sqrt{a|p|}\big).
\end{split}\]

About $\hat{\Pi}^{\text{(reg,2,2)}}_{\mu,\nu}(p';p)$, letting $F(q,ap)$ be the integrand in the r.h.s. of Eq. \eqref{eq:37}, after writing $F(q,ap)-F(q,0)= \int_0^{1}ds\, \de_s F(q, asp)$, by dimensional considerations one finds that $|F(q,ap)-F(q,0)|\le \mc{O}\Big(\frac{a|p|}{|q|^2} \Big)$, for any $|q|\ge \sqrt{a|p|}$ and $|p|\le \frac{1}{2}a^{-1}$. Therefore:
\[\big|\hat{\Pi}^{\text{(reg,2,2)}}_{\mu,\nu}(p;p)- \hat{\Pi}^{\text{(reg,2,2)}}_{\mu,\nu}(0;p) \big|\le \int_{\sqrt{a|p|}\le |q|\le \sqrt{2}\pi} \frac{d^2q}{(2\pi)^2} \mc{O}\Big(\tfrac{a|p|}{|q|^2} \Big)= \mc{O}\Big(a|p| \log\big(\tfrac{1}{a|p|}\big) \Big)\le \mc{O}\big(\sqrt{a|p|}\big),\]
which concludes the proof of Eq. \eqref{eq:36b}.
$\hfill\square$

\section{The tree expansion}
\label{app_tree}

In this appendix we prove Eq. \eqref{eq_tree4}. It is convenient to introduce the \emph{dimensionless norm} $\|\wt{W}^{(h)|q,q}_{\ud{n};\ud{\dot{n}};\ud{\ddot{n}}}\|_{(h)}:= 2^{-h D_{sc}} \|\wt{W}^{(h)|q,q}_{\ud{n};\ud{\dot{n}};\ud{\ddot{n}}}\|_1^w$. In terms of such norm one readily finds the following partial result. 

\begin{lemma}
\label{lemma_tree2}
There exists $C\ge1$ such that for every $\t\in\mc{T}_{h,N}$ and $\mb{P}\in\mc{P}_{\t}$,
\begin{equation}
\label{eq_tree3}
\|\wt{W}[\t,\mb{P}]\|_{(h_{v_0}-1)}\le  \prod_{v\in V(\t)\setminus\{v_0\}} 2^{(h_v-h_{v'})D_{sc}(P_v)} \prod_{v\in V_e(\t)} C^{|P_v|} \|W^{(h_v-1)}_{P_v}\|_{(h_v-1)},
\end{equation}

with the understanding that $v'$ is the vertex which precedes $v$ on $\t$.
\end{lemma}

We postpone the proof of Lemma \ref{lemma_tree2} to the end of this section. Using the bound of Eq.  \eqref{eq_tree3} inside Eq. \eqref{eq_tree2}, and recalling that only trees with $h_{v_0}=h+1$ contribute, we have:\footnote{Note that the operator $\mathscr{S}$ in the r.h.s. of Eq. \eqref{eq_tree2} plays no role in view of an upper bound.}
\begin{equation}
\label{eq_tree5}
\begin{split}
\|\wt{W}^{(h)|q,q}_{\ud{n};\ud{\dot{n}}; \ud{\ddot{n}}}\|_{(h)}&\le  \sum_{\t\in\mc{T}_{h,N}}^* \sum_{\mb{P}\in \mc{P}_{\t}}^{**} \big\|\wt{W}[\t,\mb{P}]\big\|_{(h_{v_0}-1)}\\
&\le 4^{|P_{v_0}|} \sum_{\t\in\mc{T}_{h,N}}^* \sum_{\mb{P}\in \mc{P}_{\t}}^{**} \prod_{v\in V(\t)} 2^{(h_v-h_{v'})D_{sc}(P_v)} \prod_{v\in V_e(\t)} C^{|P_v|} \|W^{(h_v-1)}_{P_v}\|_{(h_v-1)}, \end{split}
\end{equation}

where in order to simplify the expression, we have extended the product $\prod_{v}2^{(h_v-h_{v'})D_{sc}(P_v)}$ also to $v=v_0$, at the cost of a factor $4^{|P_{v_0}|}$, which can be reabsorbed in the constant $C$. Now by construction $D_{sc}(P_v)\le-1$ for every $v\in V(\t)$ (here is why it is important to stop the expansion of the tree as soon as a non-irrelevant kernel, i.e. $D_{sc}\ge0$, is encountered), except eventually for the vertex following the root and for the endpoints associated with non-irrelevant terms, for which, however, the quantity $h_v-h_{v'}$ is fixed to be $1$. Moreover, if $D_{sc}(P_v)\le-1$, $D_{sc}(P_v)\le-\frac{1}{8}|P_v^{\psi}|$ too, with $P_v^{\psi}$ the set of field labels associated with $\psi$ variables. Hence one can write:
\[\prod_{v\in V(\t)} 2^{(h_v-h_{v'})D_{sc}(P_v)}\le 8^{|V_e|+1} \prod_{v\in V(\t)} 2^{-\frac{1+\th}{2}(h_v-h_{v'})} \prod_{v\in V(\t)} 2^{-\frac{1-\th}{16}|P_v^{\psi}|}.\]

The pre-factor $8^{|V_e|+1}$ can be absorbed into the constant $C$ in the r.h.s. of Eq.  \eqref{eq_tree5}, while the factor $\prod_{v\in V(\t)} 2^{-\frac{1-\th}{16}|P_v^{\psi}|}$ can be used to bound the sum over $\mb{P}$ for any fixed configuration $\mb{P}_{V_e}\equiv \{P_v\}_{v\in V_e(\t)}$. All in all \cite[App. A.6]{GM01}:
\begin{equation}
\label{eq_tree6}
\|\wt{W}^{(h)|q,q}_{\ud{n};\ud{\dot{n}}; \ud{\ddot{n}}}\|_{(h)}\le  \sum_{\t\in\mc{T}_{h,N}}^* \sum_{\mb{P}_{V_e}}^{**} \prod_{v\in V(\t)} 2^{-\frac{1+\th}{2}(h_v-h_{v'})} \prod_{v\in V_e(\t)} C'^{|P_v|} \|W^{(h_v-1)}_{P_v}\|_{(h_v-1)}, \end{equation}

with a suitable constant $C'\ge1$. Now the quantity $\|W^{(h_v-1)}_{P_v}\|_{(h_v-1)}$ can be bounded as follows, for every $\e>0$:
\begin{equation}
\label{eq_tree9}
\|W^{(h_v-1)}_{P_v}\|_{(h_v-1)}\le \left\{\begin{array}{cc}
 R\l2^{-(1+\vth)(h_v-1)}\le R\l,    & P_v \sim \psi^2; \\
C_1^q \l^{q-1},     & P_v \sim \psi^{2q} \text{ with }q\ge2;\\
C_1\l^{\frac{n-1}{2}} \l^{\frac{1-\e}{2}(q-1)}, & P_v \sim G^{(n)}\psi^{2q} \text{ with } q\in\{1,n\};\\
C_1\l^{\frac{n-1}{2}} \l^{\frac{1-\e}{2}(q-1)}, & P_v \sim \dot{G}^{(n)}\psi^{2q} \text{ with } q\in\{1,n-1\};\\
C_1\l^{\frac{5}{8}n} \l^{\frac{3}{8}(1-\e)q}, & P_v \sim \ddot{G}^{(n)}\psi^{2n-2}.
\end{array}\right.
\end{equation}

with a suitable $C_1\ge1$, and for $\l$ small enough, say\footnote{The dependence $R^{-4}$ comes by requiring $\|W^{(h)|2,2}_{\eset;\eset;\eset}\|_1^w\le C\l$, with $C$ independent of $R$, cf. the second of Eq. \eqref{eq_IB_1}.} $\l\le (C'_1 R^4)^{-1}$. The convenience for losing a power $\e$ in $\l$ is to get in general a pre-factor $C_1$ instead of $C_1^{|P_v|}$, in the r.h.s. above. The bounds in Eq. \eqref{eq_tree9} can be checked to imply the following estimate for the r.h.s. of Eq. \eqref{eq_tree6}, after performing the sum $\sum_{\mb{P}_{V_e}}^{**}$:
\begin{equation}
\label{eq_tree7}
\begin{split}
\hspace{-8pt}\|\wt{W}^{(h)|q,q}_{\ud{n};\ud{\dot{n}}; \ud{\ddot{n}}}\|_{(h)}&\le C_2^{1+\dim\ud{n}+\dim\ud{\dot{n}}+\dim\ud{\ddot{n}}} \l^{d} \hspace{-3pt}\sum_{\substack{\t\in\tilde{\mc{T}}_{h,N}}} 2^{-\frac{1+\th}{2}(h_v-h_{v'})} \l^{\e'\big(|V_e|- \max\{1, \dim\ud{n}+\dim\ud{\dot{n}}+\dim\ud{\ddot{n}} \} \big)} \times \\
& \times \left\{ \begin{array}{cc}
R     & q=1, |\ud{n}|=|\ud{\dot{n}}|=|\ud{\ddot{n}}|=0 \\
1     & \text{otherwise},
\end{array}\right. 
\end{split}
\end{equation}

for some $C_2\ge1$ and $\e'$ small enough, where $\tilde{\mc{T}}_{h,N}$ is the set of Gallavotti-Nicolò trees with number of endpoints $\ge \max\big\{1,\dim\ud{n}+\dim\ud{\dot{n}}+\dim\ud{\ddot{n}} \big\}$, and such that all the vertices except the endpoints are branching nodes. Note that the parameter $R$ appears in the r.h.s. of Eq. \eqref{eq_tree7} only if we are considering the kernel $W^{(h)|1,1}_{\eset;\eset;\eset}$: for every other kernel, indeed, every factor $R$ can be checked to be always accompanied with some extra factor $\l^{\frac{1}{2}}$, and $R\l^{\frac{1}{2}}\le 1$ for $\l\le R^{-2}$. Now the sum in the r.h.s. of Eq.  \eqref{eq_tree7} can be bounded by
\begin{equation}
\label{eq_tree8}
\begin{split}
&  \sum_{\t\in\tilde{\mc{T}}_{h,N}} 2^{-\th(h_{\text{max}}(\t)-h)} 2^{-\frac{1-\th}{2}(h_v-h_{v'})} \l^{\e'\big(|V_e|- \max\{1, \dim\ud{n}+\dim\ud{\dot{n}}+\dim\ud{\ddot{n}} \} \big)},
\end{split}
\end{equation}

with $h_{\text{max}}(\t)$ the highest scale over all the endpoints of $\t$. Note that in the case the kernel has labels $\ud{\ddot{n}}\ne \eset$, i.e. there is at least one field $\ddot{G}$, we must have $h_{\text{max}}(\t)=N+1$, since by construction there are no endpoints with a field $\ddot{G}$ at lower scales. Finally, the sum in the r.h.s. of Eq. \eqref{eq_tree8} is readily checked to be bounded by $\big(8(2^{\frac{1-\th}{2}}-1)^{-1}\big)^{2\max\big\{1,\dim\ud{n}+\dim\ud{\dot{n}}+\dim\ud{\ddot{n}} \big\}}$ for $\l$ small enough, therefore: 
\[\|\wt{W}^{(h)|q,q}_{\ud{n};\ud{\dot{n}}; \ud{\ddot{n}}}\|_{(h)}\le C_3^{1+\dim\ud{n}+\dim\ud{\dot{n}}+\dim\ud{\ddot{n}}} \l^{d} \times \left\{\begin{array}{cc}
2^{\th(h-N)},     & |\ud{\ddot{n}}|\ge1  \\
R, & q=1, |\ud{n}|=|\ud{\dot{n}}|=|\ud{\ddot{n}}|=0\\
1,     & \text{otherwise},
\end{array}\right. \]

with $C_3= C_2\big(8(2^{\frac{1-\th}{2}}-1)^{-1}\big)^2$ and again $\l$ small enough (say $\l\le (C_4R^4)^{-1}$ for a suitable constant $C_4\ge1$). The claim of Proposition \ref{prop:SDB} follows with $C_0:= (1- 2^{\th-1})^{-1} \max\{C_3,C_4\}$.

\bigskip

In order to complete the proof of Eq. \eqref{eq_tree4}, we need to prove Lemma \ref{lemma_tree2}.

\begin{proof}[Proof of Lemma \ref{lemma_tree2}]
The starting point is Eq. \eqref{eq:12}:
\[\wt{W}[\t,\mb{P}](\ud{\eta};\ud{b};\ud{\dot{b}};\ud{\ddot{b}})= \frac{\s_{\mb{P}}}{s_{v_0}!} \int_{\L_F} d\ud{\eta}_{Q_{v_0}} \mc{E}^T_{(h_{v_0})} \big( \psi_{\tilde{P}_{w_1}},\dots,\psi_{\tilde{P}_{w_{s_{v_0}}}} \big)  \prod_{j=1}^{s_{v_0}} \wt{W}[\t,\mb{P}_{w_j}] \big((\ud{\eta};\ud{b};\ud{\dot{b}};\ud{\ddot{b}})_{P_{w_j}}\big),\]

from which we find: 
\begin{equation}
\label{eq:tree3}
\begin{split}
\big\|\wt{W}[\t,\mb{P}]\big\|_{(h_{v_0}-1)} &\le \frac{2^{-(h_{v_0}-1) D_{sc}(P_{v_0})}}{L^2 s_{v_0}!} \int_{\L_B} d\ud{b} _{v_0} \int_{\L_F} d\ud{\eta}_{v_0} e^{\frac{\kappa}{2}\sqrt{\d^D(\ud{b}_{v_0},\ud{\eta}_{v_0})}}\times\\
& \times\big|\mc{E}^T_{(h_{v_0})} \big( \psi_{\tilde{P}_{w_1}},\dots,\psi_{\tilde{P}_{w_{s_{v_0}}}} \big)\big|  \prod_{j=1}^{s_{v_0}} \big|\wt{W}[\t,\mb{P}_{w_j}] \big((\ud{\eta};\ud{b};\ud{\dot{b}};\ud{\ddot{b}})_{P_{w_j}}\big)\big|,
\end{split}\end{equation}

where $\ud{b}_{v_0}\equiv \ud{b}_{P_{v_0}}$ and $\ud{\eta}_{v_0}= \big(\ud{\eta}_{P_{v_0}},\ud{\eta}_{Q_{v_0}})$. As a standard corollary of the Battle-Brydges-Federbush-Kennedy formula \cite[Eq. (4.43)]{GM01} and the Gram-Hadamard inequality \cite[Sect. A.3.4]{GM01}, we have the following bound:
\begin{equation}
\label{eq:tree4}
\big|\mc{E}^T_{(h_{v_0})} \big( \psi_{\tilde{P}_{w_1}},\dots,\psi_{\tilde{P}_{w_{s_{v_0}}}} \big)\big|\le C_1^{|Q_{v_0}|} 2^{\frac{1}{2}h_{v_0}(\sum_{j=1}^{s_{v_0}}|P_{w_j}^{\psi}|-|P_{v_0}^{\psi}|-2(s_{v_0}-1))} \sum_{T} \prod_{(f,f')\in T} |g^{(h_{v_0})}_{\eta_f,\eta_{f'}}|,
\end{equation}

with a suitable $C_1\ge1$, where $\sum_T$ is the sum over all the spanning trees over the \it{clusters} $\tilde{P}_{w_1},\dots,\tilde{P_{w_s}}$, made of $s_{v_0}-1$ edges, each of them connecting a field label $f\in \tilde{P}_{w_i}$, associated with a Grassmann variable $\psi^-_{\eta_f}$, with a field label $f'\in \tilde{P}_{w_j}$ (with $j\ne i$), associated with $\psi^+_{\eta_{f'}}$.

Plugging the r.h.s. of Eq. \eqref{eq:tree4} into the r.h.s. of Eq. \eqref{eq:tree3}, and using that
\[\sqrt{\d^D(\ud{b}_{v_0},\ud{\eta}_{v_0})}\le \sum_{j=1}^{s_{v_0}} \sqrt{\d^D(\ud{b}_{P_{w_j}},\ud{\eta}_{P_{w_j}})}+ \sum_{(f,f')\in T} \sqrt{\d^D(\eta_f,\eta_{f'})}, \]

we find:
\begin{equation}
\begin{split}
&\big\|\wt{W}[\t,\mb{P}]\big\|_{(h_{v_0}-1)} \le\\
& 2^{-(h_{v_0}-1) D_{sc}(P_{v_0})} 2^{\frac{1}{2}h_{v_0}(\sum_{j=1}^{s_{v_0}}|P_{w_j}^{\psi}|-|P_{v_0}^{\psi}|-2(s_{v_0}-1))} \frac{C_1^{|Q_{v_0}|}}{L^2 s_{v_0}!} \sum_T \int_{\L_B} d\ud{b} _{v_0} \int_{\L_F} d\ud{\eta}_{v_0} \times\\
& \times \prod_{(f,f')\in T} \Big(e^{\frac{\kappa}{2}\sqrt{\d^D(\eta_f,\eta_{f'})}} |g^{(h_{v_0})}_{\eta_f,\eta_{f'}}|\Big) \; \prod_{j=1}^{s_{v_0}} \Bigg( e^{\frac{\kappa}{2} \sqrt{\d^D\big(\ud{b}_{P_{w_j}},\ud{\eta}_{P_{w_j}}\big)}} \big|\wt{W}[\t,\mb{P}_{w_j}] \big((\ud{\eta};\ud{b};\ud{\dot{b}};\ud{\ddot{b}})_{P_{w_j}}\big)\big| \Bigg).
\end{split}\end{equation}

By performing standard \it{tree-stripping estimates} on $T$ \cite[Sect. 6.2]{GM01}, we get:
\begin{equation}
\begin{split}
\big\|\wt{W}[\t,\mb{P}]\big\|_{(h_{v_0}-1)} &\le 2^{-(h_{v_0}-1) D_{sc}(P_{v_0})} 2^{\frac{1}{2}h_{v_0}(\sum_{j=1}^{s_{v_0}}|P_{w_j}^{\psi}|-|P_{v_0}^{\psi}|-2(s_{v_0}-1))} \frac{C_1^{|Q_{v_0}|}}{s_{v_0}!}\times \\
&\times \sum_T  \prod_{(f,f')\in T} \big\|g^{(h_{v_0})}\big\|_1^w \; \prod_{j=1}^{s_{v_0}}  \big\|\wt{W}[\t,\mb{P}_{w_j}]\big\|_1^w.
\end{split}\end{equation}

Now, using the fact that $\|g^{(h)}\|_1^w =\mc{O}(2^{-h})$ (cf. Eq. \eqref{eq_bound_propagator}), that the number of spanning trees $T$ is bounded by $s_{v_0}! e^{2|Q_{v_0}|}$ \cite[Lemma A.5]{GM01}, and recalling the definition of the norm $\|\cdot\|_{(h)}$ at the beginning of this section, we find:
\begin{equation}
\label{eq:tree5}
\begin{split}
\big\|\wt{W}[\t,\mb{P}]\big\|_{(h_{v_0}-1)} &\le 2^{(h_{v_0}-1) \big[ -D_{sc}(P_{v_0})+ \frac{1}{2}\sum_{j=1}^{s_{v_0}}|P_{w_j}^{\psi}|-\frac{1}{2}|P_{v_0}^{\psi}|-2(s_{v_0}-1)\big]} C_2^{|Q_{v_0}|}\times \\
& \times \prod_{j=1}^{s_{v_0}}  2^{(h_{w_j}-1)D_{sc}(P_{w_j})} \big\|\wt{W}[\t,\mb{P}_{w_j}]\big\|_{(h_{w_j}-1)},
\end{split}\end{equation}

with a suitable $C_2\ge1$. Now in the r.h.s. of Eq. \eqref{eq:tree5} we write $h_{w_j}-1= (h_{w_j}- h_{v_0})+ (h_{v_0}-1)$, and using the fact that $\sum_{j=1}^{s_{v_0}} \big( |P_{w_j}^G|+ |P_{w_j}^{\dot{G}}|+ 2|P_{w_j}^{\ddot{G}}| \big)= |P_{v_0}^G|+ |P_{v_0}^{\dot{G}}|+ 2|P_{v_0}^{\ddot{G}}|$ (no source fields can be integrated), we note that
\[
\begin{split}
& -D_{sc}(P_{v_0})+ \tfrac{1}{2}\sum_{j=1}^{s_{v_0}} |P_{w_j}^{\psi}| -\tfrac{1}{2}|P_{v_0}^{\psi}|-2s_{v_0} +2 +\sum_{j=1}^{s_{v_0}} D_{sc}(P_{w_j})=\\
&= -2+\tfrac{1}{2}|P_{v_0}^{\psi}| + |P_{v_0}^G|+ |P_{v_0}^{\dot{G}}| +2|P_{v_0}^{\ddot{G}}| +\tfrac{1}{2}\sum_{j=1}^{s_{v_0}} |P_{w_j}^{\psi}| -\tfrac{1}{2}|P_{v_0}^{\psi}|-2s_{v_0} +2 + \\
&+\sum_{j=1}^{s_{v_0}} \left( 2- \tfrac{1}{2}|P_{w_j}^{\psi}| - |P_{w_j}^G|- |P_{w_j}^{\dot{G}}| -2|P_{w_j}^{\ddot{G}}| \right)=0.
\end{split}
\]

Therefore Eq. \eqref{eq:tree5} becomes
\begin{equation}
\begin{split}
\big\|\wt{W}[\t,\mb{P}]\big\|_{(h_{v_0}-1)} \le  C_2^{|Q_{v_0}|} \prod_{j=1}^{s_{v_0}}  2^{(h_{w_j}-h_{v_0})D_{sc}(P_{w_j})} \big\|\wt{W}[\t,\mb{P}_{w_j}]\big\|_{(h_{w_j}-1)},
\end{split}\end{equation}

which, once iterated over the vertices of $\t$, yields the desired bound in Eq. \eqref{eq_tree3}.
\end{proof}

\section{Symmetries}
\label{app:symmetries}

In this section we present the symmetries of the massive lattice QED$_2$: for definiteness we state them only for model with external sources $G(J),B,\vphi$, obtained right after the integration of the boson field $A$. However, the same symmetries hold for the model which includes the auxiliary external fields $G,\dot{G},\ddot{G}$, Eqs. \eqref{def_potential}, \eqref{def_potential_2}, with the understanding that the transformation for the variables $G^{(n)}_b, \dot{G}^{(n)}_b, \ddot{G}^{(n)}_b$ is the same as the one for $G(J)$ prescribed in the following lemma.

\label{app_bubbleq}
\begin{lemma}
\label{lemma:symm}
The potential $\mc{V}(\psi;G(J)) + (\vphi,\psi)+ (B,O_5)$ and the Gaussian integrations $P^{(h)}$ (cf. Eq. \eqref{eq:2} and comments after Eq. \eqref{eq:v_hbackward}) are separately invariant under the following symmetries, with the understanding that the transformation for $\vphi^{\pm}$ is the same as for $\psi^{\pm}$.

\begin{enumerate}
\item\label{item:symm:1} Global $U(1)$: $(\mc{U}_{\a}\psi)^{\pm}_{x,s}:=  e^{\pm i\a} \psi^{\pm}_{x,s}$, for every $\a\in[0,2\pi)$.

\item\label{item:symm:2} Charge conjugation: 

\begin{equation}
\label{symm:cc}
\left\{\begin{array}{cc}
(\mc{Q}\psi)^+_{x,s}:=  \sum_{s'}\psi^-_{x,s'} (\gamma_1)_{s,s'} \\
(\mc{Q}\psi)^-_{x,s}:= \sum_{s'} (\gamma_1)_{s,s'} \psi^+_{x,s'}\end{array} \right. ,\; \qquad \left\{\begin{array}{cc}
(\mc{Q}G)_{b}(J): = G_{\ov{b}}(J) \\
(\mc{Q}B)_{b}: = B_{\ov{b}},
\end{array} \right. 
\end{equation}

where $\ov b\equiv \ov{(x,\mu,\e)}= (x,\mu,-\e)$.

\item\label{item:symm:3} Axes flip:
\begin{equation}
\label{symm:rot}
\left\{\begin{array}{cc}
(\mc{F}\psi)^+_{x,s}:=  \sum_{s'}\psi^+_{Tx,s'} F^*_{s,s'} \\
(\mc{F}\psi)^-_{x,s}:=  \sum_{s'} F_{s,s'} \psi^-_{Tx,s'}
\end{array} \right., \qquad \; \left\{\begin{array}{cc}
(\mc{F}G)_{b}(J): = G_{Tb}(J) \\
(\mc{F}B)_{b}: = -B_{Tb},
\end{array} \right.  \end{equation}
where $F\equiv \tfrac{1}{\sqrt{2}}(\gamma_0-\gamma_1)\gamma_5\equiv \tfrac{1}{\sqrt{2}}\left(\begin{array}{cc}
    i & 1 \\
    -1 &-i 
\end{array}\right)$, $Tx\equiv (x_1,x_0)$ and $T(x,\mu,\e)\equiv (Tx,1-\mu,\e)$.

\item\label{item:symm:4} Parity:
\begin{equation}
\label{sym:par}
\left\{\begin{array}{cc}
(\mc{P}\psi)^+_{x,s}:=  \sum_{s'}\psi^+_{Px,s'} (\gamma_5\gamma_1)_{s',s} \\
(\mc{P}\psi)^-_{x,s}:=  \sum_{s'} (\gamma_1\gamma_5)_{s,s'} \psi^-_{Px,s'}
\end{array} \right., \qquad  \; \left\{\begin{array}{cc}
(\mc{P}G)_{b}(J): = G_{Pb}(J) \\
(\mc{P}B)_{x,\mu,\e}: = (-1)^{\mu-1} B_{Px,\mu,\e},
\end{array} \right. \end{equation}

where $Px\equiv(x_0,-x_1)$ and $P(x,\mu,\e)\equiv(Px-\delta_{\mu,1}2^{-N}\hat{e}_1, \mu, (-1)^{\mu}\e)$.

\end{enumerate}
\end{lemma}

\begin{proof}$\,$

\begin{itemize}
\item[\ref{item:symm:1}.] Global $U(1)$ is straightforward since every field $\psi^+$ or $\vphi^+$ is always accompanied with a field $\psi^-$ or $\vphi^-$.

\item[\ref{item:symm:2}.] We note that $\gamma_1^T\gamma_{\mu}^T\gamma_1=\gamma_{\mu}$ and $\gamma_1^T\gamma_1=-1$, hence, recalling that $O_{x,\mu,+}(\psi)\equiv \tfrac{1}{2}\psi^+_x(\gamma_{\mu}-r)\psi^-_{x+2^{-N}\hat{e}_{\mu}}$ and
$O_{x,\mu,-}(\psi)= -\tfrac{1}{2}\psi^+_{x+2^{-N}\hat{e}_{\mu}}(\gamma_{\mu}+r)\psi^-_x$, we find:
\[\left\{\begin{array}{cc}
O_{x,\mu,+}(\mc{Q}\psi)=  -\tfrac{1}{2}\psi^+_{x+2^{-N}\hat{e}_{\mu}}\gamma_1^T(\gamma_{\mu}^T-r)\gamma_1\psi^-_{x}=- \tfrac{1}{2}\psi^+_{x+2^{-N}\hat{e}_{\mu}}(\gamma_{\mu}+r)\psi^-_{x}= O_{x,\mu,-}(\psi)\\
O_{x,\mu,-}(\mc{Q}\psi)= \tfrac{1}{2}\psi^+_{x}\gamma_1^T(\gamma_{\mu}^T+r)\gamma_1\psi^-_{x+2^{-N}\hat{e}_{\mu}}= \tfrac{1}{2}\psi^+_{x}(\gamma_{\mu}-r)\psi^-_{x+2^{-N}\hat{e}_{\mu}} = O_{x,\mu,+}(\psi).
\end{array}\right. \]

Besides, recalling Eqs. \eqref{eq:4a}, \eqref{eq:4b}, and noting that $g^A_{b,b'}=g^A_{\ov{b},\ov{b}'}$, we see that $w_{n,m}(\ud{b},\ud{b}')= w_{n,m}(\ov{\ud{b}},\ov{\ud{b}}')$, from which the invariance of $\mc{V}(\psi;G(J))$ follows. Moreover, since $\gamma_1^T(\gamma_{\mu}\gamma_5)^T\gamma_1= -\gamma_1^T\gamma_5^T \gamma_1 \gamma_1^T\gamma_{\mu}^T\gamma_1= -\gamma_5\gamma_{\mu}= \gamma_{\mu}\gamma_5$, we see that $O_{5;x,\mu,-\e}(\mc{Q}\psi)= O_{5;x,\mu,\e}(\psi)$, which implies the invariance of the term $(B,O_5)$. The invariance of the Grassmann integration $P^{(h)}$ means $\mc{E}_{(h)}(f(\cdot))= \mc{E}_{(h)}(f(\mc{Q}\cdot))$ for every Grassmann polynomial $f$. In force of the Grassmannian Wick rule \cite[Eq. (4.20)]{GM01}, it is sufficient to focus on the case in which $f$ is quadratic, so that the simple expectation reduces to the single-scale propagator.
\[\begin{split}
&\mc{E}_{(h)}\big((\mc{Q}\psi)^-_{x,s} (\mc{Q}\psi)^+_{y,s'} \big)= -\big(\gamma_1g^{(h)}(y-x)\gamma_1^T\big)_{s',s}= \\
&-\int_{\L^*} dk e^{-ik\cdot(y-x)} f_h(|k|)  \Big(\gamma_1 \big(-i\slashed{s}(k)+ m_N+ M_N(k)\big)^{-1}\gamma_1^T\Big)_{s',s} =\\
&-\int_{\L^*} dk e^{-ik\cdot(y-x)} f_h(|k|)   \Big(-i\slashed{s}(k)^T - m_N- M_N(k)\Big)_{s',s}^{-1}=\\
&\int_{\L^*} dk e^{-ik\cdot(x-y)} f_h(|k|)  \Big(-i\slashed{s}(k)+ m_N+ M_N(k)\Big)^{-1}_{s,s'}= \mc{E}_{(h)}(\psi^-_{x,s}\psi^+_{y,s'}),
\end{split}\]

where we introduced $f_h(|k|)\equiv \chi_h(|k|)- \chi_{h-1}(|k|)$ for $0\le h< N$ and $f_N(|k|)\equiv 1-\chi_{N-1}(|k|)$, and we used that $s_{\mu}(-k)=-s_{\mu}(k)$ and $M_N(-k)=M_N(k)$.

\item[\ref{item:symm:3}.] We observe that $F^{\dg}\gamma_{\mu} F= \gamma_{1-\mu}$, with $F^{\dg}$ the Hermitian conjugate of $F$, therefore:
\[\left\{\begin{array}{cc}
O_{x,\mu,+}(\mc{F}\psi)= \tfrac{1}{2}\psi^+_{Tx}(\gamma_{1-\mu}-r) \psi^-_{Tx + 2^{-N}\hat{e}_{1-\mu}} = O_{Tx,1-\mu,+}(\psi)  \\
O_{x,\mu,-}(\mc{F}\psi)=  -\tfrac{1}{2}\psi^+_{Tx+2^{-N}\hat{e}_{1-\mu}}(\gamma_{1-\mu}+r)\psi^-_{Tx}= O_{Tx,1-\mu,-}(\psi).
\end{array}\right. \]

Recalling also that $g^A_{b,b'}= g^A_{Tb,Tb'}$, which implies that $w_{n,m}(\ud{b},\ud{b}')= w_{n,m}(T\ud{b}, T\ud{b}')$, the invariance of $\mc{V}$ follows. For the chiral current we use the fact that $F^{\dg}\gamma_5F=-\gamma_5$, and so $O_{5;Tx,\mu,\e}(\mc{F}\psi)= - O_{5;x,1-\mu,\e}(\psi)$. For the Grassmann integration, we have
\[\begin{split}
&\mc{E}_{(h)}\big( (\mc{F}\psi)^-_{x,s} (\mc{F}\psi)^+_{y,s'} \big)= \Big(F g^{(h)}(T(x-y)) F^{\dg}\Big)_{s,s'}= \\
&\int_{\L^*} dk e^{-iTk\cdot(x-y)} f_h(|k|)  \Big(F \big(-i\slashed{s}(k)+ m_N+ M_N(k)\big)^{-1}F^{\dg}\Big)_{s,s'}=\\
&\int_{\L^*} dk e^{-ik\cdot(x-y)} f_h(|k|)  \Big(F \big(-i\slashed{s}(Tk)+ m_N+ M_N(Tk)\big)F^{\dg}\Big)^{-1}_{s,s'} =\\
&\int_{\L^*} dk e^{-ik\cdot(x-y)} f_h(|k|)  \Big(-i\slashed{s}(k)+ m_N+ M_N(k)\Big)^{-1}= \mc{E}_{(h)}(\psi^-_{x,s}\psi^+_{y,s'}),
\end{split}\]

where we used $s_{\mu}(k)= s_{1-\mu}(Tk)$ and $M_N(k)=M_N(Tk)$.

\item[\ref{item:symm:4}.] We observe that $\gamma_5\gamma_1\gamma_{\mu}\gamma_1\gamma_5= (-1)^{\mu}\gamma_{\mu}$, which implies, by inspection, $O_b(\mc{P}\psi)= O_{Pb}(\psi)$. On the other hand we have $g^A_{b,b'}= g^A_{Pb,Pb'}$, so $w_{n,m}(\ud{b},\ud{b}')= w_{n,m}(P\ud{b},P\ud{b}')$, which implies the invariance of $\mc{V}$ under \it{parity}. Moreover, we note that $(\gamma_5\gamma_1)\gamma_{\mu}\gamma_5 (\gamma_1\gamma_5)= (-1)^{\mu-1}\gamma_{\mu}\gamma_5$, hence $O_{5;Px,\mu,\e}(\mc{P}\psi)= (-1)^{\mu-1}O_{5;x,\mu,\e}(\psi)$. Finally:
\[\begin{split}
&\mc{E}_{(h)}\big((\mc{P}\psi)^-_{x,s} (\mc{P}\psi)^+_{y,s'} \big)= \Big(\gamma_1\gamma_5 g^{(h)}(P(x-y)) \gamma_5\gamma_1\Big)_{s,s'}= \\
&\int_{\L^*} dk e^{-iPk\cdot(x-y)} f_h(|k|)  \Big((\gamma_5\gamma_1)^{-1} \big(-i\slashed{s}(k)+ m_N+ M_N(k)\big)^{-1}(\gamma_1\gamma_5)^{-1}\Big)_{s,s'}=\\
&\int_{\L^*} dk e^{-iPk\cdot(x-y)} f_h(|k|)  \Big((\gamma_1\gamma_5) \big(-i\slashed{s}(Pk)+ m_N+ M_N(Pk)\big)(\gamma_5\gamma_1)\Big)^{-1}_{s,s'}=\\
&\int_{\L^*} dk e^{-ik\cdot(x-y)} f_h(|k|)  \Big(-i\slashed{s}(k)+ m_N+ M_N(k)\Big)^{-1}= \mc{E}_{(h)}(\psi^-_{x,s}\psi^+_{y,s'}),
\end{split}\]

where we used that $(\gamma_{\mu}\gamma_5)^{-1}= \gamma_5\gamma_{\mu}$, $s_{\mu}(Pk)= (-1)^{\mu}s_{\mu}(k)$ and $M_N(Pk)=M_N(k)$.
\end{itemize}
\end{proof}

\section{Overview of the continuum limit}
\label{app:limit}

We briefly present the scheme of the proof for showing the existence of the limit $N\to\infty$ of the kernels of the generating functional in Eq. \eqref{def_functional_UV}. 

\par One would like to show the existence of the $N\to\infty$ limit of the Fourier transforms of the kernels $W^{(u.v.;N)}_{q,q';\tilde{q},\tilde{q}';p;p'}$ of the generating functional, cf.  Eqs. \eqref{def_functional_UV}-\eqref{eq_main_potential}, given  (for instance if $\tilde{q}=\tilde{q}'=p=p'=0$) by
\begin{equation}
\label{eq:38}
\begin{split} \hspace{-7pt}&\hat{W}^{(u.v.;N)}_{q,q;0,0;0;0}\big((\ud{k},\ud{s}),(\ud{k}',\ud{s}')\big):= \frac{1}{L^2} \int_{\L_N^{2q}} d\ud{x} d\ud{x}' e^{i(\ud{k}\cdot\ud{x} +\ud{k}'\cdot\ud{x}')} W^{(u.v.;N)}_{q,q;0,0;0;0}\big(\ud{\eta},\ud{\eta}';\ud{b}), \;\; \ud{k},\ud{k}' \in (\L^*_N)^q
\end{split}
\end{equation}

with $\L_N\equiv \L$ and $\L_N^*\equiv \L^*$ for which we have made the dependence on $N$ explicit in every expression, recall the notations below Eq.  \eqref{eq_V&G}.

It is then enough to show a \emph{Cauchy property}, i.e. (assuming with no loss of generality $h^*_M=0$) that for every $\ud{k},\ud{k}'\in \big(\frac{2\pi}{L}\mbb{Z}^2\big)^q$,
\begin{equation}
\label{eq:39}
\begin{split}
& \big|\hat{W}^{(u.v.;N)}_{q,q;0,0;0;0}\big((\ud{k},\ud{s}),(\ud{k}',\ud{s}')\big)- \hat{W}^{(u.v.;N')}_{q,q;0,0;0;0}\big((\ud{k},\ud{s}),(\ud{k}',\ud{s}')\big)\big|\le  C^{1+q} 
\l^{\max\{1,\frac{2}{5}q\}} 2^{-\th \min\{N,N'\}},
\end{split}\end{equation}

for any $N,N'$ large enough, with a suitable $C>0$ and $\th\in(0,1)$. Notice that the r.h.s of Eq. \eqref{eq:39} is the same as the r.h.s. of Eq. \eqref{eq_50} times a small factor $2^{-\th \min\{N,N'\}}$.

Since any kernel in Fourier space is bounded in absolute value by the $L^1$ norm of its position-space counterpart (cf. Eq. \eqref{def_Lp_norm}), and since for a non-perturbative multiscale analysis we are induced to work in position space, a natural way to derive Eq.  \eqref{eq:39}  
is by proving that the kernels $W^{(u.v.;N)}$ are  \emph{Cauchy sequences} in $L^1$ norm. 
\paragraph{Comparing different lattices: smearing.} 
The first issue is that, in presence of a lattice regularization,
the kernels $W^{(u.v.;N)}$ and $W^{(u.v.;N')}$ are apriori defined on  different lattices whenever, without loss of generality, $N<N'$. In order to give a meaning to $W^{(u.v.;N)}-W^{(u.v.;N')}$, we introduce a mollifier $\mf{d}_N$ given by 

\[\mf{d}_N(x):= \frac{1}{L^2} \int_{\L^*_{\infty}}dk\, e^{-ik\cdot x} \chi(\tfrac{\gamma^{-N}}{2}|k|), \qquad x\in\L_{\infty}:=\mbb{R}^2/(L\mbb{Z}^2),\]

with $\L^*_{\infty}\equiv \frac{2\pi}{L}\mbb{Z}^2$, $\chi$ defined after Eq. \eqref{eq:7} and $1<\gamma\le 2$ (e.g. $\gamma=\sqrt{2}$), and we define a smeared version of the kernels, denoted by  $\mathsf{W}^{(u.v;N)}$,   obtained via convolution with $\mf{d}_N$, namely:  
\[
\left\{ \begin{array}{cc}
\mathsf{W}^{(u.v;N)}(x_1,\dots,x_q)
:= \int_{\L_N^q} dy_1\dots dy_q W^{(u.v.;N)}(y_1,\dots,y_q)\prod_{j=1}^q \mf{d}_N(x_j-y_j)  \\
\mathsf{W}^{(u.v;N')}(x_1,\dots,x_q)
:= \int_{\L_{N'}^q} dy_1\dots dy_q W^{(u.v.;N')}(y_1,\dots,y_q)\prod_{j=1}^q \mf{d}_N(x_j-y_j)  .
\end{array}\right.
\]

Note that
 Eq. \eqref{eq:39} follows for $\max_{j=1}^q \{|k_j|,|k'_j|\}\le \gamma^N,$\footnote{Recall indeed that the cut-off $\chi$ equals 1 on $[0,\frac{1}{2}]$, hence the convolution by $\mf{d}_N$ does not affect the Fourier transform for momenta smaller than $\gamma^N$.} if we show that

\begin{equation}
\label{eq:40}
\big\|\mathsf{W}^{(u.v.);N}_{q,q;0,0;0;0} - \mathsf{W}^{(u.v.);N'}_{q,q;0,0;0;0} \big\|_{L^1(\L_{\infty})} \le C^{1+q} 
\l^{\max\{1,\frac{2}{5}q\}} 2^{-\th N},
\end{equation}

where $\|\, \cdot \, \|_{L^1(\Lambda_\infty)}$  stands for Eq. \eqref{def_Lp_norm} on $\L_{\infty}$. 

\medskip 
\par We stress that, since the smearing procedure cuts off all the modes $|k|\ge 2\gamma^N$ ($\chi=0$ on $[1,\infty)$), the nested use of the smearing by $\mf{d}_N$ within the multiscale scheme would implicitly exclude all the contributions to correlation functions coming from modes at least greater than $2\gamma^N$, and this would be equivalent to saying that the lattice theory is well approximated by its continuum counterpart with a momentum cut-off. The latter property turns out to be actually incorrect, as one readily realizes that the there are three irreducible diagrams, $\hat{\Pi}_{\mu,\nu}^{(0,N]}(p)$, $\hat{\Pi}_{5;\mu,\nu}^{(0,N]}(p)$ and $\mc{T}^{(0,N]}$, defined in Eqs.  \eqref{eq_bubble6} and \eqref{def:tadpole}, which admit non-negligible contributions from momenta of size $\sim 2^N$. In fact it is possible to show that
\begin{equation}
\label{eq:25}
\lim_{N\to\infty}\mc{T}^{(0,N]}= -\int_{(-\pi,\pi]^2} \frac{d^2q}{(2\pi)^2} e^{-i q_0} \frac{i \tilde{s}_0(q) -r \tilde{M}(q)}{|\tilde{s}(q)|^2+ \tilde{M}(q)^2}, 
\end{equation}

where $\tilde{M}(q)\equiv 2r\sum_{\mu} \sin^2(\frac{q_{\mu}}{2})$ and $\tilde{s}_{\mu}(q)\equiv \sin(q_{\mu})$. In the integral in the r.h.s. of Eq. \eqref{eq:25}, $q$ has the interpretation of \virg{rescaled momentum}, namely $q=2^{-N}k\in 2^{-N}\L^*_N$: this shows that in the limit $N\to\infty$ all the contribution comes from momenta $k\in\L^*$ such that $|k|\ge 2^{N(1-\e)}$, for any $\e>0$. Similar considerations apply to the \it{bubble graphs} $\hat{\Pi}^{(0,N]}_{\mu,\nu}(0)$ and $\hat{\Pi}^{(0,N]}_{5;\mu,\nu}(0)$.

\medskip

\paragraph{The Cauchy property.}

In virtue of the previous discussion one realizes that in order to prove Eq. \eqref{eq:40}, a refinement of Theorem \ref{thm:IB} must be taken into account. In particular (see last paragraph for more details), it turns out that each non-irrelevant kernel can be decomposed as
\[W^{(h;N)|q,q}_{\ud{n};\ud{\dot{n}};\ud{\ddot{n}}}=w^{(h;N)|q,q}_{\ud{n};\ud{\dot{n}};\ud{\ddot{n}}} + \mc{R}W^{(h;N)|q,q}_{\ud{n};\ud{\dot{n}};\ud{\ddot{n}}},\]

where $w^{(h;N)|q,q}_{\ud{n};\ud{\dot{n}};\ud{\ddot{n}}}$ is explicit, expressed in terms of the three diagrams above Eq. \eqref{eq:25}, while 
 $\mc{R}W^{(h;N)|q,q}_{\ud{n};\ud{\dot{n}};\ud{\ddot{n}}}$  admits a bound with an extra improvement which is at least $2^{-h}$ smaller than the naive, dimensional one, $2^{h D_{sc}}$ (cf. Eq. \eqref{def_Dsc}).\footnote{For instance, in the first nontrivial case $
 \big\|\mc{R}W^{(h;N)|2,2}_{\eset;\eset;\eset}\big\|_1^w \leq R\lambda^{\frac54}2^{-h}$, which must be compared with Eq. \eqref{eq_IB_1}).}

With this decomposition, fix e.g. $\gamma=\sqrt{2}$ and assume for simplicity $N,N'$ to be multiples of $4$, with $N'> N$; we can divide the range of scales into two different regimes as follows.

\begin{enumerate}
\item $N/4\le h\le N$. In this regime,  all the information about the theory is carried by the dominant part $w^{(h;N)|q,q}_{\ud{n};\ud{\dot{n}};\ud{\ddot{n}}}$ of the marginal kernels, since the  terms coming from $\mc{R}W^{(h;N)|q,q}_{\ud{n};\ud{\dot{n}};\ud{\ddot{n}}}$  contribute with an extra factor $2^{-{h\vth}}\leq 2^{-\frac{N\vth}{4}}$. Then in order to obtain Eq.  \eqref{eq:40}, one uses the splitting of $\mathsf{W}^{(h;N)|q,q}_{\ud{n};\ud{\dot{n}};\ud{\ddot{n}}}$ and the fact that  \begin{equation}
 \big\|\mathsf{w}^{(h;N)|q,q}_{\ud{n};\ud{\dot{n}};\ud{\ddot{n}}}- \mathsf{w}^{(h;N')|q,q}_{\ud{n};\ud{\dot{n}};\ud{\ddot{n}}} \big\|_{L^1(\L_{\infty})}^w=\mc{O}(2^{-N/4}),
 \label{eq:F5}
 \end{equation} 
   which follows from the  explicit convergence of the graphs $\hat{\Pi}^{(h,N]}(0)$ and $\mc{T}^{(h,N]}$ to their continuum limit, namely
\begin{align}
& \hat{\Pi}^{(h,N]}(0)- \lim_{N\to\infty}\hat{\Pi}^{(h,N]}(0)= \mc{O}(2^{-\frac{2}{3}N})+ \mc{O}(2^{h-N}), \nonumber\\
& \mc{T}^{(h,N]}- \lim_{N\to\infty}\mc{T}^{(h,N]}= \mc{O}(2^{-\frac{2}{3}N})+ \mc{O}(2^{h-N}). \nonumber
\end{align}

\item $0\le h\le N/4$. This regime is analyzed via the same inductive structure adopted in Sections \ref{sect_multiscale}, \ref{sect_improved} for the proof of Theorem \ref{thm:IB}, with the difference that now the bounds in Eqs. \eqref{eq_IB_1}-\eqref{eq_IB_6} and Eqs.  \eqref{eq_SDB_1}-\eqref{eq_SDB_3} involve, in their l.h.s., the differences $\mathsf{W}^{(h;N)|q,q}_{\ud{n};\ud{\dot{n}};\ud{\ddot{n}}}- \mathsf{W}^{(h;N')|q,q}_{\ud{n};\ud{\dot{n}};\ud{\ddot{n}}}$ rather than the kernels $W^{(h;N)|q,q}_{\ud{n};\ud{\dot{n}};\ud{\ddot{n}}}$ themselves, and in the r.h.s. an extra smallness factor, say $2^{-N/8}$, is also present. The strategy is again based on a two-fold procedure.

 As a first step one proves the counterpart of Proposition \ref{prop:SDB}, establishing the smallness of the difference $\mathsf{W}^{(h;N)|q,q}_{\ud{n};\ud{\dot{n}};\ud{\ddot{n}}}- \mathsf{W}^{(h;N')|q,q}_{\ud{n};\ud{\dot{n}};\ud{\ddot{n}}}$ for all the kernels, assuming it true for the relevant and marginal ones. This task is carried out by exploiting the tools of the tree expansion (cf. Section \ref{sect_tree_expansion}), combined with the dimensional bounds for the difference of fermionic propagators:
\begin{equation}
\label{eq:29}
\Big(\mf{d}_N*g^{(h)}_N*\mf{d}_N\Big)(x,y)- \Big(\mf{d}_N*g^{(h)}_{N'}*\mf{d}_N\Big)(x,y)= \mc{O}\Big(2^{2h-N} e^{-\kappa_4 \sqrt{2^h\dist{x-y}_L}} \Big)
\end{equation}

for a suitable $\kappa_4>0$. Note that the r.h.s. of Eq. \eqref{eq:29} has a gain factor $2^{h-N}$ w.r.t. the dimensional bound for $g^{(h)}_N$ itself (cf. Eq. \eqref{eq_bound_propagator}), and for, say, $h\le \frac{N}{2}$, such gain is better than $2^{-\frac{N}{2}}$.

Then, as a second step, we are left with proving the smallness of $\mathsf{W}^{(h;N)|q,q}_{\ud{n};\ud{\dot{n}};\ud{\ddot{n}}}- \mathsf{W}^{(h;N')|q,q}_{\ud{n};\ud{\dot{n}};\ud{\ddot{n}}}$ for the non-irrelevant kernels, and this can be done in the same spirit as Section \ref{sect_improved}, starting from the exact identities among the kernels (see e.g. Lemma \ref{lemma_id_self_energy}), and exploiting inductively the bounds from the tree expansion, together with the bound for the difference of boson propagators:
\[\begin{split} &\Big(\mf{d}_N*g^A_N*\mf{d}_N\Big)(x,y)- \Big(\mf{d}_N*g^A_{N'}*\mf{d}_N\Big)(x,y)=\\
& \mc{O}\Big( 2^{-N} e^{-\kappa_4 \sqrt{\dist{x-y}_L}}\Big)+  \mc{O}\Big(\tfrac{2N'-N}{2} e^{-\kappa_4 \sqrt{2^{-\frac{N}{2}}\dist{x-y}_L}}\Big),\end{split}\]
and of fermion propagators (Eq. \eqref{eq:29}).  
\end{enumerate}

\paragraph{Extracting the leading terms: refinement of Theorem \ref{thm:IB}.}
As explained above, in order for the previous strategy above to work, one needs to isolate the contribution from the diagrams $\hat{\Pi}^{(h,N]}_{\mu,\nu}$, $\hat{\Pi}^{(h,N]}_{5;\mu,\nu}$ and $\mc{T}^{(h,N]}$ in the multiscale analysis, and analyze their continuum limit separately. Such graphs typically appear in the kernels $W^{(h)}$ in the form of a geometrical series (corresponding to the sum of all the reducible diagrams), whose sum gives raise to
\begin{align}
&\label{eq:26a} \Upsilon^{(h,N]}(b,b'):= \d_{\mu,\nu}\e\e' \int_{\L_N^*} dk \frac{e^{-ik\cdot(x-y)}}{|\s(k)|^2+ M^2-2\l\big(2\hat{\Pi}^{(h,N]}(0)+\mc{T}^{(h,N]}\big)},
\end{align} 

with the understanding that $b=(x,\mu,\e)$, $b'=(y,\nu,\e')$ and $\hat{\Pi}^{(h,N]}(0)\equiv \hat{\Pi}^{(h,N]}_{0,0}(0)$ (recall that by Lemma \ref{lemma_bubble}, $\hat{\Pi}^{(h,N]}_{\mu,\nu}(0)= \d_{\mu,\nu} \hat{\Pi}^{(h,N]}_{0,0}(0)$).\footnote{As established by Eqs. \eqref{eq:41} and \eqref{eq:42}, if $h=0$ the quantity $-2\l \big(2\hat{\Pi}^{(h,N]}(0)+\mc{T}^{(h,N]}\big)$ appearing in the r.h.s. of Eq. \eqref{eq:imp1} is exactly the lowest order contribution to the correlation function $\hat{\Sigma}^{(u.v.)}_{\mu,\nu}(0)$.} Let us show how to establish the dimensional improvement for a restricted class of kernels, namely $W^{(h;N)|0,0}_{(n,n');\eset;\eset}$. It turns out that the dominant part of these kernels is given by
\begin{equation}
w^{(h;N)|0,0}_{(n,n');\eset;\eset}= (nn')^2 \left( \l \mc{T}^{(h,N]} g^A\right)^{*(n-1)} *\big(\mathds{1}+\l\Omega^{(h,N]}\big)*\Pi^{(h,N]}* \left( \l \mc{T}^{(h,N]} g^A\right)^{*(n'-1)},
\end{equation}

so that
\begin{equation}
\label{eq:imp1}
\big\|g^A*\mc{R}W^{(h;N)|0,0}_{(1,1);\eset;\eset} \big\|_1^w \le R\l 2^{-h}; \;\; \big\|\mc{R}W^{(h;N)|0,0}_{(n,n');\eset;\eset} \big\|_1^w \le \l^{\frac{n+n'-2}{2}} 2^{-h} \; \forall n+n'\ge3.
\end{equation}

The strategy for showing the dimensional improvement for $\mc{R}W^{(h;N)|0,0}_{(n,n');\eset;\eset}$ is qualitatively the same as that discussed in Sections \ref{sect_multiscale}, \ref{sect_improved}. However, while the part related to Proposition \ref{prop:SDB} is proved almost identically, the inductive part (the part of Section \ref{sect_improved}) is way more involved, as we are going to show.

\medskip

As a first step, using the same expansions introduced in Section \ref{sect_improved} (see Fig. \ref{fig_kernel_0200b} for the kernel $W^{(h;N)|0,0}_{(1,1);\eset;\eset}$), one realizes that most of the terms have the right dimensional factor $2^{-h}$ as in the r.h.s of Eq. \eqref{eq:imp1}. Regarding the kernel $W^{(h;N)|0,0}_{(1,1);\eset;\eset}$, the only contributions without such dimensional improvement are the non-interacting bubble $\Pi^{(h,N]}$ and graph $(c')$ of Fig. \ref{fig_kernel_0200b}. Then one proceeds by isolating all the terms of the expansion which do not possess the dimensional improvement $2^{-h}$; in the case of the kernels $W^{(h;N)|0,0}_{(n,n');\eset;\eset}$ we have:
\begin{align}
&\label{eq:27a} W^{(h;N)|0,0}_{(1,1);\eset;\eset}=  \Pi^{(h,N]} + 2 \sum_{n\ge1} \tfrac{1}{n^2}\l \Pi^{(h,N]}*g^A*W^{(h;N)|0,0}_{(n,1);\eset;\eset} + \mc{O}(\l 2^{-h}),\\
&\label{eq:27b}W^{(h;N)|0,0}_{(n,n');\eset;\eset}=  \tfrac{n^2}{(n-1)^2} \l \mc{T}^{(h,N]} g^A* W^{(h;N)|0,0}_{(n-1,n');\eset;\eset} +\mc{O}\big( \l^{\frac{n+n'-1}{2}} 2^{-h} \big), \;\; n\ge2,
\end{align}

where $\mc{O}(\cdot)$ is understood w.r.t. the weighted $L^1$ norm (cf. Eq. \eqref{def_Lp_norm}). Eqs. \eqref{eq:27a},  \eqref{eq:27b} can be regarded as a system of linear equations in the variables $\big\{W^{(h;N)|0,0}_{(n,n');\eset;\eset} \big\}_{n\ge n'\ge1}$. First, by an iterative use of Eq. \eqref{eq:27b}, we find that
\begin{equation}
\label{eq_dom_5}
\begin{split}
&W^{(h;N)|0,0}_{(n,n');\eset;\eset}=\\
&(nn')^2  (\l\mc{T}^{(h,N]} g^A)^{*(n-1)} *W^{(k)|0,0}_{(1,1);\eset;\eset}* (\l\mc{T}^{(h,N]} g^A)^{*(n'-1)} + \mc{O}\big(\l^{\frac{1}{2}} \l^{\frac{n+n'-2}{2}} 2^{-h}\big) .
\end{split}
\end{equation}

Now, combining Eq. \eqref{eq:27a} with Eq. \eqref{eq_dom_5}, and letting $\Gamma^{(h,N]}:= \sum_{n\ge0} (\l \mc{T}^{(h,N]} g^A)^{*n}$, we find that
\begin{equation}
\label{eq_dom_6}
W^{(h;N)|0,0}_{(1,1);\eset;\eset} = \Pi^{(h,N]} + 2 \l \Pi^{(h,N]}* g^A* \Gamma^{(h,N]} * W^{(h;N)|0,0}_{(1,1);\eset;\eset}  +\mc{O}(\l 2^{-h}). 
\end{equation}

Moreover, by convoluting with $g^A$ both sides of Eq. \eqref{eq_dom_6}, writing $g^A*\Pi^{(h,N]}$ as $\hat{\Pi}^{(h,N]}(0)g^A+ \big( g^A*\Pi^{(h,N]}- \hat{\Pi}^{(h,N]}(0)g^A \big)$ and applying Item \ref{it:bubble:3} of Lemma \ref{lemma_bubble}, we further obtain:
\begin{equation}
g^A*\big(W^{(h;N)|0,0}_{(1,1);\eset;\eset} - \Pi^{(h,N]} - 2 \l \hat{\Pi}^{(h,N]}(0) g^A* \Gamma^{(h,N]} * W^{(h;N)|0,0}_{(1,1);\eset;\eset} \big) = \mc{O}(\l 2^{-h}). 
\end{equation}

Finally, using that $\Omega^{(h,N]}=\big(\mathds{1}- 2\l \hat{\Pi}^{(h,N]}(0) g^A*\Gamma^{(h,N]}\big)^{-1}$, as well as $\Omega^{(h,N]}*g^A= g^A*\Omega^{(h,N]}$, and $\|\Omega^{(h,N]}\|_1^w=\mc{O}(1)$, we get the first of Eq. \eqref{eq:imp1}:
\begin{equation}
\label{eq_dom_7}
\Big\|g^A*\big(W^{(h;N)|0,0}_{(1,1);\eset;\eset} - \Omega^{(h,N]}*\Pi^{(h,N]} \big)\Big\|_1^w \equiv \big\|g^A*\mc{R}W^{(h;N)|0,0}_{(1,1);\eset;\eset}\big\|_1^w \le R\l 2^{-h}, 
\end{equation}

for some $R$ large enough. Besides, combining Eqs. \eqref{eq_dom_5}, \eqref{eq_dom_7}, one obtains the second of Eq. \eqref{eq:imp1} concerning the kernels $W^{(k)|0,0}_{(n,n');\eset;\eset}$ with $n+n'\ge3$:
\[\begin{split}
& \big\|\mc{R}W^{(h;N)|0,0}_{(n,n');\eset;\eset}\big\|_1^w\equiv\\
& \big\|W^{(h;N)|0,0}_{(n,n');\eset;\eset}- (nn')^2  (\l\mc{T}^{(h,N]} g^A)^{*(n-1)} *w^{(h;N)|0,0}_{(1,1);\eset;\eset}* (\l\mc{T}^{(h,N]} g^A)^{*(n'-1)}\big\|_1^w\le\\
& \big\|W^{(h;N)|0,0}_{(n,n');\eset;\eset}- (nn')^2  (\l\mc{T}^{(h,N]} g^A)^{*(n-1)} *W^{(h;N)|0,0}_{(1,1);\eset;\eset}* (\l\mc{T}^{(h,N]} g^A)^{*(n'-1)}\big\|_1^w + \\
&+(nn')^2 \big\|(\l\mc{T}^{(h,N]} g^A)^{*(n-1)} *\mc{R}W^{(h;N)|0,0}_{(1,1);\eset;\eset}* (\l\mc{T}^{(h,N]} g^A)^{*(n'-1)} \big\|_1^w\le  \l^{\frac{n+n'-2}{2}} 2^{-h},
\end{split}\]

for $\l$ small enough.

\bigskip

\paragraph{Acknowledgments.} We thank A. Giuliani and M. Porta for useful comments, advices and discussions. We gratefully acknowledge financial support from the MUR, PRIN 2022 project MaIQuFi cod. 20223J85K3. This work has been carried out under the auspices of the GNFM of INdAM.

\clearpage

\bibliographystyle{ieeetr}
\bibliography{bibliography}

\end{document}